\documentclass[sigconf]{acmart}

\usepackage{graphicx}
\usepackage{balance}
\usepackage[normalem]{ulem}
\usepackage{booktabs}
\usepackage{listings}
\usepackage{fancyvrb}
\usepackage{caption}
\usepackage[font=small]{subcaption}
\usepackage{braket}
\usepackage{bbold}
\usepackage{bm}
\usepackage{clrscode3e}
\usepackage{colortbl}
\usepackage{forest}
\usepackage{algorithm}
\usepackage[noend]{algpseudocode}
\usepackage{multirow}
\usepackage{rotating}
\usepackage{xspace}
\usepackage{url}
\usepackage{textcomp}
\usepackage{float}
\usepackage{hyperref}
\usepackage{color}
\usepackage{tcolorbox}
\usepackage{cleveref}
\usepackage{amsthm}
\usepackage{makecell}
\usepackage{etoolbox}
\usepackage{soul}
\usepackage{enumitem}
\usepackage{seqsplit}
\usepackage[framemethod=tikz]{mdframed}

\usetikzlibrary{arrows,automata,trees,shapes,positioning,shapes.multipart,calc}

\newcommand{\cut}[1]{{}}

\cut{
\DeclareRobustCommand{\BG}[1]{{\todo[inline,color=red!40,size=\footnotesize,fancyline]{\textbf{Boris says: }{#1}}}}

\DeclareRobustCommand{\SR}[1]{{\todo[inline,color=red!40,size=\footnotesize,fancyline]{\textbf{Sudeepa says: }{#1}}}}

\DeclareRobustCommand{\CL}[1]{{\todo[inline,color=red!40,size=\footnotesize,fancyline]{\textbf{Chenjie says: }{#1}}}}

\DeclareRobustCommand{\ZM}[1]{{\todo[inline,color=red!40,size=\footnotesize,fancyline]{\textbf{Zhengjie says: }{#1}}}}

}

\newcommand{\BG}[1]{{\color{red} \bf [BG: #1]}\\}
\newcommand{\SR}[1]{{\color{blue} \bf [SR: #1]}\\}
\newcommand{\CL}[1]{{\color{magenta} \bf [CL: #1]}\\}
\newcommand{\ZM}[1]{{\color{darkgreen} \bf [ZM: #1]}\\}

\definecolor{black}{rgb}{0,0,0}
\definecolor{grey}{rgb}{0.8,0.8,0.8}
\definecolor{red}{rgb}{1,0,0}
\definecolor{green}{rgb}{0,1,0}
\definecolor{brown}{rgb}{0.8,0.1,0.1}
\definecolor{darkgreen}{rgb}{0,0.5,0}
\definecolor{darkpurple}{rgb}{0.5,0,0.5}
\definecolor{darkdarkpurple}{rgb}{0.3,0,0.3}
\definecolor{blue}{rgb}{0,0,1}
\definecolor{shadegreen}{rgb}{0.95,1,0.95}
\definecolor{shadeblue}{rgb}{0.95,0.95,1}
\definecolor{shadered}{rgb}{1,0.85,0.85}
\definecolor{shadegrey}{rgb}{0.85,0.85,0.85}
\definecolor{oddRowGrey}{rgb}{0.80,0.80,0.80}
\definecolor{evenRowGrey}{rgb}{0.85,0.85,0.85}

\newcommand{\red}[1]{{\color{red} #1}}
\newcommand{\blue}[1]{{\color{blue} #1}}

\newcommand{\termSchemaGraph}{schema graph\xspace}
\newcommand{\termSchemaGraphs}{schema graphs\xspace}

\newcommand{\capSchemaGraph}{Schema Graph\xspace}

\newcommand{\termJoinGraph}{join graph\xspace}

\newcommand{\capJoinGraph}{Join Graph\xspace}
\newcommand{\abbrJGresult}{JGR\xspace}

\newcommand{\abbrPT}{PT\xspace}

\newcommand{\abbrAPT}{APT\xspace}
\newcommand{\abbrAPTs}{APTs\xspace{}}
\newcommand{\abbrJG}{JG\xspace}
\newcommand{\abbrF}{F-score\xspace}
\newcommand{\abbrFs}{F-scores\xspace}

\newcommand{\abbrAggJoinQ}{AGQ\xspace}

\newcommand{\val}{{\tt val}}

\newcommand{\rels}{{\tt \bf rels}}
\newcommand{\relsin}[1]{{\tt {\bf rels}_{#1}}(D)}
\newcommand{\relsinQ}{\relsin{Q}}
\newcommand{\attrs}{{\tt \bf attrs}}

\newcommand{\nbadata}{NBA\xspace}

\newcommand{\mypar}[1]{\smallskip\noindent\textbf{{#1}.}}

\DeclareMathAlphabet{\mathbbold}{U}{bbold}{m}{n}

\newtheorem{Definition}{Definition}

\newtheorem{example}{Example}

\newcommand{\mathtext}[1]{\ensuremath{\,\text{#1}\,}}
\newcommand{\card}[1]{|{#1}|}
\newcommand{\btrue}{\ensuremath{\mathbf{true}}}
\newcommand{\bfalse}{\ensuremath{\mathbf{false}}}

\newcommand{\ptime}{PTIME\xspace}

\newcommand{\playerSB}{{\tt S. Battier}}

\newcommand{\playerLJ}{{\tt L. James}}
\newcommand{\playerSC}{{\tt S. Curry}}
\newcommand{\playerKT}{{\tt K. Thompson}}
\newcommand{\playerDG}{{\tt D. Green}}

\newcommand{\abbrGSW}{GSW}

\newcommand{\abbrMed}{Medicare}
\newcommand{\abbrPriv}{Private}

\newcommand{\sexpl}{E}
\newcommand{\dexpl}{ED}
\newcommand{\explset}{\mathcal{E}}
\newcommand{\topk}{k} 

\newcommand{\FCov}{{\tt FCov}}
\newcommand{\ICov}{{\tt Cov}} 
\newcommand{\FTP}{{\tt TP_{F}}}
\newcommand{\ITP}{{\tt TP}} 
\newcommand{\FFP}{{\tt FP_{F}}}
\newcommand{\IFP}{{\tt FP}} 
\newcommand{\FFN}{{\tt FN_{F}}}
\newcommand{\IFN}{{\tt FN}} 
\newcommand{\FRec}{{\tt Rec_{F}}}
\newcommand{\FPrec}{{\tt prec_{F}}}
\newcommand{\Ffscore}{{\tt Fscore_{F}}}

\newcommand{\Rec}{{\tt Rec}}
\newcommand{\Prec}{{\tt Prec}}
\newcommand{\fscore}{{\tt Fscore}}

\newcommand{\uquestion}{UQ}

\newcommand{\pat}{\Phi}
\newcommand{\patrnset}{{\mathcal P}}
\newcommand{\pmatch}{\vDash} 
\newcommand{\matches}{\textsc{match}}
\newcommand{\placeh}{\ast}
\newcommand{\pleq}[1]{({#1},\mathbf{\leq})}
\newcommand{\pgeq}[1]{({#1},\mathbf{\geq})}
\newcommand{\peq}[1]{({#1},\mathbf{=})}
\newcommand{\patop}[2]{({#1},{#2})}
\newcommand{\patset}{\mathcal{P}}

\newcommand{\schemaOf}[1]{\ensuremath{\mathbf{#1}}}
\newcommand{\dom}{\textsc{Dom}}
\newcommand{\domOf}[2]{\dom_{#1.#2}}

\newcommand{\tup}{t}
\newcommand{\att}{A}
\newcommand{\rel}{R}
\newcommand{\rsch}[1]{{\bf #1}}

\newcommand{\db}{D}
\newcommand{\dbsch}{{\bf D}}
\newcommand{\arity}[1]{\mid{#1}\mid}
\newcommand{\ddom}{\mathbb{D}}
\newcommand{\query}{Q}

\newcommand{\projection}{\ensuremath{\Pi}}

\newcommand{\sedges}{E_{S}}
\newcommand{\snodes}{V_{S}}
\newcommand{\sglabel}{l_{Sedge}}
\newcommand{\asgraph}{G}
\newcommand{\sgraph}{(\snodes{}, \sedges{}, \sglabel{})}

\newcommand{\cond}{\textsc{Cond}}
\newcommand{\acond}{\theta}

\newcommand{\jedge}{e}
\newcommand{\sedge}{u}
\newcommand{\node}{n}
\newcommand{\eend}[1]{end(#1)}
\newcommand{\estart}[1]{start(#1)}

\newcommand{\pt}[1]{\mathcal{PT}(#1)}

\newcommand{\ptlabel}{PT}
\newcommand{\jgraphnodes}{V_{J}}
\newcommand{\jgraphedges}{E_{J}}

\newcommand{\tlabelv}{l_{Jnode}}
\newcommand{\tlabele}{l_{Jedge}}
\newcommand{\jjgraph}{(\jgraphnodes,\jgraphedges,\tlabelv,\tlabele)}
\newcommand{\jgraph}{\Omega}
\newcommand{\jgraphs}[1]{{\bm \Omega_{#1}}}

\newcommand{\jtr}{\mathcal{APT}}

\newcommand{\renameA}{\chi}

\newcommand{\oursystem}{\textsc{CaJaDE}}
\newcommand{\oursys}{{\textsc{CaJaDE}}\xspace}

\newcommand{\algMineAPT}{MineAPT\xspace}

\newcommand{\JGGen}{EnumerateJoinGraphs\xspace}

\newcommand{\maptSample}{\mathcal{S}}
\newcommand{\maptAPT}{APT}
\newcommand{\maptAnum}{A_{num}}
\newcommand{\maptAcat}{A_{cat}}
\newcommand{\maptCluster}{clusters}
\newcommand{\maptTODO}{todo}
\newcommand{\maptDONE}{done}
\newcommand{\maptPat}{\pat}
\newcommand{\maptPatnew}{\pat_{new}}
\newcommand{\maptKcat}{k_{cat}}
\newcommand{\maptAfilter}{A_{filtered}}
\newcommand{\maptArepresent}{A_{repr}}
\newcommand{\maptFilterAttrs}{filterAttrs}

\newcommand{\aParam}{\lambda}
\newcommand{\paramDefault}[1]{{#1}}
\newcommand{\paramDBsize}{\ensuremath{\aParam_{db-size}}\xspace}
\newcommand{\paramRecallThresh}{\ensuremath{\aParam_{recall}}\xspace}
\newcommand{\paramPatsamplerate}{\ensuremath{\aParam_{pat-samp}}\xspace}
\newcommand{\paramQualitysamplerate}{\ensuremath{\aParam_{quality-sample-rate}}\xspace}
\newcommand{\paramFscoresamplerate}{\ensuremath{\aParam_{F1-samp}}\xspace}
\newcommand{\paramNumRefineFrags}{\ensuremath{\aParam_{\#frag}}\xspace}
\newcommand{\paramAfilterrate}{\ensuremath{\aParam_{\#sel-attr}}\xspace}
\newcommand{\paramAattrnum}{\ensuremath{\aParam_{attrNum}}\xspace}
\newcommand{\paramJoinGraphSize}{\ensuremath{\aParam_{\#edges}}\xspace}
\newcommand{\paramMaxQcost}{\ensuremath{\aParam_{qCost}}\xspace}
\newcommand{\patLCA}{\patset_{cat}\xspace}

\newcommand{\pattopk}{\patset_{top-k}\xspace}

\newcommand{\bdFeature}{Feature Selection\xspace}
\newcommand{\bdLCA}{Gen. Pat. Cand.\xspace}
\newcommand{\bdF}{\abbrF Calc.\xspace}
\newcommand{\bdAPT}{Materialize \abbrAPTs\xspace}
\newcommand{\bdFsamp}{Sampling for F1\xspace}
\newcommand{\bdrefine}{Refine Patterns\xspace}
\newcommand{\bdJGenum}{\abbrJG Enum.\xspace}

\newcommand{\frst}{\ensuremath{1^{st}}\xspace}
\newcommand{\scnd}{\ensuremath{2^{nd}}\xspace}
\newcommand{\thrd}{\ensuremath{3^{rd}}\xspace}

\definecolor{darkgreen}{rgb}{0,0.5,0}
\definecolor{darkpurple}{rgb}{0.5,0,0.5}

\definecolor{revgreen}{rgb}{0,0.5,0}

\newrobustcmd{\reva}[1]{{#1}}
\newrobustcmd{\revb}[1]{{#1}}
\newrobustcmd{\revc}[1]{{#1}}
\newrobustcmd{\revm}[1]{{#1}}
\newrobustcmd{\revcj}[1]{{#1}}
\newrobustcmd{\todel}[1]{}
\newrobustcmd{\revs}[1]{{#1}}
\definecolor{purple}{rgb}{0.65,0.05,0.3}
\definecolor{red}{rgb}{1,0,0}
\definecolor{RED}{rgb}{1,0,0}

\definecolor{GrayRew}{gray}{0.85}

\newrobustcmd{\ans}[1]{\textbf{Ans:} {#1}}

\setcopyright{acmcopyright}
\copyrightyear{2018}
\acmYear{2021}
\acmDOI{10.1145/1122445.1122456}

\acmConference[SIGMOD '21]{SIGMOD 21: }{June 03--05, 2021}{Woodstock, NY}
\acmBooktitle{SIGMOD '21: ACM Symposium on Neural Gaze Detection,
  June 03--05, 2018, Woodstock, NY}
\acmPrice{15.00}
\acmISBN{978-1-4503-XXXX-X/18/06}

\newbool{Techreport}
\newrobustcmd{\ifnottechreport}[1]{\ifbool{Techreport}{}{#1}}
\newrobustcmd{\iftechreport}[1]{\ifbool{Techreport}{#1}{}}

\booltrue{Techreport}

\begin{document}

\title{Using Context for Creating Rich Explanations for Query Answers}
\title{Augmenting Provenance with Contextual Information for Rich Explanations}
\title{Putting Things into Context: Rich Explanations for Query Answers using Join Graphs}


\author{Chenjie Li}
\affiliation{\institution{IIT}}
\email{cli112@hawk.iit.edu}

\author{Zhengjie Miao}
\affiliation{\institution{Duke University}}
\email{zjmiao@cs.duke.edu}

\author{Qitian Zeng}
\affiliation{\institution{IIT}}
\email{qzeng3@hawk.iit.edu}

\author{Boris Glavic}
\affiliation{\institution{Illinois Institute of Technology}}
\email{bglavic@iit.edu}

\author{Sudeepa Roy}
\affiliation{\institution{Duke University}}
\email{sudeepa@cs.duke.edu}


\begin{abstract}
  In many data analysis applications, there is a need to 
  explain why a surprising or interesting result
  was produced by a query.
  Previous approaches to
  explaining results
  have directly or indirectly used data provenance (input tuples contributing to the result(s) of interest), which
  is limited by the fact that relevant
  information for explaining an answer may not be fully contained in the provenance.
  We propose a new approach for explaining 
  query results by augmenting
  provenance with information from other related tables in the database. 
\revm{
We develop a suite of optimization techniques, and demonstrate experimentally using real datasets and through a user study that our approach produces meaningful results by efficiently navigating the large search space of possible explanations.}

  \cut{
  and produces new informative explanations compared to
  previous approaches.
  }

\end{abstract}

\cut{
  In many application domains there is a need to explain why and how a result
  was produced by a query by combining relevant input data.  Past work for
  explaining 
  results is typically based on the concept
  of interventions or more generally provenance. 
  While
  this line of work has been successful, it is limited by the fact that relevant
  information for explaining an answer may not be contained in
  the provenance -- the set of input tuples contributing to the result(s) of interest.
  We investigate how to automatically augment
  provenance with related information from other tables in the database with the goal to produce more
  comprehensive summaries for interesting outcomes.
  With such summaries, we explain the difference between two
  query results of interest by summarizing the augmented provenance that accurately represents one result but not the other.
\revm{Using an implementation of approach that utilizes a suite of optimization techniques, we demonstrate experimentally and through a user study that our approach produces meaningful results by efficiently navigating the large search space of possible explanations.}

}


\cut{

}

\cut{
}
\maketitle
\definecolor{lstpurple}{rgb}{0.5,0,0.5}
\definecolor{lstred}{rgb}{1,0,0}
\definecolor{lstreddark}{rgb}{0.7,0,0}
\definecolor{lstredl}{rgb}{0.64,0.08,0.08}
\definecolor{lstmildblue}{rgb}{0.66,0.72,0.78}
\definecolor{lstblue}{rgb}{0,0,1}
\definecolor{lstmildgreen}{rgb}{0.42,0.53,0.39}
\definecolor{lstgreen}{rgb}{0,0.5,0}
\definecolor{lstorangedark}{rgb}{0.6,0.3,0}
\definecolor{lstorange}{rgb}{0.75,0.52,0.005}
\definecolor{lstorangelight}{rgb}{0.89,0.81,0.67}
\definecolor{lstbeige}{rgb}{0.90,0.86,0.45}

\DeclareFontShape{OT1}{cmtt}{bx}{n}{<5><6><7><8><9><10><10.95><12><14.4><17.28><20.74><24.88>cmttb10}{}

\lstdefinestyle{psql}
{
tabsize=2,
basicstyle=\small\upshape\ttfamily,
language=SQL,
morekeywords={PROVENANCE,BASERELATION,INFLUENCE,COPY,ON,TRANSPROV,TRANSSQL,TRANSXML,CONTRIBUTION,COMPLETE,TRANSITIVE,NONTRANSITIVE,EXPLAIN,SQLTEXT,GRAPH,IS,ANNOT,THIS,XSLT,MAPPROV,cxpath,OF,TRANSACTION,SERIALIZABLE,COMMITTED,INSERT,INTO,WITH,SCN,UPDATED,GROUPING},
extendedchars=false,
keywordstyle=\bfseries,
mathescape=true,
escapechar=@,
sensitive=true
}

\lstdefinestyle{psqlcolor}
{
tabsize=2,
basicstyle=\footnotesize\upshape\ttfamily,
language=SQL,
morekeywords={PROVENANCE,BASERELATION,INFLUENCE,COPY,ON,TRANSPROV,TRANSSQL,TRANSXML,CONTRIBUTION,COMPLETE,TRANSITIVE,NONTRANSITIVE,EXPLAIN,SQLTEXT,GRAPH,IS,ANNOT,THIS,XSLT,MAPPROV,cxpath,OF,TRANSACTION,SERIALIZABLE,COMMITTED,INSERT,INTO,WITH,SCN,UPDATED,GROUPING},
extendedchars=false,
keywordstyle=\bfseries\color{lstpurple},
deletekeywords={count,min,max,avg,sum},
keywords=[2]{count,min,max,avg,sum},
keywordstyle=[2]\color{lstblue},
stringstyle=\color{lstreddark},
commentstyle=\color{lstgreen},
mathescape=true,
escapechar=@,
sensitive=true
}

\lstdefinestyle{datalog}
{
basicstyle=\footnotesize\upshape\ttfamily,
language=prolog
}

\lstdefinestyle{pseudocode}
{
  tabsize=3,
  basicstyle=\small,
  language=c,
  morekeywords={if,else,foreach,case,return,in,or},
  extendedchars=true,
  mathescape=true,
  literate={:=}{{$\gets$}}1 {<=}{{$\leq$}}1 {!=}{{$\neq$}}1 {append}{{$\listconcat$}}1 {calP}{{$\cal P$}}{2},
  keywordstyle=\color{lstpurple},
  escapechar=&,
  numbers=left,
  numberstyle=\color{lstgreen}\small\bfseries,
  stepnumber=1,
  numbersep=5pt,
}

\lstdefinestyle{xmlstyle}
{
  tabsize=3,
  basicstyle=\small,
  language=xml,
  extendedchars=true,
  mathescape=true,
  escapechar=£,
  tagstyle=\color{keywordpurple},
  usekeywordsintag=true,
  morekeywords={alias,name,id},
  keywordstyle=\color{lstred}
}


\lstset{style=psqlcolor}

\section{Introduction}\label{sec:introduction}
Today's world is dominated by data. Recent advances in complex analytics enable businesses, governments, and scientists to extract value from their data. However, results of such operations are often hard to interpret and debugging such applications is challenging, motivating the need to develop approaches that can automatically interpret and explain results to data analysts in a meaningful way. Data provenance~\cite{GKT07-semirings, DBLP:journals/ftdb/CheneyCT09}, which has been studied for several decades, is an immediate form of explanations that describes how an answer is derived from input data. However,
provenance is often insufficient for unearthing interesting insights from the data that led to a surprising result, especially for aggregate query answers. 
In the last few years, several \emph{``explanation''} methods have  emerged in the data management literature \cite{DBLP:conf/icde/BarmanKSGYA07, WuM13, RS14, DBLP:journals/pvldb/GebalyAGKS14, DBLP:conf/pods/CateCST15, ROS15, DBLP:conf/sigmod/MiaoZGR19, DBLP:journals/pvldb/WangM19} that return insightful answers in response to questions from a user.
However, real world data often exhibits complex correlations and inter-relationships that connect the provenance of a query with data that has not been accessed by the query. Current approaches do not take these crucial inter-relationships into account. Thus, the explanations they produce may lack important contextual information that can aid the user in developing a deeper understanding of the results.
We illustrate how to use context to explain a user's question using data extracted from the official website of the NBA~\cite{nba2019}.

\begin{example}\label{exp:GSW}
Consider a simplified \textbf{\nbadata database} with the following relations (the keys are \underline{underlined}, the full schema has 11 relations). Some tuples from each relation are shown in \Cref{fig:running-base}. Each team participating in a game can use multiple \emph{lineups} consisting of five players. {\tt Home} refers to the home team in a game.
\begin{itemize}[leftmargin=*]
\item \texttt{Game(\underline{year, month, day, home}, away, home\_pts, away\_pts, winner, season)}: participating teams and the winning team.
\item \texttt{PlayerGameScoring(\underline{player, year, month, day, home}, pts)}: 
the points each player scored in each game he played in.
\item \texttt{LineupPerGameStats(\underline{lineupid, year, month, day, home},}\\ \texttt{mp)}: 
the minutes played by each  lineup.
\item \texttt{LineupPlayer(\underline{lineupid, player})}: players for a lineup.
\end{itemize}
Suppose we are interested in exploring the winning records of the team \emph{GSW (Golden State Warriors)} 
in every season. The following query $\query_1$ returns this information:
\begin{lstlisting}
SELECT winner as team, season, count(*) as win
FROM Game g  WHERE winner = 'GSW' GROUP BY winner, season
\end{lstlisting}

\cut{
\CL{I think we should include all games in provenance here instead of only the games won in order to find more convincing patterns, so i changed the query to query the win ratio, which is the number of wins/ total number of games, should simplify this query \ldots}
}

\Cref{table:GSWrecord} shows the number of games team $\abbrGSW$ won for each season. $\abbrGSW$ made history in the NBA in the 2015-16 season to be the team that won the most games in a single season.
Observe that team GSW improved its performance significantly from season 2012-13 ($t_1$) to season  2015-16 ($t_2$).
 Such a drastic increase in a relatively short period of time naturally raises the question of what changed between these 2 seasons (denoted as the \emph{user question $\uquestion_1$ in \Cref{fig:uq1}}). Note that only the {\tt Game} table (shown in \Cref{table:game}) was accessed by $Q_1$.
This table provides the user with information about each game such as the name of the opponent team or whether $\abbrGSW$ was the home team  or not. However, such information is not enough for understanding why $\abbrGSW$ won or lost more games than in the other seasons, since in each season a team would play the same number of games and home games, and roughly the same number of times against each opponent. 

\begin{figure}[t]
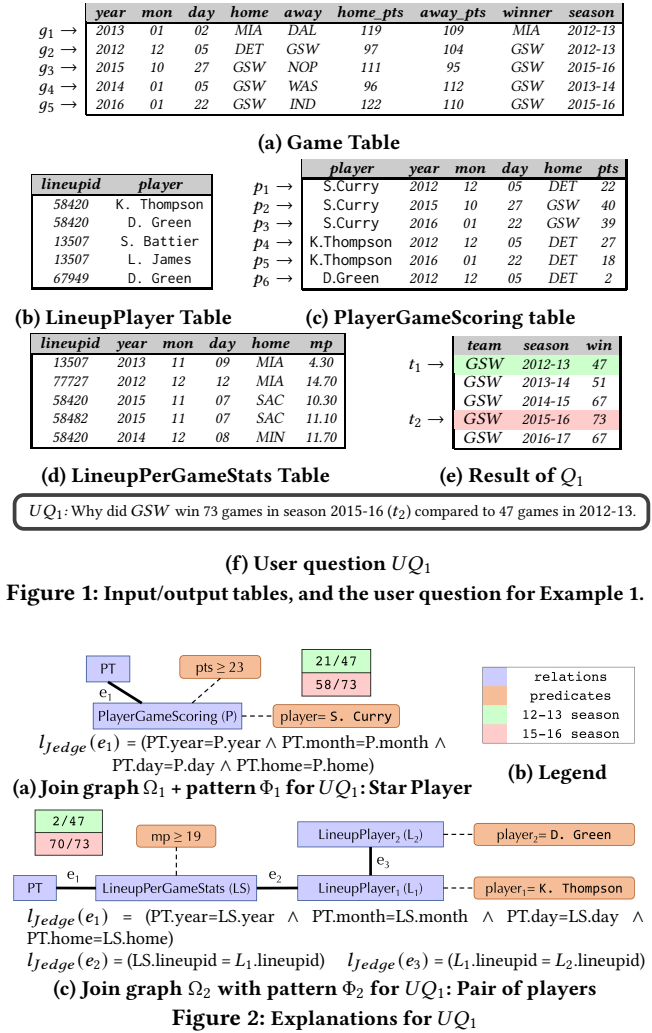
\scriptsize\setlength{\tabcolsep}{3pt}
 \vspace{-2mm}
   \subfloat[\small Game Table]{
    \begin{minipage}[b]{\linewidth}\centering
     {\scriptsize
      \begin{tabular}{l|ccccccccc|} \cline{2-10}
      & \cellcolor{grey}\textbf{year} & \cellcolor{grey}\textbf{mon} & \cellcolor{grey}\textbf{day} & \cellcolor{grey}\textbf{home} & \cellcolor{grey}\textbf{away} & \cellcolor{grey}\textbf{home\_pts} &  \cellcolor{grey}\textbf{away\_pts} & \cellcolor{grey}\textbf{winner} & \cellcolor{grey}\textbf{season} \\\cline{2-10}
      $g_1 \to$ & 2013&01&02 & MIA & DAL &119 &109 & MIA & 2012-13 \\
      $g_2 \to$ & 2012&12&05 & DET &\abbrGSW  & 97 & 104 & \abbrGSW & 2012-13 \\
      $g_3 \to$ & 2015&10&27 & \abbrGSW & NOP & 111 & 95 & \abbrGSW & 2015-16 \\
      $g_4 \to$ & 2014&01&05 & \abbrGSW & WAS &  96 & 112 & \abbrGSW & 2013-14 \\
      $g_5 \to$ & 2016&01&22 & \abbrGSW & IND & 122 & 110 & \abbrGSW & 2015-16 \\
      \cline{2-10}
      \end{tabular}
      }
      \label{table:game}
    \end{minipage}
  }
  \\
  \subfloat[\small LineupPlayer Table]{
    \begin{minipage}[b]{0.35\linewidth}\centering
     {\scriptsize
      \begin{tabular}{|cc|} \hline
      \rowcolor{grey} \textbf{lineupid} & \textbf{player} \\ \hline
      58420 & \playerKT{}\\
      58420 & \playerDG{}\\
      13507 & \playerSB{} \\
      13507 & \playerLJ{}\\
      67949 & \playerDG{} \\
      \hline
      \end{tabular}
      }
      \label{table:lp}
    \end{minipage}
 }
  \subfloat[\small PlayerGameScoring table]{
    \begin{minipage}[b]{0.62\linewidth}
    {\scriptsize
      \begin{tabular}{l|cccccc|} \cline{2-7}
       & \cellcolor{grey}\textbf{player} & \cellcolor{grey}\textbf{year} & \cellcolor{grey}\textbf{mon} & \cellcolor{grey}\textbf{day} & \cellcolor{grey}\textbf{home} &\cellcolor{grey}\textbf{pts} \\ \cline{2-7}
      $p_1 \to$ & $\playerSC{}$ & 2012 & 12 & 05 & DET & 22 \\
      $p_2 \to$ & $\playerSC{}$ & 2015 & 10 & 27 & GSW & 40 \\
      $p_3 \to$ & $\playerSC{}$ & 2016 & 01 & 22 & GSW & 39 \\
      $p_4 \to$ & $\playerKT{}$ & 2012 & 12 & 05 & DET & 27 \\
      $p_5 \to$ & $\playerKT{}$ & 2016 & 01 & 22 & DET & 18 \\
      $p_6 \to$ & $\playerDG{}$ & 2012 & 12 & 05 & DET & 2 \\
   \cline{2-7}
      \end{tabular}
      }
      \label{table:pgs}
    \end{minipage}
    }\\
  \subfloat[\small LineupPerGameStats Table]{
    \begin{minipage}[b]{0.55\linewidth}
    \begin{center}
     {\scriptsize
      \begin{tabular}{|cccccc|} \hline
      \cellcolor{grey}\textbf{lineupid} & \cellcolor{grey}\textbf{year} & \cellcolor{grey}\textbf{mon} & \cellcolor{grey}\textbf{day} &
      \cellcolor{grey}\textbf{home} & \cellcolor{grey}\textbf{mp}    \\ \hline
      13507 & 2013 & 11 & 09 & MIA & 4.30   \\
      77727 & 2012 & 12 & 12 & MIA & 14.70  \\
      58420 & 2015 & 11 & 07 & SAC & 10.30  \\
      58482 & 2015 & 11 & 07 & SAC & 11.10  \\
      58420 & 2014 & 12 & 08 & MIN & 11.70  \\
      \hline
      \end{tabular}
      }
      \end{center}
      \label{table:lpgs}
    \end{minipage}
  }\hfill
   \subfloat[\small Result of $\query_1$]{
   \begin{minipage}[b]{0.40\linewidth}\centering
     {\scriptsize
      \begin{tabular}{l|ccc|} \cline{2-4}
& \cellcolor{grey}\textbf{team} & \cellcolor{grey}\textbf{season} & \cellcolor{grey}\textbf{win} \\ \cline{2-4}
$t_1 \to$ & \cellcolor{green!20} $\cellcolor{green!20} \abbrGSW$ & \cellcolor{green!20}2012-13 & \cellcolor{green!20}47 \\
& $\abbrGSW$ & 2013-14  & 51 \\
& $\abbrGSW$ & 2014-15 & 67 \\
$t_2 \to$ & \cellcolor{red!20}$\abbrGSW$ & \cellcolor{red!20}2015-16 & \cellcolor{red!20}73 \\
& $\abbrGSW$ & 2016-17 & 67 \\\cline{2-4}
\end{tabular}
}
\label{table:GSWrecord}
    \end{minipage}
}
\\
   \subfloat[\small User question $\uquestion_1$]{
\begin{minipage}[b]{\linewidth}
    \centering
\begin{tcolorbox}[colback=white,left=2pt,right=2pt,top=0pt,bottom=0pt]
$\uquestion_1$: \emph{Why did $\abbrGSW$ win 73 games in season $2015 \mbox{-} 16$ ($t_2$) compared to 47 games in $2012 \mbox{-} 13$.} 
\end{tcolorbox}
    \label{fig:uq1}
\end{minipage}
}
  \vspace{-4mm}
  \caption{\label{fig:running-base}\small
  Input/output tables, and the user question for Example~\ref{exp:GSW}.}
    \vspace{-5mm}
 \end{figure}

\end{example}

\cut{
\CL{Removed the missing/new player explanation, for now}

\SR{let's discuss -- it might be better to give the boxed explanations first -- the representation is still confusing -- the red/green boxes are saying the same thing}
}

In this paper,
we present an approach that answer questions like $\uquestion_1$ (\Cref{fig:uq1}). Our approach produces  insightful explanations that
are based on contextual information mined from tables that are related to the tables accessed by a user's query.
To give a flavor of the explanations produced by our approach, we present two of the top explanations for $\uquestion_1$ in
\Cref{Exp: curry_points,exp:green_thompson}
(the formal definitions and scoring function are presented in the next section). Each explanation consists of three elements:
(1) A \emph{join graph} consisting of a node labeled PT representing the table(s) accessed by the user's query (we refer to this as the  provenance table, or \abbrPT for short), and nodes representing other tables that were joined with the provenance table to provide the context. Edges in a join graph represent joins \revc{between two tables. Each edge is labeled with the join condition that was used to connect the tables.}
(2)   A \emph{pattern}, which is a conjunction of predicates on attributes from the provenance or any table from the context.
   (3) \reva{The \emph{support} of the pattern in terms of the number of tuples from the provenance of each of the two result tuples from the user question that are covered by the pattern (bold and underlined in the  explanations shown below, formally defined in \Cref{sec:quality-measure}).}
\begin{figure}[t!]
\begin{minipage}[t]{0.72\linewidth}
\centering
\begin{minipage}[b]{\linewidth}
\centering
\includegraphics[width=0.7\textwidth]{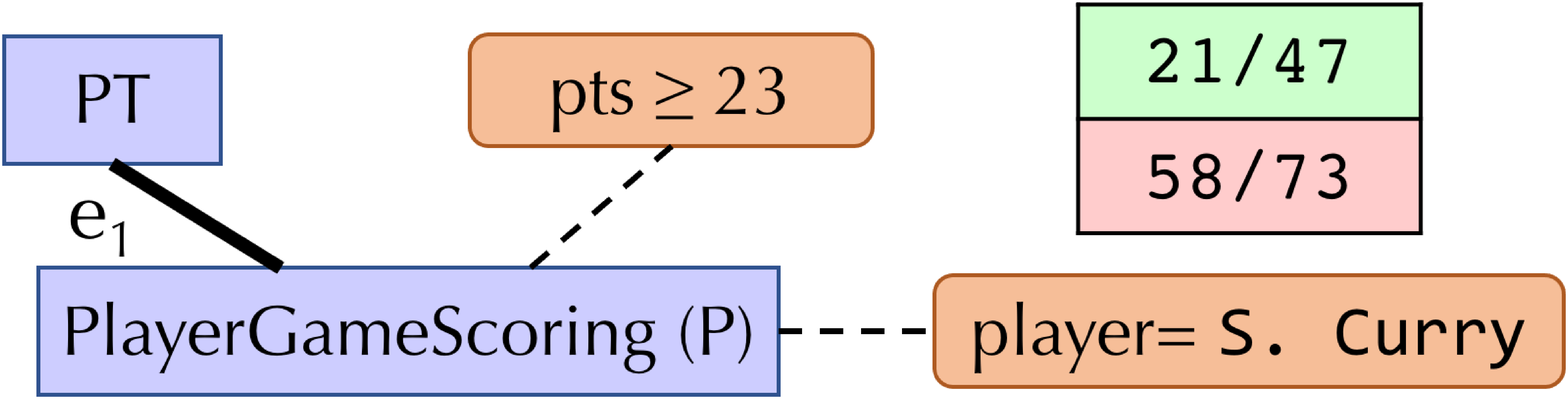}
\end{minipage}
\begin{minipage}[b]{\linewidth}
\centering
\footnotesize
{$\tlabele(\jedge_1)$ = (PT.year=P.year
  $\wedge$ PT.month=P.month $\wedge$ PT.day=P.day $\wedge$ PT.home=P.home)}
\end{minipage}\\[-2mm]
\subcaption{Join graph $\jgraph_1$ + pattern $\pat_1$ for $\uquestion_1$: Star Player}
\label{Exp: curry_points}
\end{minipage}
\begin{minipage}{.25\linewidth}
\centering
\includegraphics[width=\textwidth]{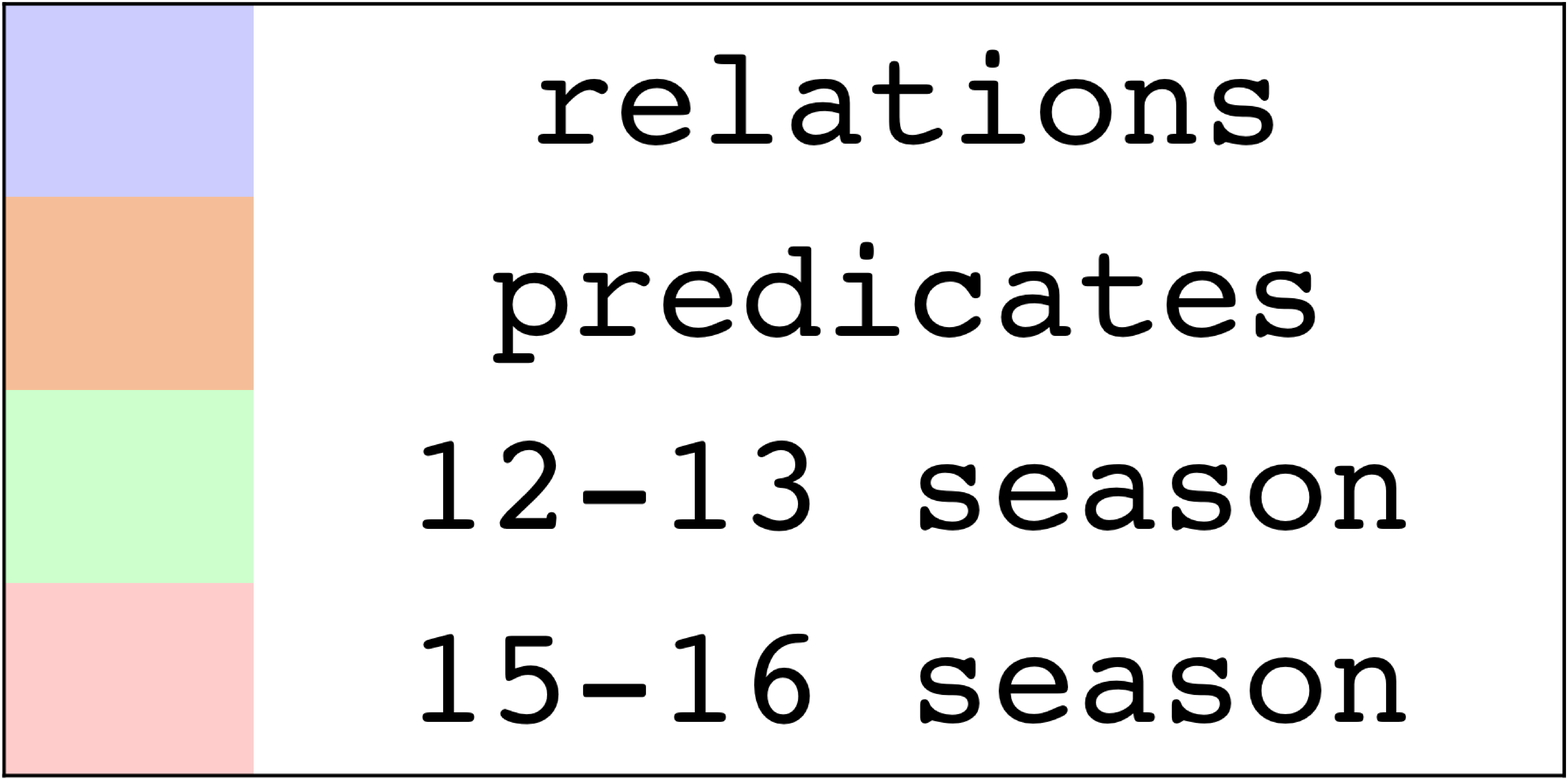}
\subcaption{Legend}
\label{table:season}
\end{minipage}

\begin{minipage}[t]{\linewidth}
\begin{minipage}{\linewidth}
\centering
\includegraphics[width=\textwidth]{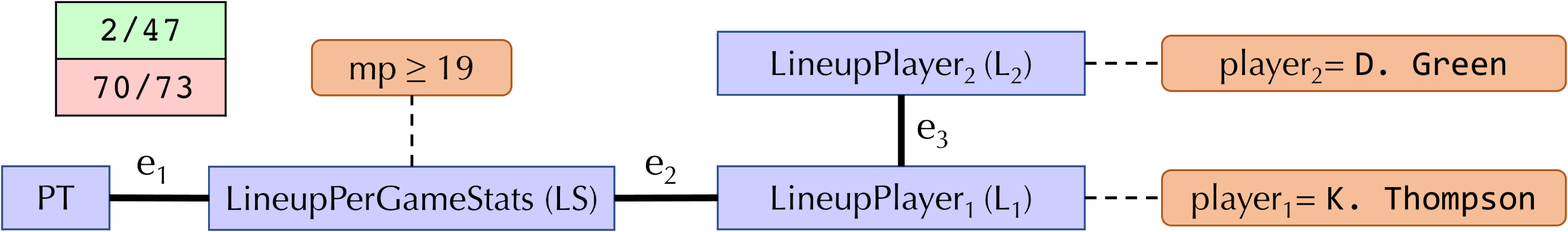}
\end{minipage}
\begin{minipage}{\linewidth}
\centering
\footnotesize
\begin{itemize}[leftmargin=*]
    \item[] $\tlabele(\jedge_1)$ = (PT.year=LS.year $\wedge$  PT.month=LS.month $\wedge$ PT.day=LS.day $\wedge$ PT.home=LS.home)
    \item[] $\tlabele(\jedge_2)$ = (LS.lineupid = $L_1$.lineupid) \quad $\tlabele(\jedge_3)$ = ($L_1$.lineupid = $L_2$.lineupid)
\end{itemize}
\end{minipage}\\[-2mm]
\subcaption{Join graph $\jgraph_2$ with pattern $\pat_2$ for $\uquestion_1$: Pair of players }
\label{exp:green_thompson}
\end{minipage}\\[-4mm]
\caption{\small Explanations for $\uquestion_1$}
\vspace{-5mm}
\label{exp: exps for uq1}
\end{figure}

Intuitively, the explanation from \Cref{Exp: curry_points} can be interpreted as: 
\begin{tcolorbox}[colback=white,left=2pt,right=2pt,top=0pt,bottom=0pt]
$\abbrGSW$ won more games in season $2015 \mbox{-} 16$ because Player \playerSC{} scored $ \geq 23$ points in \underline{$\mathbf{58}$ {\bf out of 73 games}} in $2015 \mbox{-} 16$ compared to \underline{$\mathbf{21}$ {\bf out of 47 games}} in $2012 \mbox{-} 13$. 
\end{tcolorbox}

\cut{
\SR{Add 2018-19 season}
\CL{I used 2012-13 season instead}
\ZM{Curry's average points in 15-16 season is 30.1, so I'm afraid that people may not like this explanation that his points between [24, 27] leads to more wins}
\CL{Seems only certain kinds of patterns are useful in summarization for comparison, e.g. in the above example explanation Exp \ref{Exp: curry_points}, with different pts, same player name make better sense than with different players in the pattern, but hard to generalize this for all the contexts. Proposal: put user in the loop?}
}

Given this explanation, the user can infer that 
\playerSC{} was one of the key contributors for the improvement of GSW's winning record since his 
points significantly improved
in the $2015 \mbox{-16}$ season.
Similarly, the explanation  in \Cref{exp:green_thompson} can be interpreted as:
\begin{tcolorbox}[colback=white,left=2pt,right=2pt,top=0pt,bottom=0pt]
$\abbrGSW$ won more games in season $2015 \mbox{-} 16$ because Player \playerDG{} and Player \playerKT{}'s on-court minutes together were  $\geq 19$ minutes in \underline{$\mathbf{70}$ {\bf out of 73 games}} in the $2015 \mbox{-} 16$ season compared to only \underline{$\mathbf{2}$ {\bf out of 47 games}} in $2012 \mbox{-}13$ season.
\end{tcolorbox}
\cut{
\SR{The ``winning games'' part is unclear in 2c}
\CL{Agree this might be hard to capture without prior information given, removed them}

\SR{Add 2018-19 season}
\CL{I used 2021-13 season instead}
}

This 
implies that Green and Thompson's increase of playing time together might have helped improve $\abbrGSW$'s 
record. \revs{We will discuss other example queries, user questions, and explanations returned by our approach using the NBA and the MIMIC hospital records dataset~\cite{johnson2016mimic} in \Cref{sec:case-study}.}


\cut{
\revm{
\begin{example}\label{exp:MIMIC}
Consider the (simplified) \textbf{MIMIC-III Critical Care database} \cite{johnson2016mimic} with the following relations (the keys are underlined).
\begin{itemize}
\item {\tt Admissions(\underline{adid}, dischargeloc, adtype, insurance, isdead, HospitalStayLength)} contains information 
about hospital admissions.
\item {\tt Diagnoses(\underline{pid, adid, did}, diagnosis, category)} records the diagnoses for each patient (identified by {\tt pid}) during each admission, one patient could have multiple diagnoses during one admission, which are identified by \emph{did}.
\item {\tt PatientsAdmissionInfo(\underline{pid,  admid}, age, religion, ethnicity)} records the information from a patient upon admission to the hospital. Note that one patient could have multiple entries in this table because one patient could have multiple admissions during their lifetime.
\item {\tt ICUStays(\underline{pid, admid, iid}, staylength, icutype)} records the information of the ICU stays of patients, similarly, \textit{iid} can be used to identify different ICU stays within one admission.
\end{itemize}
\end{example}
Suppose a data analyst writes the following query to find out the relationship between the insurance type 
and the death rate:
}
\begin{lstlisting}[escapeinside=!!]
!$\query_2=\,$! SELECT insurance, SUM(isdead)/count(*) AS death_rate,
            count(*) AS admit_cnt
     FROM Admissions GROUP BY insurance
\end{lstlisting}

\revm{
Given the result of this query in Figure ~\ref{Result of Q_2}, the analyst may be interested in $UQ_2$: \emph{given similar numbers of admissions, why is the death rate of patients with insurance type=$\abbrMed$ more than $2$ times larger ($t_1$) than that of patients with insurance=$\abbrPriv$}. \Cref{topk:mimic} lists the top-2 explanations returned by \oursys{}, which show that among the deaths, $\abbrMed$ has a larger fraction of admissions in emergency and male older patients ($age > 64$) compared to $\abbrPriv$. We will revisit this example in our experiments.}

\begin{figure}[t]\scriptsize\setlength{\tabcolsep}{3pt}
 \vspace{-2mm}
\subfloat[\small Result of $\query_2$]{
\begin{minipage}{0.45\linewidth}
	{\scriptsize
	\begin{tabular}{l|ccc|} \cline{2-4}
	& \cellcolor{grey}\textbf{insurance} & \cellcolor{grey}\textbf{death rate} & \cellcolor{grey}\textbf{admit\_cnt} \\ \cline{2-4}
	$t_1 \to$ & \cellcolor{green!20} $\cellcolor{green!20} \abbrMed$ & \cellcolor{green!20} 0.14 & \cellcolor{green!20} 28215 \\
	& Self Pay & 0.16 & 611 \\
	& Government & 0.05 & 1783 \\
	$t_2 \to$ & \cellcolor{red!20}$ \abbrPriv $ & \cellcolor{red!20} 0.06 & \cellcolor{red!20} 22582 \\
	& Medicaid & 0.07 & 5785 \\\cline{2-4}
	\end{tabular}
	}
	\label{Result of Q_2}
\end{minipage}
}\hfill
   \subfloat[\small User question $\uquestion_2$]{
\begin{minipage}[b]{0.45\linewidth}
    \centering
\begin{tcolorbox}[colback=white,left=1pt,right=1pt,top=0pt,bottom=0pt]
$\uquestion_2$: \emph{Given close number of admissions why do patients using $\abbrMed$ insurance have higher death rate ($t_1$, 14\%) compared to patienst with
$\abbrPriv$ insurance ($t_2$, 6\%)?}
\end{tcolorbox}
    \label{fig:uq2}
\end{minipage}
}
\\

\subfloat[\small Top-2 Explanations from \oursys{}]{
\begin{minipage}[b]{\linewidth}\centering
	{\footnotesize
      \begin{tabular}{|ccc|}
      \cellcolor{grey}\textbf{rank} & \cellcolor{grey}\textbf{pattern desc} & \cellcolor{grey}\textbf{\abbrF} \\
      $1$ & \makecell{\bf (prov.adtype=\textit{emergency}):\\ \abbrMed=23634/28215 {\bf (83.7\%)} VS \abbrPriv=    12636/22582 {\bf (56\%)}} & 0.73 \\
      \hline
      $2$ & \makecell{\bf (patient.gender=\textit{male},~admissioninfo.age<64): \\ \abbrMed=2576/28215 {\bf (9.1\%)} VS \abbrPriv=11533/22582 {\bf (51.1\%)}} & 0.63 \\
      \hline
      \end{tabular}
      }
       \label{topk:mimic}
\end{minipage}
}
 \vspace{-2mm}
 \caption{\label{fig:running-mimic}\small
  $\query_2$ results and user question for Example ~\ref{topk:mimic}.}
    \vspace{-4mm}
\end{figure}

}



\noindent
\textbf{Our Contributions.}
In this paper, we develop \oursys{} (Context-Aware Join-Augmented Deep Explanations), the first system that automatically augments provenance data with related contextual information from other tables to produce more informative summaries of the difference between the values of two tuples in the answer of an aggregate query, or, the high/low value of a single outlier tuple. We make the following contributions in this paper.

\mypar{(1) Join-augmented provenance summaries as explanations}
We propose the notion of join-augmented provenance and use summaries of augmented provenance as explanations.
The \emph{join-augmented provenance} is generated based on a \emph{join graph} that encodes how the provenance should be joined with tables that provide context.
We use patterns, i.e., conjunctions of equality and inequality predicates, to summarize the difference between the \emph{join-augmented provenance} of two tuples $t_1, t_2$ from a query's output selected by the user's question.
We adapt the notion of \emph{\abbrF} to evaluate the quality of patterns . A high \abbrF is likely to combine two desired properties for the summary  distinguishing $t_1$ and $t_2$ by giving preference to patterns with (i) high \emph{recall} (the pattern covers many tuples in the provenance of $t_1$) and (ii) high \emph{precision} (the pattern does not cover many tuples in the provenance of $t_2$). {\bf (\Cref{sec:definition})}

\mypar{(2) Mining patterns over augmented provenance}
We present algorithms for mining patterns for a given join graph and discuss a number of optimizations.
Even if we fix a single join graph to compute the augmented provenance, the large number of possible patterns poses challenges to efficiently mining patterns with high \abbrF values. Our optimizations include clustering and filtering attributes using machine learning methods, using a monotonicity property for the recall of patterns to prune \textit{refinements} of patterns (patterns are refined by adding additional predicates), and finding useful patterns on categorical attributes before considering numeric attributes to reduce the search space. {\bf (\Cref{sec:algo_pattern})}.
%
%

\mypar{(3) Mining join graphs giving useful patterns}
We also address the challenge of mining patterns over different join graphs that are based on a \emph{schema graph} which encodes which joins are permissible in a schema.
We prune the search space by estimating the cost of pattern mining as well as detecting from the available join patterns if the join graph is unlikely to generate high quality patterns. {\bf (\Cref{sec:join-graph-enumeration})}

\mypar{\revs{(4) Qualitative and quantitative evaluation}} \revs{We quantitatively evaluated the explanations produced by our approach using a case study using two real world datasets: NBA and MIMIC. We further conducted a user study to evaluate how useful the explanations generated by our approach are and how they compare with explanations generated based on the original provenance alone {\bf (\Cref{sec:case-study})}.}
We conducted performance experiments using the NBA and MIMIC database to evaluate scalability varying parameters of our algorithms, and  demonstrate the effectiveness of our optimizations.  {\bf (\Cref{sec:experiments})}

\cut{
A naive solution using existing provenance-based frameworks would try to join all possible tables to create a huge universal table, and then summarize tuples in the universal table consisting of tuples in the provenance of the original query result. This naive approach is apparently inefficient because of the huge space of possible joins and the exponential time complexity in the number of attributes for the provenance-based explanation techniques.

Therefore, we developed \oursys{}, a first system that automatically augment provenance data with related information with the goal to produce more comprehensive summaries. Given a schema graph that encodes the semantic relationships of tables in a database schema, from the base tables in the user query, \oursys{} automatically finds joins with connected tables in the schema graph enriching the provenance of a query result. By summarizing the results of such joins, \oursys{} can generate rich, high-level explanations to interesting query outcomes.
}


  \vspace{-2mm}
\section{Join-Augmented Provenance}
\label{sec:definition}

\cut{
A database schema $\dbsch$
contains a set of relations $\rsch{\rel_1}, \ldots, \rsch{\rel_k}$.
A relation schema $\rsch{\rel_i} = (A_1, \ldots, A_{m_i})$ consists of a name ($\rel$) and a  list of attribute names.
The \emph{arity} $\arity{\rsch{\rel_i}} =  m_i$ of a relation schema $\rsch{\rel_i}$ is its number of attributes.
A relation $\rel$ of relation schema $\rsch{\rel}$ is a subset of $\ddom^{\arity{\rsch{\rel}}}$ where $\ddom$ is a universal domain of attribute values.}

A database $\db$ 
comprises a set of relations $\rels(D) = \{\rel_1, \ldots, \rel_k\}$. We will use $\db$ and $\rel_1, \ldots, \rel_k$ both for the schema and the instance where it is clear from the context. For a relation $R$, $\attrs(R)$ denotes the set of attributes in $R$; similarly, for a set of relations $\mathbf S$, $\attrs(\mathbf S) = \cup_{R \in {\mathbf S}}~ \attrs(R)$ denotes the set of attributes in relations in $\mathbf S$. Without loss of generality, we assume the attribute names are distinct and use $R.A$ for disambiguation if an attribute $A$ appears in multiple relations.
\revm{In this work we focus on simple single-block SQL queries with a single aggregate function (select-from-where-group by), or, equivalently extended relational algebra queries with the same restriction.\footnote{\revc{Extensions are discussed in Section~\ref{sec:conclusions}}}.} 
\cut{
\reva{focus on relational queries expressed in SQL. We do not restrict that class of queries per se, but require that we have access to a provenance system that can compute the provenance to any query used as input to our system using the provenance representation discussed in~\Cref{sec:provenance-table}. That being said, our approach summarizes augmented provenance and, thus, works best for queries that involve aggregation or other types of data reduction operations for which a single result depends on many inputs.}
}
%
Given a query  $\query$, 
 $\query(\db)$ 
 denotes the result of evaluating the query over a database $\db$. We use 
 $\relsinQ \subseteq \rels(D)$
 to denote the relations  accessed by $Q$.
 %



\subsection{Provenance Table}
\label{sec:provenance-table}

A large body of work has studied provenance semantics for various classes of queries (e.g., \cite{GKT07-semirings, DBLP:journals/ftdb/CheneyCT09}). Here we resort to a simple \emph{why-provenance} \cite{DBLP:journals/ftdb/CheneyCT09} model sufficient for our purpose. We define the  provenance of an output tuple $t \in Q(D)$ of \reva{a query} $Q$
\reva{as a subset of the cross product ($\times$) of all relations in $\relsinQ$. For instance, Perm~\cite{glavic2009perm} can produce this type of provenance for queries in relational algebra plus aggregation and nested subqueries. In our implementation we use the GProM system~\cite{AF18}.}

\begin{Definition}[Provenance Table]\label{def:prov-table}
  Given a $\query$ with $\relsinQ = \{R_{j_1}, \cdots, R_{j_p}\}$, we define the provenance table
$\pt{\query,\db}$ for $\db$ and $\query$ \reva{to be a subset of $R_{j_1} \times \cdots \times$} $\revc{R_{j_p}}$.
\reva{We assume the existence of a provenance model that determines which tuples from the cross product belong to $\pt{\query,\db}$.}
  For a tuple $\tup \in \query(\db)$, we define the provenance table $\pt{\query,\db,\tup}$ 
  to be the subset \reva{of the provenance that contributes to $\tup$ (decided by the provenance model).}
\end{Definition}

\begin{example}
In Example ~\ref{exp:GSW}, $\pt{\query_1,\db}$ contains all the tuples from \Cref{table:game} which has $\abbrGSW$ as the winner, i.e., $g_2, g_3, g_4$, and $g_5$.
For $t_1, t_2 \in \query_1(D)$ as shown in Table~\ref{table:GSWrecord},
$\pt{\query_1,\db,t_1}$ includes all the tuples where \abbrGSW{} won in the $2012 \mbox{-} 13$ season, i.e., $g_2$, and  $\pt{\query_1,\db,t_2}$ contains
$g_3$ and $g_5$.
\end{example}


\subsection{\capSchemaGraph s and \capJoinGraph s}\label{sec:schema-graph}

\begin{figure}
\begin{minipage}{0.8\linewidth}
\centering
\tikzset{font=\scriptsize}
\begin{tikzpicture}[
node distance=1cm,
longnode/.style={rectangle,draw,thin,fill=blue!20, minimum height=3mm, text width=2.3cm, align=center, on grid},
shortnode/.style={rectangle,draw,thin,fill=blue!20, minimum height=3mm, text width=1cm, align=center, on grid},
every loop/.style={min distance=2mm,looseness=7},
longdist/.style={distance=2cm},
shortdist/.style={distance=0.5cm}
]
  \tikzset{}
    \node[longnode] (PGS)   {PlayerGameScoring (P)};
    \node[longnode] (LS)   [right=3cm of PGS] {LineupPerGameStats (LS)};
    \node[shortnode]  (G)  [below of=PGS] {Game (G)};
    \node[longnode]  (LP) [below of=LS] {LineupPlayer (L)};

  \path
        (PGS)  edge  node[left]  {$\sedge_1$} (G)
        (LS)  edge  node {$\sedge_2$} (G)
              edge  node[right] {$\sedge_3$} (LP)
         (LP) edge[loop right]  node[left] {$\sedge_4$} (LP)
              ;
\end{tikzpicture}
\end{minipage}

\begin{minipage}{\linewidth}
\footnotesize
\begin{flushleft}
$\sglabel(\sedge_1) = \{(P.year=G.year \wedge P.month=G.month \wedge P.day=G.day$ \\
\quad $\wedge P.home=G.home)$, $(P.year=G.year \wedge P.month=G.month$ \\
\quad$\wedge P.day=G.day \wedge P.home=G.home \wedge P.home=G.winner)\}$\\
$\sglabel(\sedge_2) = \{(G.year=LS.year \wedge G.month=LS.month \wedge G.day=LS.day \wedge G.home=LS.home)\}$\\
$\sglabel(\sedge_3) = \{(LS.lineupid = L.lineupid)\}$\\
$\sglabel(\sedge_4) = \{(L.lienupid = L.lineupid)\}$
\end{flushleft}
\vspace{-5mm}
\end{minipage}
\caption{Schema Graph for Example~\ref{exp:GSW}.}
\label{fig:schgraph}
  \vspace{-4mm}
\end{figure}

\textbf{Schema graphs.~} As mentioned in the introduction, we create explanations by summarizing provenance augmented with additional information produced by joining the provenance with related tables. 
We assume that a \emph{\termSchemaGraph} is given as input that models which joins are allowed.
\cut{
Given a database schema $\dbsch$, 
a \textit{\termSchemaGraph} encodes possible ways for joining the relations of the schema.
}
The vertices of \termSchemaGraph  correspond to the relations in the database. Each edge in this graph encodes a possible join between the connected relations, and is labeled with a set of possible join conditions between the two connected relations.
We use $\cond$
to denote the set of all predicates involving Boolean
conjunctions ($\wedge$) 
and equality (=) of two attributes or an attribute with a constant
that can be used for joining relations in $D$ (i.e., only \emph{equi-joins} are allowed, although all common attributes do not have to be used as in natural joins).

\begin{Definition}[\capSchemaGraph]\label{def:join-graph}
  Given a database schema $D$,  \termSchemaGraph\ $\asgraph = \sgraph$  for 
  $D$ is an undirected 
  edge-labeled graph with nodes $\snodes = \rels(D)$,
  edges $\sedges$,
  and a labeling function $\sglabel: \sedges \rightarrow 2^\cond$  that associates a set of conditions with every edge from $\sedges$. We require that for each edge $\sedge \in \sedges$, each condition in $\sglabel(\sedge)$ only references attributes from relations adjacent to $\sedge$.
\end{Definition}
\revc{Note that $\sglabel$ is an input for our method. To create \termSchemaGraphs, our system can extract join conditions from the foreign key constraints of a database and also allows the user to provide additional join conditions. Furthermore, $\sglabel$ could be determined using join discovery techniques such as~\cite{FA18,ZD19,SF12a}.}
\Cref{fig:schgraph} shows the simplified schema graph for the NBA dataset discussed in \Cref{exp:GSW}.
Unused relations in examples above are omitted.
In the schema graph, relations are represented by nodes and are connected through edges ($\sedge_1,\sedge_2, \ldots, \sedge_4$)
with conditions as labels.
For example, $\sglabel(\sedge_1)$ in \Cref{fig:schgraph}
implies that we are allowed to join \texttt{PlayerGameScoring(P)}  with \texttt{Game(G)} in two different ways: (1) through an equi-join on \emph{year, month, day}, and \emph{home} (i.e., the home-team of a game), which form the key of a game and therefore gives players' stats in all the games they played, and (2) with an additional condition on  \emph{home = winner}, which gives players' stats in the games when the home team won. Note that there is an edge $\sedge_4$ which suggests  node \texttt{LineupPlayer(L)} can be-joined with itself on condition \textit{L.lineupid = L.lineupid} (renaming of $L$ is needed in the actual join) to find players in the same lineup.
\cut{
For any edge $\sedge=(n_1,n_2)$, we define $\estart{\sedge}=n_1$, $\eend{\sedge}=n_2$.
To disambiguate attribute names within a join condition $\sglabel(\sedge)$ of an edge $\sedge$, we prefix attributes based on the following convention: attributes of the starting relation $\estart{\sedge}$ are prefixed with \textit{left}  while attributes of end relation $\eend{\sedge}$  are prefixed with \textit{right}.
\SR{should discuss directed/undirected.}
}


\mypar{Join graphs}
A \textit{\termJoinGraph} $\jgraph$ 
encodes \emph{one possible way} of augmenting 
$\pt{\query,\db}$ with related tables in the schema.
%
It contains a distinguished node $\ptlabel$ representing the relations in $\relsinQ$.
The other nodes of 
$\jgraph$ are labeled with relations,
edges in 
$\jgraph$ are labeled with join conditions
allowed by the \termSchemaGraph\ $\asgraph$, and there can be multiple parallel edges between two nodes ($\jgraph$ is a multi-graph).
\begin{Definition}[\termJoinGraph]\label{def:join-tree}
Given a database $D$, \termSchemaGraph\ $\asgraph=\sgraph$ and query $\query$, a \termJoinGraph\ $\jgraph$ for $\asgraph$ is a node- and edge-labeled undirected multigraph $\jjgraph$ with nodes  $\jgraphnodes$, edges $\jgraphedges$, a node labeling function $\tlabelv: \jgraphnodes \rightarrow \rels(D) \cup \{PT\}$, and edge labeling function $\tlabele: \jgraphedges \rightarrow \cond$. For any \termJoinGraph\ we require that it contains exactly one node labeled with $\ptlabel$ and there are no edges with $PT$ as both end-points.
For every edge $\jedge =(\node_1,\node_2) \in \jgraphedges$ we require that there exists a corresponding edge $\sedge = (\rel_1,\rel_2) \in \sedges$ such that all of the following conditions hold:
\begin{itemize}
\item
$\tlabele(\jedge) \in \sglabel(\sedge)$
(modulo renaming relations using their aliases for disambiguation as discussed below) 
\item If\ $\tlabelv(\node_1) = \ptlabel$, then $\rel_1 \in \relsinQ$, else, $\tlabelv(\node_1) = \rel_1$
\item If\  $\tlabelv(\node_2) = \ptlabel$, then $\rel_2 \in \relsinQ$, else, $\tlabelv(\node_2) = \rel_2$
\end{itemize}
\end{Definition}
The first condition above says that the join condition between two relations in the  \termJoinGraph\ $\jgraph$ should be one of the allowed conditions in the \termSchemaGraph\ $\asgraph$.
The second and third conditions say that edges adjacent to node $PT$ should correspond to an edge adjacent to a relation 
accessed by query $Q$.
Note that multiple nodes in $\jgraphnodes$ may be labelled with the same relation and also relations from $\relsinQ$ may appear node labels.

\mypar{Disambiguation of relations and attributes 
in a join graph using aliases and multigraph}
In join graphs corresponding to a schema graph, we may need to address
some ambiguity in attribute names and relation names. (1) Unlike the schema graph $\asgraph$, the join graph $\jgraph$ may contain the same relation $R_i$ multiple times with node label $\neq \ptlabel$. We give each such occurrence of $R_i$ a fresh label $R_{i1}, R_{i2}, \cdots$ in $\jgraph$. Each original attribute $R_i.A$ in the conditions in labels $\acond$ are now renamed as $R_{i1}.A, R_{i2}.A,$ and so on in the edges incident on $R_{i1}, R_{i2}, \cdots$ respectively in $\jgraph$. (2)  In addition to the join graph $\jgraph$, even in the original query $Q$ and therefore in the provenance table $\pt{Q, D}$, the same relation $R_i \in \relsinQ$ can appear multiple times using different aliases, say,  $R_{i1}, R_{i2}$. Suppose in the schema graph $\asgraph$ there is an edge between $R_i, R_j$. Then in a join graph $\jgraph$, there can be two parallel edges between node $\ptlabel$ and $R_j$, one corresponding to the join condition between $R_{i1}$ and $R_j$, and the second one corresponding to the join condition between $R_{i2}$ and $R_j$. The labels of these edges will use the corresponding aliases ($R_{i1.A}$ on one edge and $R_{i2.A}$ on the other) for disambiguation. Note that in a join graph, there can be a combination of (1) and (2).

\begin{example}
Consider the join graph $\jgraph_2$ from \Cref{exp:green_thompson}. Since $\relsin{Q_1}$ = \{\texttt{Game}\}, $\ptlabel$ represents the one relation accessed by $Q$.
Nodes from this join graph are connected through edges ($\jedge_1, \jedge_2, \jedge_3$), where each edge has a corresponding condition in the schema graph shown in \Cref{fig:schgraph}. For example, join condition on $\jedge_1$ from \revc{$\jgraph_2$} is the first condition
in the label of \revc{$\sedge_2$} from the schema graph, 
i.e., $\tlabele(\jedge_1) \in \sglabel(\revc{\sedge_2})$. Similarly, $\tlabele(\jedge_2) \in \sglabel(\revc{\sedge_3})$. 
As discussed above, \texttt{LineupPlayer} appears more than once in the join graph renamed as \texttt{LineupPlayer$_1$} ($L_1$) and \texttt{LineupPlayer$_2$} ($L_2$). 
\end{example}

\subsection{Augmented Provenance Table}\label{sec:apt}
We now describe the process of generating the relation produced for a given \termJoinGraph\ $\jgraph$ --- the result of joining the relations in the graph $\jgraph$ based on the encoded join conditions (after renaming relations and attributes as described in the previous section).
\cut{
\blue{Since a relation may appear more than once as a node label in a \termJoinGraph\ and the attribute names of relations may not be disjoint, we define for a join graph $\jgraph$ a renaming function $\renameA$ that takes as input a $\node$ from $\jgraphnodes$ and an attribute from the relation corresponding to $\node$ and return a fresh attribute name $\att$ from a set of attribute names $\{ \att_1, \ldots, \}$ that we assume to be disjoint with the set of all attribute names that appear in the database schema. We use $\renameA^{-1}$ to denote the inverse of $\renameA$. Furthermore, we use $\renameA_{\node}$ to denote that restriction of $\renameA$ to $\node$, i.e., the function $\renameA_{\node}(\att) = \renameA(\node, \att)$. We assume that for a provenance table with $n$ attributes, $\renameA$ renames the attributes of the provenance table to $A_1, \ldots, A_n$. Furthermore, we assume that $\renameA$ is deterministic which allows us to extend a \termJoinGraph\ with new nodes and edges without affecting the renamed attributes for any existing nodes in the graph.}
}

\begin{Definition}[Augmented Provenance Table]\label{def:apt}
 Consider a database $\db$, a query $\query$, and a \termJoinGraph $\jgraph = \jjgraph$. Let $S_{j_1}, \cdots, S_{j_p} = \jgraphnodes - \{\ptlabel\}$,
 i.e., all the relations that appear in $\jgraph$ with labels $\neq \ptlabel$.
  Furthermore, let $t \in Q(D)$ and tuple $t' \in \pt{Q, D, t}$, we define:
  The \emph{augmented provenance table} (\emph{\abbrAPT}) for $D$, $Q$, and $\jgraph$ (and $\tup$, $\tup'$) is defined as
  \begin{align*}
    \jtr(Q, \db, \jgraph) &=
    \sigma_{\theta_\jgraph}
    ( \pt{Q, D} \times S_{j_1}
    \times \cdots \times S_{j_p} )\\
    \jtr(Q, \db, \jgraph, \tup) &=
                                  \sigma_{\theta_\jgraph}
                                  ( \pt{Q, D, \tup} \times S_{j_1}
                                  \times \cdots \times S_{j_p} )\\
 \jtr(Q, \db, \jgraph, \tup, t')  &=
                                  \sigma_{\theta_\jgraph}
                                  ( \{\tup'\} \times S_{j_1}
                                  \times \cdots \times S_{j_p} )
  \end{align*}

Here $\theta_\jgraph = \bigwedge_{(S_a, S_b) \in \jgraphedges} \tlabele((S_a, S_b))$ is the conjunction of join conditions in the \termJoinGraph\ $\jgraph$.
 The join conditions only use equality comparisons between two attributes. 
 \revm{We assume} that duplicate (renamed) columns are removed from  $\jtr(Q, \db, \jgraph)$.

\end{Definition}

\begin{example}
Consider $\jgraph_1$ in \Cref{Exp: curry_points} that combines provenance table PT with {\tt PlayerGameScoring} through an equi-join on \emph{year, month,  day}, and \emph{home}. 
\Cref{exp:apt-table} shows the join result $\jtr(Q_1, \db, \jgraph_1)$ using the tuples from \Cref{table:game,table:pgs}. 

\begin{figure}[t]\scriptsize\setlength{\tabcolsep}{2pt}
    \begin{minipage}[b]{\linewidth}\centering
     {\scriptsize
      \begin{tabular}{|ccccccccccc|}  \cellcolor{grey}\textbf{year} & \cellcolor{grey}\textbf{mon} & \cellcolor{grey}\textbf{day} & \cellcolor{grey}\textbf{home} & \cellcolor{grey}\textbf{away} & \cellcolor{grey}\textbf{home\_pts} &  \cellcolor{grey}\textbf{away\_pts} & \cellcolor{grey}\textbf{winner} & \cellcolor{grey}\textbf{season} & \cellcolor{grey}\textbf{player} & \cellcolor{grey}\textbf{pts} \\\cline{1-11}
       2012&12&05 & DET &\abbrGSW  & 97 & 104 & \abbrGSW & 2012-13 & \playerSC{} & 22\\
       2012&12&05 & DET &\abbrGSW  & 97 & 104 & \abbrGSW & 2012-13 & \playerKT{} & 27\\
       2012&12&05 & DET &\abbrGSW  & 97 & 104 & \abbrGSW & 2012-13 & \playerDG{} & 2\\
       2015&10&27 & \abbrGSW & NOP & 111 & 95 & \abbrGSW & 2015-16 & \playerSC{} & 40\\
       2016&01&22 & \abbrGSW & IND & 122 & 110 & \abbrGSW & 2015-16 & \playerSC{} & 39\\
       2016&01&22 & \abbrGSW & IND & 122 & 110 & \abbrGSW & 2015-16 & \playerKT{} & 18\\
      \cline{1-11}
      \end{tabular}
      }
        \vspace{-2mm}
    \caption{$\jtr(Q_1, \db, \jgraph_1)$ result using example tuples}
    \label{exp:apt-table}
\vspace{-5mm}
    \end{minipage}
\end{figure}

\end{example}

  \vspace{-2mm}
\subsection{Explanations with Augmented Provenance}
\label{sec:explanations}
\oursystem's approach for generating explanations is based on summarizing augmented provenance tables. In particular, given a database and a query, the user identifies interesting or surprising tuples in the query answer (e.g., the aggregate value is high/low, or the value of a tuple is higher/lower than another).
To explain such interesting results, \oursystem{} returns \emph{patterns} (i.e., predicates) that each summarize the difference between the augmented provenance for two query result tuples (or the provenance of one result tuple).

\mypar{User questions} Given a database $D$ and a 
query $Q$,   \oursystem\ supports \textbf{two-point questions} or comparisons, which we will discuss by default:
        {\em Given
    $t_1, t_2 \in Q(D)$, summarize
    input tuples in $D$ that differentiate $t_1$ from $t_2$.}
    \cut{
    Given 
    $t_1, t_2 \in Q(D)$, 
    what in the input can best
    explain the difference of
    $t_1.\val$ and  $t_2.val$. 
    }
However, 
\oursystem\ also works for 
\textbf{single-point questions:} {\em Given a single
tuple $t \in
    Q(D)$, summarize input tuples in $D$ that differentiate $t$ from the rest of the tuples. }
    \cut{
    what in the input can best
    explain 
    $t.\val$.
    }
    Here the intuitive idea is to treat $t$ as $t_1$, and all tuples $t' \neq t \in Q(D)$ as $t_2$.

\mypar{Explaining aggregates vs. summarizing 
provenance vs. non-provenance}
\revs{Instead of directly explaining why an aggregate value $t.\val$ is high/low or $t_1.\val$ higher/lower than another value $t_2.\val$ \cite{WuM13, RS14, DBLP:conf/sigmod/MiaoZGR19}, the goal of \oursystem{} is to use ``patterns'' (discussed below) to summarize the input tuples that contributed the most to an output tuple  as well as distinguish it from the other outputs. Therefore, unlike the approaches in~\cite{WuM13, RS14, DBLP:conf/sigmod/MiaoZGR19}, in \oursystem, the aggregate values $t_1.\val, t_2.\val, t.\val$  do not play a role in the explanations}\footnote{Taking the amount of contribution (responsibility/sensitivity) of input tuples into account as in~\cite{WuM13, RS14, DBLP:conf/sigmod/MiaoZGR19} is an interesting direction for future work.}.

\mypar{Summarization patterns and explanations}
In the \oursystem\ framework, explanations are provided as \emph{patterns} or \emph{conjunctive predicates} to compactly represent sets of tuples from the augmented provenance tables based on different join graphs.
This type of patterns has been used widely for  explanations~\cite{WuM13, RS14, DBLP:journals/pvldb/GebalyAGKS14, DBLP:conf/sigmod/MiaoZGR19,lee2019pug}.

\begin{Definition}[Summarization Pattern and Matching Tuples]
Let $\rel$ be a relation with attributes $(\att_1, \ldots, \att_m)$, and let $\ddom_i$ denote the \emph{active domain} of attribute $\att_i$ in $R$.
A \emph{summarization pattern} (or simply a \emph{pattern}) $\pat$ is an $m-$ary tuple such that
for every $\att_i \in \rel$,
    (i) if $\att_i$ is a numerical or ordinal attribute:
    $\pat.\att_i \in  \bigcup_{X \in \ddom_i}\{\pleq{X}, \pgeq{X}, \peq{X}\} \cup \{\placeh\}$,
    (ii) if $\att_i$ is a categorical attribute:
      $\pat.\att_i \in \bigcup_{X \in \ddom_i}\{\peq{X}\} \cup \{\placeh\}$.

Here $*$ denotes that the attribute is not being used in the pattern and $X \in \ddom_i$ denotes a threshold for numeric attributes. If $\pat.A_i \neq \placeh$, then $\pat.A_i[0]$ denotes the threshold $X$ and $\pat.A_i[1]$ denotes the comparison operator $\leq, \geq$, or $=$.
\par
A tuple $\tup \in \rel$ \emph{matches} a pattern
$\pat$, written as $\tup \pmatch \pat$, if $\tup$ and $\pat$ agree on all conditions, i.e., $\forall i \in \{1, \ldots, m\}$, one of the following must hold:
    (i) $\pat.\att_i = \placeh$, or
    (ii) 
    ($\tup.\att_i \geq \pat.\att_i[0]) \land (\pat.\att_i[1] = `\geq\text{'})$, or
    (iii) ($\tup.\att_i \leq \pat.\att_i[0]) \land (\pat.\att_i[1] = `\leq\text{'})$, or
    (iv) ($\tup.\att_i = \pat.\att_i[0]) \land (\pat.\att_i[1] = `=\text{'})$.
We use {\upshape$\matches(\pat,\rel)$} to denote $\{ \tup \in \rel \mid  \tup \pmatch \pat \}$.
\end{Definition}

When presenting textual descriptions of summarization patterns, we omit attributes which are set to $\placeh$, and instead include the attribute name as $(A_i: \pat.\att_i[0], \pat.\att_i[1])$ to avoid ambiguity. Also, since the group-by attributes exactly capture the answer tuples $t_1, t_2$, and do not provide any additional information, patterns are not allowed to include attributes used in grouping in the query $Q$.




As discussed in the introduction, the explanations 
given by \oursystem\ consist of a \termJoinGraph\ $\jgraph$, a pattern $\pat$ over  $\jtr(\jgraph,\db)$, and 
statistics on \emph{support} of $\pat$
to show how it differentiates one tuple from the others by {\em augmenting} the provenance using $\jgraph$, and thereby including additional contextual information from other tables in $D$.

\cut{
. Intuitively, by `{\em augmenting}' the provenance table using $\jgraph$, such an explanation includes additional contextual information from other tables.
}
\begin{Definition}[Explanations from Augmented Provenance]\label{def:explanation}
Given a database $\db$,
\termSchemaGraph\ $\asgraph$,  query $\query$,
and a two-point question with \revc{$\tup_1, \tup_2$}$ \in \query(\db)$,  an explanation
is a tuple $\sexpl$ = $(\jgraph$, $\pat$, $(v_1, a_1)$, $(v_2, a_2))$, where $\jgraph$ is a join graph for $\asgraph$; $\pat$ is a pattern over the augmented provenance table $\jtr(Q, \db, \jgraph)$;
and $(v_1, a_1)$ and $(v_2, a_2)$ denote the relative support of
$\pat$ for $t_1, t_2$.\footnote{We will discuss the relative support in the next section.}
\end{Definition}

 For simplicity, we will often drop $(v_1,a_1)$ and $(v_2,a_2)$  as the supports can be easily computed with this information.

\begin{example}\label{eg:pattern}
Consider the explanation from Figure ~\ref{Exp: curry_points}. The pattern $\pat_1$ is found from $\jgraph_1$ by \oursystem\ as the following  tuple: \{(\emph{player}: S.Curry, =), (\emph{pts}: 23, $\geq$) \}. Here \textit{player} is a categorical attribute and \textit{pts} is a numeric attribute (the other attributes are $*$), both coming from the
\texttt{PlayerGameScoring} table.
Any tuple from the $\jtr(Q_1, \db, \jgraph_1)$ which fulfills \emph{player = 'S.Curry'} and \emph{$pts \leq 23$} will be included in {\upshape$\matches(\pat_1,\jtr(Q_1, \db, \jgraph_1))$}. 
Rephrasing the text box, for
$UQ_1$ in Figure~\ref{fig:uq1} one 
explanation is:
$(\pat_1, \jgraph_1, (58, 73), (21, 47))$.
\end{example}

It can be noted that the explanations for two-point questions are asymmetric, as one of the tuples is chosen as the \emph{primary tuple} whose relative support is given by $(v_1, a_1)$, and the second one is chosen as the \emph{secondary tuple}, whose relative support is given by $(v_2, a_2)$. Switching these two tuples may result in a different set of top explanations using quality measures that we discuss next. 

\cut{
Similarly, we will use join graphs and patterns to explain the difference between two query results provided by the user.  For that we generate the \abbrJGresult as before, but instead of one provide two patterns --- one for each of the two results whose difference we should explain. For a pattern $\pat$ over a schema $(\att_1, \ldots, \att_n)$, the schema of the pattern are the attributes $\att_i$ for which $\pat$ is constant. For instance, $\pat = (\placeh, \placeh, 3, 4)$ over schema $(A,B,C,D)$ has schema $(C,D)$.


\begin{Definition}[Difference Explanations]\label{def:diff-expl}
  Let $\db$ be a database, $\sgraph$ a \termSchemaGraph for $\db$, $\query$ a \abbrAggJoinQ, and let $\tup_1, \tup_2 \in \query(\db)$. An explanation for the difference of $\tup_!$ and $\tup_@$ is triple $\dexpl = (\jgraph, \pat_1, \pat_2)$ where $\jgraph$ is a \abbrJoinGraph for $\sgraph$, $\pat_1$ and $\pat_2$ are patterns for $\jtr(\jgraph, \db)$ and $\pat_1$ and $\pat_2$ that (i)  have the same schema and (ii) differ in at least one constant.
\end{Definition}
}


\subsection{Quality Measure of Explanations}\label{sec:quality-measure}
First, we discuss the quality measures for explanations when the join graph is given, and then discuss how to find top explanations across all join graphs mined by our algorithms.

\subsubsection{Quality Measures Given a Join Graph}\label{sec:quality-measure-jg}
For 
a two-point user question focusing on the difference between $t_1, t_2 \in Q(D)$, we would like the pattern in a good explanation to
match as much provenance of $t_1$ as possible, and not match much in the provenance of $t_2$. For this purpose, we adapt the standard notion of \emph{\abbrF}. 
Recall that {\upshape$\matches(\pat,\rel)$} denotes $\{ \tup \in \rel \mid  \tup \pmatch \pat \}$.


\cut{
\begin{itemize}
    \item Define TP (fractional and integral)
    \item Define FP (fractional and integral)
    \item Define FN (fractional and integral)
     \item Define Recall as TP/TP+FN (fractional and integral)
      \item Define Precision as TP/TP+FP (fractional and integral)
       \item Define \abbrF as harmonic mean of Recall and Precision (fractional and integral)
\end{itemize}
}

\begin{Definition}[Quality Metrics of a Pattern]\label{def:metrics}  
Consider a database $\db$, a query 
$\query$, a join graph $\jgraph$, two output tuples in the user question
$\tup_1, \tup_2 \in Q(D)$,
and an explanation pattern $\sexpl = (\jgraph, \pat)$.

(a) A tuple $\tup' \in \pt{\query, \db, \tup_1}$ (similarly for $t_2$) is said to be {\bf covered} by $\sexpl$ if there exists $\tup''\in \jtr(\query, \db, \jgraph, \tup_1, \tup')$ (ref. Definition~\ref{def:apt}) such that $\tup'' \pmatch \pat$.
The \emph{coverage} of $\sexpl$ on $\tup_1, \tup'$ in $\jtr(Q, \jgraph, \db)$ is: 
{\footnotesize
$$\ICov(\sexpl, \jgraph, \tup_1, \tup') = \mathbb{1}[\matches \big(\pat,  \jtr(Q, \jgraph, \db, \tup_1, \tup') \big) \neq \emptyset]$$
}
  where $\mathbb{1}[]$ is the indicator function.

\par
(b) The {\bf coverage} (or, \emph{\bf true positive}) of $\sexpl$ for $t_1$ is defined as the sum of its coverage on all tuples in the provenance table:
{\footnotesize
\begin{align*}
\ITP(\sexpl, \jgraph, \tup_1) &= \sum_{\tup' \in \pt{\query, \db, \tup_1}} \ICov(\sexpl, \jgraph, \tup_1, \tup')
\end{align*}
}

(c) The {\bf false positive} of $\sexpl$ for $t_1$ \revc{in comparison to $t_2$} is the sum of its coverage on all tuples in $\pt{Q, D}$ that are in the provenance of $\tup_2$ \revc{($t_1$ does not appear on the right-hand side here)}:
{\footnotesize
\begin{align*}
\IFP(\sexpl, \jgraph, \tup_1, \tup_2) &=
\sum_{\tup' \in \pt{\query, \db, \tup_2}} \ICov(\sexpl, \jgraph,  \tup_2, \tup')
\end{align*}
}

(d) The {\bf false negative} of $\sexpl$ for $t_1$ is defined as the sum of the uncovered tuples in the provenance of $\tup_1$:
{\footnotesize
\begin{align*}
\IFN(\sexpl, \jgraph,  \tup_1) &= \sum_{\tup' \in \pt{\query, \db, \tup_1}} 1-\ICov(\sexpl, \jgraph, \tup_1, \tup')
\end{align*}
}

(e) 
Using (b)-(d), we define {\bf precision, recall}, and {\bf \abbrF} for $t_1$  in comparison to $t_2$ as usual:\\[-4mm]
{\footnotesize
\begin{align*}
\Prec(\sexpl, \jgraph,  \tup_1, \tup_2) = \frac{\ITP(\sexpl, \jgraph,  \tup_1)}{\ITP(\sexpl, \jgraph,  \tup_1) + \IFP(\sexpl, \jgraph,  \tup_1, \tup_2)}\\
\Rec(\sexpl, \jgraph,  \tup_1) = \frac{\ITP(\sexpl, \jgraph,  \tup_1)}{\ITP(\sexpl, \jgraph,  \tup_1) + \IFN(\sexpl, \jgraph,  \tup_1)}\\
\fscore(\sexpl, \jgraph,  \tup_1, \tup_2) =\frac{2}{\frac{1}{\Prec(\sexpl, \jgraph,  \tup_1, \tup_2)} +\frac{1}{\Rec(\sexpl, \jgraph,  \tup_1)}}
\end{align*}
}
\end{Definition}

A high recall implies that the pattern $\pat$ describes the tuples contributing to $t_1$ well. A high precision implies that $\pat$ covers few tuples in the provenance of $t_2$. A high \abbrF indicates both. This definition can be easily adapted to single-point questions involving a single output tuple $t \in Q(D)$ by summing over $t' \in \pt{\query, \db} \setminus \pt{\query, \db, \tup}$ instead of summing over $t' \in \pt{\query, \db, \tup_2}$ in the false positives definition above. The other definitions remain the same.

\mypar{Support of explanation patterns} As described in the running examples and in \Cref{def:explanation}, an explanation $\sexpl = (\jgraph, \pat, (v_1, a_1), (v_2, a_2))$ includes the relative support of the pattern $\pat$ for $t_1, t_2 \in Q(D)$ to illustrate how this pattern differentiates the output tuples $t_1, t_2$. Here $v_1 = \ITP(\pat, \jgraph, t_1)$ and $a_1 = \ITP(\pat, \jgraph, t_1) + \IFN(\pat, \jgraph, t_1) = |\pt{\query, \db, \tup_1}|$ as defined in Definition~\ref{def:metrics}, denoting the set of tuples in the provenance of $t_1$ covered by the pattern $\pat$, and the set of all tuples in \revc{the} provenance of $t_1$ respectively. Similarly, we define $v_2, a_2$ for the output tuple $t_2$ to illustrate the difference with $t_1$.

\mypar{Finding Top-$k$ Patterns with Highest \abbrFs} Given a join graph $\jgraph$, our goal is to find the top-$k$ patterns $\pat$ in terms of their individual \abbrFs according to  \Cref{def:metrics}.
However, in practice there are additional considerations that we should take into account, e.g., the maximum number of attributes appearing in a pattern.

\mypar{Complexity} Finding top-$k$ patterns given a join graph $\jgraph$ has polynomial {\em data complexity} \cite{Vardi82} (fixed size schema and query). The provenance table and \abbrAPTs{} can be computed in \ptime in the size of the data.
Given a pattern, its matches can be determined in \ptime and therefore, 
all metrics in \Cref{def:metrics}
can be computed in \ptime. If there are $p$ attributes in the augmented provenance table, the number of possible patterns is bounded by $O(n^p)$ (the number of distinct attribute values is bounded by $n$ = total number of tuples 
in the database, and each attribute can appear as don't care $*$ and with at most three comparison operators), so even a naive approach of computing the top-$k$ patterns with the highest \abbrF values is  polynomial in data size. However, this naive approach is not scalable in practice and therefore we adopt a number of  heuristic optimizations to solve this problem as described in \Cref{sec:algo_pattern}.

 \cut{
\begin{enumerate}
    \itemsep0em
    \item {\bf Find top-$k$ patterns with  individual \abbrFs:} Here the goal is to simply find the top-$k$ patterns $\pat$ in terms of their individual \abbrFs.
    \item {\bf Generating a given combined \abbrF with a small set of patterns:} Given a set of patterns $\patrnset$, we adapt Definition~\ref{def:metrics} (a) as follows:  a tuple $\tup' \in \pt{\query, \db, \tup_1}$  is said to be {covered} by $\patrnset$ \emph{if there exists a pattern $\pat \in \patset$ and a $\tup''\in \jtr(\query, \db, \jgraph, \tup_1, \tup')$} such that $\tup'' \patrnset \pat$. Therefore, a tuple $t'$
    in the provenance is covered, as long as any pattern from $\patrnset$ covers any tuple $t''$ from the augmented provenance of $t_1, t'$. The other definitions (b)-(e) remain the same but use this new definition of coverage from (a). The combined \abbrF of a set of patterns is denoted by $\fscore(\patrnset, \jgraph, t_1, t_2)$.  In this objective, the goal is to find a set of patterns $\patrnset$ with the  $\fscore(\patrnset, \jgraph, t_1, t_2) \geq \lambda$ for a given threshold $\lambda$ such that $|\patrnset|$ is minimized. \red{May have to remove}
\end{enumerate}
The difference between the two objectives is that, the second one intends to find explanations that are not highly overlapping or similar so that the user can see different types of explanation patterns in a small set of pattern set. On the other hand, as we show in the next section, this objective is computationally hard. Therefore, the first objective in addition to other optimizations in our algorithms, serves as a heuristic for the second objective. \blue{CHECK}
Indeed, there can be additional considerations in practice, e.g., the maximum number of attributes appearing in a pattern, that are considered by our algorithms.
}
\mypar{Explanations over All Join Graphs} When mining multiple join graphs $\jgraph$, there are several options for finding top patterns across all join graphs, e.g., penalizing patterns from complex join graphs. However, for simplicity, and for an interactive user experience, we find top-$k$ patterns for each individual join graph and present a global ranking of all patterns. Thus, the user can explore explanations generated from more than one join graph (see \Cref{sec:join-graph-enumeration}).

\cut{\par
In Section~\ref{sec:algo_pattern}, we discuss our algorithms to mine patterns when a join graph is fixed, and in Section~\ref{sec:join-graph-enumeration} we discuss how we enumerate all join graphs and mine patterns over them.}

\cut{
\par
{\bf Quality metrics for single-point and double-point user questions.~} Given a join graph $\jgraph$, for a single-point user question mentioning an output tuple $t \in Q(D)$,  we simply compute $\fscore(\sexpl, \jgraph, \tup)$ to compute the score of a pattern $\pat$. For a two-point question with output tuples $t_1, t_2$, we choose one, say $t_1$, as the primary tuple.  True positive and false negatives of $t_1$ are computed as described in Definition~\ref{def:metrics}. For false positives OF $t_1$, instead of summing over $t' \in \pt{\query, \db} \setminus \pt{\query, \db, \tup}$, we sum over $t' \in \pt{\query, \db, \tup_2}$, to differentiate the provenance of $t_1$ and $t_2$. \blue{CHECK + ADD SUPPORT}.
}

\cut{
Compared to the previous single answer version, the quality of explanations for the difference between two answer tuples $\tup_1$ and $\tup_2$ is also measured using $\fscore$, with a major difference lies in the calculation of the False Positive. Firstly, we need to mark one of the two result tuple as ``True" tuple \CL{could change this later...} i.e., let's choose $\tup_1$ as ``True" tuple. The False Positive for $\tup_1$ is counted over tuples in the provenance of $\tup_2$ instead of rest tuples in the entire provenance, which intuitively means that in order to get a high \abbrF, we want to capture patterns should cover as much provenance of one result tuple(``True" tuple) ($\tup_1$) as possible but not the provenance of the other result tuple ($t_2$).

\begin{Definition}[Quality Measure for Difference]\label{def:quality-measure-diff}
Let $\dexpl = (\jgraph, \pat)$ be an explanations for $\query$,  $\db$, and two tuples $\tup_1, \tup_2$,  where the precision is defined as the percentage of tuples in $\pt{\query, \db, \tup_1}$ matching $\pat$ in the augmented provenance table over those matched tuples in $\pt{\query, \db, \tup_1} \cup \pt{\query, \db, \tup_2}$.
\begin{align*}
\FPrec(\dexpl) &= \frac{\FTP(\dexpl)}{\FTP(\dexpl) + \FFP(\dexpl)}  \\
&= \frac{\sum_{\tup' \in \pt{\query, \db, \tup_1}} \frac{\mid \matches \big(\pat,  \jtr(Q, \db, \jgraph, \tup') \big)\mid}{\mid { \jtr(Q, \db, \jgraph,  \tup')} \mid}}{\sum_{\tup' \in U} \frac{\mid \matches \big(\pat,  \jtr(Q, \db, \jgraph, \tup') \big)\mid}{\mid { \jtr(Q, \db, \jgraph,  \tup')} \mid}}
\end{align*}
where $U=\pt{\query, \db, \tup_1} \cup \pt{\query, \db, \tup_2}$.

The recall and \fscore\ are defined in the same way as Definition~\ref{def:fscore}.
\begin{align*}
\FRec(\dexpl) &= \frac{\FTP(\dexpl)}{\FTP(\dexpl) + \FFN(\dexpl)}  \\
&= \frac{\sum_{\tup' \in \pt{\query, \db, \tup_1}} \frac{\mid \matches \big(\pat,  \jtr(Q, \db, \jgraph, \tup') \big)\mid}{\mid { \jtr(Q, \db, \jgraph,  \tup')} \mid}}{\mid \pt{\query, \db, \tup_1} \mid} \\
\end{align*}
$$\Ffscore(\dexpl) =\frac{2}{\frac{1}{\FPrec(\dexpl)} +\frac{1}{\FRec(\dexpl)}} $$

\end{Definition}
}

\cut{
\begin{Definition}[Coverage and Recall]\label{def:coverage-recall}
Consider a database $D$, a 
query $\query$, a $\tup \in \query(\db)$, a join graph $\jgraph$, and an explanation $\sexpl = (\jgraph, \pat)$.
\begin{itemize}
    \item The \emph{fractional coverage} of $\sexpl$ is defined as $$\FCov(\sexpl) = \sum_{\tup' \in \pt{\query, \db, \tup}} \frac{\mid \matches \big(\pat,  \jtr(Q, \db, \jgraph, \tup') \big)\mid}{\mid { \jtr(Q, \db, \jgraph,  \tup')} \mid}$$
    \item The \emph{recall} of $\sexpl$ is defined as
$$\Rec(\sexpl) = \frac{\FCov(\sexpl)}{|\pt{\query, \db, \tup}|} $$
This recall definition captures the coverage of tuples in the provenance table and the property that a tuple in the provenance table $\pt{Q, D}$ can generate multiple tuples in the augmented provenance table $\jtr(Q, \db, \jgraph,  \tup')$.
A high recall implies that the pattern $\pat$ is a good description of the tuples contributing to the output tuple $t$.   
\end{itemize}
\end{Definition}

\begin{Definition}[Precision]\label{def:precision}
\item The \emph{precision} of $\sexpl$, $\Prec(\sexpl)$, is defined as the fraction of tuples in $\jtr(\query, \db, \jgraph, \tup)$ matching $\pat$:
$$\Prec(\sexpl) = \frac{\FCov(\sexpl)}{\sum_{\tup \in \pt{\query, \db}}\mathbb{1}[ \tup \in \projection_{\att_{\ptlabel}}\matches \big(\pat, \jtr(Q, \jgraph, \db, \tup) \big) ]} $$

   The \emph{precision} of $\sexpl$, $\Prec(\sexpl)$, is defined as the fraction of tuples in $\pt{\query, \db, \tup}$ \SR{should be APT?}
    \ZM{should be PT. E.g., to summarize GSW wins, the precision of pattern (Curry pts > 30) should be \\
    $\frac{|\text{GSW win} \land  \text{Curry pts }> 30|}{|\text{GSW games with Curry pts} > 30|}$, showing how accurate the pattern is in predicting GSW win, or differentiating GSW win games and other games.}
    matching $\pat$ after being augmented using the join graph: 
    $\Prec(\dexpl) = \frac{TP}{TP+FP}$ , where the numerator TP (true positive) is the total number of tuples in $\pt(Q, \db, \tup)$ matching $\pat$ after augmented $$TP = \sum_{\tup \in \pt{\query, \db, \tup_1}} \mathbb{1}[\tup \in \projection_{\att_{\ptlabel}} \matches \big(\pat,  \jtr(Q, \jgraph, \db, \tup) \big)]$$
    and the denominator TP + FP (true positive + false positive) is the total number of tuples in the same context of $\tup_1$ matching $\pat$ after augmented. When the user query \query contains selection predicates in its {\tt WHERE}-condition $\theta$, the context of $\tup_1$ is defined as the all tuples in the provenance of other outputs of \query.
    $$TP+FP = \sum_{\tup \in \pt{\query, \db}}\mathbb{1}[ \tup \in \projection_{\att_{\ptlabel}}\matches \big(\pat, \jtr(Q, \jgraph, \db, \tup) \big) ]$$
    While when the user query \query\ does not have any selection predicates, the result of \query\ would enumerate all value combinations of grouping attributes, and thus we only consider $\tup_1$ itself as its context, leading to an 1 precision. A high precision implies that the pattern $\pat$ is a good discriminator for tuples contributing to the output tuple $t$ and other output tuples when the context is given.
    \end{Definition}

    \SR{don't follow this -- we talked about the impact of `WHERE' -- no where means precision 1, otherwise we need to divide by something.}
    \ZM{Correct me if wrong: if the query selects on `winner = GSW', grouping on team and season, then the context should be GSW wins in every season, and we are summarizing the difference between season 2015-16 and other seasons. }

\begin{Definition}

    As usual, the \abbrF of $\sexpl$, denoted $\fscore(\sexpl)$, is the harmoning mean of $\Rec(\sexpl)$ and $\Prec(\sexpl)$.

\end{Definition}
}
\cut{
\begin{Definition}[Quality Metrics for Difference]\label{def:quality-metrics-diff}
Let $\dexpl = (\jgraph, \pat)$ be an explanations for $\query$,  $\db$, and two tuples $\tup_1, \tup_2$, the quality of $\dexpl$ is also measured using $F_1$ score, where the precision is defined as the percentage of tuples in $\pt{\query, \db, \tup_1}$ matching $\pat$ in the join graph result over those matched tuples in $\pt{\query, \db, \tup_1} \cup \pt{\query, \db, \tup_2}$.


$Pre(\dexpl) = \frac{TP}{TP+FP}$ , where $$TP = \sum_{\tup \in \pt{\query, \db, \tup_1}} \mathbb{1}[\tup \in \projection_{\att_{\ptlabel}} \matches \big(\pat,  \jtr(Q, \jgraph, \db, \tup) \big)]$$
$$FP =  \sum_{\tup \in \pt{\query, \db, \tup_2}}\mathbb{1}[ \tup \in \projection_{\att_{\ptlabel}}\matches \big(\pat, \jtr(Q, \jgraph, \db, \tup) \big) ]$$

where $\mathbb{1}[]$ is the indicator function, $\att_{\ptlabel}$ is the schema of the provenance table $\pt{\query, \db, \tup_1}$; the recall is defined as
$$
Rec(\dexpl) = \sum_{\tup \in \pt{\query, \db, \tup_1}} \frac{\mid \matches \big(\pat,  \jtr(\jgraph, \db, \tup) \big)\mid}{\mid { \jtr(\jgraph, \db, \tup)} \mid}
$$


\end{Definition}
}


\cut{
\subsubsection{Quality Measures Over All Join Graphs}
\SR{I HAVE NOT EDITED THIS YET}

\CL{This section is not implemented yet, the current framework only generates the Fscores, without considering the effects of the number of edges in the score}
Furthermore, we prefer explanations that are simpler, i.e., that have smaller join graphs. We incorporate these requirements into a quality measure for explanations, since using only \fscore\ may fail to reveal such simple but informative patterns. Similar to the widely-used weight function that considers the number of non-$\placeh$ values in the pattern in in prior data summarization works, we consider the number of edges in the schema join graph as the weight.

\begin{Definition}[Score Function]\label{def:score-function-all}
The score of an explanation $\sexpl = (\jgraph, \pat)$ is defined as
$$Score(\sexpl) = \fscore(\sexpl) \times \frac{B-|\jgraph.E|+1}{B-A+1}$$,
where $\jgraph.E$ denotes set of edges of the schema join graph, $A = \min_{\jgraph'\in \mathcal{E}} |\jgraph'.E|$, $B = \max_{\jgraph'\in \mathcal{E}} |\jgraph'.E|$, $\mathcal{E}$ is the set of of all candidate explanations.
\end{Definition}

\BG{Add formal problem definition: Given database D, schema graph, two-point or one-point question, generate top-k explanations for all join graphs (or subset thereof)}
}





\vspace{-2mm}
\section{Mining Patterns given a Join Graph}
\label{sec:algo_pattern}

\iftechreport{
\begin{algorithm}[t]\caption{{\bf \algMineAPT to find top-$k$ patterns given a join graph $\jgraph$}\\
{\bf Other inputs:} Database $D$, Query $Q$, Answer tuples in user question $t_1, t_2$,  limit on the number of categorical attributes in the pattern $\maptKcat$, and input parameters $\lambda$s described in text. }\label{alg:basic-pattern-mine}

  {\small
    \begin{codebox}
      \Procname{\proc{\algMineAPT} ($\db, \jgraph, \query, \tup_1, \tup_2, k, \maptKcat, \revc{\paramAfilterrate}$)}
      \li $\explset \gets \emptyset$ \Comment{min-heap: to store top-$\topk$ explanations by their score}
      \li $\maptAPT \gets \jtr({\jgraph, \db})$
      \li $\maptSample \gets \Call{createSample}{\maptAPT}$ \label{alg-line:mapt-sample}
      \li $(\maptAnum, \maptAcat, \maptCluster) \gets \Call{\maptFilterAttrs}{\maptAPT, \revc{\paramAfilterrate}}$
      \li $\patLCA \gets \Call{LCA}{\maptSample, A_{cat}}$ \label{alg-line:mapt-lca} \Comment{from \cite{DBLP:journals/pvldb/GebalyAGKS14}}
      \li $\maptTODO \gets \Call{pickTopK}{\patLCA,\maptKcat,\maptAPT}$ \label{alg-line:cat-top-k}
      \li $\maptDONE \gets \emptyset$
      \li \While $\maptTODO \neq \emptyset$ \label{alg-line:mapt-refine}
         \Do
         \li $\maptPat \gets \Call{pop}{todo}$
         \li $\maptDONE \gets \maptDONE \cup \{ \maptPat \}$
         \li \For $\tup_{cur} \in \{ \tup_1, \tup_2 \}$ \Do
             \li \If $\Call{recall}{\maptPat,\tup_{cur}} > \paramRecallThresh$ \label{alg-line:mapt-filter-recall}\Then
                 \li \If $\Call{F1}{\maptPat,\tup_{cur}} > \Call{F1}{\textsc{peek}(\explset)}$ \Then
                     \li $\explset \gets \Call{deleteMin}{\explset}$
                     \li $\explset \gets \Call{insert}{\explset, (\maptPat, \tup_{cur})}$
                     \End
                 \li \For $A \in \maptAnum$ \Do
                     \li \If $\maptPat.A = \placeh$ \Then
                         \li \For $op \in \{=, \leq\}, c \in \Call{getFragments}{\domOf{APT}{A}}$ \Do
                              \li $\maptPatnew \gets \maptPat$
                              \li $\maptPatnew.A \gets \patop{op}{c}$
                              \li \If $\maptPatnew \not\in done$ \Then
                                  \li $\maptTODO \gets \maptTODO \cup \{ \maptPatnew \}$
                              \End
                         \End
                      \End
                  \End
             \End
          \End
      \End
    \li \Return $\explset$
    \end{codebox}
    \begin{codebox}
      \Procname{$\proc{\maptFilterAttrs}(\maptAPT, \paramAfilterrate)$}
      \li $A \gets \schemaOf{\maptAPT}$
      \li $\maptAfilter \gets \Call{filterBasedOnRelevance}{\maptAPT, \paramAfilterrate}$ \label{alg-line:mapt-filter-attrs-relevance}
      \li $\maptCluster \gets \Call{clusterAttributes}{\maptAfilter, \maptAPT}$ \label{alg-line:mapt-filter-attrs-cluster}
      \li $\maptArepresent \gets \Call{pickClusterRepresentatives}{\maptCluster, \maptAPT}$ \label{alg-line:mapt-filter-attrs-pick-representatives}
      \li $\maptAnum \gets \{ A \mid A \in \maptArepresent \land A \mathtext{is numeric} \}$
      \li $\maptAcat \gets \{ A \mid A \in \maptArepresent \land A \mathtext{is categorical} \}$
      \li \Return $(\maptAnum, \maptAcat, \maptCluster)$
    \end{codebox}
    }
\end{algorithm}
}

In this section, we give an overview of
our algorithm 
for mining patterns from an augmented provenance table
(\abbrAPT) $\jtr(Q, \db, \jgraph)$ generated based on a given join graph $\jgraph$. Recall that we are dealing with patterns that may contain
equality comparisons (for categorical attributes) and/or inequality comparisons (for numeric attributes). 
We mine patterns in multiple
phases. (i) In a preprocessing step 
Then we cluster attributes that
are highly correlated to reduce redundancy in the generated patterns. The
output of this preprocessing step are the generated clusters and one
representative selected for each cluster. Then,
we use random forests to determine the
relevance of each attribute on predicting tuples to belong to the provenance of
one of the two data points from the user's question. The purpose of this step is to remove attributes that
are unlikely to yield patterns of high quality.   (ii) In the next phase of pattern mining
we only consider categorical attributes and mine pattern candidates using a
variation of the \emph{LCA (Least Common Ancestor)} method from~\cite{DBLP:journals/pvldb/GebalyAGKS14} that can only handle categorical attributes. From the set of patterns returned by the LCA method, we then select the  $\maptKcat$ patterns with the highest scores for the next step. 
(iii) These patterns are then \emph{refined}
by adding conditions on
numerical attributes that can improve precision at the potential cost of
reducing recall. (iv) The refined patterns are filtered to remove patterns with
recall below a threshold $\paramRecallThresh$
and are ranked by their score according
to~\Cref{def:metrics}. Finally, the top-k patterns according to this
ranking are returned. 

Before describing the individual steps of our pattern mining algorithm, we first
need to introduce additional notation used in this section. Given a pattern
$\pat$, we call a pattern $\pat'$ a {\bf refinement} of $\pat$ and $\pat'$ can
be derived from $\pat$ by replacing one or more placeholders ($\placeh$) with
comparisons. For instance, pattern $\pat_2 = (\peq{X}, \pleq{Y})$ is a
refinement of $\pat_1 = (\peq{X}, \placeh)$.
The following observation holds:
  \begin{proposition}
  Using Definition~\ref{def:metrics}, given a tuple $t \in Q(D)$ and a join graph $\jgraph$, $\Rec(\sexpl_2, \jgraph,  \tup)
\leq \Rec(\sexpl_1, \jgraph,  \tup)$, where $\sexpl_1 = (\jgraph, \pat_1)$, $\sexpl_2 = (\jgraph, \pat_2)$, and $\pat_2$ is a refinement of $\pat_1$.
  \end{proposition}
  \begin{proof}
  Following Definition~\ref{def:metrics} (a), if a $\tup' \in \pt{\query, \db, \tup}$ is covered by $\sexpl_2$, then it is also covered by $\sexpl_1$. Hence, by (b) and (d),  $\ITP(\sexpl_1, \jgraph, \tup) \geq \ITP(\sexpl_2, \jgraph, \tup)$ and $\IFN(\sexpl_1, \jgraph, \tup) \leq \IFN(\sexpl_2, \jgraph, \tup)$. Therefore, using (e), $\Rec(\sexpl_2, \jgraph,  \tup)
\leq \Rec(\sexpl_1, \jgraph,  \tup)$.
  \end{proof}
We exploit the above fact to
exclude patterns and their refinements early on in the process if their recall
is below a threshold.

\subsection{Clustering and Filtering Attributes}
\label{sec:filt-clust-attr}

Before mining patterns, we analyze the attributes of an \abbrAPT to (1) determine the attributes that are unlikely to contribute to top-k patterns because they are not helpful in distinguishing between the two outputs of interested provided as part of the user's question, and (2) to identify clusters of attributes with strong mutual correlations, because such attributes can lead to redundancy in explanations as discussed below.

\mypar{Clustering Attributes based on Correlations}
Redundancy in patterns can be caused by attributes that are highly correlated. As an extreme example, consider an \abbrAPT containing both the birth date and age of a person. For any pattern containing a predicate on birth date there will be an (almost) equivalent pattern using age instead and a pattern restricting  both age and birth date. To reduce the prevalence of such redundant patterns, we cluster attributes based on their mutual correlation
and pick a single representative for each cluster
We use VARCLUS~\cite{sarle1990sas}, a clustering algorithm closely related to principal component analysis and other dimensionality reduction techniques~\cite{roweis2000nonlinear}. However, any technique that can cluster correlated attributes would be applicable.

\mypar{Filtering Attributes based on Relevance}
Random forests have been successfully used in machine learning applications  to determine the relevance of a feature (attribute) to the outcome of a classification task.
We train a random forest  classifier
that predicts whether a row belongs to the augmented provenance of one of the two outputs from the user's question~\cite{breiman2001random}. 
We then rank attributes based on the relevance and find the fraction $\paramAfilterrate$ of attributes with the highest relevance ($\paramAfilterrate$ is a threshold used by the system).
The rationale for this step is to avoid generating patterns involving attributes that are irrelevant for distinguishing the two output tuples $t_1, t_2$ in the user question. This reduces the search space for patterns and additionally has the advantage of excluding attributes that are mostly constant in the rows contributing to the two outputs. Such attributes can be added to any pattern with minimal effect on the recall and precision of patterns, since they essentially do not affect the matches for the patterns. This could mislead users into thinking that the value of this attribute is a distinguishing factor for the tuples in the user question when in fact the value of this attribute has very limited or no effect.

\subsection{
Patterns over Categorical Attributes}
\label{sec:gener-patt-over}

We then generate a sample of size $\paramPatsamplerate$ (an input parameter) from 
$\jtr(Q, \db, \jgraph)$
and
generate a set of candidate patterns over categorical attributes (ignoring all
numerical attributes at this stage) using the LCA method
from~\cite{DBLP:journals/pvldb/GebalyAGKS14}
The LCA method generates pattern
candidates from a sample by computing the cross product of the sample with
itself.
A candidate pattern is generated for each pair $(t,t')$ of tuples from
the sample by replacing values of attributes $A$ where $t.A \neq t'.A$ with a
placeholder $\placeh$ and by keeping constants that $t$ and $t'$ agree upon
($t.A = t'.A$). Note that in our case each element of a pattern is a
predicate. Thus, using a constant as done in the LCA method corresponds to using
an equality predicate, i.e., $A = c$ for $t.A = t'.A = c$. By keeping constants
that frequently co-occur, the LCA method that only
works for categorical attributes (equality comparisons) generates patterns that reflect common
combinations of constants in the data.  Note that we ignore numeric attributes
at this step, i.e., we use $\placeh$ for all numeric attributes.  The
rationale of focusing on categorical attributes first is that we can (i) use the
established heuristic of the LCA method to generate pattern candidates from categorical attributes, and (ii) we can
significantly reduce the search space by pruning all refinements of patterns
that fail to have a sufficiently high recall.
\subsection{Filtering Categorical Pattern Candidates}
\label{sec:filt-categ-patt}

Next we calculate the recall for each pattern by filtering the input
table to determine the matches of the pattern. As an optimization, we can calculate the
recall over a sample of the data (using a sample size parameter
$\paramQualitysamplerate$). 
This may require using a separate sample size ($\maptKcat$) since we found that a small sample is sufficient for generating a meaningful set of
patterns, but may not be sufficient for estimating recall with high enough
accuracy.
Irrespective of whether the recall is estimated or calculated precisely, we then filter
out patterns whose recall is below a threshold $\paramRecallThresh$. Out of these
patterns, we then pick the top patterns based on their recall.
In the following we use $\patLCA$ to denote the set of patterns
that are returned by this step.

\subsection{Refinement and Numeric Attributes}
\label{sec:pattern-refinement}
In the next step, we then generate refinements of patterns from $\patLCA$ by replacing placeholders on numerical attributes with predicates. Even though such refinements can at best have the same recall as the pattern they originate from, their precision may be higher resulting in greater \abbrFs. Recall that for numerical attributes we allow for both equality as well as inequality predicates. Domains of numerical attributes are typically large, resulting in large number of possible constants to use in inequality comparisons. To reduce the size of the search space, we split the domain of each numerical attribute into a fixed number $\paramNumRefineFrags$ of
fragments (e.g., quartiles) and only use boundaries of these fragments when generating refinements. For example, for $\paramNumRefineFrags = 3$ we would use the minimum, median, and maximum value of an attribute's domain.
We systematically enumerate all refinements of a pattern by extending it by one predicate at a time. For each such refinement we calculate its recall as described above. Patterns whose recall is below $\paramRecallThresh$ are not further refined. We use $\patset_{refined}$ to denote the union of  $\patset_{LCA}$ with the set of patterns generated in this step.

\subsection{Computing Top-k Patterns}
\label{sec:computing-top-k}

\newcommand{\wscore}{\textsc{wscore}}
\newcommand{\dscore}{\mathcal{D}}
\newcommand{\patres}{\mathcal{R}}
\newcommand{\dmatchs}{\textsc{matchscore}}

Finally, we calculate the \abbrFs for each pattern in $\patset_{refined}$ and return $k$ patterns.
\revm{To improve the diversity of the returned patterns we rank them based on a score that combines \abbrFs (\Cref{def:metrics}) with a diversity score $\dscore(\pat,\pat')$ that measure how similar two patterns are. The first pattern to be returned is always the one with the highest \abbrF. Then we determine the $i+1^{th}$ pattern based on this score. Note that the diversity score $\mathcal{D}(\pat)$ for a pattern $\pat$ depends on the set $\mathcal{R}$ of patterns we have selected to far, because it measure how close the pattern $\pat$ is to the most similar pattern in $\mathcal{R}$. We calculate this as a score $\dscore(\pat,\pat')$ that ranges between -2 and 1 (larger scores means that the two patterns are more dissimilar). The formula for calculating the score is shown below.  For each attribute $A$ of pattern $\pat$ we add 1 if the attribute does not appear in $\pat'$, we add a penalty (-0.3) if the attribute appears in both patterns (but with different constants), and a larger penalty (-2) if it appears in both patterns with the same constant. We repeatedly add the pattern with the highest score to $\mathcal{R}$ until we have $k$ patterns to return.} 

 {\small
\revm{
  \[ \wscore(\pat) = \fscore(\sexpl, \jgraph,  \tup_1, \tup_2) + \min_{\pat' \in \patres} \dscore(\pat,\pat') \]
  \[ \dscore(\pat,\pat') = \frac{\sum_{A: \pat.A \neq \placeh} \dmatchs(\pat,\pat',A) }{\card{\pat}}\]
  \[
    \dmatchs(\pat,\pat',A) =
    \begin{cases}
      1    & \pat'.A = \placeh        \\
      -0.3 & \pat.A[0] \neq \pat.A[0] \\
      -2   & \pat.A[0] = \pat.A[0] \\
    \end{cases}
  \]
}
}
We denote the result of this step as $\pattopk$.
%
\ifnottechreport{
To avoid repetitive attributes and values appear in the top-$k$ results, we implemented a pattern diversification approach, which would add penalties to the same attributes and same values that have appeared in the previous results within top-$k$ results.
}


\section{Join Graph Enumeration}\label{sec:join-graph-enumeration}

  \newcommand{\jgZero}{\jgraph_0}
  \newcommand{\algjgJSize}{JSize}
  \newcommand{\algjgIsValid}{isValid}
  \newcommand{\algjgAddEdge}{AddEdge}
  \newcommand{\algjgGenNewJG}{ExtendJG}
  \newcommand{\algjgExpls}{\mathcal{E}}
  \newcommand{\algjgPK}{getPK}
  \newcommand{\algjgnone}{n_1}
  \newcommand{\algjgntwo}{n_2}
  \newcommand{\algjgRels}{rels}
  \newcommand{\algjgStart}{start}
\iftechreport{
\begin{algorithm}
  \caption{Join Graph Enumeration}\label{alg:join-graph-enumeration}
{\small
  \begin{codebox}
    \Procname{\proc{\JGGen}$(\jgZero, \asgraph, \query,k, k_{cat})$}
    \li $\algjgExpls \gets \emptyset$ \Comment{maps join graphs to explanations}
	\li $\jgraphs{} \gets \{\jgZero\}$ \Comment{generated join graphs}
	\li $\jgraphs{prev} \gets \{\jgZero\}$ \Comment{join graphs from the previous iteration}
	\li \While $\algjgJSize \in \{ 1, \ldots, \paramJoinGraphSize \}$ \label{alg-line:jge-main-loop}
	\li \Do $\jgraphs{new} \gets \emptyset$ \Comment{join graphs generated in this iteration}
	\li \For $\jgraph \in \jgraphs{prev}$ \Do
    \li $\jgraphs{new} \gets \jgraphs{new} \cup \proc{\algjgGenNewJG}(\jgraph,\asgraph, \query$)  \label{alg-line:jge-extend-jg}
    \End
        \li \For $\jgraph \in \jgraphs{new}$ \Do
    \li \If $\Call{\algjgIsValid}{\jgraph}$ \Then \label{alg-line:jge-test-valid}
    \li $\algjgExpls[\jgraph] \gets \Call{\algMineAPT}{\db, \jgraph, \query, t_1, t_2, k, k_{cat}}$
    \End
    \End
    \li $\jgraphs{} \gets \jgraphs{} \cup \jgraphs{new}$
    \li $\jgraphs{prev} = \jgraphs{new}$
    \End
    \li \Return $\algjgExpls$
      \end{codebox}
      	\begin{codebox}
	\Procname{\proc{\algjgGenNewJG}($\asgraph, \jgraph, \query$)}
	\li $\jgraphs{new} \gets \emptyset$
	\li \For $v \in \jgraphnodes$ \Do
	\li \If $\tlabelv(v) = \ptlabel$ \Then
    \li $\algjgRels \gets \relsinQ$
    \li \Else
    \li $\algjgRels \gets \tlabelv(v)$
    \End
	\li  \For  $r \in \algjgRels, n \in \snodes{}$ \Do
       \li $e \gets (r,n)$
        	\li \If $e \in \sedges{}$ \Then \label{alg-line:extjg-iterate-edges}
        	\li \For  $c \in \sglabel(e)$ \Do
                \li $\jgraphs{new} \gets \jgraphs{new} \cup \Call{\algjgAddEdge}{\jgraph, v, n, c}$
                \End
                \End
                \End
    \End
    \li \Return $\jgraphs{new}$
  \end{codebox}
   	\begin{codebox}
   	\Procname{\proc{\algjgAddEdge}($\jgraph, v, end, cond$)}
   	\li $\jgraphs{added} \gets \emptyset$
    \li $\jgraph_{curr} \gets \jgraph$ \Comment{Add $end$ as new node}
    \li $v_{new} \gets \Call{newNode}{}, e \gets (v,v_{new})$
    \li $\jgraph_{curr}.\jgraphnodes \gets \jgraph_{curr}.\jgraphnodes \cup \{v_{new}\}$
    \li $\jgraph_{curr}.\tlabelv(v_{new}) \gets end$
    \li $\jgraph_{curr}.\jgraphedges \gets \jgraph_{curr}.\jgraphedges \cup \{e\}$
    \li $\jgraph_{curr}.\tlabele(e) \gets cond$
    \li $\jgraphs{added} \gets \jgraphs{added} \cup \{\jgraph_{curr}\}$
   	\li \For $v' \in \jgraph.\jgraphnodes: \tlabelv(v') = end$ \Comment{Add edge connecting existing nodes} \Do
           	\li \If $\neg \exists e \in \jgraph.\jgraphedges: e = (v,v') \land \jgraph.\tlabele(e) = cond$ \Then
           	\li $\jgraph_{curr} \gets \jgraph$
            \li $e \gets (v,v')$
            \li $\jgraph_{curr}.\jgraphedges \gets \jgraph_{curr}.\jgraphedges \cup \{e\}$
            \li $\jgraph_{curr}.\tlabele(e) \gets cond$
           	\li $\jgraphs{added} \gets \jgraphs{added} \cup \{\jgraph_{curr}\}$
        \End
             \End
    \li \Return $\jgraphs{added}$
   	\end{codebox}
      \begin{codebox}
        \Procname{\proc{\algjgIsValid}$(\jgraph)$}
        \li \For $v \in \jgraphnodes$ \Do
        \li \For $A \in \Call{\algjgPK}{\tlabelv}$ \Do
        \li \If $\neg \exists e \in \jgraphedges: A \in \attrs(\tlabele(e))$ \Then
        \li \Return $\bfalse$
        \End
        \End
        \End
        \li \If $\Call{estimateCost}{\jgraph} > \paramMaxQcost$ \Then
        \li \Return $\bfalse$
        \End
        \li \Return $\btrue$
      \end{codebox}
}
\end{algorithm}
}

In this section,
we describe an algorithm
\iftechreport{present in  \Cref{alg:join-graph-enumeration}}
that enumerates join graphs  of iteratively increasing size.  In iteration $i$, we enumerate all join
graphs with $i$ edges by adding a single edge conforming to the schema graph to
one of the join graphs of size $i-1$ produced in the previous iteration. The
maximum size of join graphs considered by the algorithm is determined by
parameter \paramJoinGraphSize. We employ several heuristic tests to determine
whether a join graph generated by the algorithm should be considered for pattern
mining. The rationale for not considering all join graphs for pattern mining is
that pattern mining can be significantly more expensive than just generating a
join graph, so  
we skip pattern mining for join graphs
that are either unlikely to yield patterns of high quality or for which
generating the \abbrAPT and patterns are likely to be expensive.  For join
graphs that pass these tests we materialize the corresponding \abbrAPT and
apply the pattern mining algorithm from \Cref{sec:algo_pattern} to compute the
top-k patterns for the \abbrAPT.

\label{sec:gener-join-graphs}

\mypar{Generating Join Graphs} We enumerate join graphs in iteration $i$ by extending every join graph produced in iteration $i-1$ with all possible edges. For each join graph $\jgraph$ produced in iteration $i-1$, we consider two types of extensions: (i) we add an additional edge between two existing nodes of the graph and (ii) we add a new node and connect it via a new edge to an existing node.
\iftechreport{
The main entry point is function \textsc{\JGGen} from \Cref{alg:join-graph-enumeration}. We first initialize the result to contain $\jgZero$, the join graph consisting of a single node labeled ($\ptlabel$).
Within each iteration of the main loop (line~\ref{alg-line:jge-main-loop}), each join graph $\jgraph$ from the set of join graphs generated in the previous iteration ($\jgraphs{prev}$) is passed to function \proc{\algjgGenNewJG} which computes the set of all possible extensions of $\jgraph$ (line~\ref{alg-line:jge-extend-jg}). For any such join graph produced in this step, we then check whether it should be considered for patterning mining using function described below (line~\ref{alg-line:jge-test-valid}). For any join graph passing this check we call \proc{\algMineAPT} to generate the top-k patterns and store them in $\algjgExpls$.


Function \proc{\algjgGenNewJG} enumerates all possible extension of a join graph with one additional edge. For that we consider all existing nodes $v$ in the join graph as possible extension points. We then enumerate all edges from the schema graph $\sgraph$ that are adjacent to nodes $r$ for relations that are represented by $v$. For that we have to distinguish two cases.  If $v$ is labeled $\ptlabel$, then $r$ can be any relation accessed by the query. Otherwise, there is a single node $r$ that is determined based on $v$'s label (the relation represented by $v$). For each edge $e$ adjacent to a node $r$ (line~\ref{alg-line:extjg-iterate-edges}), we then iterate over all conditions $c$ from the label of $e$ and then use function \proc{\algjgAddEdge} to enumerate all join graphs generated from $\jgraph$ by connecting $v$ with an edge labeled with $c$ to another node (either already in the join graph or new node added to the graph).

Function \proc{\algjgAddEdge} first generates a join graph by adding a new node $v_{new}$ labeled $end$ (the relation that is the end point of the edge in the schema graph) and connects it to $v$. Afterwards, for each node in the join graph $\jgraph$ labeled $end$ that is not connected to $v$ though an edge labeled $cond$ already, we generate a new join graph by adding such an edge.
}




\mypar{Checking Join Graph Connectivity and Skipping Expensive \abbrAPT Computations}
We filter out join graphs based on \emph{lack of connectivity} or on \emph{high estimated computation costs}.

\iftechreport{
We use the function \proc{\algjgIsValid} shown in~\Cref{alg:join-graph-enumeration} to filter out join graphs based on lack of connectivity (described in the following) or based on high estimated computation costs.
}

Note that schema graphs may contain tables with multiple primary key attributes
that are connected to edges which join on part of the key. This is typical in
``mapping'' tables that represent relationships. For instance, consider the
{\tt PlayerGameScoring} table from our running example that stores the number
of points a player scored in a particular game. Assume that there exists another
table player that is not part of the running example.  The primary key of this
table consists of a foreign key to the game table and the name of the player for
which we are recording stats. For this current example also assume that the
schema graph permits that {\tt PlayerGameScoring} to be joined with the {\tt Player} table.
Consider a query that joins the {\tt Game} table with {\tt LineupPerGameState},
{\tt LineupPlayer}, and {\tt Games} tables and selects games played by team GSW. A valid
join graph for this query would be to join the node $\ptlabel$ with
{\tt PlayerGameScoring} on the primary key of game. Note that while the result of the
query contains rows that pair a player of GSW with one lineup they played in a
particular game, the \abbrAPT would pair each row with any player that played in
that game irrespective of their team. Join graphs like this can lead to
redundancy and large \abbrAPT tables. One reason for this redundancy is that not
all primary key attributes of the {\tt PlayerGameScoring} table are joined with
another table. To prevent such join graphs that because of their redundancy
often lead to large \abbrAPT tables,
our algorithm checks that
for every node in the join graph, the primary key attributes of the relation
corresponding to that node are joined with at least one other node from the join
graph. For instance, in the example above the join graph could be modified to
pass this check by also joining the {\tt PlayerGameScoring} table with the Player
table.

Even though we filter join graphs that are not fully connected, some generated
join graphs will result in \revc{\abbrAPT}s of significant size which are expensive to
materialize and have significantly high cost for pattern mining. Recall that we
use queries to materialize join graphs.  We use the DBMS to estimate the cost of this query upfront. We
skip pattern mining for join graphs where the estimated cost of this query is
above a threshold $\paramMaxQcost$. While we may lose explanations by skipping
such join graphs, experimental results demonstrate that it is necessary for
reasonable performance to apply this check. \reva{Further,
keeping the join graph size relatively small is likely not to overload the user with too complex explanations.
}



\mypar{Ranking Results} After we have enumerated join graphs and have computed top-k patterns for each join graph, we rank the union of all pattern sets based on their \abbrF. We decided to rank patterns 
to reduce the load on the user by increasing the likelihood that good patterns are shown early, but without having the risk of completely filtering out patterns that have lower scores.

\section{Experimental Evaluation}\label{sec:experiments}
In this section, \revs{we evaluate the implementation of our algorithms and \revs{optimizations} in \oursystem{}.} We evaluate both performance in terms of runtime and quality of results with respect to different parameters and \revs{compare against systems from related work}. 

\mypar{Datasets} We used the same \textit{NBA} and \textit{MIMIC} dataset as described in \Cref{sec:case-study}. We created several scaled versions of these two datasets preserving the relative sizes of most tables and join results.  
\iftechreport{In detail, subsets of the dataset \reva{were created} by random sampling, preserving relative sizes of most tables and ensuring that the result sizes of joins are scaled appropriately. Similarly, for scaling up the dataset size we duplicate rows appending identifiers to primary key columns and other selected columns to ensure that the constraints of the schema are not violated and the join result sizes are scaled too. We use scale factors from 0.1 ($\sim$ 17MB) up to 10.0 ($\sim$ 1.7GB).}
\iftechreport{
\begin{figure*}
    \begin{minipage}{\linewidth}
      \centering
      \includegraphics[width=\linewidth]{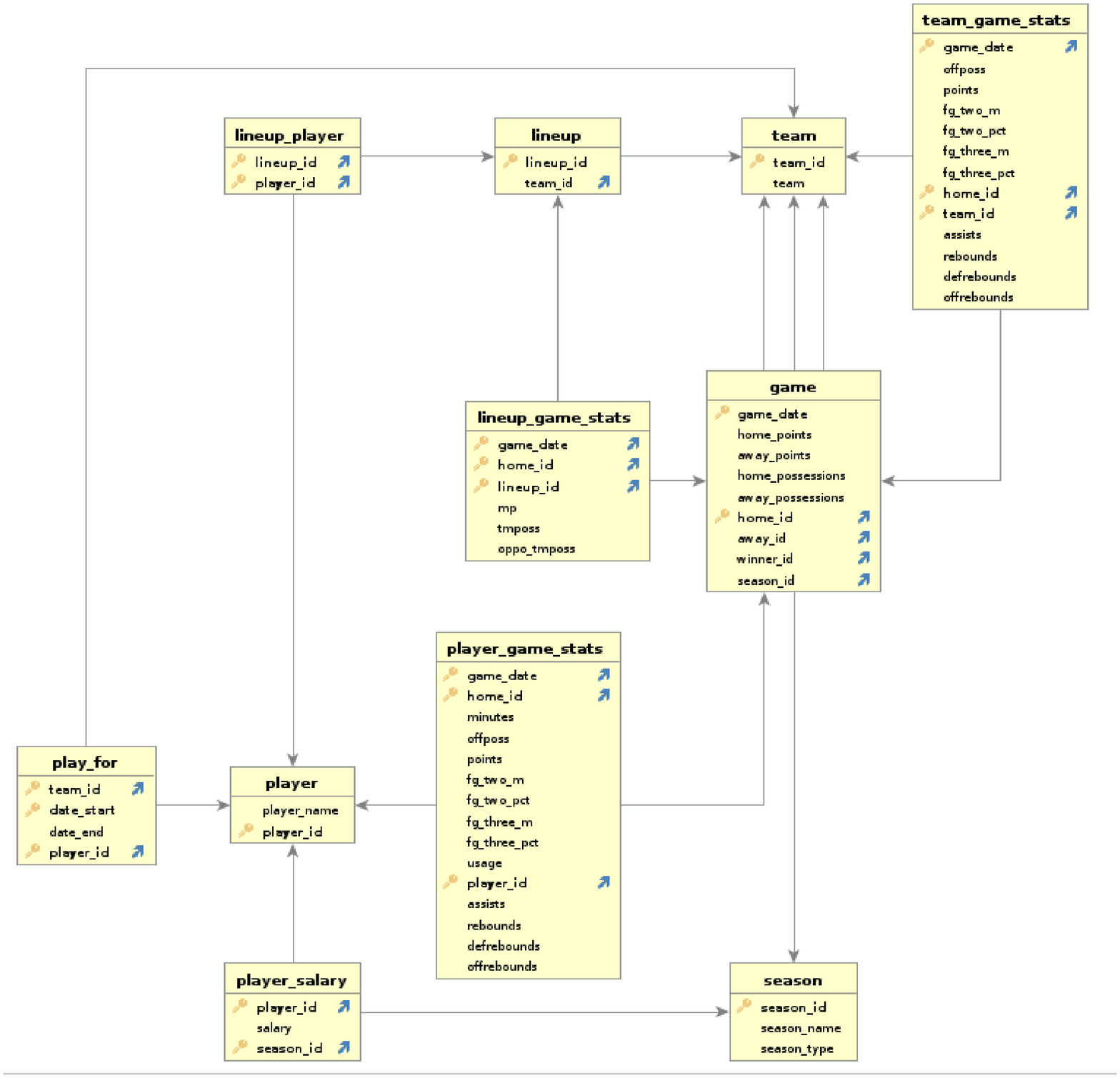}
      \caption{Schema Graph for NBA}
      \label{fig:schema_nba}
    \end{minipage}
\end{figure*}

\Cref{fig:schema_nba} shows the schema of the NBA dataset. Primary key attributes are marked with a key symbol and foreign key attributes with a blue arrow on the right of an attribute's name.

\mypar{Game}
Table \textbf{game} stores information about games including the date when the game took place (\texttt{game\_date}), the number points scored by the home and visiting team (\texttt{home\_points} and \texttt{away\_pointsw}), the number of ball possessions by each team (\texttt{home\_possessions} and \texttt{away\_possessions}), the ids of the home, visiting, and winning team (\texttt{home\_id}, \texttt{away\_id}, \texttt{winner\_id}), and the season when the team took place (\texttt{season\_id}).
Games are uniquely identified by their date and the home team.

\mypar{team}
Table \textbf{team} stores an artificial key (\texttt{team\_id} and the name (\texttt{team}) of NBA teams.

\mypar{player}
Table \textbf{player} records the name (\texttt{player\_name}) and artificial identifiers (\texttt{player\_id}) of NBA players.

\mypar{player\_salary}
Table \textbf{player\_salary} stores the \texttt{salary} a player is earning in a particular season.

\mypar{play\_for}
Table \textbf{play\_for} stores which player played for which team for which time period (\texttt{start\_date} to \texttt{end\_date}).

\mypar{line\_up and lineup\_player}
Table \textbf{lineup} records lineups. Each lineup is a set of $5$ players (table \textbf{lineup\_player}) from a team that are together on the field during a game.

\mypar{team\_game\_stats}
This table stores statistics related to the performance of a team in a particular game. The following statistics are reported. The number of points scored by the team (\texttt{points}), offensive possesions (\texttt{offposs}), number of field goals made (\texttt{fg\_two\_m}),
two point field goal percentage (\texttt{fg\_two\_pct}), three point field goal made(\texttt{fg\_three\_m}), three point percentage (\texttt{fg\_three\_pct}), team assists total(\texttt{assists}), team total rebounds (\texttt{rebounds}), team defensive rebounds(\texttt{defrebounds}), team offensive rebounds (\texttt{offrebounds}). We used this simplified version of team\_game\_stats table in user study. In experiments we used a richer number of columns for this table and because of the space constraint we only report the list of their names here:
\texttt{\seqsplit{fg\_two\_a, fg\_three\_a, nonheavefg\_three\_pct, ftpoints, ptsassisted\_two\_s, ptsunassisted\_two\_s, ptsassisted\_three\_s, ptsunassisted\_three\_s, assisted\_two\_spct, nonputbacksassisted\_two\_spct, assisted\_three\_spct, fg\_three\_apct, shotqualityavg, efgpct, tspct, ptsputbacks, fg\_two\_ablocked, fg\_two\_apctblocked, fg\_three\_ablocked, fg\_three\_apctblocked, assistpoints, two\_ptassists, three\_ptassists, atrimassists, shortmidrangeassists, longmidrangeassists, corner\_three\_assists, arc\_three\_assists, ftdefrebounds, defftreboundpct, def\_two\_ptrebounds, def\_two\_ptreboundpct, def\_three\_ptrebounds, def\_three\_ptreboundpct, deffgreboundpct, ftoffrebounds, offftreboundpct, off\_two\_ptrebounds, off\_two\_ptreboundpct, off\_three\_ptrebounds, off\_three\_ptreboundpct, offfgreboundpct, defatrimreboundpct, defshortmidrangereboundpct, deflongmidrangereboundpct, defarc\_three\_reboundpct, defcorner\_three\_reboundpct, offatrimreboundpct, offshortmidrangereboundpct, offlongmidrangereboundpct, offarc\_three\_reboundpct, offcorner\_three\_reboundpct}}

\mypar{lineup\_game\_stats}
This table stores statistics about the performance of a lineup in a game.
total number of minutes the lineup played in the game (\texttt{mp}), total number of possessions this lineup have in the game (\texttt{tmposs}), total number of opponent lineup possesions in the game (\texttt{oppo\_tmposs}).

\mypar{player\_game\_stats}
This table records statistics about a player's performance in a game. The most of the attributes contained in this table has same names and meanings as those in team\_game\_stats except it has a unique attribute \texttt{usage} which describes  what percentage of team plays a player was involved in while he was on the floor. Similarly, we used this simplified table in user study and richer version (please refer to team\_game\_stats for details) in the experiments.

\begin{figure*}
  \begin{minipage}{\linewidth}
    \centering
    \includegraphics[width=\linewidth]{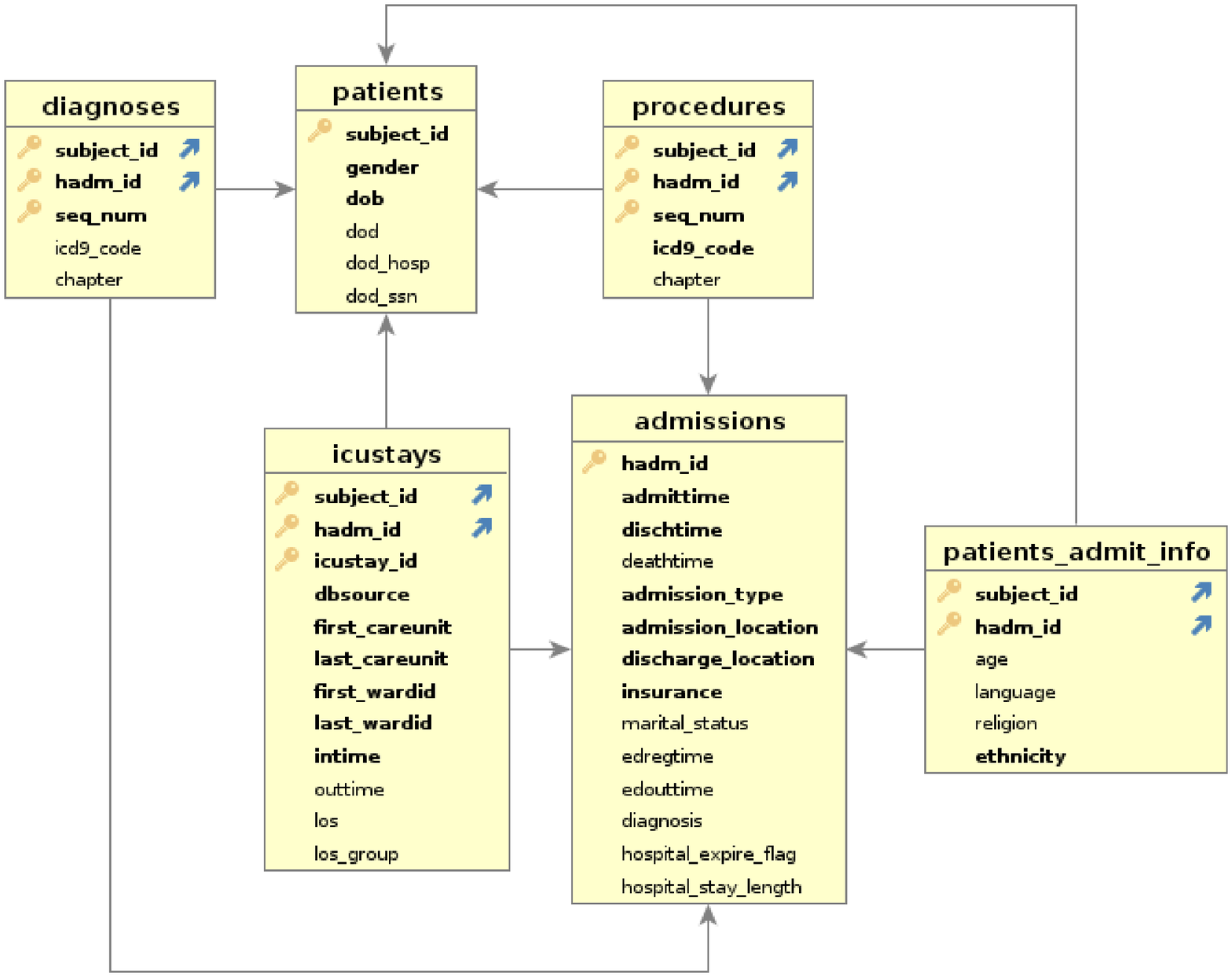}
    \caption{Schema Graph for MIMIC}
    \label{fig:schema_mimic}
  \end{minipage}
\end{figure*}

\Cref{fig:schema_mimic} shows the schema of the MIMIC dataset.

\mypar{admissions}
This table stores information about hospital admissions.
\texttt{hadm\_id} is a unique identifier for an admission.
\texttt{admittime} is the time when the patient got admitted.
\texttt{dischtime} is the time when the patient was discharged from the hospital.
\texttt{admission\_type} describes the type of the admission: \textit{ELECTIVE}, \textit{URGENT}, \textit{NEWBORN} or \textit{EMERGENCY}.
\texttt{admission\_location} provides information about the previous location of the patient prior to arriving at the hospital.
\texttt{discharge\_location} provides information about the location of the patient after visiting the hospital.
\texttt{insurance} is the type of insurance the patient has.
\texttt{martial\_status} is the patient's martial status.
\texttt{edregtime} is the time that the patient was registered in the emergency department
\texttt{edouttime} is the time that the patient was discharged from the emergency department.
\texttt{diagnosis} provides a preliminary, free text diagnosis for the patient on hospital admission
\texttt{hospital\_expire\_flag} indicates whether the patient died within the given hospitalization. 1 indicates death in the hospital, and 0 indicates survival to hospital discharge.
\texttt{hospital\_stay\_length} Is the length of the patient's stay in the hospital in days.

\mypar{procedures}
Table procedures records information of medical procedures for each patient.
\texttt{seq\_num} is part of the primary key in procedure table to differentiating one procedure from another during one admission for the same patient.
\texttt{icd9\_code}: ICD-9 Procedure Codes to represent the procedure
\texttt{chapter}: a broader category which contains certain \textit{icd9 codes} as a group.

\mypar{patients}
This table records information about patients.
\texttt{subject\_id} is a unique identifier for a patient.
For each patient we record their \texttt{gender} and date of birth (\texttt{dob}). \texttt{dod}is the date of death for the given patient. \texttt{dod\_hosp} is the date of death as recorded in the hospital database. \texttt{dod\_ssn} is the date of death from the social security database

\mypar{patients\_admit\_info}
This table records additional information about patients at the time of an hospital admission. For each hospital admission the patient's \texttt{age}, \texttt{language} of choice, \texttt{religion}, and \texttt{ethnicity} are recorded.

\mypar{icustays}
This table records information about intensive care unit (ICU) stays of a patient during an admission. A patient may stay multiple times in ICU during an admission.
\texttt{dbsource}: ‘carevue’ indicates the record was sourced from \textit{CareVue}, while \textit{metavision} indicated the record was sourced from Metavision.

\texttt{first\_careunit} and \texttt{last\_careunit} contain, respectively, the first and last ICU type in which the patient was cared for. \texttt{first\_wardid} and \texttt{last\_wardid} contain the first and last ICU unit in which the patient stayed. \texttt{intime} provides the date and time the patient was transferred into the ICU. \texttt{outtime} provides the date and time the patient was transferred out of the ICU. \texttt{los} is the length of the patient's stay in intensive care. \texttt{los\_group} is a categorized length of stay where it is divided into $5$ groups.

\mypar{diagnoses}
This table stores information about diagnosis for patients for a particular admission.
\texttt{seq\_num} is part of the primary key in procedure table to differentiating one diagnosis from another during one admission for the same patient.
\texttt{icd9\_code} is ICD-9 Diagnoses Codes to represent the procedure
\texttt{chapter}  a broader category which contains certain \textit{icd9 codes} as a group.
    }


\mypar{Experimental setup}
 \oursystem{} is implemented in Python (version 3.6) and runs on top of PostgreSQL (version 10.14). All experiments were run on a machine with 2 x AMD Opteron 4238 CPUs, 128GB RAM, and 4 x 1 TB 7.2K RPM HDDs in hardware RAID 5.

 \begin{table}[t]
   \centering
   \scriptsize
   \begin{tabular}{|l|l|c|}
     \hline
     {\bf Parameter} & {\bf Description} & \textbf{Default}  \\\hline
     \paramDBsize & the size of the database (scale factor) & \paramDefault{1.0} \\
     \paramJoinGraphSize & maximum number of edges per join graph (\Cref{sec:join-graph-enumeration}) & \paramDefault{3} \\
     \paramAfilterrate & \#attributes returned by feature selection (\Cref{sec:filt-clust-attr}) & \paramDefault{3} \\
     \paramAattrnum & max number of \revcj{numerical} attributes allowed in a pattern & \paramDefault{3} \\
     \paramPatsamplerate & sample rate for LCA pattern candidate generation
                           (\Cref{sec:gener-patt-over})  & \paramDefault{0.1}\\
     \revcj{\paramFscoresamplerate} & \revcj{sample rate for calculating \abbrFs of patterns} (\Cref{sec:filt-categ-patt}) & \paramDefault{0.3} \\
     \hline
   \end{tabular}
   \caption{Parameters of our approach and default values}
   \label{tab:exper_para}
    \vspace{-7mm}
 \end{table}

\mypar{Parameters and Optimizations}
\revm{
  \Cref{tab:exper_para} shows the parameters used in our experiments and their default values.
  We vary the following: (1) the size of the database; (2) the maximum number of join graph edges $\paramJoinGraphSize$; (3) sample rate for \abbrF  $\paramFscoresamplerate$;
  and (4)  the sample rate for pattern candidate generation ($\paramPatsamplerate$).
}
\ifnottechreport{
  \revs{
    In~\cite{chenjie2021} we also compare our approach with and without feature selection and evaluate how the maximum allowed number of edges in join graphs (\paramJoinGraphSize) affects performance. Based on these results we activated feature selection and set $\paramJoinGraphSize = 3$ for all experiments.
Unless stated otherwise we use queries $\query_1$ from \Cref{sec:introduction} (for NBA experiments) and $\query_{mimic4}$ from~\Cref{sec:qual-eval-mimic} (for MIMIC experiments) with their respective user questions and use the default values for all other parameters.
  }}

\iftechreport{
\begin{figure}[t]
  \centering
  \begin{minipage}{1\linewidth}
    \centering
    {\scriptsize
      \begin{tabular}{|c|r|r|r|r|r|}
        \hline
        \multirow{2}{*}{\textbf{Step}} & \multicolumn{4}{c|}{\textbf{feature sel.: \paramFscoresamplerate =}}  
                                       & \multirow{2}{*}{\textbf{w/o feature sel.}}                                                                                      \\ \cline{2-5}
                                       & \textbf{0.1}                                                            & \textbf{0.3} & \textbf{0.5} & \textbf{1.0} &          \\ \hline
        \textbf{\bdFeature}            & 84.96                                                                   & 87.39        & 86.86        & 84.80        & N/A      \\
        \textbf{\bdLCA}                & 9.43                                                                    & 9.21         & 9.25         & 9.31         & 9.39     \\
        \textbf{\bdF}                  & 33.19                                                                   & 91.50        & 149.21       & 226.53       & 16749.36 \\
        \textbf{\bdAPT}                & 21.96                                                                   & 21.29        & 20.87        & 20.47        & 20.51    \\
        \textbf{\bdrefine}             & 15.94                                                                   & 21.10        & 22.93        & 23.69        & 128      \\
        \textbf{\bdFsamp}              & 15.52                                                                   & 21.03        & 23.44        & N/A          & N/A      \\
        \textbf{\bdJGenum}             & 17.57                                                                   & 17.76        & 17.59        & 17.43        & 17.77    \\
        \hline
        \hline
        \textbf{total}                 & 214.46                                                                  & 285.19       & 346.61       & 399.07       & 17017.44 \\
        \hline
      \end{tabular}
    }
    \subcaption{Feature sel.: runtime in sec. (NBA, $\paramJoinGraphSize = 3$)}
    \label{table:naive_vs_opt_nba}
  \end{minipage}
  \begin{minipage}{1\linewidth}
    \centering
    {\scriptsize
      \begin{tabular}{|c|r|r|r|r|r|}
        \hline
        \multirow{2}{*}{\textbf{Step}} & \multicolumn{4}{c|}{\textbf{feature sel.: \paramFscoresamplerate =}}  
                                       & \multirow{2}{*}{\textbf{w/o feature sel.}}                                                                                    \\ \cline{2-5}
                                       & \textbf{0.1}                                                            & \textbf{0.3} & \textbf{0.5} & \textbf{1.0} &        \\ \hline
        \textbf{\bdFeature}            & 19.05                                                                   & 18.96        & 19.09        & 19.05        & N/A    \\
        \textbf{\bdLCA}                & 16.91                                                                   & 16.50        & 16.71        & 16.35        & 16.44  \\
        \textbf{\bdF}                  & 18.34                                                                   & 74.62        & 131.43       & 226.35       & 209.28 \\
        \textbf{\bdAPT}                & 6.62                                                                    & 7.08         & 6.74         & 6.65         & 6.98   \\
        \textbf{\bdrefine}             & 13.01                                                                   & 16.66        & 16.86        & 15.59        & 15.22  \\
        \textbf{\bdFsamp}              & 6.46                                                                    & 8.80         & 9.93         & N/A          & N/A    \\
        \textbf{\bdJGenum}             & 0.24                                                                    & 0.24         & 0.24         & 0.24         & 0.24   \\
        \hline
        \hline
        \textbf{total}                 & 82.87                                                                   & 145.15       & 203.22       & 286.50       & 250.41 \\
        \hline
      \end{tabular}
    }
    \end{minipage}
    \caption{Feature sel.: runtime in sec. (MIMIC, $\paramJoinGraphSize = 3$)}
    \label{table:naive_vs_opt_mimic}
  \end{figure}


\subsection{Feature Selection}\label{sec:feature-selection}
%
In this set of experiments, we evaluate how feature selection affects our approach.
We compare our approach without feature selection as discussed in \Cref{sec:algo_pattern} (\textbf{Naive}) against with feature selection (\textbf{opt}). For \textbf{opt} we also vary the sample rate for \abbrF calculation (\paramFscoresamplerate). We measured the runtime of the individual steps of our algorithm: (i) we apply feature selection to determine which attributes to use for explanations (\textbf{\bdFeature}); (ii) we use the LCA method to generate candidate patterns for an \abbrAPT (\textbf{\bdLCA}); (iii) we materialize \abbrAPTs (\textbf{\bdAPT}); (iv) we create a sample of \abbrAPTs (\textbf{\bdFsamp}) and then calculate \abbrFs of patterns using this sample (\textbf{\bdF}), (v) in the refinement step we create patterns generated over categorical attributes by adding numerical attributes (\textbf{\bdrefine}).
\Cref{table:naive_vs_opt_mimic} (\Cref{table:naive_vs_opt_nba}) shows the runtime breakdown for MIMIC (NBA).

}

\iftechreport{
\begin{figure}[t]
  \centering
  \includegraphics[trim=0mm 10mm 0mm 0mm,width=0.5\linewidth]{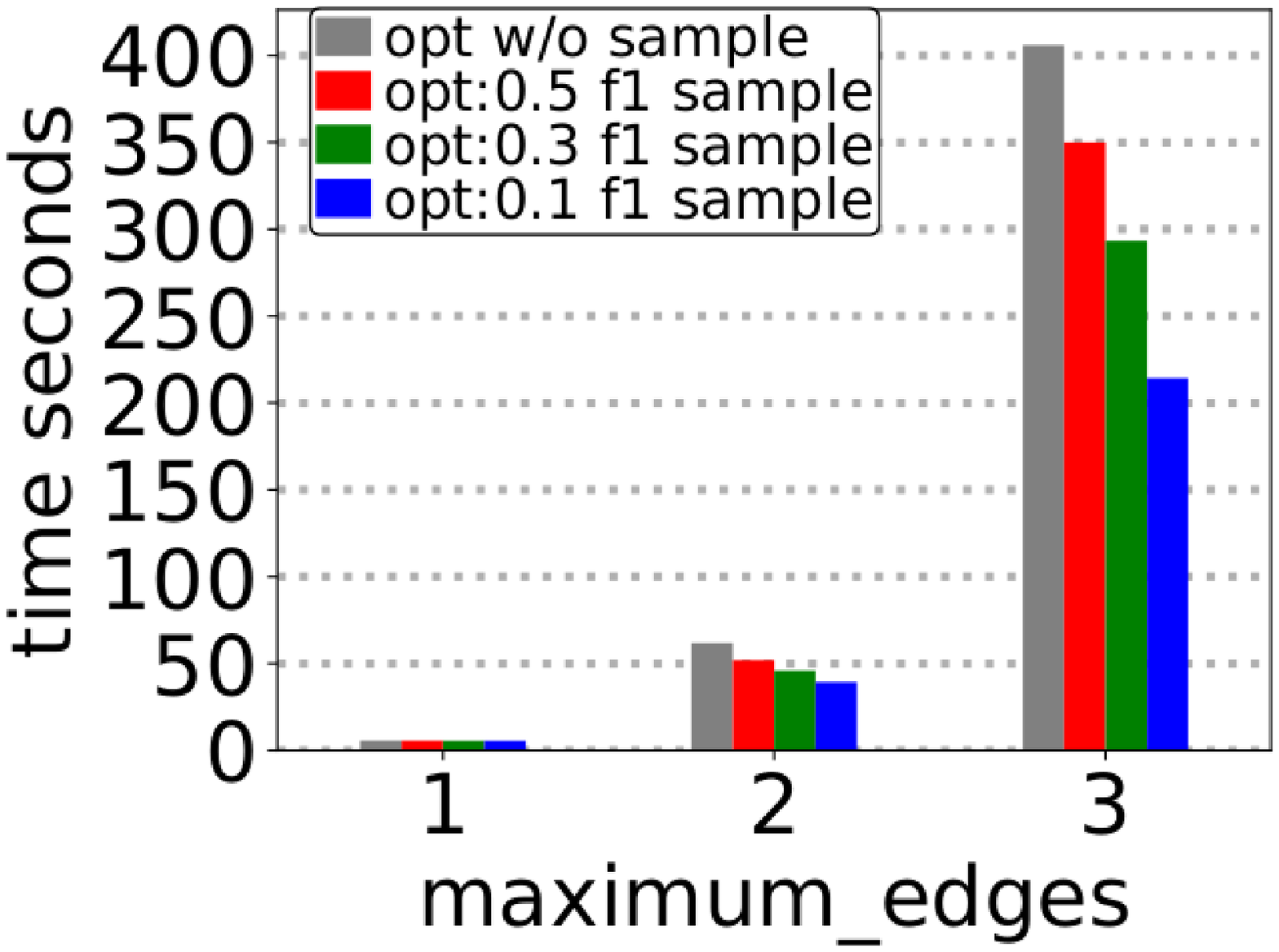}
  \caption{Varying \paramFscoresamplerate and $\paramJoinGraphSize$}
  \label{fig:exp-varying-join-graph-edges}
\end{figure}

\subsection{Join Graph Size}\label{sec:join-graph-size}
We first evaluate how the maximum number of join graphs affects performance by varying \paramJoinGraphSize from 1 up to 3. We compare the runtime of our algorithm with feature selection varying the sample rate for \abbrF calculation (\paramFscoresamplerate= $\{0.1, 0.3, 0.5, 1.0\}$). We use query $\query_1$ and user question (NBA dataset) from the running example in the introduction. The results of this experiment are shown in \Cref{fig:exp-varying-join-graph-edges}. As expected runtime is increases significantly in \paramJoinGraphSize, because the number of join graphs to be considered increases dramatically when we allow for more edges per join graph. Sampling for \abbrF calculation improves performance by up to $\sim 50\%$ for $\paramJoinGraphSize > 1$ when \paramFscoresamplerate is set to 10\%.

}

\begin{figure}[t]
    \begin{minipage}{0.493\linewidth}
    \centering
    \includegraphics[width=1.0\linewidth,trim=0 3mm 0 2mm,clip=true]{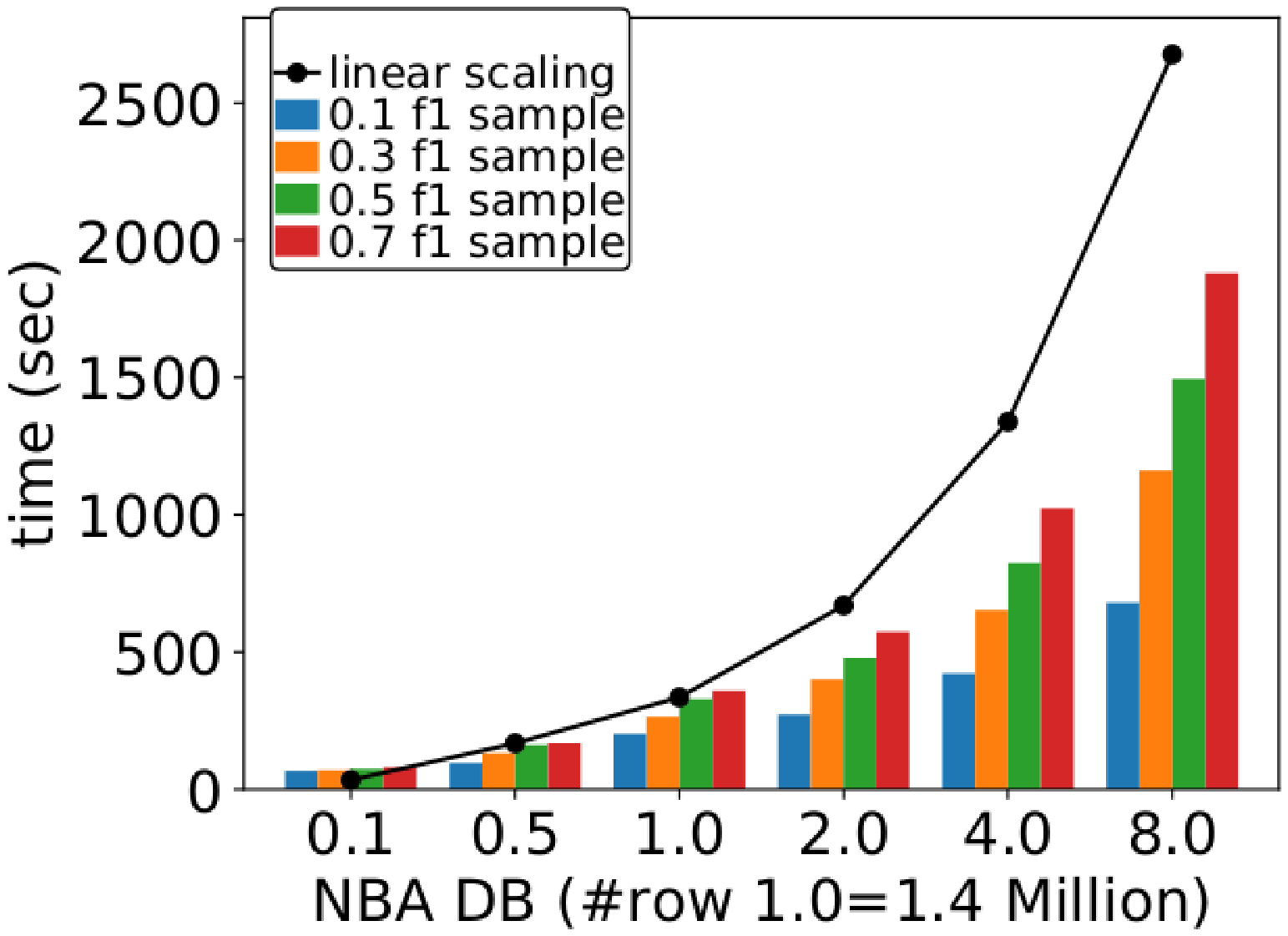}
    \vspace{-5mm}
    \subcaption{NBA, varying \paramFscoresamplerate}
    \label{fig:scale_nba_sample}
    \end{minipage}
    \begin{minipage}{0.493\linewidth}
        \centering
        \includegraphics[width=1.0\linewidth,trim=0 3mm 0 2mm,clip=true]{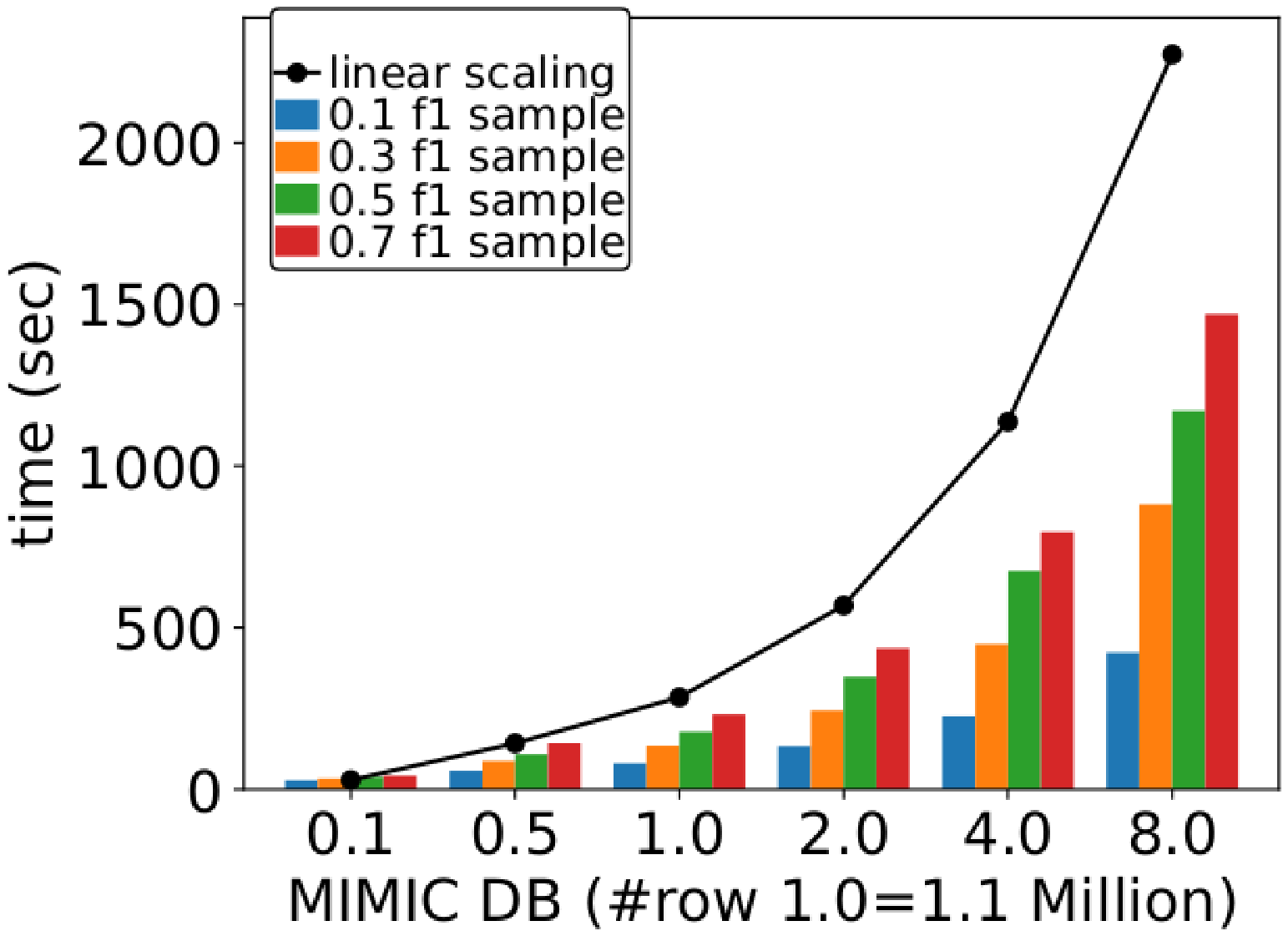}
        \vspace{-5mm}
        \subcaption{MIMIC, varying \paramFscoresamplerate}
        \label{fig:scale_mimic_sample}
    \end{minipage}
\iftechreport{
    \begin{minipage}{1.0\linewidth}
     {
      \centering
      \scriptsize
      \begin{tabular}{|c|r|r|r|r|r|r|}
        \hline
        \textbf{db size}                              & \textbf{0.1} & \textbf{0.5} & \textbf{1} & \textbf{2} & \textbf{4} & \textbf{8} \\
        \hline
        \textbf{\bdFeature} & 31.71        & 35.43        & 84.63      & 97.5       & 112.49     & 121.19     \\
        \textbf{\bdLCA}     & 1.25         & 2.76         & 9.18       & 16.84      & 24.92      & 42.17      \\
        \textbf{\bdF}       & 12.58        & 75.99        & 174.77     & 315.31     & 635.12     & 1248.13    \\
        \textbf{\bdAPT}     & 13.53        & 16.29        & 22         & 27.75      & 42.65      & 71.02      \\
        \textbf{\bdrefine}  & 1.98         & 10.42        & 23.25      & 46.52      & 82.08      & 176.23     \\
        \textbf{\bdFsamp}   & 1.11         & 8.97         & 26.98      & 51.71      & 108.59     & 203.91     \\
        \hline
        \hline
        \textbf{totals}                               & 62.16        & 149.86       & 340.81     & 555.63     & 1005.85    & 1862.65    \\
        \hline
      \end{tabular}
      \subcaption{Time breakdown (NBA, $\paramJoinGraphSize = 3$, $\paramFscoresamplerate = 0.7$)}
      \label{table:bkdown-scale-nba-sample}
    }
    \end{minipage}
     \begin{minipage}{1.0\linewidth}
     {
      \centering
      \scriptsize
      \begin{tabular}{|c|r|r|r|r|r|r|}
        \hline
        \textbf{db size}    & \textbf{0.1} & \textbf{0.5} & \textbf{1} & \textbf{2} & \textbf{4} & \textbf{8} \\
        \hline
        \textbf{\bdFeature} & 13.49        & 18.93        & 18.99      & 19.63      & 20.85      & 23.05      \\
        \textbf{\bdLCA}     & 6.4          & 15.13        & 16.46      & 19.17      & 24.6       & 34.74      \\
        \textbf{\bdF}       & 15.86        & 89.63        & 161.91     & 328.08     & 611.08     & 1133.01    \\
        \textbf{\bdAPT}     & 2.57         & 4.37         & 6.66       & 12.41      & 21.31      & 43.18      \\
        \textbf{\bdrefine}  & 1.67         & 8.22         & 15.49      & 32.63      & 66.84      & 133.31     \\
        \textbf{\bdFsamp}   & 1.6          & 5.58         & 11.68      & 23.91      & 51.25      & 102.16     \\
        \hline
        \hline
        \textbf{totals}     & 41.59        & 141.86       & 231.19     & 435.83     & 795.93     & 1469.45    \\
        \hline
      \end{tabular}
      \subcaption{Time breakdown (MIMIC, $\paramJoinGraphSize = 3$, $\paramFscoresamplerate = 0.7$)}
      \label{table:bkdown-scale-mimic-sample}
    }
  \end{minipage}
  }
  \caption{Scalability in database size}
  \label{fig:exp-scalability}
\vspace{-2mm}
\end{figure}


\subsection{Scalability}\label{sec:exp-scalability}
\revs{
  To evaluate the scalability of our approach, we use scaled versions of the NBA and MIMIC datasets ranging ($\sim 10\%$ to $\sim 800\%$). We varied the \abbrF sample rate ($\paramFscoresamplerate$) from $0.1$ to $0.7$. The results are shown in \Cref{fig:exp-scalability} comparing against linear scaling (black line).
  The effect of database size on runtime is similar for both datasets. Our approach shows sublinear scaling for both datasets (note the log-scale x-axis). The benefits of sampling are more produced for larger database sizes: $\paramFscoresamplerate = 0.1$ is more than 60\% (70\%) faster than $\paramFscoresamplerate = 0.7$ for scale factor $8$ on the NBA (MIMIC) dataset.
  \ifnottechreport{We present a detailed breakdown on where time is spent in \cite{chenjie2021}. \abbrF calculation  turned out to be the most significant factor.}
}
\iftechreport{
  A detailed breakdown of the runtime of the individual steps of our algorithm for $\paramFscoresamplerate = 0.7$  is shown in \Cref{table:bkdown-scale-nba-sample} and \Cref{table:bkdown-scale-mimic-sample}.
  Based on the results of our study of how sampling rate for pattern generating ($\paramPatsamplerate$) affects performance and quality presented in~\Cref{sec:exp-parameter-analysis} (see  \Cref{fig:lca_apt1} to \Cref{fig:lca_apt4}), we capped the number of rows sampled for LCA at $1000$. Recall that we measure the following steps of our approach: (i) we apply feature selection to determine which attributes to use for explanations (\textbf{\bdFeature}); (ii) we use the LCA method to generate candidate patterns for an \abbrAPT (\textbf{\bdLCA}); (iii) we materialize \abbrAPTs (\textbf{\bdAPT}); (iv) we create a sample of \abbrAPTs (\textbf{\bdFsamp}) and then calculate \abbrFs of patterns using this sample (\textbf{\bdF}), (v) in the refinement step we create patterns generated over categorical attributes by adding numerical attributes (\textbf{\bdrefine}). The major contributing factor for larger database sizes is \abbrF calculation which  makes up than 50\% of the runtime. The step with the largest growth rate is sampling for \abbrF calculations for the NBA dataset which is $\sim 183$ times slower for scale factor 8 compared to scale factor 0.1. Based on these results, using lower sample rates for \abbrF calculation is preferable for larger database sizes.}
\begin{figure}[t]
  \begin{minipage}{1\linewidth}
    \centering
    {\scriptsize
      \begin{tabular}{|c|c|c|c|}
        \hline
        \textbf{join graph} & \textbf{join graph structure} &\textbf{\abbrAPT (\#rows)} & \textbf{\# attributes} \\
        \hline
        $\jgraph_1$ &PT& $2621$ & $2$ \\
        \hline
        $\jgraph_2$ &PT\,-\,player\_salary\,-\,player& $66282$ & $2$ \\
        \hline
        $\jgraph_3$ &PT& $50797$ & $10$ \\
        \hline
        $\jgraph_4$ &PT\,-\,patient\_admt\_info\,-\,patients& $50797$ & $19$ \\
        \hline
      \end{tabular}
    }\\[-2mm]
    \subcaption{Join graph APTs size (LCA sampling)}
    \label{table:lca_apt_info}
  \end{minipage}
  \begin{minipage}{1\linewidth}
    \begin{minipage}{1\linewidth}
    \centering
    \begin{minipage}{0.48\linewidth}
      \centering
      \includegraphics[trim=0mm 10mm 0mm 0mm,width=\linewidth]{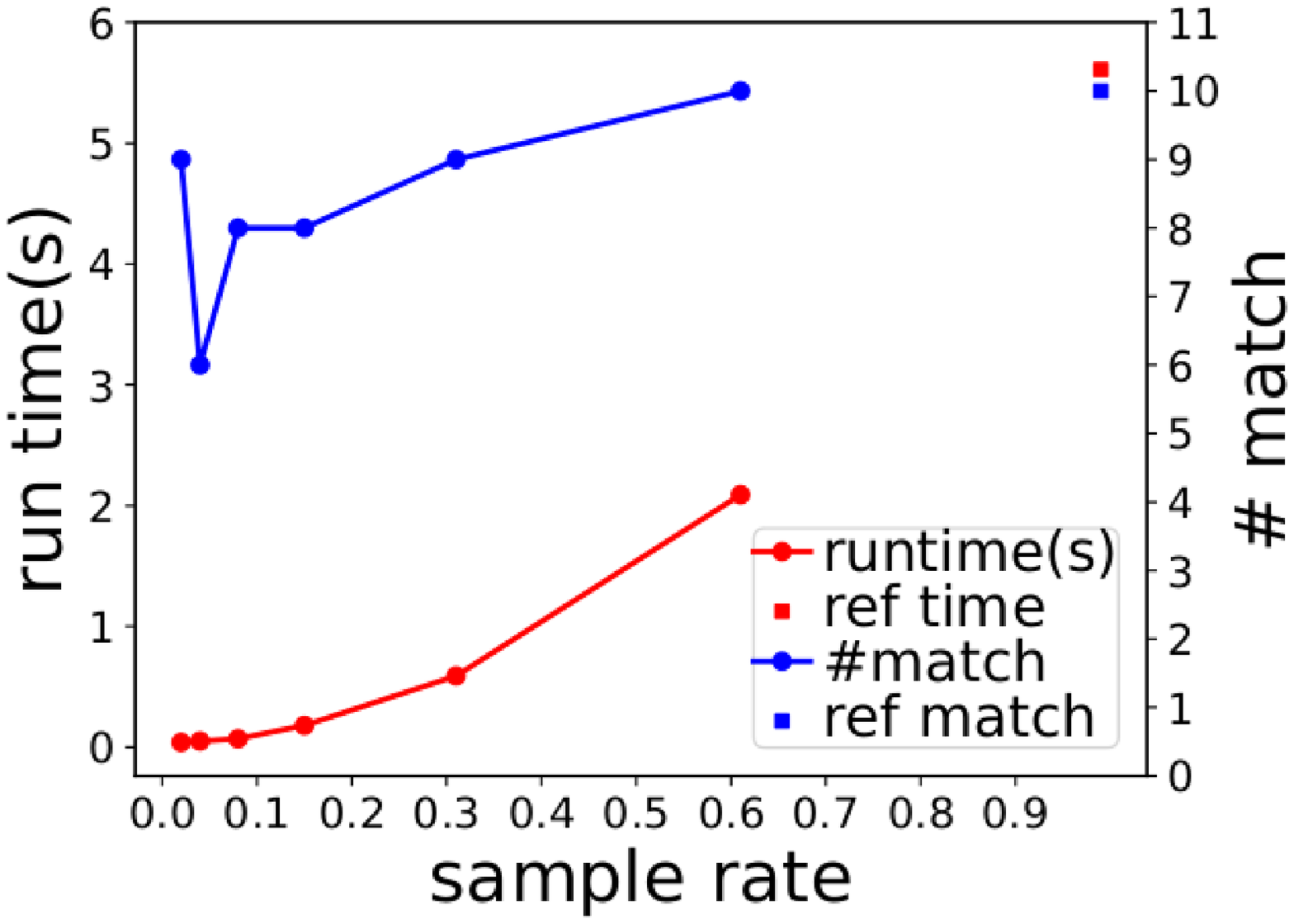}
      \subcaption{LCA sampling for $\jgraph_1$}
      \label{fig:lca_apt1}
    \end{minipage}
    \begin{minipage}{0.48\linewidth}
      \centering
      \includegraphics[trim=0mm 10mm 0mm 0mm,width=\linewidth]{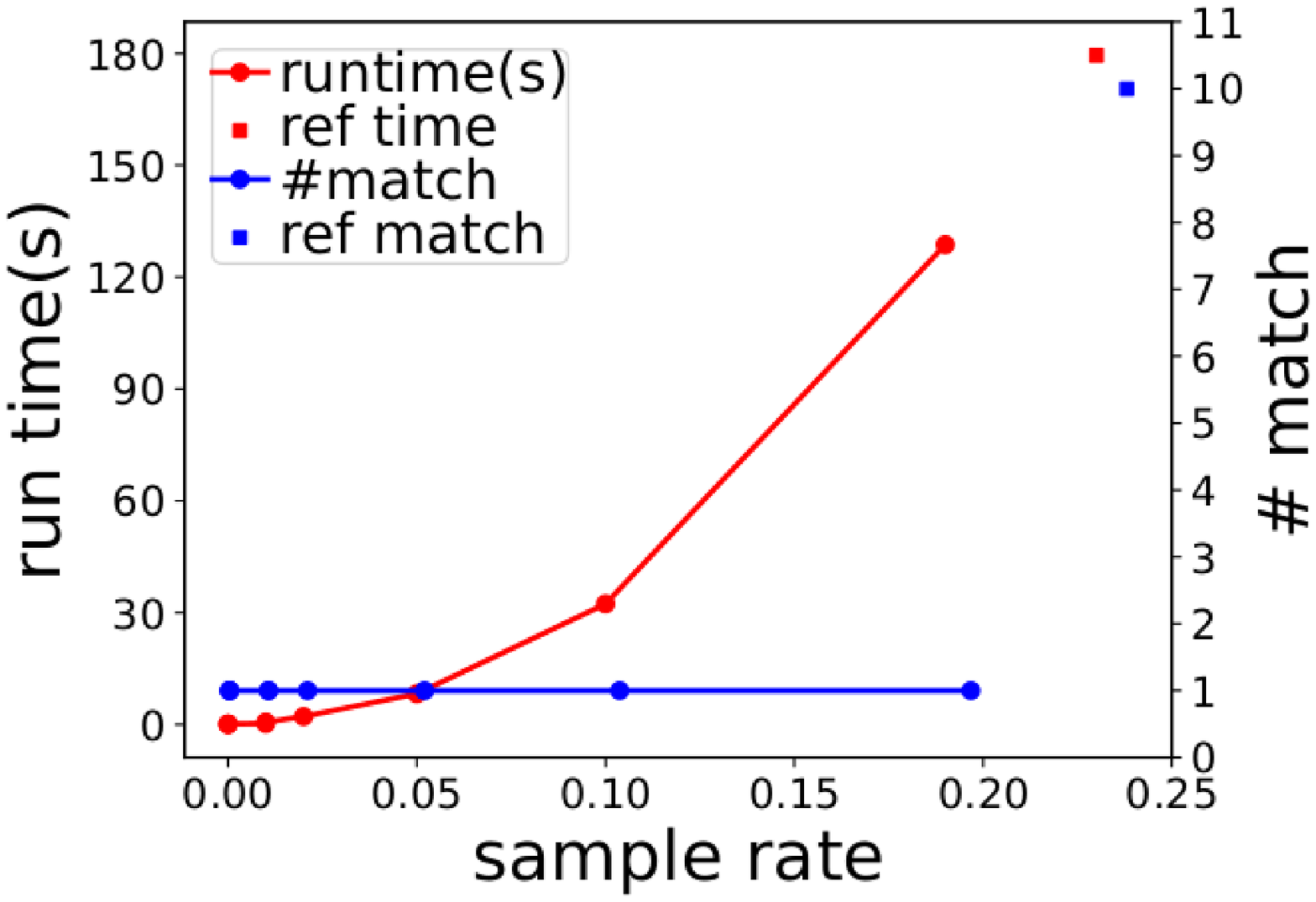}
      \subcaption{LCA sampling for $\jgraph_2$}
      \label{fig:lca_apt2}
    \end{minipage}
    \end{minipage}\\
    \begin{minipage}{1\linewidth}
    \centering
    \begin{minipage}{0.48\linewidth}
      \centering
      \includegraphics[trim=0mm 10mm 0mm 0mm,width=\linewidth]{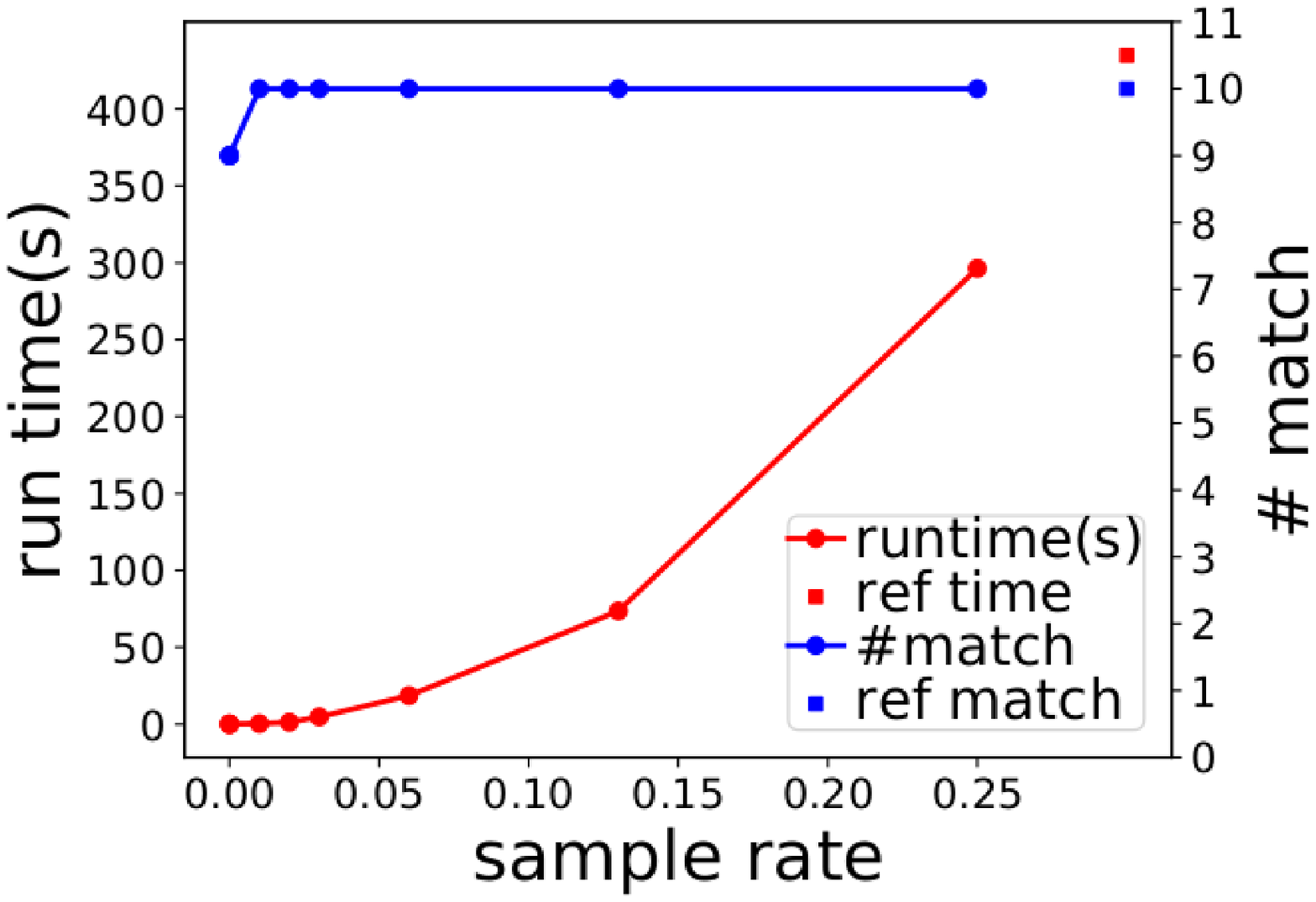}
      \subcaption{LCA sampling for $\jgraph_3$}
      \label{fig:lca_apt3}
    \end{minipage}
    \begin{minipage}{0.48\linewidth}
      \centering
      \includegraphics[trim=0mm 10mm 0mm 0mm,width=\linewidth]{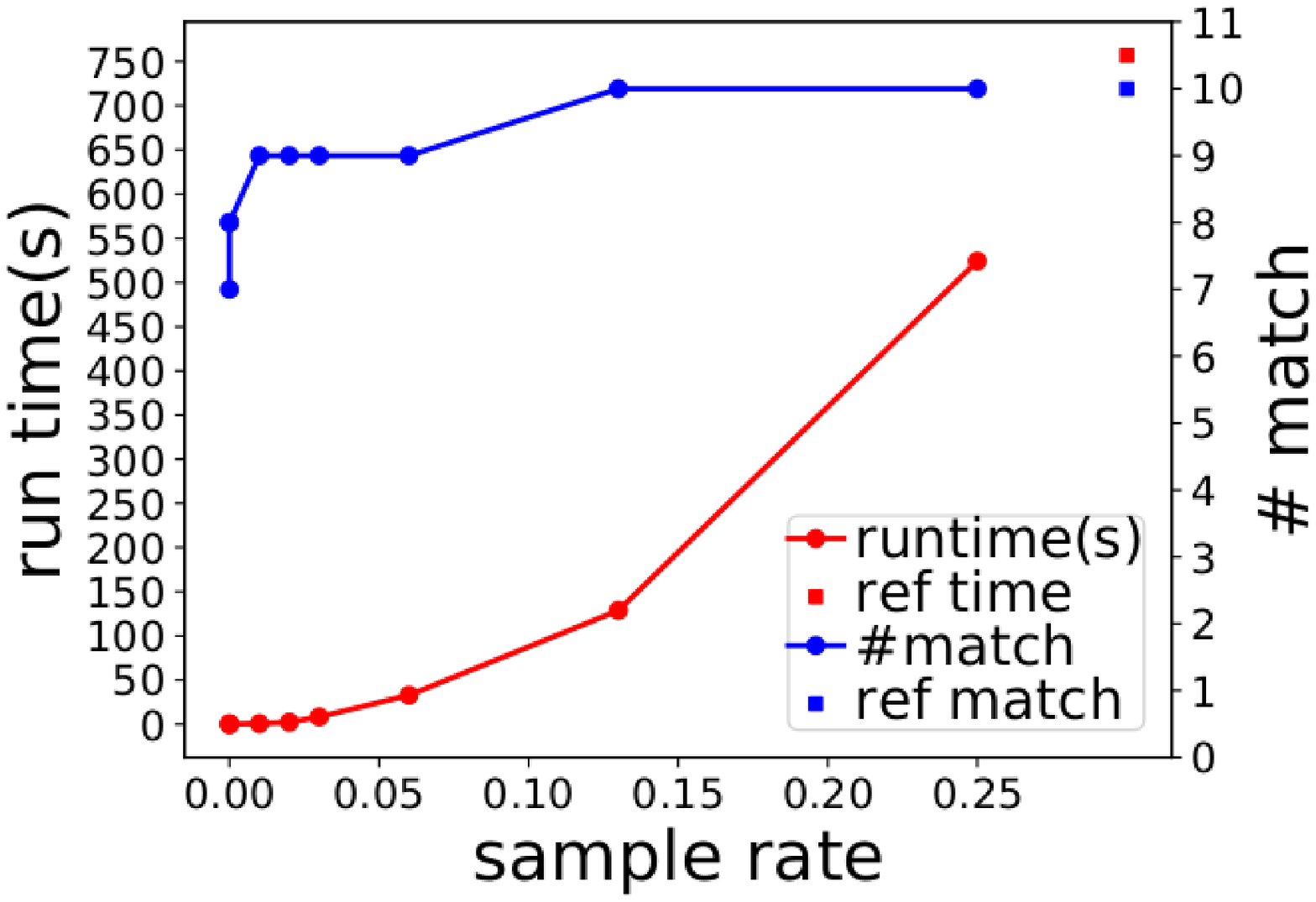}
      \subcaption{LCA sampling for $\jgraph_4$}
      \label{fig:lca_apt4}
    \end{minipage}
  \end{minipage}
  \end{minipage}
  \begin{minipage}{1.0\linewidth}
  \centering
    \iftechreport{
    \begin{minipage}{0.96\linewidth}
      \centering
      \includegraphics[trim=0mm 10mm 0mm 0mm,width=\linewidth]{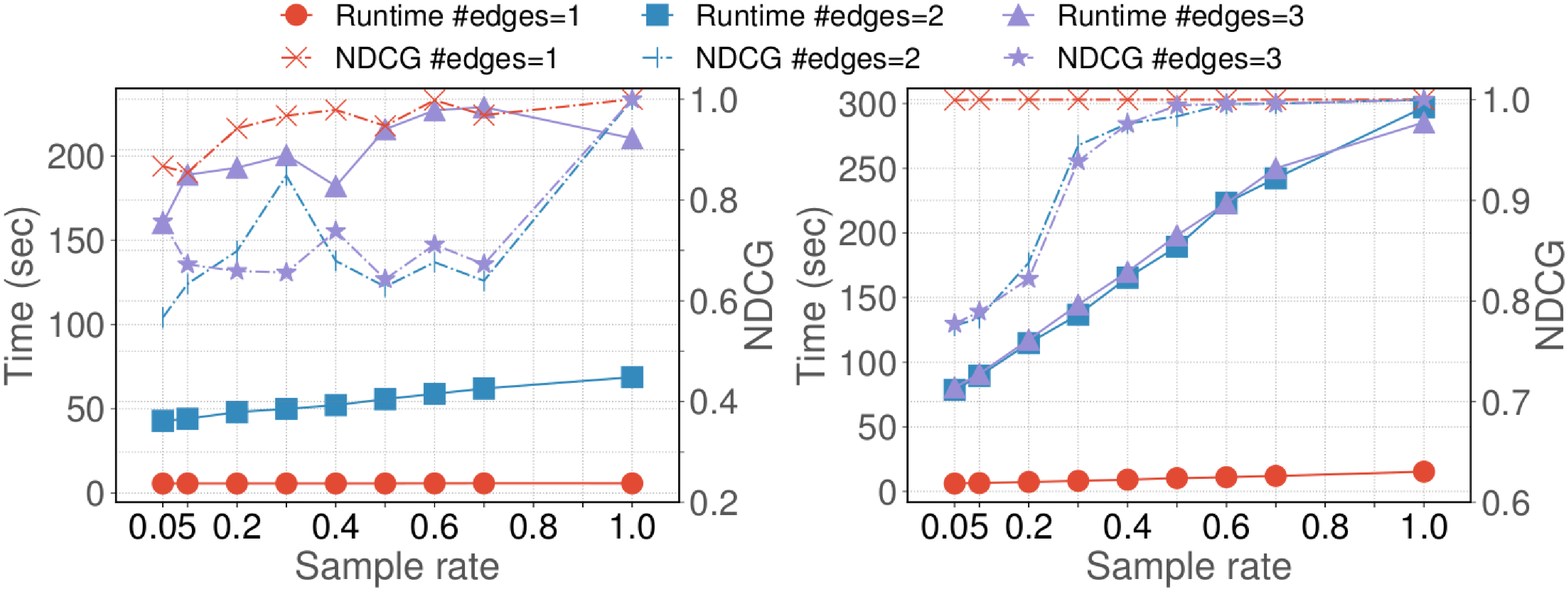}
      \subcaption{NDCG against $\paramFscoresamplerate$ (Left: NBA, right: MIMIC)}
      \label{fig:nba_mimic_ndcg}
    \end{minipage}
    \begin{minipage}{0.96\linewidth}
      \centering
      \includegraphics[trim=0mm 10mm 0mm 0mm,width=\linewidth]{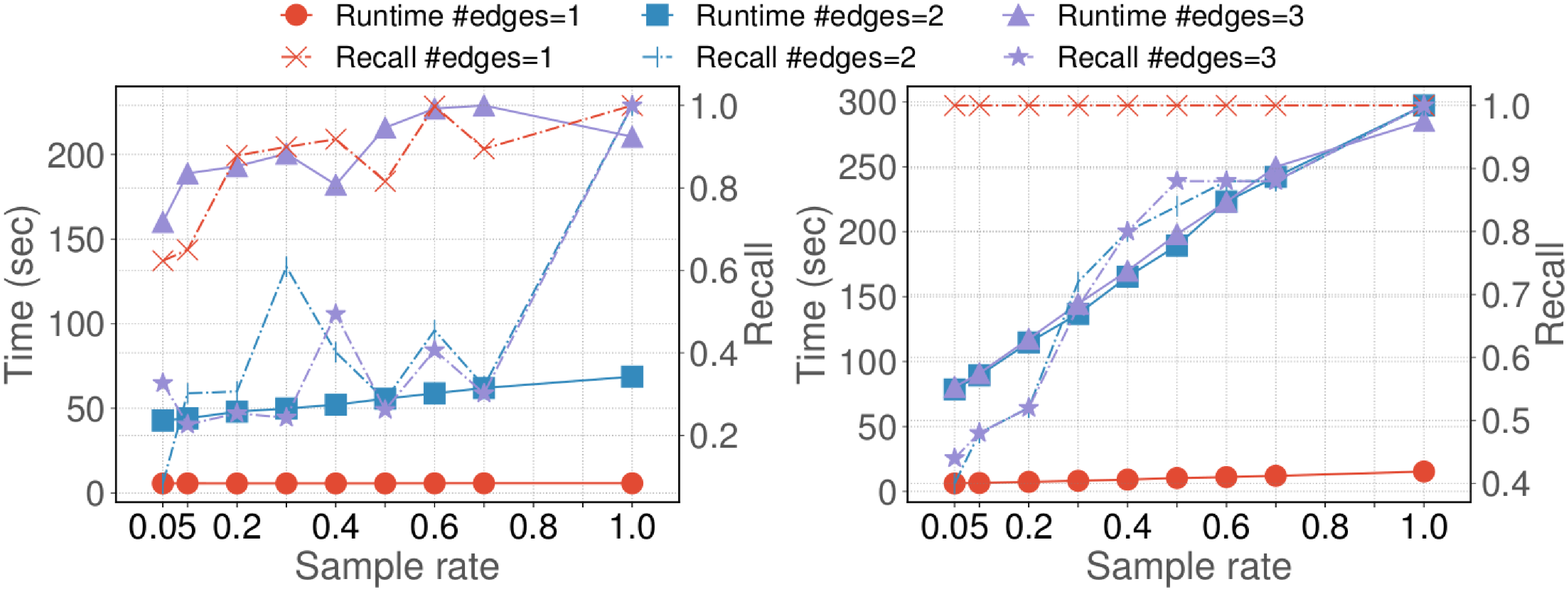}
      \subcaption{Recall against $\paramFscoresamplerate$ (left: NBA, right: MIMIC)}
      \label{fig:nba_mimic_recall}
    \end{minipage}
    }
  \end{minipage}
  \vspace{-4mm}
  \caption{Effect of sampling on runtime and pattern quality}
  \label{fig:sample_rate}
   \vspace{-4mm}
\end{figure}


\subsection{Sample Size}\label{sec:exp-parameter-analysis}

\revs{We now study the impact of sampling for \abbrF calculation (\paramFscoresamplerate) and for pattern candidate generation  (\paramPatsamplerate) on performance and pattern quality. We treat the result produced without sampling as ground truth and measure the difference between this result and the result produced by sampling.}

\mypar{Sampling for Pattern Generation}
\revm{
  Recall that we use the \textit{LCA} approach to generate candidate patterns over the categorical attributes of a database. This approach computes a cross product of a sample with itself. Our implementation of LCA ranks the pattern candidates generated by LCA by their recall, and then selects the top-k ranked patterns as input for the next step.
In this experiment, we want to determine a robust choice for the \textit{LCA} sample size parameter and, thus, compare the results produced by this step.
We selected $4$ join graphs and their \abbrAPTs: $\jgraph_1$ and $\jgraph_2$ for $\query_1$, and $\jgraph_3$ and $\jgraph_4$ for $\query_{mimic4}$.
\Cref{table:lca_apt_info} shows the number of rows and attributes for the \abbrAPTs, and join graph structure for each of these join graphs. The results are shown in \Cref{fig:lca_apt1} to \Cref{fig:lca_apt4}.
We measure pattern quality as the number of patterns from the top-10 computed over the full dataset that occur in the top-10 computed based on a sample (see the blue lines labeled \textit{match}). 
For each materialized join graph, we also measure the runtime of generating the top-10 patterns.
As expected because of the cross product computed over the sample, runtime increases  quadratically in the sample size. For \Cref{fig:lca_apt3} and \Cref{fig:lca_apt4}, all ground truth top-10 patterns are found even for just 3\% sample rate.}
\revs{Whereas as shown in \Cref{fig:lca_apt2}}, \revm{even for $20\%$ sample rate ($13000$ rows), we only find one matching pattern. The reason behind the different result observed in \Cref{fig:lca_apt3} is that one of the columns in $\jgraph_2$ has over $800$ distinct values that are roughly evenly distributed and, thus, the recall-based ranking is sensitive to small variations in frequency caused by sampling. For \cref{fig:lca_apt1}, even though this join graph also contains this attribute (over $500$ distinct values), for this \abbrAPT, the column's distribution in the \abbrAPT is skewed leading to a more stable set of high frequency values that are used in the top-10 patterns. 
Based on these observations we determine the sample size \paramPatsamplerate=0.1 for the rest of the experiments 
and set a cap number of rows in the sample as 1000. 
 }

\mypar{Sampling for \abbrF Calculation}
\reva{We also use sampling to reduce the cost of the quality measure calculation (parameter \paramFscoresamplerate). Instead of scanning all tuples in the augmented provenance table (\abbrAPT), we compute the number of matching tuples over a sample of the \abbrAPT{} for a given pattern. \Cref{fig:sample_rate} shows the running time and the quality of patterns when varying the sample rate and maximum number of edges in join graphs (\paramJoinGraphSize) for queries $\query_1$ and $\query_{mimic4}$. 
We use the normalized discounted cumulative gain (NDCG)~\cite{jarvelin2002cumulated} as the sample quality metric, which is often used in information retrieval for evaluating the ranking results of recommendation and search engine algorithms.
A high NDCG score (between 0 and 1) indicates that the ranking of the top patterns returned by the sampled result is close to the top patterns produced
without sampling.
\Cref{fig:nba_mimic_ndcg} shows that for both datasets for $\paramJoinGraphSize=1$, the similarity between the sampled result and result over the full dataset is high, even for aggressive sampling.
For larger join graphs ($\paramJoinGraphSize=2,3$), the NDCG score fluctuates around 0.7 for the NBA dataset.
For the MIMIC dataset, the NDCG converges to $\sim$ 1.0 at a sampling rate of 0.5. Overall, even for low sample rates, the NDCG score is at least $\sim 0.6$ ($\sim 0.8$) for the NBA (MIMIC) dataset. }
\iftechreport{
We also evaluate the number of patterns from sampling that are present in the results without sampling, for which we used recall as the metric as shown in \Cref{fig:nba_mimic_recall}.
}




\subsection{Comparison with Explanation Tables}\label{sec:comp-with-expl}

\reva{We also compared our approach against the approach from~\cite{DBLP:journals/pvldb/GebalyAGKS14} (referred to as ET from now on). We compared on one join graph with structure \texttt{PT - player\_game\_stats - player} for the NBA dataset using query $\query_1$ and the corresponding user question from the introduction. The corresponding \abbrAPT has $\sim$ 2600 rows and 84 columns. To be fair, we did apply our feature selection technique to filter columns for ET too, reducing the number of columns to 20. \revs{Without that step, ET took 30 seconds even for the smallest sample size (16 tuples). \Cref{fig:exp-vary-queries} lists the runtime of \oursystem{} and ET after applying feature selection}. \iftechreport{As a qualitative comparison, we list the first 20 patterns returned by ET in \Cref{append:et}.} While slower for a sample size of 16 our approach scales much better when increasing the sample size ($\sim 50x$ faster for sample size 512).} \revs{That being said, we would like to point out the major contribution of our work the efficient exploration of a schema graph for finding explanations. However, as this experiment demonstrates this would not be possible without our optimizations for mining patterns over a single \abbrAPT.}

\begin{figure}[t]
\centering
\begin{minipage}{1.0\linewidth}
\centering
\begin{minipage}[b]{0.33\linewidth}
      {\scriptsize
        \begin{tabular}{| >{\centering}p{3em}| >{\centering}p{3em}|c|}
          \hline
                     & \multicolumn{2}{c|}{\textbf{Runtime (sec)}} \\\cline{2-3}
          \textbf{Sample Size} & \textbf{\oursystem} & \textbf{ET}  \\
          \hline
          16         & 9.90        & 3.21           \\
          \hline
          64         & 14.46       & 11.65          \\
          \hline
          256        & 15.32       & 176.76         \\
          \hline
          512        & 14.81       & 855.13         \\
          \hline
        \end{tabular}
      }\\[-2mm]

  \caption{Comparison with Explanation Tables}\label{fig:comparison-with-explanation}
\end{minipage}
\hfill
\begin{minipage}[b]{0.63\linewidth}
\centering
  \includegraphics[trim=0mm 1mm 0mm 0mm,clip,width=0.95\linewidth]{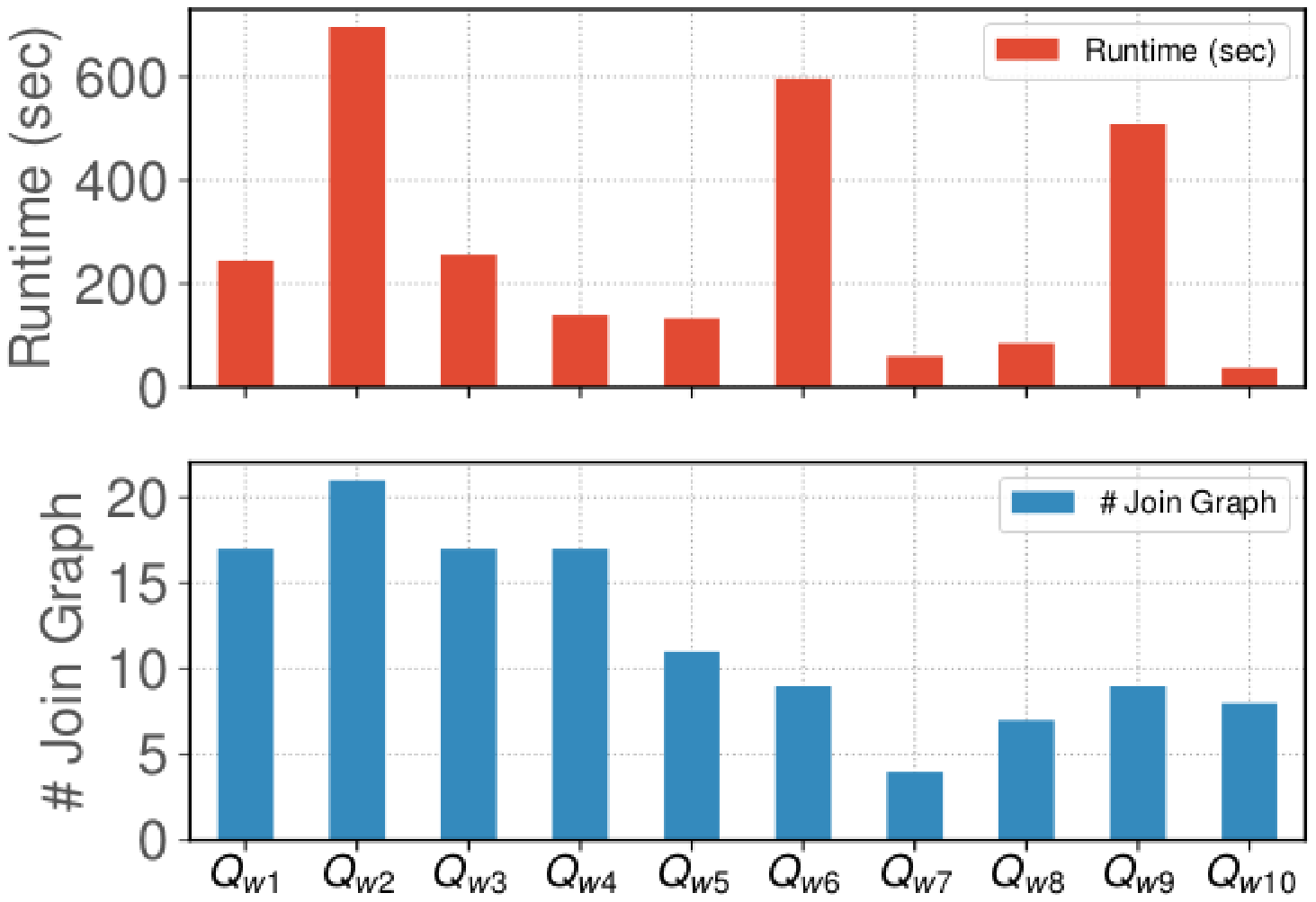}
   \\[-4mm]
  \caption{Varying Queries}
  \label{fig:exp-vary-queries}
\end{minipage}
\end{minipage}
\begin{minipage}[t]{0.9\linewidth}
  \centering
      {\scriptsize
        \begin{tabular}{| >{\centering}p{3em}|c|c|}
          \hline
        & \multicolumn{2}{c|}{\textbf{Query}} \\\cline{2-3}
          \textbf{Rank} & \textbf{$\uquestion_{cape1}$} & \textbf{$\uquestion_{cape2}$}  \\
          \hline
          $1$ & (\textit{LeBron James},$2009$-$10$,$29.7$) & (\textit{GSW},$2013$-$14$,$51$)\\
          \hline
          $2$ & (\textit{LeBron James},$2011$-$12$,$27.1$) & (\textit{GSW},$2014$-$2015$,$67$) \\
          \hline
          $3$ & (\textit{LeBron James},$2013$-$14$,$27.1$)  & (\textit{GSW},$2015$-$16$,$73$)\\
          \hline
        \end{tabular}
      }\\[-4mm]
  \caption{CAPE's explanations for the NBA questions }\label{fig:comparison-with-cape}
\end{minipage}
\ifnottechreport{
 \vspace{-3mm}
 }
\end{figure}


\begin{table}[t]
  \centering
  {\footnotesize
    \begin{tabular}{|c|p{17em}|p{9em}|}
      \hline
      \textbf{Query} & \textbf{Description}  & \textbf{Tables used}              \\\hline
      $\query_{w1}$  & The average points change over the years for player \textit{Draymond Green} &  player,  game, season, player\_game\_stats \\
      $\query_{w2}$  & Team \textit{GSW} average assists over the years & team\_game\_stats, game, team, season \\
      $\query_{w3}$  & Average points for player \textit{Lebron James} over the years & player, game, season , player\_game\_stats \\
      $\query_{w4}$  & Team \textit{GSW} wins over the years & team, game, season \\
      $\query_{w5}$ & Average points by player \textit{Jimmy Butler} & player, game, player\_game\_stats \\
      \hline
      $\query_{w1}$  & Return the number of diagnosis group by chapter(group of procedure type) & diagnoses   \\
      $\query_{w2}$  & Returns the death rate of patients grouped by their insurance.   & admissions   \\
      $\query_{w3}$  & Number of ICU stays grouped by the length of stays (\texttt{los\_group}). & icustays \\
      $\query_{w4}$  & Number of procedures for a particular chapter (group of diagnosis types). & procedures   \\
      $\query_{w5}$ & Number admissions of different ethnicities. & patients\_admit\_info \\
      \hline
    \end{tabular}
  }
  \caption{NBA and MIMIC Queries}
  \label{tab:workload-queries}
  \ifnottechreport{
   \vspace{-9mm}
   }
\end{table}

\subsection{Comparison with CAPE}
\revs{
We compared our approach against CAPE (\cite{DBLP:conf/sigmod/MiaoZGR19}).  
The question we used are from NBA running example $\query_1$ which asks question about number of wins for \textit{GSW} over the years and $\query_{nba3}$ from the case study which asks for player \textit{LeBron James}'s average points over the seasons.
CAPE expects as input one data point plus a direction \textit{high} or \textit{low}. We select the following question $\uquestion_{cape_1}$ for $\query_1$: \textit{Why was GSW number of wins high in 2015-16 season?} and $\uquestion_{cape_2}$ for $\query_{nba3}$: \textit{Why was LeBron James' average points low in 2010-11 season?}.
Since CAPE does not explore related tables, we constructed $2$ join graphs as input to CAPE, which are \texttt{PT} ($\uquestion_{cape_1}$) and \texttt{PT - team\_game\_stats} ($\uquestion_{cape_2}$).
\Cref{fig:comparison-with-cape} shows the top-3 explanations produced by CAPE. The system identifies a trend in the data (using regression) according to which the user question is an outlier in the user-provided direction and then returns a similar outlier in the other direction. For our experiment, this means that CAPE returns seasons with low wins for GSW and high averages points for LeBron James. This experiment demonstrates that CAPE is orthogonal to our technique. The system identifies counter-balances while we find features that are related to the difference between two query results. Nonetheless, our techniques for exploring schema graphs may be of use for finding counterbalances too.
}

\subsection{Varying Queries}\label{sec:exp-vary-queries}
\reva{
To evaluate how the runtime of our approach is affected by the choice of query, we measured the runtime for  10 different queries ($5$ for NBA and $5$ for MIMIC) shown in \Cref{tab:workload-queries}.
We designed these queries to access different relations and use different group-by attributes.} \ifnottechreport{\revs{The SQL code for these queries is shown~\cite{chenjie2021}.}}\iftechreport{\revs{The SQL code for these queries is shown below.}}
\reva{All queries were run with  $\paramFscoresamplerate = 0.3$ and $\paramJoinGraphSize = 3$. The results are shown in \Cref{fig:exp-vary-queries}. We observe that the runtime is relatively stable for different queries and is to some degree correlated to the number of join graphs for the query (shown on top of \Cref{fig:exp-vary-queries}).}
\iftechreport{
\mypar{SQL code for Query Workload}
Note that the queries are based on the original schema graph which are more complex than the ones presented in the running example. The real schema can be found in \Cref{fig:schema_mimic} and \Cref{fig:schema_nba}.

\mypar{NBA}
\mypar{$\query_{w1}$}
The average points change over the years for player Draymond Green.

\begin{lstlisting}

SELECT AVG(points) as avp_pts, s.season_name 
FROM player p, player_game_stats pgs, game g, season s 
WHERE p.player_id=pgs.player_id AND 
g.game_date = pgs.game_date AND 
g.home_id = pgs.home_id AND 
s.season_id = g.season_id 
AND p.player_name='Draymond Green' 
GROUP BY s.season_name

\end{lstlisting}

\mypar{$\query_{w2}$}
GSW average assists over the years

\begin{lstlisting}
SELECT AVG(tgs.assists) as avgast, s.season_name 
FROM team_game_stats tgs, game g, team t, season s 
WHERE s.season_id = g.season_id AND
tgs.game_date = g.game_date AND 
tgs.home_id=g.home_id AND 
tgs.team_id = t.team_id AND
t.team='GSW' 
GROUP BY s.season_name
\end{lstlisting}

\mypar{$\query_{w3}$}
Average points for Lebron James over the years

\begin{lstlisting}

SELECT AVG(points) AS avp_pts, s.season_name 
FROM player p, player_game_stats pgs, game g, season s WHERE
p.player_id=pgs.player_id AND 
g.game_date = pgs.game_date AND 
g.home_id = pgs.home_id AND 
s.season_id = g.season_id AND 
p.player_name='LeBron James' 
GROUP BY s.season_name

\end{lstlisting}

\mypar{$\query_{w4}$}

GSW wins over the years, but we used different 2 seasons 

\begin{lstlisting}
SELECT COUNT(*) AS win, s.season_name 
FROM team t, game g, season s 
WHERE t.team_id = g.winner_id AND
g.season_id = s.season_id AND 
t.team= 'GSW' 
GROUP BY s.season_name
\end{lstlisting}

\mypar{$\query_{w5}$}
Average points by Jimmy Butler over the years

\begin{lstlisting}
SELECT AVG(points) AS avp_pts, s.season_name 
FROM player p, player_game_stats pgs, game g, season s 
WHERE p.player_id=pgs.player_id AND 
g.game_date = pgs.game_date AND
g.home_id = pgs.home_id AND 
s.season_id = g.season_id AND 
p.player_name='Jimmy Butler' 
GROUP BY season_name
\end{lstlisting}

\mypar{MIMIC}
\mypar{$\query_{w6}$}
Return the number of diagnosis by chapter. A chapter is a type of diagnosis.

\begin{lstlisting}
SELECT count(*) AS cnt, chapter
FROM diagnoses
GROUP BY chapter
\end{lstlisting}

\mypar{$\query_{w7}$}
Returns the death rate of patients grouped by their insurance.

\begin{lstlisting}
SELECT insurance,
       1.0 * sum(hospital_expire_flag)
             / count(*) AS death_rate
FROM admissions
GROUP BY insurance
\end{lstlisting}

\mypar{$\query_{w8}$}
Number of ICU stays grouped by the length of stays (\texttt{los\_group}).

\begin{lstlisting}
SELECT count(*) AS cnt, los_group
FROM icustays
GROUP BY los_group
\end{lstlisting}

\mypar{$\query_{w9}$}
Number of procedures for a particular chapter (group of diagnosis types).

\begin{lstlisting}
SELECT count(*) AS cnt, chapter
FROM procedures
GROUP BY chapter
\end{lstlisting}

\mypar{$\query_{w10}$}
Number of admits per patient ethnicity.

\begin{lstlisting}
SELECT count(*) AS cnt, ethnicity
FROM patients_admit_info
GROUP BY ethnicity
\end{lstlisting}


}


\section{Qualitative Evaluation}
\label{sec:case-study}
\iftechreport{
We now evaluate the quality of explanations produced by \oursys using case studies on the same two real datasets (NBA and MIMIC) as in \Cref{sec:experiments}. We report the SQL code of queries, query results, user questions, and returned explanations for both datasets. We also report results of a user study with the NBA dataset.
For both case studies, we report the top-3 explanations for each query and user question in \Cref{tab:nba-queries-expl} (NBA) and \Cref{tab:expl-mimic-queries} (MIMIC). Due to the space limit, we simplified some of these descriptions. Note that the same pattern may be returned for several join graphs (same attributes, but different join path). In the interest of diversity, we removed duplicates and explanations that only differ slightly in terms of constants. We show the top-3 explanations after this step. we use ``[$t_1$]" or ``[$t_2$]" in explanations as identifier of the \textit{primary tuple} for the explanation. The full sets of top-20 explanations (including join nodes and edge details) are in \Cref{sec:top20}. 
}
\ifnottechreport{
\revs{
We now evaluate the quality of explanations produced by \oursys using case studies on two real datasets (NBA and MIMIC). 
We also report results of a user study with the NBA dataset.
Due to the space limit, we simplified some of these descriptions. Note that the same pattern may be returned for several join graphs (same attributes, but different join path). In the interest of diversity, we removed duplicates and explanations that only differ slightly in terms of constants. We show the top-3 explanations after this step. we use ``[$t_1$]" or ``[$t_2$]" in explanations as identifier of the \textit{primary tuple} for the explanation. For the sets of top-20 explanations (including join nodes and edge details), SQL code of the queries, and a detailed description of the datasets  
see~\cite{chenjie2021}.
}
}


\vspace{-2mm}
\revs{\subsection{Case Study: NBA}\label{sec:qual-eval-nba}
\ifnottechreport{
\mypar{Dataset} \textit{NBA} is a dataset we extracted from the NBA website (\url{https://www.nba.com/}) and PBP stats website (\url{https://www.pbpstats.com/}).
It includes 10 seasons' worth of data (seasons \texttt{2009-10} to \texttt{2018-19}). The data set contains  $11$ relations with the largest relation having around 1.4 million rows. Its total size is $\sim 170$MB. }
%
}
\mypar{Setup}\revs{
For the NBA dataset, we use five queries calculating player's and team's stats and generated user questions based on interesting results. 
\Cref{tab:nba-queries-expl} shows the user questions, queries, and top-3 explanations 
produced by our method for these user questions.
\iftechreport{
\Cref{tab:nba-queries-desc} shows a description for each these queries, which include group-by aggregations over path joins.
  We present the SQL code for these query below. Their results and the tuples used in the user question (highlighted rows) are shown in \Cref{fig:query_results_nba}.}
}

\revs{
  \iftechreport{
    \begin{table}[t]
      \centering
      {\scriptsize
        \begin{tabular}{|c|p{15em}|c|}
          \hline
          \textbf{Q}      & \textbf{Description}                              & \textbf{Tables used}                      \\\hline
          $\query_{nba1}$   & Average points per year for Draymond Green.       & player, player\_game\_stats, game, season \\
          $\query_{nba2}$   & GSW average assists over the years.               & team\_game\_stats, game, team, season     \\
          $\query_{nba3}$   & Average points over the years for Lebron James .  & player, player\_game\_stats, game, season \\
          $\query_{nba4}$   & GSW wins over the years.                          & game, team, season                        \\
          $\query_{nba5} $  & Average points over the years for Jimmy Butler.   & player, player\_game\_stats, game, season \\
          \hline
        \end{tabular}
        \caption{NBA queries}
        \label{tab:nba-queries-desc}
      }
    \end{table}
  }
\begin{table}[t]
  {    \centering
    \scriptsize
    \begin{tabular}{|l|l|p{16em}|l|}
      \hline
      \textbf{Query}& \textbf{User question} &  \textbf{Top explanations} & \textbf{\abbrF} \\\hline
      $\query_{nba1}$ &  \multirow{3}{11em}{\textbf{\textit{Draymond Green}'s average points per year}: 14 points in season $2015 \mbox{-} 16$ ($t_1$) VS $10$ points in season $2016 \mbox{-} 17$ ($t_2$)} & player\_salary$<15330000$ [$t_1$] & $1$ \\
      \cline{3-4}
                & &  prov.tspct$<0.69$ $\wedge$ prov.usage$<20.5$ $\wedge$ salary$>14260000$ [$t_2$] & $0.71$ \\
      \cline{3-4}
                & &  prov.minutes$>31$ $\wedge$ prov.tspct$>0.4$ $\wedge$ salary$<15330000$ [$t_1$] & $0.66$ \\
      \cline{1-4}
      $\query_{nba2}$ & \multirow{3}{11em}{\textbf{\textit{GSW}'s average assists per year}: $23$ in season $2013 \mbox{-}14$ ($t_1$) VS $27$ in season $2014 \mbox{-}15$ ($t_2$)} & \makecell[l]{prov.assistpoints$<68$ $\wedge$ \\ player=Draymond Green [$t_1$]}& $0.74$ \\
      \cline{3-4}
                & & \makecell[l]{prov.assistpoints$>57$ $\wedge$ \\ prov.nonputbackast\_2\_pct$>0.55$ \\ $\wedge$ player.player=Harrison Barnes [$t_2$]}& $0.73$ \\
      \cline{3-4}
                & & \makecell[l]{prov.assistpoints$<68$ $\wedge$ \\ offreboundpct$>0.25$ [$t_1$]}& $0.72$ \\
      \cline{1-4}
      $\query_{nba3}$ & \multirow{3}{11em}{\textbf{\textit{LeBron James}'s average points per year}: $29.7$ in season $2009 \mbox{-}10$ ($t_1$) VS $26.7$ in season $2010 \mbox{-}11$($t_2$)} & player\_salary$>14500000$ [$t_1$] & $1$\\
      \cline{3-4}
                & & team=MIA [$t_2$] & $0.98$ \\
      \cline{3-4}
                & & team=CLE [$t_1$] & $0.93$ \\
                & & & \\
      \cline{1-4}
      $\query_{nba4}$ & \multirow{3}{11em}{\textbf{\textit{GSW}'s number of wins per year}: $47$ in $2012 \mbox{-}13$ ($t_1$) season VS $67$ in $2016-17$ season ($t_2$)} & \makecell[l]{player\_name=Pau Gasol $\wedge$ \\ player\_salary$<19285850$ [$t_2$]} & $1$\\
      \cline{3-4}
                & & player\_name=Andre Iguodala [$t_2$] & $0.97$ \\
      \cline{3-4}
                & & \makecell[l]{fg\_3\_apct$<0.31\wedge$ \\ team\_points$<121$ [$t_1$]} & $0.92$ \\
      \cline{1-4}
      $\query_{nba5}$ & \multirow{3}{11em}{\textbf{\textit{Jimmy Butler}'s average points per year}: $13$ points in season $2013 \mbox{-} 14$ ($t_1$) VS $20$ points in season $2014 \mbox{-} 15$ ($t_2$)} & player\_salary$>1112880$ [$t_2$] & $1$\\
      \cline{3-4}
                & & \makecell[l]{prov.away\_points$>87$ $\wedge$ \\ prov.efgpct$>0.38$ [$t_2$]} & $0.84$ \\
      \cline{3-4}
                & & \makecell[l]{prov.usage$<23$ $\wedge$ team=CHI$\wedge$ \\ team\_assisted\_2\_spct$>0.5$ [$t_1$]} & $0.77$ \\
      \cline{1-4}
    \end{tabular}
    \caption{Queries, user questions and explanations (NBA)}
    \label{tab:nba-queries-expl}
    \ifnottechreport{
    \vspace{-10.5mm}
    }
    }
\end{table}
}
\iftechreport{

\mypar{$\query_{nba1}$}
The average points change over the years for player Draymond Green.
\begin{lstlisting}
SELECT AVG(points) as avp_pts, s.season_name
FROM player p, player_game_stats pgs, game g, season s
WHERE p.player_id=pgs.player_id AND
g.game_date = pgs.game_date AND
g.home_id = pgs.home_id AND
s.season_id = g.season_id
AND p.player_name='Draymond Green'
GROUP BY s.season_name
\end{lstlisting}

%
\mypar{$\query_{nba2}$}
GSW average assists over the years
\begin{lstlisting}
SELECT AVG(tgs.assists) as avgast, s.season_name
FROM team_game_stats tgs, game g, team t, season s
WHERE s.season_id = g.season_id AND
tgs.game_date = g.game_date AND
tgs.home_id=g.home_id AND
tgs.team_id = t.team_id AND
t.team='GSW'
GROUP BY s.season_name
\end{lstlisting}

%
\mypar{$\query_{nba3}$}
Average points for Lebron James over the years
\begin{lstlisting}
SELECT AVG(points) AS avp_pts, s.season_name
FROM player p, player_game_stats pgs, game g, season s WHERE
p.player_id=pgs.player_id AND
g.game_date = pgs.game_date AND
g.home_id = pgs.home_id AND
s.season_id = g.season_id AND
p.player_name='LeBron James'
GROUP BY s.season_name
\end{lstlisting}

%
\mypar{$\query_{nba4}$}
GSW wins over the years, but we used different 2 seasons
\begin{lstlisting}
SELECT COUNT(*) AS win, s.season_name
FROM team t, game g, season s
WHERE t.team_id = g.winner_id AND
g.season_id = s.season_id AND
t.team= 'GSW'
GROUP BY s.season_name
\end{lstlisting}
%
\mypar{$\query_{nba5}$}
Average points by Jimmy Butler over the years
\begin{lstlisting}
SELECT AVG(points) AS avp_pts, s.season_name
FROM player p, player_game_stats pgs, game g, season s
WHERE p.player_id=pgs.player_id AND
g.game_date = pgs.game_date AND
g.home_id = pgs.home_id AND
s.season_id = g.season_id AND
p.player_name='Jimmy Butler'
GROUP BY season_name
\end{lstlisting}


}
\iftechreport{
\begin{figure}
\begin{minipage}{\linewidth}
    \begin{minipage}[b]{0.45\linewidth}\centering
     {\scriptsize
      \begin{tabular}{l|cc|} \cline{2-3}
  & \cellcolor{grey}\textbf{avg\_pts} & \cellcolor{grey}\textbf{season\_name} \\
    & $2.87$ & 2012-13 \\
    & $6.23$ & 2013-14 \\
    & $11.66$ & 2014-15 \\
    & \cellcolor{green!20}$13.96$ & \cellcolor{green!20}2015-16 \\
    & \cellcolor{red!20}$10.21$ & \cellcolor{red!20}2016-17 \\
    & $11.04$ & 2017-18 \\
    & $7.36$ & 2018-19 \\
\cline{2-3}
\end{tabular}
}
\subcaption{Result of $\query_{nba1}$}
\label{table:case_nba_q1_result}
\end{minipage}

\begin{minipage}[b]{0.45\linewidth}\centering
     {\scriptsize
      \begin{tabular}{l|cc|} \cline{2-3}
  & \cellcolor{grey}\textbf{avg\_ast} & \cellcolor{grey}\textbf{season\_name} \\
    & $22.43$ & 2009-10 \\
    & $22.52$ & 2010-11 \\
    & $22.27$ & 2011-12 \\
    & $22.50$ & 2012-13 \\
    & \cellcolor{green!20} $23.32$ & \cellcolor{green!20} 2013-14 \\
    & \cellcolor{red!20} $27.41$ & \cellcolor{red!20} 2014-15 \\
    & $28.94$ & 2015-16 \\
    & $30.38$ & 2016-17 \\
    & $29.29$ & 2017-18 \\
    & $29.43$ & 2018-19 \\
\cline{2-3}
\end{tabular}
}
\subcaption{Result of $\query_{nba2}$}
\label{table:case_nba_q2_result}
\end{minipage}
\begin{minipage}[b]{0.45\linewidth}\centering
     {\scriptsize
      \begin{tabular}{l|cc|} \cline{2-3}
  & \cellcolor{grey}\textbf{avg\_pts} & \cellcolor{grey}\textbf{season\_name} \\
    & \cellcolor{green!20} $29.71$ & \cellcolor{green!20} 2009-10 \\
    & \cellcolor{red!20} $26.72$ & \cellcolor{red!20} 2010-11 \\
    & $27.15$ & 2011-12 \\
    & $26.79$ & 2012-13 \\
    & $27.13$ & 2013-14 \\
    & $25.26$ & 2014-15 \\
    & $25.26$ & 2015-16 \\
    & $26.41$ & 2016-17 \\
    & $27.45$ & 2017-18 \\
    & $27.36$ & 2018-19 \\
\cline{2-3}
\end{tabular}
}
\subcaption{Result of $\query_{nba3}$}
\label{table:case_nba_q3_result}
\end{minipage}

\begin{minipage}[b]{0.45\linewidth}\centering
     {\scriptsize
      \begin{tabular}{l|cc|} \cline{2-3}
  & \cellcolor{grey}\textbf{avg\_pts} & \cellcolor{grey}\textbf{season\_name} \\
        & $26$ & 2009-10 \\
        & $36$ & 2010-11 \\
        & $23$ & 2011-12 \\
        & \cellcolor{green!20} $47$ & \cellcolor{green!20} 2012-13 \\
        & $51$ & 2013-14 \\
        & $67$ & 2014-15 \\
        & $73$ & 2015-16 \\
        & \cellcolor{red!20} $67$ & \cellcolor{red!20} 2016-17 \\
        & $58$ & 2017-18 \\
        & $57$ & 2018-19 \\
\cline{2-3}
\end{tabular}
}
\subcaption{Result of $\query_{nba4}$}
\label{table:case_nba_q4_result}
\end{minipage}
\begin{minipage}[b]{0.45\linewidth}\centering
     {\scriptsize
      \begin{tabular}{l|cc|} \cline{2-3}
  & \cellcolor{grey}\textbf{avg\_pts} & \cellcolor{grey}\textbf{season\_name} \\

  & $2.60$ & 2011-12 \\
& $8.60$ & 2012-13 \\
& \cellcolor{green!20} $13.10$ & \cellcolor{green!20} 2013-14 \\
& \cellcolor{red!20} $20.02$ & \cellcolor{red!20} 2014-15 \\
& $20.88$ & 2015-16 \\
& $23.89$ & 2016-17 \\
& $22.15$ & 2017-18 \\
& $18.69$ & 2018-19 \\
\cline{2-3}
\end{tabular}
}
\subcaption{Result of $\query_{nba5}$}
\label{table:case_nba_q5_result}
\end{minipage}
\end{minipage}
\caption{Query results and user question tuples for NBA}
\label{fig:query_results_nba}
\end{figure}
}

\mypar{Explanations and Analysis}
$\query_{nba1}$
\revs{\textit{Draymond Green} had a big average points difference between $2$ consecutive seasons.
All $3$ explanation contains salary change information. In reality, from $2015$-$16$ season to $2016$-$17$ season, \textit{Green}'s salary increased, which could result in losing incentive to play as hard as when he earns lower salary. \scnd and \thrd explanation successfully find key game related factors deciding the player's points such as minutes played and shooting percentage (e.g., in \scnd, \textit{Green} had more games where he played more than $31$ minutes and shooting percentage higher than $0.4$ in $2015-16$ season).}
$\query_{nba2}$.
\revs{The \textit{GSW} team had a sudden increase in average assists. All explanations contain \texttt{assistpoints} which has a cause-and-effect relationship with assists (more assists result in more assistpoints).}
$\query_{nba3}$ and $\query_{nba5}$.
\revs{Both players had some significant average point changes. For $\query_{w3}$, \textit{Lebron James} had an average points decrease. This occurred when he switched to a new team (from \textit{CLE} to \textit{MIA} and had less pressure offensively in the following year. \oursys successfully identified this fact as a potential cause (\scnd and \thrd).}$\query_{nba5}$
\revs{\textit{Jimmy Butler} had a big improvement in average points
. Our top explanations to this improvement include an increase of usage and minutes played.}
$\query_{nba4}$.
\revs{This query is similar to our running example $\query_{1}$ but with a question asking for different $2$ seasons. The explanations contain player changes (\scnd, \textit{Andre Iguodala} only played for \textit{GSW} in $2015-17$ season) as well as the team's points difference and 3-point percentage (\thrd). Note that while the first explanation has a high \abbrF, if we look at the join graph details, the salary and player constants can have no relation with GSW at all. This highlights the importance of making join graphs part of explanations.} 
\revs{
\subsection{Case Study: MIMIC}\label{sec:qual-eval-mimic}}
\ifnottechreport{
\mypar{Dataset} \revm{\textit{MIMIC} (\url{https://mimic.physionet.org/}) is a deidentified dataset of intensive care unit hospital admissions. The dataset consists of 6 relations and is $\sim$ 120MB large. 
}}
\iftechreport{
We constructed $5$ queries over the MIMIC dataset accessing different tables. The simplified descriptions of the queries, user questions, and explanations \Cref{tab:expl-mimic-queries}.
  To help the reader understand the queries and explanations, we first briefly introduce the MIMIC dataset using the example below.}
\ifnottechreport{
\revs{We constructed $5$ queries over this dataset accessing different tables. The simplified descriptions of the queries, user questions, and explanations \Cref{tab:expl-mimic-queries}.
  We first briefly introduce the MIMIC dataset to help the reader understand the queries and explanations. The main table of the dataset is the \texttt{Admissions} table that records hospital admissions. The \texttt{Diagnosis} table records diagnoses for patients for each admission (a patient may be admitted more than once during their lifetime). \texttt{PatientsAdmissionInfo} records information like age and religion for individual admissions of patients (e.g., age may change over time). \texttt{ICUStays} record intensive care unit stays of patients. There may be multiple ICU stays per admission.
} 
}

\iftechreport{
\begin{example}\label{exp:MIMIC}
Consider the (simplified) \textbf{MIMIC-III Critical Care database} \cite{johnson2016mimic} with the following relations (the keys are underlined).
\begin{itemize}[leftmargin=*]
\item {\tt Admissions(\underline{adid}, dischargeloc, adtype, insurance, isdead, HospitalStayLength)} contains information 
about hospital admissions.
\item {\tt Diagnoses(\underline{pid, adid, did}, diagnosis, category)} records the diagnoses for each patient (identified by {\tt pid}) during each admission, one patient could have multiple diagnoses during one admission, which are identified by \emph{did}.
\item {\tt PatientsAdmissionInfo(\underline{pid,  admid}, age, religion,} \tt{ethnicity)} records the information from a patient upon admission to the hospital. Note that one patient could have multiple entries in this table because one patient could have multiple admissions during their lifetime.
\item {\tt ICUStays(\underline{pid, admid, iid}, staylength, icutype)} records the information of the ICU stays of patients. \textit{iid}  identifies different ICU stays within one admission.
\end{itemize}
Suppose a data analyst writes the following query to find out the relationship between the insurance type 
and the death rate:
\iftechreport{\begin{lstlisting}[mathescape=true]
$\query_2 = $ SELECT insurance, sum(isdead) / count(*) AS death_rate,
            count(*) AS admit_cnt
     FROM Admissions GROUP BY insurance
\end{lstlisting}

}

Given the result of this query in \Cref{Result of Q_2}, the analyst may ask $UQ_2$: \emph{given similar numbers of admissions, why is the death rate of patients with insurance type=$\abbrMed$ more than $2$ times larger ($t_1$) than that of patients with insurance=$\abbrPriv$}. \Cref{tab:expl-mimic-queries} lists the top-3 explanations returned by \oursys{} (referred to as $Q_{mimic4}$), which shows that among the deaths, $\abbrMed$ has a larger fraction of admissions in emergency and male older patients ($age  < 71$) compared to $\abbrPriv$. 
This is aligned with the fact that \textit{Medicare} is mostly for patients who are over $65$ years old. \scnd and \thrd are stating that patients using \textit{Medicare} has more admission because of \textit{emergency} and also facts about length of hospital stays. 
\end{example}
}

%
\iftechreport{
\begin{figure}[t]\scriptsize\setlength{\tabcolsep}{3pt}
 \vspace{-2mm}
\subfloat[\small Result of $\query_2$]{
\begin{minipage}{0.45\linewidth}
	{\scriptsize
	\begin{tabular}{l|ccc|} 
              & \cellcolor{grey}\textbf{insurance} & \cellcolor{grey}\textbf{death rate} & \cellcolor{grey}\textbf{admit\_cnt} \\ \cline{2-4}
	$t_1 \to$ & \cellcolor{green!20} $\abbrMed$    & \cellcolor{green!20} 0.14           & \cellcolor{green!20} 28215          \\
              & Self Pay                           & 0.16                                & 611                                 \\
              & Government                         & 0.05                                & 1783                                \\
	$t_2 \to$ & \cellcolor{red!20} $\abbrPriv$     & \cellcolor{red!20} 0.06             & \cellcolor{red!20} 22582            \\
              & Medicaid                           & 0.07                                & 5785                                \\\cline{2-4}
	\end{tabular}
	}
	\label{Result of Q_2}
\end{minipage}
}\hfill
   \subfloat[\small User question $\uquestion_2$]{
\begin{minipage}[b]{0.45\linewidth}
    \centering
\begin{tcolorbox}[colback=white,left=1pt,right=1pt,top=0pt,bottom=0pt]
$\uquestion_2$: \emph{Given close number of admissions why do patients using $\abbrMed$ insurance have higher death rates ($t_1$, 14\%) compared to patienst with
$\abbrPriv$ insurance ($t_2$, 6\%)?}
\end{tcolorbox}
    \label{fig:uq2}
\end{minipage}
}
 \vspace{-2mm}
 \caption{\label{fig:running-mimic}\small
  $\query_2$ results and user question for Example ~\ref{exp:MIMIC}.}
    \vspace{-4mm}
\end{figure}
}

\iftechreport{\Cref{tab:mimic-queries-desc} shows a description for each these queries, which include group-by aggregations over path joins.
We present the SQL code of these queries below. Their results and the tuples used in the user question (highlighted rows) are shown in \Cref{fig:query_results_mimic}.

\mypar{$\query_{mimic1}$}
Count of diagnoses over different chapters.\\
User question: chapter=$2$ (neoplasms) VS chapter=$13$ (diseases of the musculoskeletal system and connective tissue)
\begin{lstlisting}
SELECT 
1.0*SUM(a.hospital_expire_flag)/count(*) AS death_rate,  
d.chapter FOM admissions a, diagnoses d 
WHERE a.hadm_id=d.hadm_id GROUP BY d.chapter
\end{lstlisting}
\mypar{$\query_{mimic2}$}
Death rate of patients grouped by their insurance\\
User question: insurance=`Self Pay' VS insurance=`Private'
\begin{lstlisting}
SELECT insurance,
1.0*SUM(hospital_expire_flag)/COUNT(*) AS death_rate
FROM admissions  GROUP BY insurance;
\end{lstlisting}
\mypar{$\query_{mimic3}$}
Length of stays in ICU group by length of stays (los\_group)\\
User question: los\_group=`$>8$' VS los\_group=`$0\mbox{-}1$'
\begin{lstlisting}
SELECT COUNT(*) AS cnt, los_group
FROM icustays GROUP BY los_group;
\end{lstlisting}
\mypar{$\query_{mimic4}$}
Death rate of patients grouped by their insurance\\
User question: insurance=`Medicare' VS insurance=`Private'
\begin{lstlisting}
SELECT insurance,
1.0*SUM(hospital_expire_flag)/COUNT(*) AS death_rate
FROM admissions  GROUP BY insurance;
\end{lstlisting}
\mypar{$\query_{mimic5}$}
Number of procedures group by ethnicity \\
User question: ethnicity=`ASIAN' VS ethnicity=`HISPANIC'
\begin{lstlisting}
SELECT COUNT(*) AS cnt, pai.ethnicity 
FROM patients_admit_info pai, procedures p 
WHERE p.hadm_id=pai.hadm_id AND p.subject_id=pai.subject_id 
GROUP BY pai.ethnicity
\end{lstlisting}
}

\iftechreport{
\begin{figure}
\begin{minipage}{0.35\linewidth}
\subfloat[\small Result of $\query_{mimic1}$]{
\begin{minipage}{0.9\linewidth}\centering
     {\scriptsize
      \begin{tabular}{l|cc|} \cline{2-3}
  & \cellcolor{grey}\textbf{death\_rate} & \cellcolor{grey}\textbf{chapter} \\
    & $0.01$ & $11$ \\
    & $0.02$ & $15$ \\
    & $0.05$ & $14$ \\
    & $0.08$ & $5$ \\
    & $0.09$ & V \\
    & \cellcolor{green!20}$0.09$ & \cellcolor{green!20}$13$ \\
    & $0.10$ & E \\
    & $0.12$ & $7$ \\
    & $0.12$ & $3$ \\
    & $0.13$ & $17$ \\
    & $0.13$ & $6$ \\
    & $0.14$ & $12$ \\
    & $0.14$ & $4$ \\
    & $0.14$ & $9$ \\
    & $0.15$ & $10$ \\
    & $0.16$ & $16$ \\
    & $0.18$ & $8$ \\
    & \cellcolor{red!20}$0.19$ & \cellcolor{red!20}$2$ \\
    & $0.19$ & $1$ \\
\cline{2-3}
\end{tabular}
}
\label{table:workload_mimic_q1_result}
\end{minipage}
}
\end{minipage}
\begin{minipage}{0.55\linewidth}
\subfloat[\small Result of $\query_{mimic2}$]{
\begin{minipage}{0.9\linewidth}\centering
     {\scriptsize
      \begin{tabular}{l|cc|} \cline{2-3}
  	& \cellcolor{grey}\textbf{insurance} & \cellcolor{grey}\textbf{death rate}\\
	& \cellcolor{green!20} \cellcolor{green!20} Medicare & \cellcolor{green!20} 0.14 \\
	& Self Pay & 0.16 \\
	& Government & 0.05 \\
	& \cellcolor{red!20} Private & \cellcolor{red!20} 0.06 \\
	& Medicaid & 0.07 \\ \cline{2-3}
\cline{2-3}
\end{tabular}
}
\label{table:workload_mimic_q2_result}
\end{minipage}
}\\
\subfloat[\small Result of $\query_{mimic3}$]{
\begin{minipage}{0.9\linewidth}\centering
     {\scriptsize
      \begin{tabular}{l|cc|} \cline{2-3}
  & \cellcolor{grey}\textbf{cnt} & \cellcolor{grey}\textbf{los\_group} \\
& $16901$ & $1$-$2$ \\
& \cellcolor{green!20} $8605$ & \cellcolor{green!20}x$>8$ \\
& \cellcolor{red!20} $15034$ & \cellcolor{red!20}$2$-$4$ \\
& $12311$ & $0$-$1$ \\
& $8671$ & $4$-$8$ \\
\cline{2-3}
\end{tabular}
}
\label{table:workload_mimic_q3_result}
\end{minipage}
}\\
\subfloat[\small Result of $\query_{mimic4}$]{
\begin{minipage}{0.9\linewidth}\centering
     {\scriptsize
      \begin{tabular}{l|cc|} \cline{2-3}
  	& \cellcolor{grey}\textbf{insurance} & \cellcolor{grey}\textbf{death rate}\\
	& \cellcolor{green!20} \cellcolor{green!20} Medicare & \cellcolor{green!20} 0.14 \\
	& Self Pay & 0.16 \\
	& Government & 0.05 \\
	& \cellcolor{red!20} Private & \cellcolor{red!20} 0.06 \\
	& Medicaid & 0.07 \\ \cline{2-3}
\cline{2-3}
\end{tabular}
}
\label{table:workload_mimic_q4_result}
\end{minipage}
}
\end{minipage}

\subfloat[\small Result of $\query_{mimic5}$]{
\begin{minipage}{0.45\linewidth}\centering
     {\scriptsize
      \begin{tabular}{l|rc|} \cline{2-3}
  & \cellcolor{grey}\textbf{cnt} & \cellcolor{grey}\textbf{ethnicity} \\
&  $11$ & South American \\
&  $64$ & Pacific Islander \\
& $191$ & Middle Eastern \\
& $568$ & Multi-Race Ethnicity \\
& $2641$ & Declined To Answer \\
& $4244$ & Unable To Obtain \\
& $6056$ & Other \\
& \cellcolor{green!20} $6247$ & \cellcolor{green!20} Asian \\
& \cellcolor{red!20} $7821$ & \cellcolor{red!20} Hispanic \\
& $19579$ & Black \\
& $22710$ & Unkown \\
& $169478$ & White \\
\cline{2-3}
\end{tabular}
}
\label{table:workload_mimic_q5_result}
\end{minipage}
}
\caption{Query results and user question tuples for MIMIC}
\label{fig:query_results_mimic}
\end{figure}
}


\iftechreport{
\begin{table}[t]
    \centering
    {\scriptsize
    \begin{tabular}{|c|p{15em}|c|}
  \hline
  Num. & Description & Tables used \\\hline
  $\query_{mimic1}$ & Return the death rate of diagnosis by chapter & diagnoses, admissions  \\
  $\query_{mimic2}$ & Returns the death rate of patients grouped by their insurance.  & admissions \\
  $\query_{mimic3}$ & Number of ICU stays grouped by the length of stays (\texttt{los\_group}). & icustays \\
  $\query_{mimic4}$ & Number of procedures for a particular chapter (group of diagnosis types). & procedures  \\
  $\query_{mimic5}$ & Number of procedures among different ethnicities. & patients\_admit\_info, procedures \\
  \hline
  \end{tabular}
    \caption{MIMIC queries}
    \label{tab:mimic-queries-desc}
}
\end{table}
}

\begin{table}[t]
{    \centering
    \scriptsize
    \begin{tabular}{|l|p{11em}|m{15em}|l|}
   \hline
   \textbf{Query} & \textbf{User question} &  \textbf{Top-3 explanations} & \textbf{\abbrF} \\\hline
   $\query_{mimic1}$ &  \multirow{3}{11em}{\textbf{Patient death rate grouped by diagnoses}: $0.19$ for chapter$=2$ ($t_1$) VS $0.09$ for chapter$=13$ ($t_2$)} & expire\_flag=1 [$t_1$] & $0.68$ \\
       \cline{3-4}
       & & hospital\_stay\_length$<23\wedge$expire\_flag$=1$ [$t_1$]& $0.65$ \\
       \cline{3-4}
       & & hospital\_stay\_length$<16$,expire\_flag$=1$ [$t_1$]& $0.63$ \\
   \cline{1-4}
   $\query_{mimic2}$ & \multirow{3}{11em}{\textbf{Death rate by insurance}: Medicare=$0.138$ ($t_1$) VS Medicaid=$0.066$ ($t_2$)} & prov.admission\_type=emergency [$t_1$]& $0.85$\\
    \cline{3-4}
    & & expire\_flag$=1$ [$t_1$] & $0.7$ \\
    \cline{3-4}
    & & gender=Male [$t_1$] & $0.65$ \\
    \cline{1-4}
  $\query_{mimic3}$ & \multirow{3}{11em}{\textbf{Number of patients grouped by ICU stays length}: less than $1$ day ($t_1$) VS more than $8$ days ($t_2$)} & \makecell[l]{hospital\_stay\_length$>9$ $\wedge$ \\ procedure.chapter=16 [$t_2$]} & $0.94$\\
  \cline{3-4}
  & & \makecell[l]{hospital\_stay\_length$<6$ $\wedge$ \\ los\_group=$0\mbox{-}1$ [$t_1$]} & $0.77$ \\
  \cline{3-4}
  & & \makecell[l]{prov.dbsource=carevue $\wedge$ \\  hospital\_stay\_length$>8$ [$t_2$]}& $0.71$ \\
  \cline{1-4}
   $\query_{mimic4}$ & \multirow{3}{11em}{\textbf{Death rate by insurance}: Medicare=$0.14$ ($t_1$) VS Private=$0.06$ ($t_2$)} & expire\_flag=$0\wedge$age$<71$ [$t_2$]& $0.77$ \\
   \cline{3-4}
   & & prov.admission\_type=\textit{emergency} [$t_1$] & $0.73$ \\
   \cline{3-4}
   & & \makecell[l]{prov.hospital\_stay\_length$<22.0$ $\wedge$ \\ expire\_flag=$1$ [$t_2$]} & $0.61$\\
   \cline{1-4}
   $\query_{mimic5}$& \multirow{3}{11em}{\textbf{Number of patients that did a procedure grouped by ethnicity}: $7821$ Hispanic patients [$t_1$] VS $6247$ Asian patients [$t2$]} & hospital\_stay\_length$<19$ $\wedge$ ethnicity=Asian [$t_2$]& $0.89$\\
   \cline{3-4}
   & & \makecell[l]{admission\_type=emergency $\wedge$ \\ hospital\_stay\_length$>5$ $\wedge$ age$<66$ \\ $\wedge$ ethnicity=Hispanic [$t_1$]}& $0.80$ \\
   \cline{3-4}
   & & prov.religion=Catholic [$t_1$]& $0.63$ \\
   \cline{1-4}
   \end{tabular}
    \caption{Queries, user questions and explanations (MIMIC)}
    \label{tab:expl-mimic-queries}
    }
\end{table}
\mypar{Explanations and Analysis}\revs{
  \Cref{tab:expl-mimic-queries} shows the top-3 explanations for returned by \oursys{} for each user question.}
$\query_{mimic1}$.
\revs{This question asks for the difference in death rates between two diagnosis categories (chapter 2: neoplasms vs chapter 13: musculoskeletal system and connective tissue). The death rate is the fraction of patients that died during their hospital stay. The \frst explanation  uses \textit{expire\_flag}$=1$ from patient table suggesting that this patient has passed away. This flag only indicates that the patient died, but not whether insides or outside the hospital, subsuming all hospital deaths. The \scnd and \thrd explanation add additional information about the  lengths of hospital stays which could help user imply the severity differences between two categories thus infer the possible reason for different death rates.}
$\query_{mimic2}$.
\revs{This query asks about differences between the death rate of patients based on their insurance. The \frst explanation states that there are more \textit{emergency} admissions with \textit{Medicare} than with \textit{Medicaid} which may explain the higher death rate. The \scnd explanation related death rate to the \texttt{expire\_flag}. The \thrd explanation suggests that \textit{Medicare} has more \textit{Male} patients than \textit{Medicaid}.}
$\query_{mimic3}$.
\revs{
The \frst explanation shows that most of the patients staying over $8$ days in ICU will stay in hospital for more than $9$ days and also have procedure that is from chapter $16$ (\textit{Miscellaneous Diagnostic and Therapeutic Procedures}). The \scnd and \thrd explanations are both related to the hospital stay length. The \scnd explanation suggests that most patients will be released from hospital in less than $6$ days when  their ICU stay is less than $1$ day. The \thrd explanation states the same facts for patients that have more than $8$ days of hospital stay when they stay more than $8$ days in the ICU. These explanations regarding hospital stay length can help users identify that ICU stay length may be a good indicator for ICU stay lengths}
$\query_{mimic4}$.
\revs{
This question uses the same query as $\query_{mimic2}$, but compares \textit{Private} insurance with \textit{Medicare}.  The \frst explanation states that for patients who have \textit{Private} insurance, more patients are alive and less than $71$ years old. This is aligned with the fact that \textit{Medicare} is mostly for patients over $65$ years old (this is a fact from online resources). The \scnd and \thrd explanations are stating that patients using \textit{Medicare} are more likely to be admitted because of an \textit{emergency} and also facts about length of hospital stays.
}
$\query_{mimic5}$
\revs{
The \frst explanation states that \textit{Asian} patients that had a procedure are more likely to stay less than 19 days in the hospital.
The \scnd explanation says that compared to \textit{Asian} patients there were more \textit{Hispanic}  patients younger than $66$ years old and stayed more than $5$ days in the hospital. The  \thrd explanation points out that more \textit{Hispanic} patients are \textit{Catholic}. Note that the ethnicity information appeared in explanations are not from PT, but from a different \texttt{patient\_admit\_info} table. Because we do not currently consider functional dependencies, results like this cannot be avoided. We plan to address this in future work}.


\vspace{0.2cm}
\subsection{User Study}\label{sec:user-study}
\revs{We conducted a user study for the NBA dataset to evaluate:
(S1) whether
\oursys{} provides meaningful explanations in addition to explanations that only come from the provenance, and (S2) whether
 the \oursys{}'s quality metric is consistent with user preference.}

\smallskip
\noindent\revs{\textbf{Participants.} We recruited 20 participants --- all of them are graduate students studying \iftechreport{in different areas of} computer science,
13 of them have some prior experience with SQL, and 5 were NBA fans.}

\vspace{2mm}
\noindent\revs{\textbf{Tasks.} We first presented background knowledge on the NBA 
  to each participant, and explained the schema of the dataset. Each participant was shown the SQL query $\query'_1$ (shown below) and the results of this query, and then was asked to find and evaluate explanations for the user question $\uquestion_1$ from \Cref{exp:GSW}: ``\emph{Why did $\abbrGSW$ win 73 games in season $2015 \mbox{-} 16$ compared to 47 games in $2012 \mbox{-} 13$.}''. 
}
{\scriptsize
\begin{lstlisting}[escapeinside=!!]
!$\query'_1=$! SELECT s.season_name, count(*) AS win
FROM team t, game g, season s WHERE t.team_id = g.winner_id
   AND g.season_id = s.season_id AND t.team = 'GSW'
GROUP BY s.season_name
\end{lstlisting}
}
\revs{We gave each participant 20 minutes to explore the dataset and manually find explanations. For those unfamiliar with SQL, we skipped this step. 
  Participants then were asked to rate each of the top-5 explanations with the highest \abbrFs produced based on   \textit{provenance} and by \textit{\oursys{}} (\Cref{tab:user_study_expls})
using a scale from 1 to 5. We also asked them which set of explanations makes more sense and whether they provided new insights. 
Because the top explanations by \oursys{} have higher \abbrFs
, we added one with a low \abbrF ($Expl_{10}$) as a control. By covering a wider range of \abbrF values, we can test S2: (1) can participants distinguish between low and high score explanations, and (2) do participants agree with our ranking based on our quality measure.}

\begin{table}[t]
    \centering
    \scriptsize
    \begin{tabular}{c|p{32em}}
\multicolumn{2}{c}{\textbf{Provenance-based Explanations}}\\\midrule
$Expl_1$ & In season 2015-16, among the games GSW won, they were the visiting team and had points $> 104$ in 28 games (10 games in 2012-13, resp.)  \\\midrule
$Expl_2$ & In 2015-16 season, 73 games (47 games in 2012-13, resp.) GSW won are regular season games. \\\midrule
$Expl_3$ & In 2015-16 season, among the games GSW won, they were the visiting team, had points $> 98$ and possessions $> 101$ in 17 games (0 games in 2012-13, resp.)\\\midrule
$Expl_4$ & In 2015-16 season, GSW scored more than 104 points in each of 64 games (24 games in 2012-13, resp.) GSW won. \\\midrule
$Expl_5$ & In 2015-16 season, the home teams had points $< 106$ and possessions $< 101$ in each of 29 games (40 games in 2012-13, resp.) GSW won. \\\midrule
\multicolumn{2}{c}{\textbf{\oursys{}}}\\\midrule
$Expl_6$ & In 2015-16 season, the number of games with GSW player Stephen Curry's minutes $< 38$ and usage $> 25$ is 59 games (12 games in 2012-13, resp.) GSW won. \\\midrule
$Expl_7$ &In 2015-16 season, the number of games with GSW player Draymond Green's minutes>15 is 73 games (15 games in 2012-13, resp.) GSW won.\\\midrule
$Expl_8$ & In 2015-16 season, Jarrett Jack played in 0 games (47 games in 2012-13, resp.) GSW won.\\\midrule
$Expl_9$ &In 2015-16 season, GSW had three\_pct $> 35\%$ and points $> 112$ in each of 39 games (9 games in 2012-13, resp.) GSW won.\\\midrule
$Expl_{10}$ & In 2015-16 season, GSW had fg\_three\_pct $> 48\%$ and points $> 112$ and rebounds $> 51$ in 5 games (2 games in 2012-13, resp.) GSW won.
\\\bottomrule
    \end{tabular}
    \caption{Explanations for $\uquestion_1$ used in the user study }
    \label{tab:user_study_expls}
    \ifnottechreport{
    \vspace{-8mm}
    }
\end{table}

\smallskip
\noindent
\revs{\textbf{Results and Analysis.}\label{sec:study-results}
Overall, the responses were positive: 16 out of 20 participants agreed that explanations by \oursys{} make more sense to them and seeing explanations by \oursys{} in advance will help them find explanations that they did not think about before.
}
\begin{table}[!ht]
    \small
    \resizebox{0.48\textwidth}{!}{
    \begin{tabular}{c|ccccc|ccccc}\toprule
        & \multicolumn{5}{c|}{Provenance-based} & \multicolumn{5}{c}{\oursys{}}    \\
        & Expl1 & Expl2 &Expl3 &  Expl4 &Expl5 & Expl6 & Expl7 &Expl8 &  Expl9 &Expl10  \\ \midrule
        All	users & 3.150 & 1.450 & 3.950 & 3.600 & 2.750 & 3.600 & 3.800 & 2.350 & 3.950 & 2.300\\
        Stdev & 1.040 & 0.999 & 0.759 & 1.095 & 1.410 & 0.883 & 1.196 & 1.424 & 0.999 & 1.174 \\ 
        NBA: Yes & 3.400 & 1.800 & 3.800 & 3.600 & 2.800 & 3.800 & 3.800 & 2.800 & 4.200 & 2.600 \\
        NBA: No & 3.067 & 1.333 & 4.000 & 3.600 & 2.733 & 3.533 & 3.800 & 2.200 & 3.867 & 2.200 \\\midrule
        \abbrF & 	0.69& 0.56 & 0.38 & 0.8 & 0.4 & 0.82 &0.91 & 1 & 0.64 & 0.13\\
        recall & 0.38 & 1 & 0.23 & 0.87 & 0.4 &0.81 & 1 & 0.99 & 0.53 & 0.07 \\
        precision & 0.74 & 0.61 & 1 & 0.73 &0.4 & 0.83 & 0.83 & 0.99 &0.81 & 0.7 \\\bottomrule
    \end{tabular}}
    \caption{\revs{\small Average ratings for each explanation by users with different expertise and the measures for each explanation by \oursys{}}} \label{tab:avg_score}
    \ifnottechreport{
    \vspace{-7mm}
    }
\end{table}

\revs{\Cref{tab:avg_score} shows the average user ratings and quality measures for each explanation. Regarding (S1), the average ratings of the top-1 explanation are the same for both methods ($Expl_9: 3.95$ vs $Expl_3: 3.95$, both explanations summarize the team statistics of GSW while $Expl_9$ refers to the table $team\_game\_stats$ not in the provenance). For the next two explanations, $Expl_7$ and $Expl_6$ (\oursys{}) summarize the statistics of two GSW's key players and have higher average ratings ($Expl_7: 3.8$ vs $Expl_5: 3.6$, $Expl_6: 3.6$ vs $Expl_1: 3.15$). The margin is larger for participants who are familiar with the NBA (4.2 vs 3.8, 3.8 vs 3.6, 3.8 vs 3.4). }

\revs{To answer (S2), we find that explanations with high user ratings (Expl 3, 4 , and Expl 6, 7, 9) have a positive relation with high \abbrF and precision. The only exception, $Expl_8$ of \oursys{}, is also the most controversial one, indicated by the largest standard deviation. Evaluating these explanations is subjective and requires domain knowledge: the player Jack in $Expl_8$ left GSW in 2013, and people may or may not regard this as a signal that the team had begun relying more on other players 
who play a similar position to Jack.
Next, we evaluate the ranking results of our quality measures by regarding each participant's ratings as the ground truth.
We use Kendall-Tau rank distance~\cite{kendall1938new} for measuring pairwise ranking error and normalized discounted cumulative gain (NDCG)~\cite{jarvelin2002cumulated} for the entire ranked list. As shown in \Cref{tab:ranking}, ranking by precision gives the lowest pairwise ranking error for the provenance-based method, while for our method it is ranking by \abbrF. If we drop the most controversial explanation, the pairwise error is reduced by more than half. The $NDCG_n$ for \oursys{} reaches 0.9 for all cases and even 0.95 after dropping the most controversial explanation. 
}

\begin{table}[t]
\resizebox{0.40\textwidth}{!}{
\begin{tabular}{c|cc|cc}\toprule
 \multicolumn{3}{c}{}&  Explanation Tables (All / -1) & \oursys{} (All / -1) \\ \midrule
  \multirow{6}{3.5em}{Avg. Kendall tau rank distance} & \multirow{3}{5em}{All users} & \abbrF & 3.95 / 2.2 & 3.9 / 1.4\\
& &recall & 5.9 / 3.85 & 3.3 / 1.4 \\
& &precision &  2.2 / 0.95 & 3.9 / 1.4\\\cline{2-5}
& \multirow{3}{5em}{Users with domain knowledge} & \abbrF & 3.6 / 2.0 & 3.2 / 1.8\\
& &recall & 5.2 / 3.2 & 3.8 / 1.8 \\
& &precision &  2.2 / 1.2 & 4.2 / 1.8\\
\midrule
\multirow{6}{3.5em}{Avg. $NDCG_n$} & \multirow{3}{5em}{All users} & \abbrF & 0.875 / 0.882 & 0.901 / 0.955\\
& &recall & 0.844 / 0.852 & 0.901 / 0.955 \\
& &precision &  0.933 / 0.965 & 0.901 / 0.955\\\cline{2-5}
 & \multirow{3}{5em}{Users with domain knowledge} & \abbrF & 0.897 / 0.901 & 0.903 / 0.954\\
& &recall & 0.862 / 0.878 & 0.903 / 0.954 \\
& &precision &  0.953 / 0.977 & 0.903 / 0.954 \\
\bottomrule
\end{tabular}}
\caption{\revs{Ranking quality: all 5 explanations (\emph{All}),  dropping the explanation with the largest standard deviation (\emph{-1})}.} \label{tab:ranking}
\ifnottechreport{
\vspace{-8mm}
}
\end{table}
\noindent\revs{\textbf{Takeaways.} The main findings are: (1) the majority (16/20) of participants preferred our method, thanks to the new information provided by tables not used in the query, which complements the explanations only based on provenance; (2) our quality measures are consistent with participants' preference; (3) for both methods, there can be top explanations rated low by participants, which is as expected because we did not do causal analysis, and validating such explanations may be subjective and depend on domain expertise; and (4) participants with domain knowledge had a stronger preference for our method than participants without domain knowledge.}

\noindent\revs{\textbf{Other findings and discussion.}
Finally, it is also worth noting that the participants' feedback supports the motivation of \oursys{}. For example, participants found that
\textit{``The usage of Stephen Curry increases in 2015-16.''}, \textit{```Players play both season (12-13 and 15-16) have higher point per game and assist per game''} before they saw the explanations by \oursys{}. One suggested to use health information of the players in explanation. 
Another participant remarked that \textit{``the use of other tables in the database to explore how the contributions of individual players can have an outcome on the team's performance produced explanations that were more novel or interesting''}.
}

\section{Related Work}
\label{sec:related}
\mypar{Provenance and summarization}
Provenance \cite{DBLP:journals/ftdb/CheneyCT09} for relational queries refers to records of the inputs that contribute to the results of a query, which has been studied extensively.
For non-aggregate queries, why-provenance~\cite{buneman2001and} returns a set of input tuples responsible for a given output tuple;
how-provenance~\cite{green2007provenance} encodes how the query combined input tuples to generate the answers.
For aggregate queries, symbolic representation of semiring~\cite{amsterdamer2011provenance} is used to express how aggregate results are computed.
Given the significant cost of managing provenance information in practical DBMS,
different provenance-management frameworks to store and retrieve relevant provenance information have been proposed in the literature \cite{glavic2009perm, arab2014generic,chapman2008efficient,psallidas2018smoke,lee2019pug}. Some of them  support provenance for aggregate queries
using simplified models and query plan optimizations~\cite{karvounarakis2010querying, psallidas2018smoke,lee2019pug}. A number of recent papers have  proposed \emph{summarization} techniques to
represent provenance approximately \cite{ainy2015approximated,re2008approximate,lee2017integrating,LL20}, or using summarization rules for better usability \cite{alomeirsummarizing}.
Factorized and summarized provenance in natural language \cite{deutch2016nlprov,deutch2017provenance} has also been studied for ease of user comprehension.

\mypar{Data summarization}
There is a line of work on summarizing relational data
where informative summarization is provided with focuses on relevance, diversity, and coverage~\cite{DBLP:journals/pvldb/QinYC12, joglekar2017interactive, wen2018interactive,kim2020summarizing}.
For a relation augmented with a binary outcome attribute,
Gebaly et al. ~\cite{DBLP:journals/pvldb/GebalyAGKS14} developed solutions to find optimally-informative summarizations
of \reva{categorical} attributes affecting the outcome attribute \reva{only considering the provenance (and not other relevant relations like our work), while its extension \cite{vollmer-19-insnd} considers numeric data.}  We adopt one of the optimizations in \cite{DBLP:journals/pvldb/GebalyAGKS14}, the use of \emph{lowest common ancestor} (LCA) patterns in our algorithms to prune the search space,
to generate pattern candidates from a sample to avoid the costly computation of the \emph{data cube} of the sample. \reva{We discussed the potential problems in adapting the approach in \cite{DBLP:journals/pvldb/GebalyAGKS14} using a materialized augmented provenance table in our experiments.}
\cut{Note that \cite{DBLP:journals/pvldb/GebalyAGKS14} only supports categorical attributes and only considers the summarization of a single table, while our framework supports numerical attributes in summarization patterns, differentiates two output tuples by the summarized provenance, and provides additional information by considering the context in other relevant relations.
}

\mypar{Explanations for query answers}
This line of work aims at explaining unexpected outcomes in a query result,
including outlier values, missing tuples, or existing tuples that should not exist.
Provenance and provenance summaries provide a straightforward form of explanations \cite{WuM13,RS14, ROS15, abuzaid2018diff} by characterizing a set of tuples whose removal or modification would affect the
query answer of interest;
while query-based explanations, i.e., changes to queries, are investigated for ``why'' and ``why-not'' questions~\cite{chapman2009not, DBLP:conf/edbt/BidoitHT14}.
Explanations for  outliers have been studied in \cite{DBLP:conf/sigmod/MiaoZGR19,bessa2020effective}.
\reva{We share with \cite{DBLP:conf/sigmod/MiaoZGR19} the motivation of considering explanations that are not (solely) based on provenance. The difference is that in \cite{DBLP:conf/sigmod/MiaoZGR19}, the same table used in the query is considered for finding explanations that ``counterbalance'' an outlier by learning patterns and balancing a low (high) outlier with a high (low) outlier w.r.t. a pattern, whereas we find explanations in ``augmented provenance'' stemming from tables not used in the query. Therefore, \cite{DBLP:conf/sigmod/MiaoZGR19} is orthogonal to our work. 
}

\mypar{Join path discovery}
Join path discovery has been studied for finding data related to a table of interest based on inclusion dependencies or string similarity~\cite{SF12a, FA18, ZD19, DBLP:journals/pvldb/HeGC15, DBLP:journals/pvldb/ZhuHC17}. Recently, ~\cite{DBLP:conf/sigmod/KumarNPZ16, DBLP:journals/pvldb/ShahKZ17,DBLP:journals/pvldb/ChepurkoMZFKK20} studied the performance of  machine learning models trained on join results. \oursystem{} can utilize join graph discovery techniques to find more augmentation opportunities. 


\vspace{2mm}
\section{Discussions and Future Work}
\label{sec:conclusions}
\revm{Explanations for database query answers is a relatively new research topic with lots of interesting future directions. } \reva{For instance, currently our approach only considers correlations. In the future, we plan to integrate it with existing \emph{observational causal analysis} framework from AI and Statistics \cite{PearlBook2000, rubin2005causal} to find causal explanations.}
\cut{
We present a novel approach that explains the difference between two answer tuples using summarization of join-augmented provenance.
Our framework also applies to explain the high/low value of a single outlier tuple by comparing with all other answer tuples. Instead of just summarizing the standard provenance coming from the relations appearing the query, our approach automatically explores context when generating explanations by augmenting the provenance of a query result joining with related tables using a schema graph and join graphs. We present algorithms with a suite of optimizations that mine patterns with top scores from a given augmented table as well as to mine join graphs that lead to informative patterns efficiently.

Automated explanations to different types of user questions is a  relatively new area and there are several interesting directions of future work. First, we did not consider the actual value or type of the aggregate function in the explanations; by combining with some of previous works in the literature one can see if other types of explanations can be found. Second, we can study the variations of the objective functions.
Although our optimization problem
has polynomial data complexity, if we intend to find a set of patterns that together explain the provenance and define \abbrF on them, the problem resembles set cover and is likely to be NP-hard. We can see if this approach returns more informative explanations with diversity.

}
 \revb{Another interesting direction for future work is to integrate context-based explanations with join discovery techniques (e.g., \cite{FA18,ZD19}) to automatically find datasets to be used as context.} \revb{Finally, our approach is not suited well for textual and sparse} \revs{data} \revb{because such data cannot be summarized well using the type of patterns we support since values are rarely repeated. Different summarization techniques (e.g., using taxonomies) or preprocessing techniques (e.g., information extraction techniques) would have to be incorporated with our approach. }
 \revc{While we discussed simple SQL aggregate queries, extensions of our model can be studied for more general queries (with nested sub-queries, negation) if we have access to a provenance system that can compute the query provenance.
 }
 \revb{Beyond having an intuitive scoring function to rank explanations that may not always produce meaningful explanations,} \reva{a challenging future work is to evaluate the correctness of provided explanations without much human interventions,} \revb{to evaluate whether the returned explanations consider the intent of the user in the `why' question}, \reva{and to have a confidence score on the explanations by deciding whether the data has enough information to explain a user question.}
 Finally, one can explore other types of user questions than our type of comparison question (like explaining an increasing/decreasing trend or \revc{explain why two results are similar}) and applicability of context/provenance in ML applications.


\clearpage
\bibliographystyle{ACM-Reference-Format}

\begin{thebibliography}{54}


\ifx \showCODEN    \undefined \def \showCODEN     #1{\unskip}     \fi
\ifx \showDOI      \undefined \def \showDOI       #1{#1}\fi
\ifx \showISBNx    \undefined \def \showISBNx     #1{\unskip}     \fi
\ifx \showISBNxiii \undefined \def \showISBNxiii  #1{\unskip}     \fi
\ifx \showISSN     \undefined \def \showISSN      #1{\unskip}     \fi
\ifx \showLCCN     \undefined \def \showLCCN      #1{\unskip}     \fi
\ifx \shownote     \undefined \def \shownote      #1{#1}          \fi
\ifx \showarticletitle \undefined \def \showarticletitle #1{#1}   \fi
\ifx \showURL      \undefined \def \showURL       {\relax}        \fi
\providecommand\bibfield[2]{#2}
\providecommand\bibinfo[2]{#2}
\providecommand\natexlab[1]{#1}
\providecommand\showeprint[2][]{arXiv:#2}

\bibitem[\protect\citeauthoryear{Abuzaid, Kraft, Suri, Gan, Xu, Shenoy,
  Ananthanarayan, Sheu, Meijer, Wu, et~al\mbox{.}}{Abuzaid
  et~al\mbox{.}}{2018}]%
        {abuzaid2018diff}
\bibfield{author}{\bibinfo{person}{Firas Abuzaid}, \bibinfo{person}{Peter
  Kraft}, \bibinfo{person}{Sahaana Suri}, \bibinfo{person}{Edward Gan},
  \bibinfo{person}{Eric Xu}, \bibinfo{person}{Atul Shenoy},
  \bibinfo{person}{Asvin Ananthanarayan}, \bibinfo{person}{John Sheu},
  \bibinfo{person}{Erik Meijer}, \bibinfo{person}{Xi Wu}, {et~al\mbox{.}}}
  \bibinfo{year}{2018}\natexlab{}.
\newblock \showarticletitle{DIFF: a relational interface for large-scale data
  explanation}.
\newblock \bibinfo{journal}{\emph{Proceedings of the VLDB Endowment}}
  \bibinfo{volume}{12}, \bibinfo{number}{4} (\bibinfo{year}{2018}),
  \bibinfo{pages}{419--432}.
\newblock


\bibitem[\protect\citeauthoryear{Ainy, Bourhis, Davidson, Deutch, and
  Milo}{Ainy et~al\mbox{.}}{2015}]%
        {ainy2015approximated}
\bibfield{author}{\bibinfo{person}{Eleanor Ainy}, \bibinfo{person}{Pierre
  Bourhis}, \bibinfo{person}{Susan~B Davidson}, \bibinfo{person}{Daniel
  Deutch}, {and} \bibinfo{person}{Tova Milo}.} \bibinfo{year}{2015}\natexlab{}.
\newblock \showarticletitle{Approximated summarization of data provenance}. In
  \bibinfo{booktitle}{\emph{Proceedings of the 24th ACM International on
  Conference on Information and Knowledge Management}}.
  \bibinfo{pages}{483--492}.
\newblock


\bibitem[\protect\citeauthoryear{AlOmeir, Lai, Milani, and Pottinger}{AlOmeir
  et~al\mbox{.}}{[n. d.]}]%
        {alomeirsummarizing}
\bibfield{author}{\bibinfo{person}{Omar AlOmeir},
  \bibinfo{person}{Eugenie~Yujing Lai}, \bibinfo{person}{Mostafa Milani}, {and}
  \bibinfo{person}{Rachel Pottinger}.} \bibinfo{year}{[n. d.]}\natexlab{}.
\newblock \showarticletitle{Summarizing Provenance of Aggregation Query Results
  in Relational Databases}.
\newblock  (\bibinfo{year}{[n. d.]}).
\newblock


\bibitem[\protect\citeauthoryear{Amsterdamer, Deutch, and Tannen}{Amsterdamer
  et~al\mbox{.}}{2011}]%
        {amsterdamer2011provenance}
\bibfield{author}{\bibinfo{person}{Yael Amsterdamer}, \bibinfo{person}{Daniel
  Deutch}, {and} \bibinfo{person}{Val Tannen}.}
  \bibinfo{year}{2011}\natexlab{}.
\newblock \showarticletitle{Provenance for aggregate queries}. In
  \bibinfo{booktitle}{\emph{PODS}}. \bibinfo{pages}{153--164}.
\newblock


\bibitem[\protect\citeauthoryear{Arab, Feng, Glavic, Lee, Niu, and Zeng}{Arab
  et~al\mbox{.}}{2018}]%
        {AF18}
\bibfield{author}{\bibinfo{person}{Bahareh Arab}, \bibinfo{person}{Su Feng},
  \bibinfo{person}{Boris Glavic}, \bibinfo{person}{Seokki Lee},
  \bibinfo{person}{Xing Niu}, {and} \bibinfo{person}{Qitian Zeng}.}
  \bibinfo{year}{2018}\natexlab{}.
\newblock \showarticletitle{{GProM} - {A} Swiss Army Knife for Your Provenance
  Needs}.
\newblock \bibinfo{journal}{\emph{{IEEE} Data Engineering Bulletin}}
  \bibinfo{volume}{41}, \bibinfo{number}{1} (\bibinfo{year}{2018}),
  \bibinfo{pages}{51--62}.
\newblock


\bibitem[\protect\citeauthoryear{Arab, Gawlick, Radhakrishnan, Guo, and
  Glavic}{Arab et~al\mbox{.}}{2014}]%
        {arab2014generic}
\bibfield{author}{\bibinfo{person}{Bahareh Arab}, \bibinfo{person}{Dieter
  Gawlick}, \bibinfo{person}{Venkatesh Radhakrishnan}, \bibinfo{person}{Hao
  Guo}, {and} \bibinfo{person}{Boris Glavic}.} \bibinfo{year}{2014}\natexlab{}.
\newblock \showarticletitle{A generic provenance middleware for database
  queries, updates, and transactions}. In \bibinfo{booktitle}{\emph{{TaPP}}}.
\newblock


\bibitem[\protect\citeauthoryear{Barman, Korn, Srivastava, Gunopulos, Young,
  and Agarwal}{Barman et~al\mbox{.}}{2007}]%
        {DBLP:conf/icde/BarmanKSGYA07}
\bibfield{author}{\bibinfo{person}{Dhiman Barman}, \bibinfo{person}{Flip Korn},
  \bibinfo{person}{Divesh Srivastava}, \bibinfo{person}{Dimitrios Gunopulos},
  \bibinfo{person}{Neal~E. Young}, {and} \bibinfo{person}{Deepak Agarwal}.}
  \bibinfo{year}{2007}\natexlab{}.
\newblock \showarticletitle{Parsimonious Explanations of Change in Hierarchical
  Data}. In \bibinfo{booktitle}{\emph{Proceedings of the 23rd International
  Conference on Data Engineering, {ICDE} 2007, The Marmara Hotel, Istanbul,
  Turkey, April 15-20, 2007}}, \bibfield{editor}{\bibinfo{person}{Rada
  Chirkova}, \bibinfo{person}{Asuman Dogac}, \bibinfo{person}{M.~Tamer
  {\"{O}}zsu}, {and} \bibinfo{person}{Timos~K. Sellis}} (Eds.).
  \bibinfo{publisher}{{IEEE} Computer Society}, \bibinfo{pages}{1273--1275}.
\newblock
\urldef\tempurl%
\url{https://doi.org/10.1109/ICDE.2007.368991}
\showDOI{\tempurl}


\bibitem[\protect\citeauthoryear{Bessa, Freire, Dasu, and Srivastava}{Bessa
  et~al\mbox{.}}{2020}]%
        {bessa2020effective}
\bibfield{author}{\bibinfo{person}{Aline Bessa}, \bibinfo{person}{Juliana
  Freire}, \bibinfo{person}{Tamraparni Dasu}, {and} \bibinfo{person}{Divesh
  Srivastava}.} \bibinfo{year}{2020}\natexlab{}.
\newblock \showarticletitle{Effective Discovery of Meaningful Outlier
  Relationships}.
\newblock \bibinfo{journal}{\emph{ACM Transactions on Data Science}}
  \bibinfo{volume}{1}, \bibinfo{number}{2} (\bibinfo{year}{2020}),
  \bibinfo{pages}{1--33}.
\newblock


\bibitem[\protect\citeauthoryear{Bidoit, Herschel, and Tzompanaki}{Bidoit
  et~al\mbox{.}}{2014}]%
        {DBLP:conf/edbt/BidoitHT14}
\bibfield{author}{\bibinfo{person}{Nicole Bidoit}, \bibinfo{person}{Melanie
  Herschel}, {and} \bibinfo{person}{Katerina Tzompanaki}.}
  \bibinfo{year}{2014}\natexlab{}.
\newblock \showarticletitle{Query-Based Why-Not Provenance with NedExplain}. In
  \bibinfo{booktitle}{\emph{Proceedings of the 17th International Conference on
  Extending Database Technology, {EDBT} 2014, Athens, Greece, March 24-28,
  2014}}, \bibfield{editor}{\bibinfo{person}{Sihem Amer{-}Yahia},
  \bibinfo{person}{Vassilis Christophides}, \bibinfo{person}{Anastasios
  Kementsietsidis}, \bibinfo{person}{Minos~N. Garofalakis},
  \bibinfo{person}{Stratos Idreos}, {and} \bibinfo{person}{Vincent Leroy}}
  (Eds.). \bibinfo{publisher}{OpenProceedings.org}, \bibinfo{pages}{145--156}.
\newblock
\urldef\tempurl%
\url{https://doi.org/10.5441/002/edbt.2014.14}
\showDOI{\tempurl}


\bibitem[\protect\citeauthoryear{Breiman}{Breiman}{2001}]%
        {breiman2001random}
\bibfield{author}{\bibinfo{person}{Leo Breiman}.}
  \bibinfo{year}{2001}\natexlab{}.
\newblock \showarticletitle{Random forests}.
\newblock \bibinfo{journal}{\emph{Machine learning}} \bibinfo{volume}{45},
  \bibinfo{number}{1} (\bibinfo{year}{2001}), \bibinfo{pages}{5--32}.
\newblock


\bibitem[\protect\citeauthoryear{Buneman, Khanna, and Wang-Chiew}{Buneman
  et~al\mbox{.}}{2001}]%
        {buneman2001and}
\bibfield{author}{\bibinfo{person}{Peter Buneman}, \bibinfo{person}{Sanjeev
  Khanna}, {and} \bibinfo{person}{Tan Wang-Chiew}.}
  \bibinfo{year}{2001}\natexlab{}.
\newblock \showarticletitle{Why and where: A characterization of data
  provenance}. In \bibinfo{booktitle}{\emph{ICDT}}. \bibinfo{pages}{316--330}.
\newblock


\bibitem[\protect\citeauthoryear{Chapman and Jagadish}{Chapman and
  Jagadish}{2009}]%
        {chapman2009not}
\bibfield{author}{\bibinfo{person}{Adriane Chapman} {and} \bibinfo{person}{HV
  Jagadish}.} \bibinfo{year}{2009}\natexlab{}.
\newblock \showarticletitle{Why not?}. In \bibinfo{booktitle}{\emph{Proceedings
  of the 2009 ACM SIGMOD International Conference on Management of data}}.
  \bibinfo{pages}{523--534}.
\newblock


\bibitem[\protect\citeauthoryear{Chapman, Jagadish, and Ramanan}{Chapman
  et~al\mbox{.}}{2008}]%
        {chapman2008efficient}
\bibfield{author}{\bibinfo{person}{Adriane~P Chapman},
  \bibinfo{person}{Hosagrahar~V Jagadish}, {and} \bibinfo{person}{Prakash
  Ramanan}.} \bibinfo{year}{2008}\natexlab{}.
\newblock \showarticletitle{Efficient provenance storage}. In
  \bibinfo{booktitle}{\emph{Proceedings of the 2008 ACM SIGMOD international
  conference on Management of data}}. \bibinfo{pages}{993--1006}.
\newblock


\bibitem[\protect\citeauthoryear{Cheney, Chiticariu, and Tan}{Cheney
  et~al\mbox{.}}{2009}]%
        {DBLP:journals/ftdb/CheneyCT09}
\bibfield{author}{\bibinfo{person}{James Cheney}, \bibinfo{person}{Laura
  Chiticariu}, {and} \bibinfo{person}{Wang~Chiew Tan}.}
  \bibinfo{year}{2009}\natexlab{}.
\newblock \showarticletitle{Provenance in Databases: Why, How, and Where}.
\newblock \bibinfo{journal}{\emph{Found. Trends Databases}}
  \bibinfo{volume}{1}, \bibinfo{number}{4} (\bibinfo{year}{2009}),
  \bibinfo{pages}{379--474}.
\newblock


\bibitem[\protect\citeauthoryear{Chepurko, Marcus, Zgraggen, Fernandez, Kraska,
  and Karger}{Chepurko et~al\mbox{.}}{2020}]%
        {DBLP:journals/pvldb/ChepurkoMZFKK20}
\bibfield{author}{\bibinfo{person}{Nadiia Chepurko}, \bibinfo{person}{Ryan
  Marcus}, \bibinfo{person}{Emanuel Zgraggen}, \bibinfo{person}{Raul~Castro
  Fernandez}, \bibinfo{person}{Tim Kraska}, {and} \bibinfo{person}{David
  Karger}.} \bibinfo{year}{2020}\natexlab{}.
\newblock \showarticletitle{{ARDA:} Automatic Relational Data Augmentation for
  Machine Learning}.
\newblock \bibinfo{journal}{\emph{PVLDB}} \bibinfo{volume}{13},
  \bibinfo{number}{9} (\bibinfo{year}{2020}), \bibinfo{pages}{1373--1387}.
\newblock


\bibitem[\protect\citeauthoryear{Deutch, Frost, and Gilad}{Deutch
  et~al\mbox{.}}{2016}]%
        {deutch2016nlprov}
\bibfield{author}{\bibinfo{person}{Daniel Deutch}, \bibinfo{person}{Nave
  Frost}, {and} \bibinfo{person}{Amir Gilad}.} \bibinfo{year}{2016}\natexlab{}.
\newblock \showarticletitle{Nlprov: Natural language provenance}.
\newblock \bibinfo{journal}{\emph{Proceedings of the VLDB Endowment}}
  \bibinfo{volume}{9}, \bibinfo{number}{13} (\bibinfo{year}{2016}),
  \bibinfo{pages}{1537--1540}.
\newblock


\bibitem[\protect\citeauthoryear{Deutch, Frost, and Gilad}{Deutch
  et~al\mbox{.}}{2017}]%
        {deutch2017provenance}
\bibfield{author}{\bibinfo{person}{Daniel Deutch}, \bibinfo{person}{Nave
  Frost}, {and} \bibinfo{person}{Amir Gilad}.} \bibinfo{year}{2017}\natexlab{}.
\newblock \showarticletitle{Provenance for natural language queries}.
\newblock \bibinfo{journal}{\emph{Proceedings of the VLDB Endowment}}
  \bibinfo{volume}{10}, \bibinfo{number}{5} (\bibinfo{year}{2017}),
  \bibinfo{pages}{577--588}.
\newblock


\bibitem[\protect\citeauthoryear{Fernandez, Abedjan, Koko, Yuan, Madden, and
  Stonebraker}{Fernandez et~al\mbox{.}}{2018}]%
        {FA18}
\bibfield{author}{\bibinfo{person}{Raul~Castro Fernandez},
  \bibinfo{person}{Ziawasch Abedjan}, \bibinfo{person}{Famien Koko},
  \bibinfo{person}{Gina Yuan}, \bibinfo{person}{Samuel Madden}, {and}
  \bibinfo{person}{Michael Stonebraker}.} \bibinfo{year}{2018}\natexlab{}.
\newblock \showarticletitle{Aurum: A Data Discovery System}. In
  \bibinfo{booktitle}{\emph{ICDE}}. \bibinfo{pages}{1001--1012}.
\newblock


\bibitem[\protect\citeauthoryear{Gebaly, Agrawal, Golab, Korn, and
  Srivastava}{Gebaly et~al\mbox{.}}{2014}]%
        {DBLP:journals/pvldb/GebalyAGKS14}
\bibfield{author}{\bibinfo{person}{Kareem~El Gebaly}, \bibinfo{person}{Parag
  Agrawal}, \bibinfo{person}{Lukasz Golab}, \bibinfo{person}{Flip Korn}, {and}
  \bibinfo{person}{Divesh Srivastava}.} \bibinfo{year}{2014}\natexlab{}.
\newblock \showarticletitle{Interpretable and Informative Explanations of
  Outcomes}.
\newblock \bibinfo{journal}{\emph{Proc. {VLDB} Endow.}} \bibinfo{volume}{8},
  \bibinfo{number}{1} (\bibinfo{year}{2014}), \bibinfo{pages}{61--72}.
\newblock
\urldef\tempurl%
\url{https://doi.org/10.14778/2735461.2735467}
\showDOI{\tempurl}


\bibitem[\protect\citeauthoryear{Glavic and Alonso}{Glavic and Alonso}{2009}]%
        {glavic2009perm}
\bibfield{author}{\bibinfo{person}{Boris Glavic} {and} \bibinfo{person}{Gustavo
  Alonso}.} \bibinfo{year}{2009}\natexlab{}.
\newblock \showarticletitle{Perm: Processing provenance and data on the same
  data model through query rewriting}. In \bibinfo{booktitle}{\emph{ICDE}}.
  \bibinfo{pages}{174--185}.
\newblock


\bibitem[\protect\citeauthoryear{Green, Karvounarakis, and Tannen}{Green
  et~al\mbox{.}}{2007a}]%
        {GKT07-semirings}
\bibfield{author}{\bibinfo{person}{Todd~J. Green}, \bibinfo{person}{Grigoris
  Karvounarakis}, {and} \bibinfo{person}{Val Tannen}.}
  \bibinfo{year}{2007}\natexlab{a}.
\newblock \showarticletitle{Provenance semirings}. In
  \bibinfo{booktitle}{\emph{PODS}}. \bibinfo{pages}{31--40}.
\newblock


\bibitem[\protect\citeauthoryear{Green, Karvounarakis, and Tannen}{Green
  et~al\mbox{.}}{2007b}]%
        {green2007provenance}
\bibfield{author}{\bibinfo{person}{Todd~J Green}, \bibinfo{person}{Grigoris
  Karvounarakis}, {and} \bibinfo{person}{Val Tannen}.}
  \bibinfo{year}{2007}\natexlab{b}.
\newblock \showarticletitle{Provenance semirings}. In
  \bibinfo{booktitle}{\emph{PODS}}. \bibinfo{pages}{31--40}.
\newblock


\bibitem[\protect\citeauthoryear{He, Ganjam, and Chu}{He et~al\mbox{.}}{2015}]%
        {DBLP:journals/pvldb/HeGC15}
\bibfield{author}{\bibinfo{person}{Yeye He}, \bibinfo{person}{Kris Ganjam},
  {and} \bibinfo{person}{Xu Chu}.} \bibinfo{year}{2015}\natexlab{}.
\newblock \showarticletitle{{SEMA-JOIN:} Joining Semantically-Related Tables
  Using Big Table Corpora}.
\newblock \bibinfo{journal}{\emph{Proc. {VLDB} Endow.}} \bibinfo{volume}{8},
  \bibinfo{number}{12} (\bibinfo{year}{2015}), \bibinfo{pages}{1358--1369}.
\newblock
\urldef\tempurl%
\url{https://doi.org/10.14778/2824032.2824036}
\showDOI{\tempurl}


\bibitem[\protect\citeauthoryear{J{\"a}rvelin and
  Kek{\"a}l{\"a}inen}{J{\"a}rvelin and Kek{\"a}l{\"a}inen}{2002}]%
        {jarvelin2002cumulated}
\bibfield{author}{\bibinfo{person}{Kalervo J{\"a}rvelin} {and}
  \bibinfo{person}{Jaana Kek{\"a}l{\"a}inen}.} \bibinfo{year}{2002}\natexlab{}.
\newblock \showarticletitle{Cumulated gain-based evaluation of IR techniques}.
\newblock \bibinfo{journal}{\emph{ACM Transactions on Information Systems
  (TOIS)}} \bibinfo{volume}{20}, \bibinfo{number}{4} (\bibinfo{year}{2002}),
  \bibinfo{pages}{422--446}.
\newblock


\bibitem[\protect\citeauthoryear{Joglekar, Garcia-Molina, and
  Parameswaran}{Joglekar et~al\mbox{.}}{2017}]%
        {joglekar2017interactive}
\bibfield{author}{\bibinfo{person}{Manas Joglekar}, \bibinfo{person}{Hector
  Garcia-Molina}, {and} \bibinfo{person}{Aditya Parameswaran}.}
  \bibinfo{year}{2017}\natexlab{}.
\newblock \showarticletitle{Interactive data exploration with smart
  drill-down}.
\newblock \bibinfo{journal}{\emph{IEEE Transactions on Knowledge and Data
  Engineering}} \bibinfo{volume}{31}, \bibinfo{number}{1}
  (\bibinfo{year}{2017}), \bibinfo{pages}{46--60}.
\newblock


\bibitem[\protect\citeauthoryear{Johnson, Pollard, Shen, Li-Wei, Feng,
  Ghassemi, Moody, Szolovits, Celi, and Mark}{Johnson et~al\mbox{.}}{2016}]%
        {johnson2016mimic}
\bibfield{author}{\bibinfo{person}{Alistair~EW Johnson}, \bibinfo{person}{Tom~J
  Pollard}, \bibinfo{person}{Lu Shen}, \bibinfo{person}{H~Lehman Li-Wei},
  \bibinfo{person}{Mengling Feng}, \bibinfo{person}{Mohammad Ghassemi},
  \bibinfo{person}{Benjamin Moody}, \bibinfo{person}{Peter Szolovits},
  \bibinfo{person}{Leo~Anthony Celi}, {and} \bibinfo{person}{Roger~G Mark}.}
  \bibinfo{year}{2016}\natexlab{}.
\newblock \showarticletitle{MIMIC-III, a freely accessible critical care
  database}.
\newblock \bibinfo{journal}{\emph{Scientific data}} \bibinfo{volume}{3},
  \bibinfo{number}{1} (\bibinfo{year}{2016}), \bibinfo{pages}{1--9}.
\newblock


\bibitem[\protect\citeauthoryear{Karvounarakis, Ives, and Tannen}{Karvounarakis
  et~al\mbox{.}}{2010}]%
        {karvounarakis2010querying}
\bibfield{author}{\bibinfo{person}{Grigoris Karvounarakis},
  \bibinfo{person}{Zachary~G Ives}, {and} \bibinfo{person}{Val Tannen}.}
  \bibinfo{year}{2010}\natexlab{}.
\newblock \showarticletitle{Querying data provenance}. In
  \bibinfo{booktitle}{\emph{Proceedings of the 2010 ACM SIGMOD International
  Conference on Management of data}}. \bibinfo{pages}{951--962}.
\newblock


\bibitem[\protect\citeauthoryear{Kendall}{Kendall}{1938}]%
        {kendall1938new}
\bibfield{author}{\bibinfo{person}{Maurice~G Kendall}.}
  \bibinfo{year}{1938}\natexlab{}.
\newblock \showarticletitle{A new measure of rank correlation}.
\newblock \bibinfo{journal}{\emph{Biometrika}} \bibinfo{volume}{30},
  \bibinfo{number}{1/2} (\bibinfo{year}{1938}), \bibinfo{pages}{81--93}.
\newblock


\bibitem[\protect\citeauthoryear{Kim, Lakshmanan, and Srivastava}{Kim
  et~al\mbox{.}}{2020}]%
        {kim2020summarizing}
\bibfield{author}{\bibinfo{person}{Alexandra Kim}, \bibinfo{person}{Laks~VS
  Lakshmanan}, {and} \bibinfo{person}{Divesh Srivastava}.}
  \bibinfo{year}{2020}\natexlab{}.
\newblock \showarticletitle{Summarizing Hierarchical Multidimensional Data}. In
  \bibinfo{booktitle}{\emph{2020 IEEE 36th International Conference on Data
  Engineering (ICDE)}}. IEEE, \bibinfo{pages}{877--888}.
\newblock


\bibitem[\protect\citeauthoryear{Kumar, Naughton, Patel, and Zhu}{Kumar
  et~al\mbox{.}}{2016}]%
        {DBLP:conf/sigmod/KumarNPZ16}
\bibfield{author}{\bibinfo{person}{Arun Kumar}, \bibinfo{person}{Jeffrey~F.
  Naughton}, \bibinfo{person}{Jignesh~M. Patel}, {and} \bibinfo{person}{Xiaojin
  Zhu}.} \bibinfo{year}{2016}\natexlab{}.
\newblock \showarticletitle{To Join or Not to Join?: Thinking Twice about Joins
  before Feature Selection}. In \bibinfo{booktitle}{\emph{SIGMOD}},
  \bibfield{editor}{\bibinfo{person}{Fatma {\"{O}}zcan},
  \bibinfo{person}{Georgia Koutrika}, {and} \bibinfo{person}{Sam Madden}}
  (Eds.). \bibinfo{publisher}{{ACM}}, \bibinfo{pages}{19--34}.
\newblock


\bibitem[\protect\citeauthoryear{Lee, Lud{\"a}scher, and Glavic}{Lee
  et~al\mbox{.}}{2019}]%
        {lee2019pug}
\bibfield{author}{\bibinfo{person}{Seokki Lee}, \bibinfo{person}{Bertram
  Lud{\"a}scher}, {and} \bibinfo{person}{Boris Glavic}.}
  \bibinfo{year}{2019}\natexlab{}.
\newblock \showarticletitle{PUG: a framework and practical implementation for
  why and why-not provenance}.
\newblock \bibinfo{journal}{\emph{The VLDB Journal}} \bibinfo{volume}{28},
  \bibinfo{number}{1} (\bibinfo{year}{2019}), \bibinfo{pages}{47--71}.
\newblock


\bibitem[\protect\citeauthoryear{Lee, Ludäscher, and Glavic}{Lee
  et~al\mbox{.}}{2020}]%
        {LL20}
\bibfield{author}{\bibinfo{person}{Seokki Lee}, \bibinfo{person}{Bertram
  Ludäscher}, {and} \bibinfo{person}{Boris Glavic}.}
  \bibinfo{year}{2020}\natexlab{}.
\newblock \showarticletitle{Approximate Summaries for Why and Why-not
  Provenance}.
\newblock \bibinfo{journal}{\emph{PVLDB}} \bibinfo{volume}{13},
  \bibinfo{number}{6} (\bibinfo{year}{2020}), \bibinfo{pages}{912--924}.
\newblock


\bibitem[\protect\citeauthoryear{Lee, Niu, Lud{\"a}scher, and Glavic}{Lee
  et~al\mbox{.}}{2017}]%
        {lee2017integrating}
\bibfield{author}{\bibinfo{person}{Seokki Lee}, \bibinfo{person}{Xing Niu},
  \bibinfo{person}{Bertram Lud{\"a}scher}, {and} \bibinfo{person}{Boris
  Glavic}.} \bibinfo{year}{2017}\natexlab{}.
\newblock \showarticletitle{Integrating approximate summarization with
  provenance capture}. In \bibinfo{booktitle}{\emph{9th $\{$USENIX$\}$ Workshop
  on the Theory and Practice of Provenance (TaPP 2017)}}.
\newblock


\bibitem[\protect\citeauthoryear{Miao, Zeng, Glavic, and Roy}{Miao
  et~al\mbox{.}}{2019}]%
        {DBLP:conf/sigmod/MiaoZGR19}
\bibfield{author}{\bibinfo{person}{Zhengjie Miao}, \bibinfo{person}{Qitian
  Zeng}, \bibinfo{person}{Boris Glavic}, {and} \bibinfo{person}{Sudeepa Roy}.}
  \bibinfo{year}{2019}\natexlab{}.
\newblock \showarticletitle{Going Beyond Provenance: Explaining Query Answers
  with Pattern-based Counterbalances}. In \bibinfo{booktitle}{\emph{Proceedings
  of the 2019 International Conference on Management of Data, {SIGMOD}
  Conference 2019, Amsterdam, The Netherlands, June 30 - July 5, 2019}},
  \bibfield{editor}{\bibinfo{person}{Peter~A. Boncz}, \bibinfo{person}{Stefan
  Manegold}, \bibinfo{person}{Anastasia Ailamaki}, \bibinfo{person}{Amol
  Deshpande}, {and} \bibinfo{person}{Tim Kraska}} (Eds.).
  \bibinfo{publisher}{{ACM}}, \bibinfo{pages}{485--502}.
\newblock
\urldef\tempurl%
\url{https://doi.org/10.1145/3299869.3300066}
\showDOI{\tempurl}


\bibitem[\protect\citeauthoryear{NBA.com}{NBA.com}{2020}]%
        {nba2019}
\bibfield{author}{\bibinfo{person}{NBA.com}.} \bibinfo{year}{2020}\natexlab{}.
\newblock \bibinfo{title}{The official site of the NBA}.
\newblock
\newblock
\urldef\tempurl%
\url{https://www.nba.com/}
\showURL{%
\tempurl}


\bibitem[\protect\citeauthoryear{Pearl}{Pearl}{2000}]%
        {PearlBook2000}
\bibfield{author}{\bibinfo{person}{Judea Pearl}.}
  \bibinfo{year}{2000}\natexlab{}.
\newblock \bibinfo{booktitle}{\emph{Causality: models, reasoning, and
  inference}}.
\newblock \bibinfo{publisher}{Cambridge University Press}.
\newblock
\showISBNx{0-521-77362-8}


\bibitem[\protect\citeauthoryear{Psallidas and Wu}{Psallidas and Wu}{2018}]%
        {psallidas2018smoke}
\bibfield{author}{\bibinfo{person}{Fotis Psallidas} {and}
  \bibinfo{person}{Eugene Wu}.} \bibinfo{year}{2018}\natexlab{}.
\newblock \showarticletitle{Smoke: Fine-grained lineage at interactive speed}.
\newblock \bibinfo{journal}{\emph{PVLDB}} \bibinfo{volume}{11},
  \bibinfo{number}{6} (\bibinfo{year}{2018}), \bibinfo{pages}{719--732}.
\newblock


\bibitem[\protect\citeauthoryear{Qin, Yu, and Chang}{Qin et~al\mbox{.}}{2012}]%
        {DBLP:journals/pvldb/QinYC12}
\bibfield{author}{\bibinfo{person}{Lu Qin}, \bibinfo{person}{Jeffrey~Xu Yu},
  {and} \bibinfo{person}{Lijun Chang}.} \bibinfo{year}{2012}\natexlab{}.
\newblock \showarticletitle{Diversifying Top-K Results}.
\newblock \bibinfo{journal}{\emph{Proc. {VLDB} Endow.}} \bibinfo{volume}{5},
  \bibinfo{number}{11} (\bibinfo{year}{2012}), \bibinfo{pages}{1124--1135}.
\newblock
\urldef\tempurl%
\url{https://doi.org/10.14778/2350229.2350233}
\showDOI{\tempurl}


\bibitem[\protect\citeauthoryear{R{\'e} and Suciu}{R{\'e} and Suciu}{2008}]%
        {re2008approximate}
\bibfield{author}{\bibinfo{person}{Christopher R{\'e}} {and}
  \bibinfo{person}{Dan Suciu}.} \bibinfo{year}{2008}\natexlab{}.
\newblock \showarticletitle{Approximate lineage for probabilistic databases}.
\newblock \bibinfo{journal}{\emph{Proceedings of the VLDB Endowment}}
  \bibinfo{volume}{1}, \bibinfo{number}{1} (\bibinfo{year}{2008}),
  \bibinfo{pages}{797--808}.
\newblock


\bibitem[\protect\citeauthoryear{Roweis and Saul}{Roweis and Saul}{2000}]%
        {roweis2000nonlinear}
\bibfield{author}{\bibinfo{person}{Sam~T Roweis} {and}
  \bibinfo{person}{Lawrence~K Saul}.} \bibinfo{year}{2000}\natexlab{}.
\newblock \showarticletitle{Nonlinear dimensionality reduction by locally
  linear embedding}.
\newblock \bibinfo{journal}{\emph{Science}} \bibinfo{volume}{290},
  \bibinfo{number}{5500} (\bibinfo{year}{2000}), \bibinfo{pages}{2323--2326}.
\newblock


\bibitem[\protect\citeauthoryear{Roy, Orr, and Suciu}{Roy
  et~al\mbox{.}}{2015}]%
        {ROS15}
\bibfield{author}{\bibinfo{person}{Sudeepa Roy}, \bibinfo{person}{Laurel Orr},
  {and} \bibinfo{person}{Dan Suciu}.} \bibinfo{year}{2015}\natexlab{}.
\newblock \showarticletitle{Explaining query answers with explanation-ready
  databases}.
\newblock \bibinfo{journal}{\emph{PVLDB}} \bibinfo{volume}{9},
  \bibinfo{number}{4} (\bibinfo{year}{2015}), \bibinfo{pages}{348--359}.
\newblock


\bibitem[\protect\citeauthoryear{Roy and Suciu}{Roy and Suciu}{2014}]%
        {RS14}
\bibfield{author}{\bibinfo{person}{Sudeepa Roy} {and} \bibinfo{person}{Dan
  Suciu}.} \bibinfo{year}{2014}\natexlab{}.
\newblock \showarticletitle{A formal approach to finding explanations for
  database queries}. SIGMOD.
\newblock


\bibitem[\protect\citeauthoryear{Rubin}{Rubin}{2005}]%
        {rubin2005causal}
\bibfield{author}{\bibinfo{person}{Donald~B Rubin}.}
  \bibinfo{year}{2005}\natexlab{}.
\newblock \showarticletitle{Causal inference using potential outcomes: Design,
  modeling, decisions}.
\newblock \bibinfo{journal}{\emph{J. Amer. Statist. Assoc.}}
  \bibinfo{volume}{100}, \bibinfo{number}{469} (\bibinfo{year}{2005}),
  \bibinfo{pages}{322--331}.
\newblock


\bibitem[\protect\citeauthoryear{Sarle}{Sarle}{1990}]%
        {sarle1990sas}
\bibfield{author}{\bibinfo{person}{WS Sarle}.} \bibinfo{year}{1990}\natexlab{}.
\newblock \showarticletitle{SAS/STAT User’s Guide: The VARCLUS Procedure}.
\newblock \bibinfo{journal}{\emph{SAS Institute, Inc., Cary, NC, USA,}}
  (\bibinfo{year}{1990}), \bibinfo{pages}{134}.
\newblock


\bibitem[\protect\citeauthoryear{Sarma, Fang, Gupta, Halevy, Lee, Wu, Xin, and
  Yu}{Sarma et~al\mbox{.}}{2012}]%
        {SF12a}
\bibfield{author}{\bibinfo{person}{Anish~Das Sarma}, \bibinfo{person}{Lujun
  Fang}, \bibinfo{person}{Nitin Gupta}, \bibinfo{person}{Alon~Y. Halevy},
  \bibinfo{person}{Hongrae Lee}, \bibinfo{person}{Fei Wu},
  \bibinfo{person}{Reynold Xin}, {and} \bibinfo{person}{Cong Yu}.}
  \bibinfo{year}{2012}\natexlab{}.
\newblock \showarticletitle{Finding related tables}. In
  \bibinfo{booktitle}{\emph{SIGMOD}}. \bibinfo{pages}{817--828}.
\newblock


\bibitem[\protect\citeauthoryear{Shah, Kumar, and Zhu}{Shah
  et~al\mbox{.}}{2017}]%
        {DBLP:journals/pvldb/ShahKZ17}
\bibfield{author}{\bibinfo{person}{Vraj Shah}, \bibinfo{person}{Arun Kumar},
  {and} \bibinfo{person}{Xiaojin Zhu}.} \bibinfo{year}{2017}\natexlab{}.
\newblock \showarticletitle{Are Key-Foreign Key Joins Safe to Avoid when
  Learning High-Capacity Classifiers?}
\newblock \bibinfo{journal}{\emph{PVLDB}} \bibinfo{volume}{11},
  \bibinfo{number}{3} (\bibinfo{year}{2017}), \bibinfo{pages}{366--379}.
\newblock


\bibitem[\protect\citeauthoryear{ten Cate, Civili, Sherkhonov, and Tan}{ten
  Cate et~al\mbox{.}}{2015}]%
        {DBLP:conf/pods/CateCST15}
\bibfield{author}{\bibinfo{person}{Balder ten Cate}, \bibinfo{person}{Cristina
  Civili}, \bibinfo{person}{Evgeny Sherkhonov}, {and}
  \bibinfo{person}{Wang{-}Chiew Tan}.} \bibinfo{year}{2015}\natexlab{}.
\newblock \showarticletitle{High-Level Why-Not Explanations using Ontologies}.
  In \bibinfo{booktitle}{\emph{Proceedings of the 34th {ACM} Symposium on
  Principles of Database Systems, {PODS} 2015, Melbourne, Victoria, Australia,
  May 31 - June 4, 2015}}, \bibfield{editor}{\bibinfo{person}{Tova Milo} {and}
  \bibinfo{person}{Diego Calvanese}} (Eds.). \bibinfo{publisher}{{ACM}},
  \bibinfo{pages}{31--43}.
\newblock
\urldef\tempurl%
\url{https://doi.org/10.1145/2745754.2745765}
\showDOI{\tempurl}


\bibitem[\protect\citeauthoryear{Vardi}{Vardi}{1982}]%
        {Vardi82}
\bibfield{author}{\bibinfo{person}{Moshe~Y. Vardi}.}
  \bibinfo{year}{1982}\natexlab{}.
\newblock \showarticletitle{The Complexity of Relational Query Languages
  (Extended Abstract)}. In \bibinfo{booktitle}{\emph{Proceedings of the
  Fourteenth Annual ACM Symposium on Theory of Computing}}
  \emph{(\bibinfo{series}{STOC '82})}. \bibinfo{publisher}{ACM},
  \bibinfo{address}{New York, NY, USA}, \bibinfo{pages}{137--146}.
\newblock
\showISBNx{0-89791-070-2}
\urldef\tempurl%
\url{https://doi.org/10.1145/800070.802186}
\showDOI{\tempurl}


\bibitem[\protect\citeauthoryear{Vollmer, Golab, B{\"{o}}hm, and
  Srivastava}{Vollmer et~al\mbox{.}}{2019}]%
        {vollmer-19-insnd}
\bibfield{author}{\bibinfo{person}{Michael Vollmer}, \bibinfo{person}{Lukasz
  Golab}, \bibinfo{person}{Klemens B{\"{o}}hm}, {and} \bibinfo{person}{Divesh
  Srivastava}.} \bibinfo{year}{2019}\natexlab{}.
\newblock \showarticletitle{Informative Summarization of Numeric Data}. In
  \bibinfo{booktitle}{\emph{SSDBM}}. \bibinfo{pages}{97--108}.
\newblock


\bibitem[\protect\citeauthoryear{Wang and Meliou}{Wang and Meliou}{2019}]%
        {DBLP:journals/pvldb/WangM19}
\bibfield{author}{\bibinfo{person}{Xiaolan Wang} {and}
  \bibinfo{person}{Alexandra Meliou}.} \bibinfo{year}{2019}\natexlab{}.
\newblock \showarticletitle{Explain3D: Explaining Disagreements in Disjoint
  Datasets}.
\newblock \bibinfo{journal}{\emph{Proc. {VLDB} Endow.}} \bibinfo{volume}{12},
  \bibinfo{number}{7} (\bibinfo{year}{2019}), \bibinfo{pages}{779--792}.
\newblock
\urldef\tempurl%
\url{https://doi.org/10.14778/3317315.3317320}
\showDOI{\tempurl}


\bibitem[\protect\citeauthoryear{Wen, Zhu, Roy, and Yang}{Wen
  et~al\mbox{.}}{2018}]%
        {wen2018interactive}
\bibfield{author}{\bibinfo{person}{Yuhao Wen}, \bibinfo{person}{Xiaodan Zhu},
  \bibinfo{person}{Sudeepa Roy}, {and} \bibinfo{person}{Jun Yang}.}
  \bibinfo{year}{2018}\natexlab{}.
\newblock \showarticletitle{Interactive summarization and exploration of top
  aggregate query answers}. In \bibinfo{booktitle}{\emph{Proceedings of the
  VLDB Endowment. International Conference on Very Large Data Bases}},
  Vol.~\bibinfo{volume}{11}. NIH Public Access, \bibinfo{pages}{2196}.
\newblock


\bibitem[\protect\citeauthoryear{Wu and Madden}{Wu and Madden}{2013}]%
        {WuM13}
\bibfield{author}{\bibinfo{person}{Eugene Wu} {and} \bibinfo{person}{Samuel
  Madden}.} \bibinfo{year}{2013}\natexlab{}.
\newblock \showarticletitle{Scorpion: Explaining Away Outliers in Aggregate
  Queries}.
\newblock \bibinfo{journal}{\emph{PVLDB}} \bibinfo{volume}{6},
  \bibinfo{number}{8} (\bibinfo{year}{2013}), \bibinfo{pages}{553--564}.
\newblock


\bibitem[\protect\citeauthoryear{Zhu, Deng, Nargesian, and Miller}{Zhu
  et~al\mbox{.}}{2019}]%
        {ZD19}
\bibfield{author}{\bibinfo{person}{Erkang Zhu}, \bibinfo{person}{Dong Deng},
  \bibinfo{person}{Fatemeh Nargesian}, {and} \bibinfo{person}{Renée~J.
  Miller}.} \bibinfo{year}{2019}\natexlab{}.
\newblock \showarticletitle{JOSIE: Overlap Set Similarity Search for Finding
  Joinable Tables in Data Lakes}. In \bibinfo{booktitle}{\emph{SIGMOD}}.
  \bibinfo{pages}{847--864}.
\newblock


\bibitem[\protect\citeauthoryear{Zhu, He, and Chaudhuri}{Zhu
  et~al\mbox{.}}{2017}]%
        {DBLP:journals/pvldb/ZhuHC17}
\bibfield{author}{\bibinfo{person}{Erkang Zhu}, \bibinfo{person}{Yeye He},
  {and} \bibinfo{person}{Surajit Chaudhuri}.} \bibinfo{year}{2017}\natexlab{}.
\newblock \showarticletitle{Auto-Join: Joining Tables by Leveraging
  Transformations}.
\newblock \bibinfo{journal}{\emph{Proc. {VLDB} Endow.}} \bibinfo{volume}{10},
  \bibinfo{number}{10} (\bibinfo{year}{2017}), \bibinfo{pages}{1034--1045}.
\newblock
\urldef\tempurl%
\url{https://doi.org/10.14778/3115404.3115409}
\showDOI{\tempurl}


\end{thebibliography}

\begin{appendix}
\section{Appendix}
\newcolumntype{P}[1]{>{\centering\arraybackslash}p{#1}}

\subsection{Top 20 patterns returned by Explanation Table}\label{append:et}
 Note since ET doesn't accept numeric attributes, we did a preprocessing step by converting numeric values into categorical value. As shown in \Cref{tab:patterns_et},  the large number of predicates make the results hard comprehend. Plus, since ET is mainly focused on maximizing information gain, the result should be interpreted as a set, which make the result even harder to understand.

\begin{table}[!h]
\centering
{\scriptsize
\centering
\begin{tabular}{|c|P{30em}|} 
\hline
\cellcolor{grey}\textbf{Num} & \cellcolor{grey}\textbf{Pattern Description} \\ 
Pattern 1 & \makecell{fg\_three\_pct$\in$[$0,0$] $\wedge$ off\_three\_ptreboundpct$\in$[$0,0$]} \\  
\hline

Pattern 2 & \makecell{off\_three\_ptreboundpct$\in$[$0,0$]} \\  
\hline

Pattern 3 & \makecell{rebounds$\in$[$1,3$] $\wedge$ off\_three\_ptreboundpct$\in$[$0,0$]} \\  
\hline

Pattern 4 & \makecell{fg\_three\_pct$\in$[$0,0$] $\wedge$ three\_ptassists$\in$[$0,0$] \\ 
fg\_three\_apct$\in$[$0,0$] fg\_two\_ablocked$\in$[$0,0$]} \\  
\hline

Pattern 5 & \makecell{offfgreboundpct$\in$[$0.07,0.75$] $\wedge$ prov\_game\_away\_\_points$\in$[$96,104$] $\wedge$
\\ three\_ptassists$\in$[$0,0$]} \\  
\hline

Pattern 6 & \makecell{prov\_game\_away\_\_points$\in$[$112,135$] $\wedge$ fg\_two\_ablocked$\in$[$0,0$]} \\  
\hline

Pattern 7 & \makecell{assistpoints$\in$[$8,36$] $\wedge$ rebounds$\in$[$6,22$] $\wedge$  \\ 
player\_name$\in$Draymond Green} \\  
\hline

Pattern 8 & \makecell{minutes$\in$[$31.78,49.63$] $\wedge$ fg\_three\_apct$\in$[$0,0$]} \\  
\hline

Pattern 9 & \makecell{offfgreboundpct$\in$[$0.07,0.75$] $\wedge$ minutes$\in$[$22.20,31.78$] \\
$\wedge$ assisted\_two\_spct$\in$[$0.83,1.00$] $\wedge$ def\_three\_ptreboundpct$\in$[$0.0, 0.08$] $\wedge$\\
offftreboundpct$\in$[$0,0$]} \\  
\hline

Pattern 10 & \makecell{off\_three\_ptreboundpct$\in$[$0,0$] $\wedge$ shotqualityavg$\in$[$0.45,0.49$] 
\\ fg\_two\_ablocked$\in$[$1,5$]} \\  
\hline

Pattern 11 & \makecell{fg\_three\_pct$\in$[$0.4,1.0$] $\wedge$ tspct$\in$[$0.67,1$] $\wedge$ \\
defftreboundpct$\in$[$0,0$]} \\  
\hline

Pattern 12 & \makecell{assisted\_two\_spct$\in$[$0,0$] $\wedge$ def\_three\_ptreboundpct$\in$[$0,0$] $\wedge$ \\
offftreboundpct$\in$[$0,0$] fg\_three\_apct$\in$[$0,0.19$] usage$\in$[$0,0$]} \\  
\hline

Pattern 13 & \makecell{off\_three\_ptreboundpct$\in$[$0,0$] $\wedge$ tspct$\in$[$0,0$] $\wedge$ \\
defftreboundpct$\in$[$0,0$] $\wedge$ three\_ptassists$\in$[$0,0$] $\wedge$ fg\_three\_apct$\in$[$0,0.19$] \\
 shotqualityavg$\in$[$0,0.45$]} \\  
\hline

Pattern 14 & \makecell{assistpoints$\in$[$8,36$] fg\_three\_pct$\in$[$0.4,1.0$] $\wedge$ minutes$\in$[$31.78,49.63$] $\wedge$ \\ prov\_game\_away\_\_points$\in$[$96,104$] $\wedge$ tspct$\in$[$0.67,1$] $\wedge$ \\ defftreboundpct$\in$[$0,0$]} \\  
\hline

Pattern 15 & \makecell{rebounds$\in$[$3,6$]$\wedge$  fg\_three\_pct$\in$[$0.4,1.0$] $\wedge$ \\
 shotqualityavg$\in$[$0.45,0.49$]} \\  
\hline

Pattern 16 & \makecell{assistpoints$\in$[$0,3$] $\wedge$ assisted\_two\_spct$\in$[$0.83,1.00$] $\wedge$ \\ 
fg\_two\_ablocked$\in$[$0,0$] $\wedge$ usage$\in$[$0,0$]} \\  
\hline

Pattern 17 & \makecell{offfgreboundpct$\in$[$0.07,0.75$] $\wedge$ prov\_game\_away\_\_points$\in$[$104,112$] 
\\ $\wedge$ usage$\in$[$0,0$]} \\  
\hline

Pattern 18 & \makecell{assistpoints$\in$[$3,8$] $\wedge$ off\_three\_ptreboundpct$\in$[$0,0$] $\wedge$ \\
assisted\_two\_spct$\in$[$0.5,0.83$] tspct$\in$[$0,0$]} \\  
\hline

Pattern 19 & \makecell{assistpoints$\in$[$8,36$] $\wedge$ player\_name$\in$Jarrett Jack} \\  
\hline

Pattern 20 & \makecell{prov\_game\_away\_\_points$\in$[$112,135$] $\wedge$ fg\_two\_ablocked$\in$[$1,5$]} \\ 
\hline
\end{tabular}
}
\label{tab:patterns_et}
\caption{First 20 Patterns Returned from Explanation Table}
\end{table}

\begin{figure*}[!h]
\centering
{\scriptsize
\centering
\begin{tabular}{|c|P{25em}|P{20em}|c|c|c|P{5em}|}
\hline
\cellcolor{grey}\textbf{Nodes} & \cellcolor{grey}\textbf{Edges} & \cellcolor{grey}\textbf{Pattern Desc} & \cellcolor{grey}\textbf{Precision} & \cellcolor{grey}\textbf{Recall} & \cellcolor{grey}\textbf{\abbrF} & \cellcolor{grey}\textbf{Primary Tuple} \\
\makecell{$A_1$: PT \\ $A_2$: team \\ $A_3$: player\_salary} & \makecell{($A_2$.team\_id)=($A_1$.prov\_game\_away\_\_id)\\($A_1$.prov\_player\_player\_\_id)=($A_3$.player\_id)\\($A_1$.prov\_season\_season\_\_id)=($A_3$.season\_id)} & \makecell{$A_1$.prov\_season\_season\_\_type=regular season \\ $A_3$.salary<15330435.0} & 1.0 & 1.0 & 1.0 & 2015-16\\
\hline
\makecell{$A_1$: PT \\ $A_2$: team \\ $A_3$: player\_salary} & \makecell{($A_2$.team\_id)=($A_1$.prov\_game\_winner\_\_id)\\($A_1$.prov\_player\_player\_\_id)=($A_3$.player\_id)\\($A_1$.prov\_season\_season\_\_id)=($A_3$.season\_id)} & \makecell{$A_1$.prov\_season\_season\_\_type=regular season \\ $A_3$.salary<15330435.0} & 1.0 & 1.0 & 1.0 & 2015-16\\
\hline
\makecell{$A_1$: PT \\ $A_2$: team \\ $A_3$: player\_salary} & \makecell{($A_2$.team\_id)=($A_1$.prov\_game\_home\_\_id)\\($A_1$.prov\_player\_player\_\_id)=($A_3$.player\_id)\\($A_1$.prov\_season\_season\_\_id)=($A_3$.season\_id)} & \makecell{$A_1$.prov\_season\_season\_\_type=regular season \\ $A_3$.salary<15330435.0} & 1.0 & 1.0 & 1.0 & 2015-16\\
\hline
\makecell{$A_1$: PT \\ $A_2$: team \\ $A_3$: player\_salary} & \makecell{($A_2$.team\_id)=($A_1$.prov\_game\_away\_\_id)\\($A_1$.prov\_player\_player\_\_id)=($A_3$.player\_id)\\($A_1$.prov\_season\_season\_\_id)=($A_3$.season\_id)} & \makecell{$A_1$.prov\_player\_\_game\_\_stats\_tspct<0.69 \\ $A_1$.prov\_player\_\_game\_\_stats\_usage<20.51 \\ $A_1$.prov\_season\_season\_\_type=regular season \\ $A_3$.salary>14260870.0} & 1.0 & 0.71 & 0.83 & 2016-17\\
\hline
\makecell{$A_1$: PT \\ $A_2$: team \\ $A_3$: player\_salary} & \makecell{($A_2$.team\_id)=($A_1$.prov\_game\_home\_\_id)\\($A_1$.prov\_player\_player\_\_id)=($A_3$.player\_id)\\($A_1$.prov\_season\_season\_\_id)=($A_3$.season\_id)} & \makecell{$A_1$.prov\_player\_\_game\_\_stats\_tspct<0.69 \\ $A_1$.prov\_player\_\_game\_\_stats\_usage<20.51 \\ $A_1$.prov\_season\_season\_\_type=regular season \\ $A_3$.salary>14260870.0} & 1.0 & 0.71 & 0.83 & 2016-17\\
\hline
\makecell{$A_1$: PT \\ $A_2$: team \\ $A_3$: player\_salary} & \makecell{($A_2$.team\_id)=($A_1$.prov\_game\_winner\_\_id)\\($A_1$.prov\_player\_player\_\_id)=($A_3$.player\_id)\\($A_1$.prov\_season\_season\_\_id)=($A_3$.season\_id)} & \makecell{$A_1$.prov\_player\_\_game\_\_stats\_tspct<0.69 \\ $A_1$.prov\_player\_\_game\_\_stats\_usage<20.51 \\ $A_1$.prov\_season\_season\_\_type=regular season \\ $A_3$.salary>14260870.0} & 1.0 & 0.71 & 0.83 & 2016-17\\
\hline
\makecell{$A_1$: PT \\ $A_2$: team \\ $A_3$: player\_salary} & \makecell{($A_2$.team\_id)=($A_1$.prov\_game\_winner\_\_id)\\($A_1$.prov\_player\_player\_\_id)=($A_3$.player\_id)\\($A_1$.prov\_season\_season\_\_id)=($A_3$.season\_id)} & \makecell{$A_1$.prov\_player\_\_game\_\_stats\_minutes>30.7 \\ $A_1$.prov\_player\_\_game\_\_stats\_tspct>0.42 \\ $A_1$.prov\_season\_season\_\_type=regular season \\ $A_3$.salary<15330435.0} & 1.0 & 0.69 & 0.82 & 2015-16\\
\hline
\makecell{$A_1$: PT \\ $A_2$: team \\ $A_3$: player\_salary} & \makecell{($A_2$.team\_id)=($A_1$.prov\_game\_away\_\_id)\\($A_1$.prov\_player\_player\_\_id)=($A_3$.player\_id)\\($A_1$.prov\_season\_season\_\_id)=($A_3$.season\_id)} & \makecell{$A_1$.prov\_player\_\_game\_\_stats\_minutes>30.7 \\ $A_1$.prov\_player\_\_game\_\_stats\_tspct>0.42 \\ $A_1$.prov\_season\_season\_\_type=regular season \\ $A_3$.salary<15330435.0} & 1.0 & 0.69 & 0.82 & 2015-16\\
\hline
\makecell{$A_1$: PT \\ $A_2$: team \\ $A_3$: player\_salary} & \makecell{($A_2$.team\_id)=($A_1$.prov\_game\_home\_\_id)\\($A_1$.prov\_player\_player\_\_id)=($A_3$.player\_id)\\($A_1$.prov\_season\_season\_\_id)=($A_3$.season\_id)} & \makecell{$A_1$.prov\_player\_\_game\_\_stats\_tspct>0.42 \\ $A_1$.prov\_player\_\_game\_\_stats\_usage>14.06 \\ $A_1$.prov\_season\_season\_\_type=regular season \\ $A_3$.salary<15330435.0} & 1.0 & 0.67 & 0.8 & 2015-16\\
\hline
\makecell{$A_1$: PT \\ $A_2$: team \\ $A_3$: player\_salary} & \makecell{($A_2$.team\_id)=($A_1$.prov\_game\_winner\_\_id)\\($A_1$.prov\_player\_player\_\_id)=($A_3$.player\_id)\\($A_1$.prov\_season\_season\_\_id)=($A_3$.season\_id)} & \makecell{$A_1$.prov\_player\_\_game\_\_stats\_tspct>0.42 \\ $A_1$.prov\_player\_\_game\_\_stats\_usage>14.02 \\ $A_1$.prov\_season\_season\_\_type=regular season \\ $A_3$.salary<15330435.0 \\ $A_2$.team=GSW} & 1.0 & 0.59 & 0.74 & 2015-16\\
\hline
\makecell{$A_1$: PT \\ $A_2$: team \\ $A_3$: player\_salary} & \makecell{($A_2$.team\_id)=($A_1$.prov\_game\_winner\_\_id)\\($A_1$.prov\_player\_player\_\_id)=($A_3$.player\_id)\\($A_1$.prov\_season\_season\_\_id)=($A_3$.season\_id)} & \makecell{$A_1$.prov\_player\_\_game\_\_stats\_minutes<36.105 \\ $A_1$.prov\_player\_\_game\_\_stats\_usage<20.21 \\ $A_1$.prov\_season\_season\_\_type=regular season \\ $A_3$.salary>14260870.0 \\ $A_2$.team=GSW} & 1.0 & 0.59 & 0.74 & 2016-17\\
\hline
\makecell{$A_1$: PT \\ $A_2$: team \\ $A_3$: player\_salary} & \makecell{($A_2$.team\_id)=($A_1$.prov\_game\_home\_\_id)\\($A_1$.prov\_player\_player\_\_id)=($A_3$.player\_id)\\($A_1$.prov\_season\_season\_\_id)=($A_3$.season\_id)} & \makecell{$A_1$.prov\_player\_\_game\_\_stats\_minutes<36.57 \\ $A_1$.prov\_player\_\_game\_\_stats\_usage>14.06 \\ $A_1$.prov\_season\_season\_\_type=regular season \\ $A_3$.salary>14260870.0} & 1.0 & 0.59 & 0.74 & 2016-17\\
\hline
\makecell{$A_1$: PT \\ $A_2$: team} & \makecell{($A_2$.team\_id)=($A_1$.prov\_game\_away\_\_id)\\($A_2$.team\_id)=($A_1$.prov\_game\_winner\_\_id)} & \makecell{$A_1$.prov\_player\_\_game\_\_stats\_fg\_\_three\_\_apct<0.4875 \\ $A_1$.prov\_season\_season\_\_type=regular season \\ $A_2$.team=GSW} & 0.65 & 0.83 & 0.73 & 2015-16\\
\hline
\makecell{$A_1$: PT \\ $A_2$: team \\ $A_3$: team\_game\_stats \\ $A_4$: team} & \makecell{($A_2$.team\_id)=($A_1$.prov\_game\_home\_\_id)\\($A_1$.prov\_game\_game\_\_date)=($A_3$.game\_date)\\($A_1$.prov\_game\_home\_\_id)=($A_3$.home\_id)\\($A_4$.team\_id)=($A_3$.team\_id)} & \makecell{$A_1$.prov\_player\_\_game\_\_stats\_deflongmidrangereboundpct>0.1 \\ $A_1$.prov\_season\_season\_\_type=regular season \\ $A_4$.team=GSW} & 0.59 & 0.84 & 0.69 & 2015-16\\
\hline
\makecell{$A_1$: PT \\ $A_2$: team \\ $A_3$: team\_game\_stats \\ $A_4$: team} & \makecell{($A_2$.team\_id)=($A_1$.prov\_game\_away\_\_id)\\($A_1$.prov\_game\_game\_\_date)=($A_3$.game\_date)\\($A_1$.prov\_game\_home\_\_id)=($A_3$.home\_id)\\($A_4$.team\_id)=($A_3$.team\_id)} & \makecell{$A_1$.prov\_player\_\_game\_\_stats\_deflongmidrangereboundpct>0.1 \\ $A_1$.prov\_season\_season\_\_type=regular season \\ $A_4$.team=GSW} & 0.59 & 0.84 & 0.69 & 2015-16\\
\hline
\makecell{$A_1$: PT \\ $A_2$: team \\ $A_3$: team\_game\_stats \\ $A_4$: team} & \makecell{($A_2$.team\_id)=($A_1$.prov\_game\_winner\_\_id)\\($A_1$.prov\_game\_game\_\_date)=($A_3$.game\_date)\\($A_1$.prov\_game\_home\_\_id)=($A_3$.home\_id)\\($A_4$.team\_id)=($A_3$.team\_id)} & \makecell{$A_1$.prov\_player\_\_game\_\_stats\_deflongmidrangereboundpct>0.1 \\ $A_1$.prov\_season\_season\_\_type=regular season \\ $A_4$.team=GSW} & 0.59 & 0.84 & 0.69 & 2015-16\\
\hline
\makecell{$A_1$: PT \\ $A_2$: team \\ $A_3$: play\_for \\ $A_4$: team} & \makecell{($A_2$.team\_id)=($A_1$.prov\_game\_away\_\_id)\\($A_1$.prov\_player\_player\_\_id)=($A_3$.player\_id)\\($A_4$.team\_id)=($A_3$.team\_id)} & \makecell{$A_1$.prov\_player\_\_game\_\_stats\_minutes>30.7 \\ $A_1$.prov\_season\_season\_\_type=regular season \\ $A_3$.date\_end=2019-04-09 \\ $A_4$.team=GSW} & 0.58 & 0.84 & 0.69 & 2015-16\\
\hline
\makecell{$A_1$: PT \\ $A_2$: team \\ $A_3$: play\_for \\ $A_4$: team} & \makecell{($A_2$.team\_id)=($A_1$.prov\_game\_home\_\_id)\\($A_1$.prov\_player\_player\_\_id)=($A_3$.player\_id)\\($A_4$.team\_id)=($A_3$.team\_id)} & \makecell{$A_1$.prov\_player\_\_game\_\_stats\_minutes>30.7 \\ $A_1$.prov\_season\_season\_\_type=regular season \\ $A_3$.date\_end=2019-04-09 \\ $A_4$.team=GSW} & 0.58 & 0.84 & 0.69 & 2015-16\\
\hline
\makecell{$A_1$: PT \\ $A_2$: team \\ $A_3$: play\_for \\ $A_4$: team} & \makecell{($A_2$.team\_id)=($A_1$.prov\_game\_winner\_\_id)\\($A_1$.prov\_player\_player\_\_id)=($A_3$.player\_id)\\($A_4$.team\_id)=($A_3$.team\_id)} & \makecell{$A_1$.prov\_player\_\_game\_\_stats\_minutes>30.7 \\ $A_1$.prov\_season\_season\_\_type=regular season \\ $A_3$.date\_end=2019-04-09 \\ $A_4$.team=GSW} & 0.58 & 0.84 & 0.69 & 2015-16\\
\hline
\makecell{} & \makecell{A\_1} & \makecell{$A_1$.prov\_player\_\_game\_\_stats\_minutes>30.7 \\ $A_1$.prov\_season\_season\_\_type=regular season} & 0.58 & 0.84 & 0.69 & 2015-16\\
\hline
\end{tabular}
}
\label{tab:tops-nba1}
\caption{$\query_{nba1}$ Top-20 patterns}
\end{figure*}

\begin{figure*}[!h]
\centering
{\scriptsize
\centering
\begin{tabular}{|c|P{20em}|P{30em}|c|c|c|P{4em}|}
\hline
\cellcolor{grey}\textbf{Nodes} & \cellcolor{grey}\textbf{Edges} & \cellcolor{grey}\textbf{Pattern Desc} & \cellcolor{grey}\textbf{Precision} & \cellcolor{grey}\textbf{Recall} & \cellcolor{grey}\textbf{\abbrF} & \cellcolor{grey}\textbf{Primary Tuple} \\

\makecell{$A_1$: PT \\ $A_2$: team \\ $A_3$: player\_game\_stats \\ $A_4$: player} & \makecell{($A_2$.team\_id)=($A_1$.prov\_game\_away\_\_id)\\($A_1$.prov\_game\_game\_\_date)=($A_3$.game\_date)\\($A_1$.prov\_game\_home\_\_id)=($A_3$.home\_id)\\($A_4$.player\_id)=($A_3$.player\_id)} & \makecell{$A_1$.prov\_season\_season\_\_type=regular season \\ $A_1$.prov\_team\_\_game\_\_stats\_assistpoints<68.0 \\ $A_4$.player\_name=Draymond Green} & 0.62 & 0.9 & 0.74 & 2013-14\\
\hline
\makecell{$A_1$: PT \\ $A_2$: team \\ $A_3$: player\_game\_stats \\ $A_4$: player} & \makecell{($A_2$.team\_id)=($A_1$.prov\_game\_winner\_\_id)\\($A_1$.prov\_game\_game\_\_date)=($A_3$.game\_date)\\($A_1$.prov\_game\_home\_\_id)=($A_3$.home\_id)\\($A_4$.player\_id)=($A_3$.player\_id)} & \makecell{$A_1$.prov\_season\_season\_\_type=regular season \\ $A_1$.prov\_team\_\_game\_\_stats\_assistpoints<68.0 \\ $A_4$.player\_name=Draymond Green} & 0.62 & 0.9 & 0.74 & 2013-14\\
\hline
\makecell{$A_1$: PT \\ $A_2$: team \\ $A_3$: player\_game\_stats \\ $A_4$: player} & \makecell{($A_2$.team\_id)=($A_1$.prov\_game\_home\_\_id)\\($A_1$.prov\_game\_game\_\_date)=($A_3$.game\_date)\\($A_1$.prov\_game\_home\_\_id)=($A_3$.home\_id)\\($A_4$.player\_id)=($A_3$.player\_id)} & \makecell{$A_1$.prov\_season\_season\_\_type=regular season \\ $A_1$.prov\_team\_\_game\_\_stats\_assistpoints<68.0 \\ $A_4$.player\_name=Draymond Green} & 0.62 & 0.9 & 0.74 & 2013-14\\
\hline
\makecell{$A_1$: PT \\ $A_2$: team \\ $A_3$: player\_game\_stats \\ $A_4$: player} & \makecell{($A_2$.team\_id)=($A_1$.prov\_game\_winner\_\_id)\\($A_1$.prov\_game\_game\_\_date)=($A_3$.game\_date)\\($A_1$.prov\_game\_home\_\_id)=($A_3$.home\_id)\\($A_4$.player\_id)=($A_3$.player\_id)} & \makecell{$A_1$.prov\_season\_season\_\_type=regular season \\ $A_1$.prov\_team\_\_game\_\_stats\_assistpoints>57.0 \\ $A_1$.prov\_team\_\_game\_\_stats\_nonputback- \\ sassisted\_\_two\_\_spct>0.55 \\ $A_4$.player\_name=Harrison Barnes} & 0.78 & 0.7 & 0.74 & 2014-15\\
\hline
\makecell{$A_1$: PT \\ $A_2$: team \\ $A_3$: player\_game\_stats \\ $A_4$: player} & \makecell{($A_2$.team\_id)=($A_1$.prov\_game\_home\_\_id)\\($A_1$.prov\_game\_game\_\_date)=($A_3$.game\_date)\\($A_1$.prov\_game\_home\_\_id)=($A_3$.home\_id)\\($A_4$.player\_id)=($A_3$.player\_id)} & \makecell{$A_1$.prov\_season\_season\_\_type=regular season \\ $A_1$.prov\_team\_\_game\_\_stats\_assistpoints>57.0 \\ $A_1$.prov\_team\_\_game\_\_stats\_nonputbacksassisted\_\_two\_\_spct>0.5475 \\ $A_4$.player\_name=Harrison Barnes} & 0.78 & 0.7 & 0.74 & 2014-15\\
\hline
\makecell{$A_1$: PT \\ $A_2$: team \\ $A_3$: player\_salary \\ $A_4$: player} & \makecell{($A_2$.team\_id)=($A_1$.prov\_game\_winner\_\_id)\\($A_1$.prov\_season\_season\_\_id)=($A_3$.season\_id)\\($A_4$.player\_id)=($A_3$.player\_id)} & \makecell{$A_1$.prov\_season\_season\_\_type=regular season \\ $A_1$.prov\_team\_\_game\_\_stats\_assistpoints<68.0 \\ $A_4$.player\_name=Gal Mekel} & 0.61 & 0.9 & 0.73 & 2013-14\\
\hline
\makecell{} & \makecell{A\_1} & \makecell{$A_1$.prov\_season\_season\_\_type=regular season \\ $A_1$.prov\_team\_\_game\_\_stats\_assistpoints<68.0} & 0.61 & 0.9 & 0.73 & 2013-14\\
\hline
\makecell{} & \makecell{A\_1} & \makecell{$A_1$.prov\_season\_season\_\_type=regular season \\ $A_1$.prov\_team\_\_game\_\_stats\_assistpoints<68.0 \\ $A_1$.prov\_team\_\_game\_\_stats\_offatrimreboundpct>0.25} & 0.71 & 0.72 & 0.72 & 2013-14\\
\hline
\makecell{$A_1$: PT \\ $A_2$: team \\ $A_3$: player\_game\_stats \\ $A_4$: player} & \makecell{($A_2$.team\_id)=($A_1$.prov\_game\_away\_\_id)\\($A_1$.prov\_game\_game\_\_date)=($A_3$.game\_date)\\($A_1$.prov\_game\_home\_\_id)=($A_3$.home\_id)\\($A_4$.player\_id)=($A_3$.player\_id)} & \makecell{$A_1$.prov\_season\_season\_\_type=regular season \\ $A_1$.prov\_team\_\_game\_\_stats\_assistpoints<67.0 \\ $A_1$.prov\_team\_\_game\_\_stats\_nonputbacksassisted\_\_two\_\_spct<0.69 \\ $A_4$.player\_name=David Lee} & 0.75 & 0.67 & 0.71 & 2013-14\\
\hline
\makecell{$A_1$: PT \\ $A_2$: team \\ $A_3$: player\_salary \\ $A_4$: player} & \makecell{($A_2$.team\_id)=($A_1$.prov\_game\_winner\_\_id)\\($A_1$.prov\_season\_season\_\_id)=($A_3$.season\_id)\\($A_4$.player\_id)=($A_3$.player\_id)} & \makecell{$A_1$.prov\_season\_season\_\_type=regular season \\ $A_1$.prov\_team\_\_game\_\_stats\_assistpoints>58.5 \\ $A_1$.prov\_team\_\_game\_\_stats\_nonputbacksassisted\_\_two\_\_spct>0.55 \\ $A_4$.player\_name=Mike Muscala} & 0.74 & 0.68 & 0.71 & 2014-15\\
\hline
\makecell{$A_1$: PT \\ $A_2$: team \\ $A_3$: player\_game\_stats \\ $A_4$: player} & \makecell{($A_2$.team\_id)=($A_1$.prov\_game\_home\_\_id)\\($A_1$.prov\_game\_game\_\_date)=($A_3$.game\_date)\\($A_1$.prov\_game\_home\_\_id)=($A_3$.home\_id)\\($A_4$.player\_id)=($A_3$.player\_id)} & \makecell{$A_1$.prov\_season\_season\_\_type=regular season \\ $A_1$.prov\_team\_\_game\_\_stats\_assistpoints>50.0 \\ $A_1$.prov\_team\_\_game\_\_stats\_nonputbacksassisted\_\_two\_\_spct>0.5525 \\ $A_4$.player\_name=Stephen Curry} & 0.65 & 0.78 & 0.71 & 2014-15\\
\hline
\makecell{$A_1$: PT \\ $A_2$: team} & \makecell{($A_2$.team\_id)=($A_1$.prov\_game\_home\_\_id)\\($A_2$.team\_id)=($A_1$.prov\_game\_winner\_\_id)} & \makecell{$A_1$.prov\_season\_season\_\_type=regular season \\ $A_1$.prov\_team\_\_game\_\_stats\_assistpoints>57.0 \\ $A_2$.team=GSW} & 0.72 & 0.65 & 0.69 & 2014-15\\
\hline
\makecell{$A_1$: PT \\ $A_2$: team} & \makecell{($A_2$.team\_id)=($A_1$.prov\_game\_away\_\_id)\\($A_2$.team\_id)=($A_1$.prov\_game\_winner\_\_id)} & \makecell{$A_1$.prov\_season\_season\_\_type=regular season \\ $A_2$.team=GSW} & 0.54 & 0.93 & 0.68 & 2014-15\\
\hline
\makecell{$A_1$: PT \\ $A_2$: team \\ $A_3$: player\_game\_stats \\ $A_4$: player} & \makecell{($A_2$.team\_id)=($A_1$.prov\_game\_winner\_\_id)\\($A_1$.prov\_game\_game\_\_date)=($A_3$.game\_date)\\($A_1$.prov\_game\_home\_\_id)=($A_3$.home\_id)\\($A_4$.player\_id)=($A_3$.player\_id)} & \makecell{$A_1$.prov\_season\_season\_\_type=regular season \\ $A_1$.prov\_team\_\_game\_\_stats\_assistpoints>50.0 \\ $A_1$.prov\_team\_\_game\_\_stats\_nonputbacksassisted\_\_two\_\_spct>0.55 \\ $A_4$.player\_name=Marreese Speights} & 0.63 & 0.73 & 0.68 & 2014-15\\
\hline
\makecell{$A_1$: PT \\ $A_2$: team \\ $A_3$: player\_game\_stats \\ $A_4$: player} & \makecell{($A_2$.team\_id)=($A_1$.prov\_game\_home\_\_id)\\($A_1$.prov\_game\_game\_\_date)=($A_3$.game\_date)\\($A_1$.prov\_game\_home\_\_id)=($A_3$.home\_id)\\($A_4$.player\_id)=($A_3$.player\_id)} & \makecell{$A_1$.prov\_season\_season\_\_type=regular season \\ $A_1$.prov\_team\_\_game\_\_stats\_nonputbacksassisted\_\_two\_\_spct<0.68 \\ $A_1$.prov\_team\_\_game\_\_stats\_offatrimreboundpct>0.25 \\ $A_4$.player\_name=Klay Thompson} & 0.71 & 0.65 & 0.68 & 2013-14\\
\hline
\makecell{$A_1$: PT \\ $A_2$: team \\ $A_3$: player\_game\_stats \\ $A_4$: player} & \makecell{($A_2$.team\_id)=($A_1$.prov\_game\_away\_\_id)\\($A_1$.prov\_game\_game\_\_date)=($A_3$.game\_date)\\($A_1$.prov\_game\_home\_\_id)=($A_3$.home\_id)\\($A_4$.player\_id)=($A_3$.player\_id)} & \makecell{$A_1$.prov\_season\_season\_\_type=regular season \\ $A_1$.prov\_team\_\_game\_\_stats\_nonputbacksassisted\_\_two\_\_spct<0.68 \\ $A_1$.prov\_team\_\_game\_\_stats\_offatrimreboundpct>0.25 \\ $A_4$.player\_name=Klay Thompson} & 0.71 & 0.65 & 0.68 & 2013-14\\
\hline
\makecell{$A_1$: PT \\ $A_2$: team} & \makecell{($A_2$.team\_id)=($A_1$.prov\_game\_winner\_\_id)} & \makecell{$A_1$.prov\_season\_season\_\_type=regular season \\ $A_2$.team=GSW} & 0.57 & 0.82 & 0.67 & 2014-15\\
\hline
\makecell{$A_1$: PT \\ $A_2$: team \\ $A_3$: player\_salary \\ $A_4$: player} & \makecell{($A_2$.team\_id)=($A_1$.prov\_game\_winner\_\_id)\\($A_1$.prov\_season\_season\_\_id)=($A_3$.season\_id)\\($A_4$.player\_id)=($A_3$.player\_id)} & \makecell{$A_1$.prov\_season\_season\_\_type=regular season \\ $A_1$.prov\_team\_\_game\_\_stats\_nonputbacksassisted\_\_two\_\_spct<0.69 \\ $A_1$.prov\_team\_\_game\_\_stats\_offatrimreboundpct>0.25 \\ $A_4$.player\_name=Gal Mekel} & 0.65 & 0.67 & 0.66 & 2013-14\\
\hline
\makecell{} & \makecell{A\_1} & \makecell{$A_1$.prov\_season\_season\_\_type=regular season \\ $A_1$.prov\_team\_\_game\_\_stats\_assisted\_\_three\_\_spct>0.6975 \\ $A_1$.prov\_team\_\_game\_\_stats\_assistpoints>58.5} & 0.71 & 0.62 & 0.66 & 2014-15\\
\hline
\makecell{$A_1$: PT \\ $A_2$: team} & \makecell{($A_2$.team\_id)=($A_1$.prov\_game\_home\_\_id)\\($A_2$.team\_id)=($A_1$.prov\_game\_winner\_\_id)} & \makecell{$A_1$.prov\_season\_season\_\_type=regular season \\ $A_1$.prov\_team\_\_game\_\_stats\_assistpoints>57.0 \\ $A_1$.prov\_team\_\_game\_\_stats\_nonputbacksassisted\_\_two\_\_spct>0.59 \\ $A_2$.team=GSW} & 0.77 & 0.58 & 0.66 & 2014-15\\
\hline

\end{tabular}
}
\label{tab:tops-nba2}
\caption{$\query_{nba2}$ Top-20 patterns}
\end{figure*}

\begin{figure*}[!h]
\centering
{\scriptsize
\centering
\begin{tabular}{|c|P{20em}|P{30em}|c|c|c|P{5em}|}
\hline
\cellcolor{grey}\textbf{Nodes} & \cellcolor{grey}\textbf{Edges} & \cellcolor{grey}\textbf{Pattern Desc} & \cellcolor{grey}\textbf{Precision} & \cellcolor{grey}\textbf{Recall} & \cellcolor{grey}\textbf{\abbrF} & \cellcolor{grey}\textbf{Primary Tuple} \\

\makecell{$A_1$: PT \\ $A_2$: team \\ $A_3$: player\_salary} & \makecell{($A_2$.team\_id)=($A_1$.prov\_game\_away\_\_id)\\($A_1$.prov\_player\_player\_\_id)=($A_3$.player\_id)\\($A_1$.prov\_season\_season\_\_id)=($A_3$.season\_id)} & \makecell{$A_1$.prov\_season\_season\_\_type=regular season \\ $A_3$.salary>14500000.0} & 1.0 & 1.0 & 1.0 & 2009-10\\
\hline
\makecell{$A_1$: PT \\ $A_2$: team \\ $A_3$: player\_salary} & \makecell{($A_2$.team\_id)=($A_1$.prov\_game\_winner\_\_id)\\($A_1$.prov\_player\_player\_\_id)=($A_3$.player\_id)\\($A_1$.prov\_season\_season\_\_id)=($A_3$.season\_id)} & \makecell{$A_1$.prov\_season\_season\_\_type=regular season \\ $A_3$.salary>14500000.0} & 1.0 & 1.0 & 1.0 & 2009-10\\
\hline
\makecell{$A_1$: PT \\ $A_2$: team \\ $A_3$: player\_salary} & \makecell{($A_2$.team\_id)=($A_1$.prov\_game\_home\_\_id)\\($A_1$.prov\_player\_player\_\_id)=($A_3$.player\_id)\\($A_1$.prov\_season\_season\_\_id)=($A_3$.season\_id)} & \makecell{$A_1$.prov\_season\_season\_\_type=regular season \\ $A_3$.salary>14500000.0} & 1.0 & 1.0 & 1.0 & 2009-10\\
\hline
\makecell{$A_1$: PT \\ $A_2$: team \\ $A_3$: team\_game\_stats \\ $A_4$: team} & \makecell{($A_2$.team\_id)=($A_1$.prov\_game\_away\_\_id)\\($A_1$.prov\_game\_game\_\_date)=($A_3$.game\_date)\\($A_1$.prov\_game\_home\_\_id)=($A_3$.home\_id)\\($A_4$.team\_id)=($A_3$.team\_id)} & \makecell{$A_1$.prov\_season\_season\_\_type=regular season \\ $A_4$.team=MIA} & 0.96 & 1.0 & 0.98 & 2010-11\\
\hline
\makecell{$A_1$: PT \\ $A_2$: team \\ $A_3$: team\_game\_stats \\ $A_4$: team} & \makecell{($A_2$.team\_id)=($A_1$.prov\_game\_home\_\_id)\\($A_1$.prov\_game\_game\_\_date)=($A_3$.game\_date)\\($A_1$.prov\_game\_home\_\_id)=($A_3$.home\_id)\\($A_4$.team\_id)=($A_3$.team\_id)} & \makecell{$A_1$.prov\_season\_season\_\_type=regular season \\ $A_4$.team=MIA} & 0.96 & 1.0 & 0.98 & 2010-11\\
\hline
\makecell{$A_1$: PT \\ $A_2$: team \\ $A_3$: team\_game\_stats \\ $A_4$: team} & \makecell{($A_2$.team\_id)=($A_1$.prov\_game\_winner\_\_id)\\($A_1$.prov\_game\_game\_\_date)=($A_3$.game\_date)\\($A_1$.prov\_game\_home\_\_id)=($A_3$.home\_id)\\($A_4$.team\_id)=($A_3$.team\_id)} & \makecell{$A_1$.prov\_season\_season\_\_type=regular season \\ $A_4$.team=MIA} & 0.96 & 1.0 & 0.98 & 2010-11\\
\hline
\makecell{$A_1$: PT \\ $A_2$: team} & \makecell{($A_2$.team\_id)=($A_1$.prov\_game\_away\_\_id)\\($A_2$.team\_id)=($A_1$.prov\_game\_winner\_\_id)} & \makecell{$A_1$.prov\_season\_season\_\_type=regular season \\ $A_2$.team=CLE} & 1.0 & 0.87 & 0.93 & 2009-10\\
\hline
\makecell{$A_1$: PT \\ $A_2$: team} & \makecell{($A_2$.team\_id)=($A_1$.prov\_game\_winner\_\_id)} & \makecell{$A_1$.prov\_season\_season\_\_type=regular season \\ $A_2$.team=CLE} & 0.98 & 0.79 & 0.88 & 2009-10\\
\hline
\makecell{$A_1$: PT \\ $A_2$: team \\ $A_3$: play\_for \\ $A_4$: team} & \makecell{($A_2$.team\_id)=($A_1$.prov\_game\_winner\_\_id)\\($A_1$.prov\_player\_player\_\_id)=($A_3$.player\_id)\\($A_4$.team\_id)=($A_3$.team\_id)} & \makecell{$A_1$.prov\_season\_season\_\_type=regular season \\ $A_3$.date\_end=2018-04-11 \\ $A_2$.team=CLE \\ $A_4$.team=CLE} & 0.98 & 0.79 & 0.88 & 2009-10\\
\hline
\makecell{$A_1$: PT \\ $A_2$: team} & \makecell{($A_2$.team\_id)=($A_1$.prov\_game\_home\_\_id)\\($A_2$.team\_id)=($A_1$.prov\_game\_winner\_\_id)} & \makecell{$A_1$.prov\_season\_season\_\_type=regular season \\ $A_2$.team=MIA} & 1.0 & 0.73 & 0.85 & 2010-11\\
\hline
\makecell{$A_1$: PT \\ $A_2$: team \\ $A_3$: player\_salary} & \makecell{($A_2$.team\_id)=($A_1$.prov\_game\_winner\_\_id)\\($A_1$.prov\_player\_player\_\_id)=($A_3$.player\_id)\\($A_1$.prov\_season\_season\_\_id)=($A_3$.season\_id)} & \makecell{$A_1$.prov\_player\_\_game\_\_stats\_assisted\_\_two\_\_spct>0.19 \\ $A_1$.prov\_player\_\_game\_\_stats\_usage>28.025 \\ $A_1$.prov\_season\_season\_\_type=regular season \\ $A_3$.salary>14500000.0} & 1.0 & 0.71 & 0.83 & 2009-10\\
\hline
\makecell{$A_1$: PT \\ $A_2$: team \\ $A_3$: player\_salary} & \makecell{($A_2$.team\_id)=($A_1$.prov\_game\_away\_\_id)\\($A_1$.prov\_player\_player\_\_id)=($A_3$.player\_id)\\($A_1$.prov\_season\_season\_\_id)=($A_3$.season\_id)} & \makecell{$A_1$.prov\_player\_\_game\_\_stats\_assisted\_\_two\_\_spct>0.19 \\ $A_1$.prov\_player\_\_game\_\_stats\_usage>28.025 \\ $A_1$.prov\_season\_season\_\_type=regular season \\ $A_3$.salary>14500000.0} & 1.0 & 0.71 & 0.83 & 2009-10\\
\hline
\makecell{$A_1$: PT \\ $A_2$: team \\ $A_3$: player\_salary} & \makecell{($A_2$.team\_id)=($A_1$.prov\_game\_home\_\_id)\\($A_1$.prov\_player\_player\_\_id)=($A_3$.player\_id)\\($A_1$.prov\_season\_season\_\_id)=($A_3$.season\_id)} & \makecell{$A_1$.prov\_player\_\_game\_\_stats\_deflongmidrangereboundpct>0.1 \\ $A_1$.prov\_player\_\_game\_\_stats\_usage>28.025 \\ $A_1$.prov\_season\_season\_\_type=regular season \\ $A_3$.salary>14500000.0} & 1.0 & 0.7 & 0.82 & 2009-10\\
\hline
\makecell{$A_1$: PT \\ $A_2$: team \\ $A_3$: player\_salary} & \makecell{($A_2$.team\_id)=($A_1$.prov\_game\_away\_\_id)\\($A_1$.prov\_player\_player\_\_id)=($A_3$.player\_id)\\($A_1$.prov\_season\_season\_\_id)=($A_3$.season\_id)} & \makecell{$A_1$.prov\_player\_\_game\_\_stats\_assisted\_\_two\_\_spct<0.5 \\ $A_1$.prov\_player\_\_game\_\_stats\_def\_\_three\_\_ptreboundpct<0.25 \\ $A_1$.prov\_season\_season\_\_type=regular season \\ $A_3$.salary<15779912.0} & 1.0 & 0.65 & 0.78 & 2010-11\\
\hline
\makecell{$A_1$: PT \\ $A_2$: team \\ $A_3$: player\_salary} & \makecell{($A_2$.team\_id)=($A_1$.prov\_game\_winner\_\_id)\\($A_1$.prov\_player\_player\_\_id)=($A_3$.player\_id)\\($A_1$.prov\_season\_season\_\_id)=($A_3$.season\_id)} & \makecell{$A_1$.prov\_player\_\_game\_\_stats\_assisted\_\_two\_\_spct<0.5 \\ $A_1$.prov\_player\_\_game\_\_stats\_def\_\_three\_\_ptreboundpct<0.25 \\ $A_1$.prov\_season\_season\_\_type=regular season \\ $A_3$.salary<15779912.0} & 1.0 & 0.65 & 0.78 & 2010-11\\
\hline
\makecell{$A_1$: PT \\ $A_2$: team \\ $A_3$: player\_salary} & \makecell{($A_2$.team\_id)=($A_1$.prov\_game\_home\_\_id)\\($A_1$.prov\_player\_player\_\_id)=($A_3$.player\_id)\\($A_1$.prov\_season\_season\_\_id)=($A_3$.season\_id)} & \makecell{$A_1$.prov\_player\_\_game\_\_stats\_assisted\_\_two\_\_spct<0.5 \\ $A_1$.prov\_player\_\_game\_\_stats\_usage<36.36 \\ $A_1$.prov\_season\_season\_\_type=regular season \\ $A_3$.salary<15779912.0} & 1.0 & 0.63 & 0.78 & 2010-11\\
\hline
\makecell{$A_1$: PT \\ $A_2$: team \\ $A_3$: player\_salary} & \makecell{($A_2$.team\_id)=($A_1$.prov\_game\_home\_\_id)\\($A_1$.prov\_player\_player\_\_id)=($A_3$.player\_id)\\($A_1$.prov\_season\_season\_\_id)=($A_3$.season\_id)} & \makecell{$A_1$.prov\_player\_\_game\_\_stats\_assisted\_\_two\_\_spct>0.19 \\ $A_1$.prov\_player\_\_game\_\_stats\_deflongmidrangereboundpct<0.29 \\ $A_1$.prov\_season\_season\_\_type=regular season \\ $A_3$.salary>14500000.0} & 1.0 & 0.63 & 0.77 & 2009-10\\
\hline
\makecell{$A_1$: PT \\ $A_2$: team \\ $A_3$: play\_for \\ $A_4$: team} & \makecell{($A_2$.team\_id)=($A_1$.prov\_game\_winner\_\_id)\\($A_1$.prov\_player\_player\_\_id)=($A_3$.player\_id)\\($A_4$.team\_id)=($A_3$.team\_id)} & \makecell{$A_1$.prov\_player\_\_game\_\_stats\_deflongmidrangereboundpct<0.27 \\ $A_1$.prov\_season\_season\_\_type=regular season \\ $A_3$.date\_end=2014-04-12 \\ $A_2$.team=CLE \\ $A_4$.team=MIA} & 1.0 & 0.57 & 0.72 & 2009-10\\
\hline
\makecell{$A_1$: PT \\ $A_2$: team \\ $A_3$: team\_game\_stats \\ $A_4$: team} & \makecell{($A_2$.team\_id)=($A_1$.prov\_game\_home\_\_id)\\($A_1$.prov\_game\_game\_\_date)=($A_3$.game\_date)\\($A_1$.prov\_game\_home\_\_id)=($A_3$.home\_id)\\($A_4$.team\_id)=($A_3$.team\_id)} & \makecell{$A_1$.prov\_player\_\_game\_\_stats\_assisted\_\_two\_\_spct<0.38 \\ $A_1$.prov\_player\_\_game\_\_stats\_fg\_\_three\_\_apct<0.2475 \\ $A_1$.prov\_season\_season\_\_type=regular season \\ $A_4$.team=MIA} & 0.98 & 0.57 & 0.72 & 2010-11\\
\hline
\makecell{$A_1$: PT \\ $A_2$: team \\ $A_3$: team\_game\_stats \\ $A_4$: team} & \makecell{($A_2$.team\_id)=($A_1$.prov\_game\_winner\_\_id)\\($A_1$.prov\_game\_game\_\_date)=($A_3$.game\_date)\\($A_1$.prov\_game\_home\_\_id)=($A_3$.home\_id)\\($A_4$.team\_id)=($A_3$.team\_id)} & \makecell{$A_1$.prov\_player\_\_game\_\_stats\_assisted\_\_two\_\_spct<0.38 \\ $A_1$.prov\_player\_\_game\_\_stats\_fg\_\_three\_\_apct<0.2475 \\ $A_1$.prov\_season\_season\_\_type=regular season \\ $A_4$.team=MIA} & 0.98 & 0.57 & 0.72 & 2010-11\\
\hline

\end{tabular}
}
\label{tab:tops-nba3}
\caption{$\query_{nba3}$ Top-20 patterns}
\end{figure*}

\begin{figure*}[!h]
\centering
{\scriptsize
\centering
\begin{tabular}{|c|P{25em}|P{20em}|c|c|c|P{5em}|}
\hline
\cellcolor{grey}\textbf{Nodes} & \cellcolor{grey}\textbf{Edges} & \cellcolor{grey}\textbf{Pattern Desc} & \cellcolor{grey}\textbf{Precision} & \cellcolor{grey}\textbf{Recall} & \cellcolor{grey}\textbf{\abbrF} & \cellcolor{grey}\textbf{Primary Tuple} \\

\makecell{$A_1$: PT \\ $A_2$: player\_salary \\ $A_3$: player} & \makecell{($A_1$.prov\_season\_season\_\_id)=($A_2$.season\_id)\\($A_3$.player\_id)=($A_2$.player\_id)} & \makecell{$A_1$.prov\_season\_season\_\_type=regular season \\ $A_3$.player\_name=Monta Ellis \\ $A_2$.salary<11000000.0} & 1.0 & 1.0 & 1.0 & 2016-17\\
\hline
\makecell{$A_1$: PT \\ $A_2$: player\_salary \\ $A_3$: player} & \makecell{($A_1$.prov\_season\_season\_\_id)=($A_2$.season\_id)\\($A_3$.player\_id)=($A_2$.player\_id)} & \makecell{$A_1$.prov\_season\_season\_\_type=regular season \\ $A_3$.player\_name=Terrence Jones \\ $A_2$.salary<1485000.0} & 1.0 & 1.0 & 1.0 & 2016-17\\
\hline
\makecell{$A_1$: PT \\ $A_2$: player\_salary \\ $A_3$: player} & \makecell{($A_1$.prov\_season\_season\_\_id)=($A_2$.season\_id)\\($A_3$.player\_id)=($A_2$.player\_id)} & \makecell{$A_1$.prov\_season\_season\_\_type=regular season \\ $A_3$.player\_name=Pau Gasol \\ $A_2$.salary<19285850.0} & 1.0 & 1.0 & 1.0 & 2016-17\\
\hline
\makecell{$A_1$: PT \\ $A_2$: team \\ $A_3$: player\_salary \\ $A_4$: player} & \makecell{($A_2$.team\_id)=($A_1$.prov\_game\_home\_\_id)\\($A_1$.prov\_season\_season\_\_id)=($A_3$.season\_id)\\($A_4$.player\_id)=($A_3$.player\_id)} & \makecell{$A_1$.prov\_season\_season\_\_type=regular season \\ $A_4$.player\_name=Robert Sacre \\ $A_3$.salary>788872.0} & 1.0 & 1.0 & 1.0 & 2016-17\\
\hline
\makecell{$A_1$: PT \\ $A_2$: player\_salary \\ $A_3$: player} & \makecell{($A_1$.prov\_season\_season\_\_id)=($A_2$.season\_id)\\($A_3$.player\_id)=($A_2$.player\_id)} & \makecell{$A_1$.prov\_season\_season\_\_type=regular season \\ $A_3$.player\_name=Raymond Felton \\ $A_2$.salary<3480453.0} & 1.0 & 1.0 & 1.0 & 2016-17\\
\hline
\makecell{$A_1$: PT \\ $A_2$: player\_salary \\ $A_3$: player} & \makecell{($A_1$.prov\_season\_season\_\_id)=($A_2$.season\_id)\\($A_3$.player\_id)=($A_2$.player\_id)} & \makecell{$A_1$.prov\_season\_season\_\_type=regular season \\ $A_3$.player\_name=Evan Turner \\ $A_2$.salary>5293080.0} & 1.0 & 1.0 & 1.0 & 2016-17\\
\hline
\makecell{$A_1$: PT \\ $A_2$: team \\ $A_3$: player\_game\_stats \\ $A_4$: player} & \makecell{($A_2$.team\_id)=($A_1$.prov\_game\_away\_\_id)\\($A_1$.prov\_game\_game\_\_date)=($A_3$.game\_date)\\($A_1$.prov\_game\_home\_\_id)=($A_3$.home\_id)\\($A_4$.player\_id)=($A_3$.player\_id)} & \makecell{$A_1$.prov\_season\_season\_\_type=regular season \\ $A_4$.player\_name=Andre Iguodala} & 0.98 & 0.96 & 0.97 & 2016-17\\
\hline
\makecell{$A_1$: PT \\ $A_2$: team \\ $A_3$: player\_game\_stats \\ $A_4$: player} & \makecell{($A_2$.team\_id)=($A_1$.prov\_game\_home\_\_id)\\($A_1$.prov\_game\_game\_\_date)=($A_3$.game\_date)\\($A_1$.prov\_game\_home\_\_id)=($A_3$.home\_id)\\($A_4$.player\_id)=($A_3$.player\_id)} & \makecell{$A_1$.prov\_season\_season\_\_type=regular season \\ $A_4$.player\_name=Shaun Livingston} & 0.97 & 0.94 & 0.95 & 2016-17\\
\hline
\makecell{$A_1$: PT \\ $A_2$: team\_game\_stats} & \makecell{($A_1$.prov\_game\_game\_\_date)=($A_2$.game\_date)\\($A_1$.prov\_game\_home\_\_id)=($A_2$.home\_id)\\($A_1$.prov\_team\_team\_\_id)=($A_2$.team\_id)} & \makecell{$A_1$.prov\_season\_season\_\_type=regular season \\ $A_2$.fg\_three\_apct<0.31 \\ $A_2$.points<121.0} & 0.95 & 0.89 & 0.92 & 2012-13\\
\hline
\makecell{$A_1$: PT \\ $A_2$: team \\ $A_3$: team\_game\_stats} & \makecell{($A_2$.team\_id)=($A_1$.prov\_game\_away\_\_id)\\($A_1$.prov\_game\_game\_\_date)=($A_3$.game\_date)\\($A_1$.prov\_game\_home\_\_id)=($A_3$.home\_id)\\($A_1$.prov\_team\_team\_\_id)=($A_3$.team\_id)} & \makecell{$A_1$.prov\_season\_season\_\_type=regular season \\ $A_3$.fg\_three\_apct<0.31 \\ $A_3$.points<121.0} & 0.95 & 0.89 & 0.92 & 2012-13\\
\hline
\makecell{$A_1$: PT \\ $A_2$: team \\ $A_3$: team\_game\_stats} & \makecell{($A_2$.team\_id)=($A_1$.prov\_game\_home\_\_id)\\($A_1$.prov\_game\_game\_\_date)=($A_3$.game\_date)\\($A_1$.prov\_game\_home\_\_id)=($A_3$.home\_id)\\($A_1$.prov\_team\_team\_\_id)=($A_3$.team\_id)} & \makecell{$A_1$.prov\_season\_season\_\_type=regular season \\ $A_3$.fg\_three\_apct<0.31 \\ $A_3$.points<121.0} & 0.95 & 0.89 & 0.92 & 2012-13\\
\hline
\makecell{$A_1$: PT \\ $A_2$: team \\ $A_3$: team\_game\_stats} & \makecell{($A_2$.team\_id)=($A_1$.prov\_game\_home\_\_id)\\($A_1$.prov\_game\_game\_\_date)=($A_3$.game\_date)\\($A_1$.prov\_game\_home\_\_id)=($A_3$.home\_id)\\($A_1$.prov\_team\_team\_\_id)=($A_3$.team\_id)} & \makecell{$A_1$.prov\_season\_season\_\_type=regular season \\ $A_3$.assistpoints>55.0 \\ $A_3$.fg\_three\_apct>0.25 \\ $A_3$.points>105.0} & 0.92 & 0.9 & 0.91 & 2016-17\\
\hline
\makecell{$A_1$: PT \\ $A_2$: team \\ $A_3$: player\_game\_stats \\ $A_4$: player} & \makecell{($A_2$.team\_id)=($A_1$.prov\_game\_away\_\_id)\\($A_1$.prov\_game\_game\_\_date)=($A_3$.game\_date)\\($A_1$.prov\_game\_home\_\_id)=($A_3$.home\_id)\\($A_4$.player\_id)=($A_3$.player\_id)} & \makecell{$A_1$.prov\_season\_season\_\_type=regular season \\ $A_4$.player\_name=Draymond Green \\ $A_3$.minutes>13.735 \\ $A_3$.shotqualityavg>0.475} & 0.87 & 0.9 & 0.88 & 2016-17\\
\hline
\makecell{$A_1$: PT \\ $A_2$: team \\ $A_3$: player\_game\_stats \\ $A_4$: player} & \makecell{($A_2$.team\_id)=($A_1$.prov\_game\_home\_\_id)\\($A_1$.prov\_game\_game\_\_date)=($A_3$.game\_date)\\($A_1$.prov\_game\_home\_\_id)=($A_3$.home\_id)\\($A_4$.player\_id)=($A_3$.player\_id)} & \makecell{$A_1$.prov\_season\_season\_\_type=regular season \\ $A_4$.player\_name=Draymond Green \\ $A_3$.minutes>13.735 \\ $A_3$.shotqualityavg>0.475} & 0.87 & 0.9 & 0.88 & 2016-17\\
\hline
\makecell{$A_1$: PT \\ $A_2$: team \\ $A_3$: player\_game\_stats \\ $A_4$: player} & \makecell{($A_2$.team\_id)=($A_1$.prov\_game\_away\_\_id)\\($A_1$.prov\_game\_game\_\_date)=($A_3$.game\_date)\\($A_1$.prov\_game\_home\_\_id)=($A_3$.home\_id)\\($A_4$.player\_id)=($A_3$.player\_id)} & \makecell{$A_1$.prov\_game\_home\_\_points<115.5 \\ $A_1$.prov\_season\_season\_\_type=regular season \\ $A_4$.player\_name=Draymond Green \\ $A_3$.minutes<27.72} & 0.89 & 0.87 & 0.88 & 2012-13\\
\hline
\makecell{$A_1$: PT \\ $A_2$: team \\ $A_3$: team\_game\_stats} & \makecell{($A_2$.team\_id)=($A_1$.prov\_game\_away\_\_id)\\($A_1$.prov\_game\_game\_\_date)=($A_3$.game\_date)\\($A_1$.prov\_game\_home\_\_id)=($A_3$.home\_id)\\($A_1$.prov\_team\_team\_\_id)=($A_3$.team\_id)} & \makecell{$A_1$.prov\_season\_season\_\_type=regular season \\ $A_3$.assistpoints>55.0 \\ $A_3$.fg\_three\_m>9.0 \\ $A_3$.points>105.0} & 0.88 & 0.84 & 0.85 & 2016-17\\
\hline
\makecell{$A_1$: PT \\ $A_2$: team\_game\_stats} & \makecell{($A_1$.prov\_game\_game\_\_date)=($A_2$.game\_date)\\($A_1$.prov\_game\_home\_\_id)=($A_2$.home\_id)\\($A_1$.prov\_team\_team\_\_id)=($A_2$.team\_id)} & \makecell{$A_1$.prov\_season\_season\_\_type=regular season \\ $A_2$.assistpoints>55.0 \\ $A_2$.fg\_three\_m>9.0 \\ $A_2$.points>105.0} & 0.88 & 0.84 & 0.85 & 2016-17\\
\hline
\makecell{$A_1$: PT \\ $A_2$: team \\ $A_3$: team\_game\_stats} & \makecell{($A_2$.team\_id)=($A_1$.prov\_game\_away\_\_id)\\($A_1$.prov\_game\_game\_\_date)=($A_3$.game\_date)\\($A_1$.prov\_game\_home\_\_id)=($A_3$.home\_id)\\($A_1$.prov\_team\_team\_\_id)=($A_3$.team\_id)} & \makecell{$A_1$.prov\_season\_season\_\_type=regular season \\ $A_3$.assistpoints<67.0 \\ $A_3$.fg\_three\_m<11.0 \\ $A_3$.points<121.0} & 0.92 & 0.77 & 0.84 & 2012-13\\
\hline
\makecell{$A_1$: PT \\ $A_2$: team \\ $A_3$: team\_game\_stats} & \makecell{($A_2$.team\_id)=($A_1$.prov\_game\_home\_\_id)\\($A_1$.prov\_game\_game\_\_date)=($A_3$.game\_date)\\($A_1$.prov\_game\_home\_\_id)=($A_3$.home\_id)\\($A_1$.prov\_team\_team\_\_id)=($A_3$.team\_id)} & \makecell{$A_1$.prov\_season\_season\_\_type=regular season \\ $A_3$.assistpoints<67.0 \\ $A_3$.fg\_three\_m<11.0 \\ $A_3$.points<121.0} & 0.92 & 0.77 & 0.84 & 2012-13\\
\hline
\makecell{$A_1$: PT \\ $A_2$: team \\ $A_3$: player\_game\_stats \\ $A_4$: player} & \makecell{($A_2$.team\_id)=($A_1$.prov\_game\_home\_\_id)\\($A_1$.prov\_game\_game\_\_date)=($A_3$.game\_date)\\($A_1$.prov\_game\_home\_\_id)=($A_3$.home\_id)\\($A_4$.player\_id)=($A_3$.player\_id)} & \makecell{$A_1$.prov\_game\_home\_\_points>100.25 \\ $A_1$.prov\_season\_season\_\_type=regular season \\ $A_4$.player\_name=Stephen Curry \\ $A_3$.minutes<37.4275} & 0.88 & 0.79 & 0.83 & 2016-17\\
\hline

\end{tabular}
}
\label{tab:tops-nba4}
\caption{$\query_{nba4}$ Top-20 patterns}
\end{figure*}

\begin{figure*}[!h]
\centering
{\scriptsize
\centering
\begin{tabular}{|c|P{25em}|P{20em}|c|c|c|P{5em}|}
\hline
\cellcolor{grey}\textbf{Nodes} & \cellcolor{grey}\textbf{Edges} & \cellcolor{grey}\textbf{Pattern Desc} & \cellcolor{grey}\textbf{Precision} & \cellcolor{grey}\textbf{Recall} & \cellcolor{grey}\textbf{\abbrF} & \cellcolor{grey}\textbf{Primary Tuple} \\

\makecell{$A_1$: PT \\ $A_2$: team \\ $A_3$: player\_salary} & \makecell{($A_2$.team\_id)=($A_1$.prov\_game\_winner\_\_id)\\($A_1$.prov\_player\_player\_\_id)=($A_3$.player\_id)\\($A_1$.prov\_season\_season\_\_id)=($A_3$.season\_id)} & \makecell{$A_1$.prov\_season\_season\_\_type=regular season \\ $A_3$.salary>1112880.0} & 1.0 & 1.0 & 1.0 & 2014-15\\
\hline
\makecell{$A_1$: PT \\ $A_2$: team \\ $A_3$: player\_salary} & \makecell{($A_2$.team\_id)=($A_1$.prov\_game\_away\_\_id)\\($A_1$.prov\_player\_player\_\_id)=($A_3$.player\_id)\\($A_1$.prov\_season\_season\_\_id)=($A_3$.season\_id)} & \makecell{$A_1$.prov\_season\_season\_\_type=regular season \\ $A_3$.salary>1112880.0} & 1.0 & 1.0 & 1.0 & 2014-15\\
\hline
\makecell{$A_1$: PT \\ $A_2$: team \\ $A_3$: player\_salary} & \makecell{($A_2$.team\_id)=($A_1$.prov\_game\_home\_\_id)\\($A_1$.prov\_player\_player\_\_id)=($A_3$.player\_id)\\($A_1$.prov\_season\_season\_\_id)=($A_3$.season\_id)} & \makecell{$A_1$.prov\_season\_season\_\_type=regular season \\ $A_3$.salary>1112880.0} & 1.0 & 1.0 & 1.0 & 2014-15\\
\hline
\makecell{$A_1$: PT \\ $A_2$: team \\ $A_3$: player\_salary} & \makecell{($A_2$.team\_id)=($A_1$.prov\_game\_winner\_\_id)\\($A_1$.prov\_player\_player\_\_id)=($A_3$.player\_id)\\($A_1$.prov\_season\_season\_\_id)=($A_3$.season\_id)} & \makecell{$A_1$.prov\_player\_\_game\_\_stats\_minutes<42.6675 \\ $A_1$.prov\_season\_season\_\_type=regular season \\ $A_3$.salary>1112880.0} & 1.0 & 0.83 & 0.91 & 2014-15\\
\hline
\makecell{$A_1$: PT \\ $A_2$: team \\ $A_3$: player\_salary} & \makecell{($A_2$.team\_id)=($A_1$.prov\_game\_winner\_\_id)\\($A_1$.prov\_player\_player\_\_id)=($A_3$.player\_id)\\($A_1$.prov\_season\_season\_\_id)=($A_3$.season\_id)} & \makecell{$A_1$.prov\_game\_away\_\_points>86.75 \\ $A_1$.prov\_player\_\_game\_\_stats\_efgpct>0.38 \\ $A_1$.prov\_season\_season\_\_type=regular season \\ $A_3$.salary>1112880.0} & 1.0 & 0.72 & 0.84 & 2014-15\\
\hline
\makecell{$A_1$: PT \\ $A_2$: team \\ $A_3$: player\_salary} & \makecell{($A_2$.team\_id)=($A_1$.prov\_game\_home\_\_id)\\($A_1$.prov\_player\_player\_\_id)=($A_3$.player\_id)\\($A_1$.prov\_season\_season\_\_id)=($A_3$.season\_id)} & \makecell{$A_1$.prov\_game\_away\_\_points>86.75 \\ $A_1$.prov\_player\_\_game\_\_stats\_efgpct>0.38 \\ $A_1$.prov\_season\_season\_\_type=regular season \\ $A_3$.salary>1112880.0} & 1.0 & 0.72 & 0.84 & 2014-15\\
\hline
\makecell{$A_1$: PT \\ $A_2$: team \\ $A_3$: player\_salary} & \makecell{($A_2$.team\_id)=($A_1$.prov\_game\_away\_\_id)\\($A_1$.prov\_player\_player\_\_id)=($A_3$.player\_id)\\($A_1$.prov\_season\_season\_\_id)=($A_3$.season\_id)} & \makecell{$A_1$.prov\_game\_away\_\_points>86.75 \\ $A_1$.prov\_player\_\_game\_\_stats\_minutes>36.3425 \\ $A_1$.prov\_season\_season\_\_type=regular season \\ $A_3$.salary>1112880.0} & 1.0 & 0.71 & 0.83 & 2014-15\\
\hline
\makecell{$A_1$: PT \\ $A_2$: team \\ $A_3$: player\_salary} & \makecell{($A_2$.team\_id)=($A_1$.prov\_game\_away\_\_id)\\($A_1$.prov\_player\_player\_\_id)=($A_3$.player\_id)\\($A_1$.prov\_season\_season\_\_id)=($A_3$.season\_id)} & \makecell{$A_1$.prov\_game\_away\_\_points<102.0 \\ $A_1$.prov\_player\_\_game\_\_stats\_efgpct<0.58 \\ $A_1$.prov\_season\_season\_\_type=regular season \\ $A_3$.salary<2008748.0} & 1.0 & 0.66 & 0.79 & 2013-14\\
\hline
\makecell{$A_1$: PT \\ $A_2$: team \\ $A_3$: player\_salary} & \makecell{($A_2$.team\_id)=($A_1$.prov\_game\_home\_\_id)\\($A_1$.prov\_player\_player\_\_id)=($A_3$.player\_id)\\($A_1$.prov\_season\_season\_\_id)=($A_3$.season\_id)} & \makecell{$A_1$.prov\_game\_away\_\_points<102.0 \\ $A_1$.prov\_player\_\_game\_\_stats\_efgpct<0.58 \\ $A_1$.prov\_season\_season\_\_type=regular season \\ $A_3$.salary<2008748.0} & 1.0 & 0.66 & 0.79 & 2013-14\\
\hline
\makecell{$A_1$: PT \\ $A_2$: team \\ $A_3$: team\_game\_stats \\ $A_4$: team} & \makecell{($A_2$.team\_id)=($A_1$.prov\_game\_winner\_\_id)\\($A_1$.prov\_game\_game\_\_date)=($A_3$.game\_date)\\($A_1$.prov\_game\_home\_\_id)=($A_3$.home\_id)\\($A_4$.team\_id)=($A_3$.team\_id)} & \makecell{$A_1$.prov\_player\_\_game\_\_stats\_usage<22.92 \\ $A_1$.prov\_season\_season\_\_type=regular season \\ $A_4$.team=CHI \\ $A_3$.assisted\_two\_spct>0.5} & 0.72 & 0.84 & 0.77 & 2013-14\\
\hline
\makecell{$A_1$: PT \\ $A_2$: team \\ $A_3$: team\_game\_stats \\ $A_4$: team} & \makecell{($A_2$.team\_id)=($A_1$.prov\_game\_home\_\_id)\\($A_1$.prov\_game\_game\_\_date)=($A_3$.game\_date)\\($A_1$.prov\_game\_home\_\_id)=($A_3$.home\_id)\\($A_4$.team\_id)=($A_3$.team\_id)} & \makecell{$A_1$.prov\_player\_\_game\_\_stats\_usage<22.92 \\ $A_1$.prov\_season\_season\_\_type=regular season \\ $A_4$.team=CHI \\ $A_3$.assisted\_two\_spct>0.5} & 0.72 & 0.84 & 0.77 & 2013-14\\
\hline
\makecell{$A_1$: PT \\ $A_2$: team \\ $A_3$: play\_for \\ $A_4$: team} & \makecell{($A_2$.team\_id)=($A_1$.prov\_game\_home\_\_id)\\($A_1$.prov\_player\_player\_\_id)=($A_3$.player\_id)\\($A_4$.team\_id)=($A_3$.team\_id)} & \makecell{$A_1$.prov\_player\_\_game\_\_stats\_usage<22.92 \\ $A_1$.prov\_season\_season\_\_type=regular season \\ $A_3$.date\_end=2019-04-09 \\ $A_4$.team=PHI} & 0.63 & 0.93 & 0.75 & 2013-14\\
\hline
\makecell{$A_1$: PT \\ $A_2$: team \\ $A_3$: play\_for \\ $A_4$: team} & \makecell{($A_2$.team\_id)=($A_1$.prov\_game\_winner\_\_id)\\($A_1$.prov\_player\_player\_\_id)=($A_3$.player\_id)\\($A_4$.team\_id)=($A_3$.team\_id)} & \makecell{$A_1$.prov\_player\_\_game\_\_stats\_usage<22.92 \\ $A_1$.prov\_season\_season\_\_type=regular season \\ $A_3$.date\_end=2019-04-09 \\ $A_4$.team=PHI} & 0.63 & 0.93 & 0.75 & 2013-14\\
\hline
\makecell{$A_1$: PT \\ $A_2$: team \\ $A_3$: team\_game\_stats \\ $A_4$: team} & \makecell{($A_2$.team\_id)=($A_1$.prov\_game\_away\_\_id)\\($A_1$.prov\_game\_game\_\_date)=($A_3$.game\_date)\\($A_1$.prov\_game\_home\_\_id)=($A_3$.home\_id)\\($A_4$.team\_id)=($A_3$.team\_id)} & \makecell{$A_1$.prov\_player\_\_game\_\_stats\_usage<22.92 \\ $A_1$.prov\_season\_season\_\_type=regular season \\ $A_4$.team=CHI} & 0.63 & 0.93 & 0.75 & 2013-14\\
\hline
\makecell{} & \makecell{A\_1} & \makecell{$A_1$.prov\_player\_\_game\_\_stats\_usage<22.92 \\ $A_1$.prov\_season\_season\_\_type=regular season} & 0.63 & 0.93 & 0.75 & 2013-14\\
\hline
\makecell{$A_1$: PT \\ $A_2$: team \\ $A_3$: play\_for \\ $A_4$: team} & \makecell{($A_2$.team\_id)=($A_1$.prov\_game\_away\_\_id)\\($A_1$.prov\_player\_player\_\_id)=($A_3$.player\_id)\\($A_4$.team\_id)=($A_3$.team\_id)} & \makecell{$A_1$.prov\_player\_\_game\_\_stats\_usage<22.92 \\ $A_1$.prov\_season\_season\_\_type=regular season \\ $A_3$.date\_end=2017-04-12 \\ $A_4$.team=CHI} & 0.63 & 0.93 & 0.75 & 2013-14\\
\hline
\makecell{$A_1$: PT \\ $A_2$: team \\ $A_3$: player\_salary} & \makecell{($A_2$.team\_id)=($A_1$.prov\_game\_home\_\_id)\\($A_1$.prov\_player\_player\_\_id)=($A_3$.player\_id)\\($A_1$.prov\_season\_season\_\_id)=($A_3$.season\_id)} & \makecell{$A_1$.prov\_player\_\_game\_\_stats\_efgpct<0.58 \\ $A_1$.prov\_player\_\_game\_\_stats\_minutes<42.6675 \\ $A_1$.prov\_season\_season\_\_type=regular season \\ $A_3$.salary>1112880.0} & 1.0 & 0.57 & 0.73 & 2014-15\\
\hline
\makecell{$A_1$: PT \\ $A_2$: team \\ $A_3$: player\_salary} & \makecell{($A_2$.team\_id)=($A_1$.prov\_game\_winner\_\_id)\\($A_1$.prov\_player\_player\_\_id)=($A_3$.player\_id)\\($A_1$.prov\_season\_season\_\_id)=($A_3$.season\_id)} & \makecell{$A_1$.prov\_player\_\_game\_\_stats\_efgpct<0.58 \\ $A_1$.prov\_player\_\_game\_\_stats\_minutes>36.3425 \\ $A_1$.prov\_season\_season\_\_type=regular season \\ $A_3$.salary<2008748.0} & 1.0 & 0.57 & 0.72 & 2013-14\\
\hline
\makecell{$A_1$: PT \\ $A_2$: team \\ $A_3$: team\_game\_stats \\ $A_4$: team} & \makecell{($A_2$.team\_id)=($A_1$.prov\_game\_away\_\_id)\\($A_1$.prov\_game\_game\_\_date)=($A_3$.game\_date)\\($A_1$.prov\_game\_home\_\_id)=($A_3$.home\_id)\\($A_4$.team\_id)=($A_3$.team\_id)} & \makecell{$A_1$.prov\_player\_\_game\_\_stats\_usage>18.85 \\ $A_1$.prov\_season\_season\_\_type=regular season \\ $A_4$.team=CHI \\ $A_3$.offposs>88.0} & 0.75 & 0.66 & 0.7 & 2014-15\\
\hline
\makecell{$A_1$: PT \\ $A_2$: team \\ $A_3$: team\_game\_stats \\ $A_4$: team} & \makecell{($A_2$.team\_id)=($A_1$.prov\_game\_winner\_\_id)\\($A_1$.prov\_game\_game\_\_date)=($A_3$.game\_date)\\($A_1$.prov\_game\_home\_\_id)=($A_3$.home\_id)\\($A_4$.team\_id)=($A_3$.team\_id)} & \makecell{$A_1$.prov\_player\_\_game\_\_stats\_usage>18.85 \\ $A_1$.prov\_season\_season\_\_type=regular season \\ $A_4$.team=CHI \\ $A_3$.assisted\_two\_spct<0.58} & 0.88 & 0.58 & 0.7 & 2014-15\\
\hline

\end{tabular}
}
\label{tab:tops-nba5}
\caption{$\query_{nba5}$ Top-20 patterns}
\end{figure*}

\subsection{Top-20 explanations for Case Study queries}
We present the detailed top-20 explanations for each case study queries mentioned in \Cref{sec:case-study}.

\textbf{Nodes} are identifiers for relations presented in the join graph. \textbf{Edges} are the edge details with their join conditions. \textbf{Pattern Desc} are pattern descriptions. Each line represents a predicate of the pattern. The text descriptions that contain ``\texttt{prov\_}" prefix suggest that these predicates are generated from PT node followed by the name of the relation that was part of the user query, then followed by the attribute from which the constant value was generated. For example \texttt{$A_1$.prov\_player\_\_game\_\_stats\_minutes<36} should be interpreted as \textit{"minutes coming from player\_game\_stats table in PT smaller than 36}. \textbf{Precision, Recall and \abbrF} are calculated based on \textbf{Primary Tuple} which was identified using one of two user question constants.

\label{sec:top20}
\begin{figure*}[!h]
\centering
{\scriptsize
\centering
\begin{tabular}{|c|P{22em}|P{25em}|c|c|c|P{5em}|}
\hline
\cellcolor{grey}\textbf{Nodes} & \cellcolor{grey}\textbf{Edges} & \cellcolor{grey}\textbf{Pattern Desc} & \cellcolor{grey}\textbf{Precision} & \cellcolor{grey}\textbf{Recall} & \cellcolor{grey}\textbf{\abbrF} & \cellcolor{grey}\textbf{Primary Tuple} \\
\hline
\makecell{$A_1$: PT \\ $A_2$: patients} & \makecell{($A_2$.subject\_id)=($A_1$.prov\_diagnoses\_subject\_\_id)} & \makecell{$A_2$.expire\_flag=1} & 0.65 & 0.72 & 0.68 & 2\\
\hline
\makecell{$A_1$: PT \\ $A_2$: patients \\ $A_3$: patients\_admit\_info \\ $A_4$: patients} & \makecell{($A_2$.subject\_id)=($A_1$.prov\_diagnoses\_subject\_\_id)\\($A_1$.prov\_admissions\_hadm\_\_id)=($A_3$.hadm\_id)\\($A_4$.subject\_id)=($A_3$.subject\_id)} & \makecell{$A_2$.expire\_flag=1 \\ $A_4$.expire\_flag=1} & 0.65 & 0.72 & 0.68 & 2\\
\hline
\makecell{$A_1$: PT \\ $A_2$: patients\_admit\_info \\ $A_3$: patients} & \makecell{($A_1$.prov\_admissions\_hadm\_\_id)=($A_2$.hadm\_id)\\($A_3$.subject\_id)=($A_2$.subject\_id)} & \makecell{$A_3$.expire\_flag=1} & 0.65 & 0.72 & 0.68 & 2\\
\hline
\makecell{$A_1$: PT \\ $A_2$: procedures \\ $A_3$: patients} & \makecell{($A_1$.prov\_admissions\_hadm\_\_id)=($A_2$.hadm\_id)\\($A_3$.subject\_id)=($A_2$.subject\_id)} & \makecell{$A_3$.expire\_flag=1} & 0.66 & 0.71 & 0.68 & 2\\
\hline
\makecell{$A_1$: PT \\ $A_2$: patients \\ $A_3$: icustays \\ $A_4$: patients} & \makecell{($A_2$.subject\_id)=($A_1$.prov\_diagnoses\_subject\_\_id)\\($A_1$.prov\_admissions\_hadm\_\_id)=($A_3$.hadm\_id)\\($A_4$.subject\_id)=($A_3$.subject\_id)} & \makecell{$A_2$.expire\_flag=1 \\ $A_4$.expire\_flag=1} & 0.65 & 0.72 & 0.68 & 2\\
\hline
\makecell{$A_1$: PT \\ $A_2$: icustays \\ $A_3$: patients} & \makecell{($A_1$.prov\_admissions\_hadm\_\_id)=($A_2$.hadm\_id)\\($A_3$.subject\_id)=($A_2$.subject\_id)} & \makecell{$A_3$.expire\_flag=1} & 0.65 & 0.72 & 0.68 & 2\\
\hline
\makecell{$A_1$: PT \\ $A_2$: patients \\ $A_3$: procedures \\ $A_4$: admissions} & \makecell{($A_2$.subject\_id)=($A_1$.prov\_diagnoses\_subject\_\_id)\\($A_2$.subject\_id)=($A_3$.subject\_id)\\($A_4$.hadm\_id)=($A_3$.hadm\_id)} & \makecell{$A_4$.hospital\_stay\_length<23.0 \\ $A_2$.expire\_flag=1} & 0.65 & 0.66 & 0.65 & 2\\
\hline
\makecell{$A_1$: PT \\ $A_2$: procedures \\ $A_3$: patients} & \makecell{($A_1$.prov\_admissions\_hadm\_\_id)=($A_2$.hadm\_id)\\($A_3$.subject\_id)=($A_2$.subject\_id)} & \makecell{$A_1$.prov\_admissions\_hospital\_\_stay\_\_length<24.0 \\ $A_3$.expire\_flag=1} & 0.66 & 0.63 & 0.65 & 2\\
\hline
\makecell{$A_1$: PT \\ $A_2$: patients \\ $A_3$: icustays \\ $A_4$: admissions} & \makecell{($A_2$.subject\_id)=($A_1$.prov\_diagnoses\_subject\_\_id)\\($A_2$.subject\_id)=($A_3$.subject\_id)\\($A_4$.hadm\_id)=($A_3$.hadm\_id)} & \makecell{$A_4$.hospital\_stay\_length<16.0 \\ $A_2$.expire\_flag=1} & 0.65 & 0.61 & 0.63 & 2\\
\hline
\makecell{$A_1$: PT \\ $A_2$: patients \\ $A_3$: icustays \\ $A_4$: admissions} & \makecell{($A_2$.subject\_id)=($A_1$.prov\_diagnoses\_subject\_\_id)\\($A_2$.subject\_id)=($A_3$.subject\_id)\\($A_4$.hadm\_id)=($A_3$.hadm\_id)} & \makecell{$A_1$.prov\_admissions\_admission\_\_type=EMERGENCY \\ $A_4$.admission\_type=EMERGENCY} & 0.5 & 0.81 & 0.62 & 13\\
\hline
\makecell{} & \makecell{A\_1} & \makecell{$A_1$.prov\_admissions\_admission\_\_type=EMERGENCY} & 0.5 & 0.81 & 0.62 & 13\\
\hline
\makecell{$A_1$: PT \\ $A_2$: patients \\ $A_3$: patients\_admit\_info \\ $A_4$: admissions} & \makecell{($A_2$.subject\_id)=($A_1$.prov\_diagnoses\_subject\_\_id)\\($A_2$.subject\_id)=($A_3$.subject\_id)\\($A_4$.hadm\_id)=($A_3$.hadm\_id)} & \makecell{$A_1$.prov\_admissions\_admission\_\_type=EMERGENCY \\ $A_4$.admission\_type=EMERGENCY} & 0.5 & 0.81 & 0.62 & 13\\
\hline
\makecell{$A_1$: PT \\ $A_2$: patients \\ $A_3$: procedures \\ $A_4$: admissions} & \makecell{($A_2$.subject\_id)=($A_1$.prov\_diagnoses\_subject\_\_id)\\($A_2$.subject\_id)=($A_3$.subject\_id)\\($A_4$.hadm\_id)=($A_3$.hadm\_id)} & \makecell{$A_1$.prov\_admissions\_admission\_\_type=EMERGENCY \\ $A_4$.admission\_type=EMERGENCY} & 0.5 & 0.8 & 0.61 & 13\\
\hline
\makecell{$A_1$: PT \\ $A_2$: patients} & \makecell{($A_2$.subject\_id)=($A_1$.prov\_diagnoses\_subject\_\_id)} & \makecell{$A_1$.prov\_admissions\_admission\_\_type=EMERGENCY \\ $A_2$.expire\_flag=1} & 0.63 & 0.59 & 0.61 & 2\\
\hline
\makecell{$A_1$: PT \\ $A_2$: patients \\ $A_3$: icustays \\ $A_4$: admissions} & \makecell{($A_2$.subject\_id)=($A_1$.prov\_diagnoses\_subject\_\_id)\\($A_2$.subject\_id)=($A_3$.subject\_id)\\($A_4$.hadm\_id)=($A_3$.hadm\_id)} & \makecell{$A_4$.insurance=Medicare} & 0.54 & 0.67 & 0.6 & 13\\
\hline
\makecell{$A_1$: PT \\ $A_2$: patients \\ $A_3$: procedures \\ $A_4$: admissions} & \makecell{($A_2$.subject\_id)=($A_1$.prov\_diagnoses\_subject\_\_id)\\($A_2$.subject\_id)=($A_3$.subject\_id)\\($A_4$.hadm\_id)=($A_3$.hadm\_id)} & \makecell{$A_4$.insurance=Medicare} & 0.54 & 0.66 & 0.59 & 13\\
\hline
\makecell{$A_1$: PT \\ $A_2$: patients \\ $A_3$: patients\_admit\_info \\ $A_4$: admissions} & \makecell{($A_2$.subject\_id)=($A_1$.prov\_diagnoses\_subject\_\_id)\\($A_2$.subject\_id)=($A_3$.subject\_id)\\($A_4$.hadm\_id)=($A_3$.hadm\_id)} & \makecell{$A_4$.hospital\_stay\_length>4.0 \\ $A_3$.language=ENGL} & 0.54 & 0.63 & 0.58 & 13\\
\hline
\makecell{$A_1$: PT \\ $A_2$: patients \\ $A_3$: procedures \\ $A_4$: admissions} & \makecell{($A_2$.subject\_id)=($A_1$.prov\_diagnoses\_subject\_\_id)\\($A_2$.subject\_id)=($A_3$.subject\_id)\\($A_4$.hadm\_id)=($A_3$.hadm\_id)} & \makecell{$A_1$.prov\_admissions\_marital\_\_status=MARRIED \\ $A_4$.marital\_status=MARRIED} & 0.59 & 0.57 & 0.58 & 2\\
\hline
\makecell{} & \makecell{A\_1} & \makecell{$A_1$.prov\_admissions\_marital\_\_status=MARRIED} & 0.59 & 0.57 & 0.58 & 2\\
\hline
\makecell{$A_1$: PT \\ $A_2$: procedures \\ $A_3$: patients} & \makecell{($A_1$.prov\_admissions\_hadm\_\_id)=($A_2$.hadm\_id)\\($A_3$.subject\_id)=($A_2$.subject\_id)} & \makecell{$A_1$.prov\_admissions\_hospital\_\_stay\_\_length<21.0 \\ $A_3$.expire\_flag=0} & 0.65 & 0.52 & 0.58 & 13\\
\hline
\end{tabular}
}
\label{tab:tops-mimic1}
\caption{$\query_{mimic1}$ Top-20 patterns}
\end{figure*}

\begin{figure*}[!h]
\centering
{\scriptsize
\centering
\begin{tabular}{|c|P{20em}|P{29.5em}|c|c|c|P{5em}|}
\hline
\cellcolor{grey}\textbf{Nodes} & \cellcolor{grey}\textbf{Edges} & \cellcolor{grey}\textbf{Pattern Desc} & \cellcolor{grey}\textbf{Precision} & \cellcolor{grey}\textbf{Recall} & \cellcolor{grey}\textbf{\abbrF} & \cellcolor{grey}\textbf{Primary Tuple} \\
\makecell{} & \makecell{A\_1} & \makecell{$A_1$.prov\_admissions\_admission\_\_type=EMERGENCY} & 0.85 & 0.84 & 0.85 & Medicare\\
\hline
\makecell{$A_1$: PT \\ $A_2$: procedures \\ $A_3$: patients} & \makecell{($A_1$.prov\_admissions\_hadm\_\_id)=($A_2$.hadm\_id)\\($A_3$.subject\_id)=($A_2$.subject\_id)} & \makecell{$A_3$.expire\_flag=1} & 0.91 & 0.57 & 0.7 & Medicare\\
\hline
\makecell{$A_1$: PT \\ $A_2$: patients\_admit\_info \\ $A_3$: patients} & \makecell{($A_1$.prov\_admissions\_hadm\_\_id)=($A_2$.hadm\_id)\\($A_3$.subject\_id)=($A_2$.subject\_id)} & \makecell{$A_3$.expire\_flag=1} & 0.91 & 0.57 & 0.7 & Medicare\\
\hline
\makecell{$A_1$: PT \\ $A_2$: icustays \\ $A_3$: patients} & \makecell{($A_1$.prov\_admissions\_hadm\_\_id)=($A_2$.hadm\_id)\\($A_3$.subject\_id)=($A_2$.subject\_id)} & \makecell{$A_1$.prov\_admissions\_admission\_\_type=EMERGENCY \\ $A_3$.expire\_flag=1} & 0.91 & 0.5 & 0.65 & Medicare\\
\hline
\makecell{$A_1$: PT \\ $A_2$: procedures \\ $A_3$: patients} & \makecell{($A_1$.prov\_admissions\_hadm\_\_id)=($A_2$.hadm\_id)\\($A_3$.subject\_id)=($A_2$.subject\_id)} & \makecell{$A_3$.gender=M} & 0.82 & 0.53 & 0.65 & Medicare\\
\hline
\makecell{$A_1$: PT \\ $A_2$: patients\_admit\_info \\ $A_3$: patients} & \makecell{($A_1$.prov\_admissions\_hadm\_\_id)=($A_2$.hadm\_id)\\($A_3$.subject\_id)=($A_2$.subject\_id)} & \makecell{$A_3$.gender=M} & 0.82 & 0.52 & 0.64 & Medicare\\
\hline
\makecell{$A_1$: PT \\ $A_2$: icustays \\ $A_3$: patients} & \makecell{($A_1$.prov\_admissions\_hadm\_\_id)=($A_2$.hadm\_id)\\($A_3$.subject\_id)=($A_2$.subject\_id)} & \makecell{$A_3$.gender=M} & 0.82 & 0.52 & 0.64 & Medicare\\
\hline
\makecell{} & \makecell{A\_1} & \makecell{$A_1$.prov\_admissions\_marital\_\_status=MARRIED} & 0.93 & 0.45 & 0.61 & Medicare\\
\hline
\makecell{$A_1$: PT \\ $A_2$: procedures \\ $A_3$: patients} & \makecell{($A_1$.prov\_admissions\_hadm\_\_id)=($A_2$.hadm\_id)\\($A_3$.subject\_id)=($A_2$.subject\_id)} & \makecell{$A_1$.prov\_admissions\_admission\_\_type=EMERGENCY \\ $A_1$.prov\_admissions\_hospital\_\_stay\_\_length<22.0 \\ $A_3$.expire\_flag=1} & 0.92 & 0.43 & 0.59 & Medicare\\
\hline
\makecell{$A_1$: PT \\ $A_2$: patients\_admit\_info \\ $A_3$: patients} & \makecell{($A_1$.prov\_admissions\_hadm\_\_id)=($A_2$.hadm\_id)\\($A_3$.subject\_id)=($A_2$.subject\_id)} & \makecell{$A_3$.gender=F \\ $A_2$.age>62.38} & 0.97 & 0.41 & 0.58 & Medicare\\
\hline
\makecell{$A_1$: PT \\ $A_2$: patients\_admit\_info \\ $A_3$: patients} & \makecell{($A_1$.prov\_admissions\_hadm\_\_id)=($A_2$.hadm\_id)\\($A_3$.subject\_id)=($A_2$.subject\_id)} & \makecell{$A_1$.prov\_admissions\_admission\_\_type=EMERGENCY \\ $A_3$.expire\_flag=1 \\ $A_2$.age>66.96} & 0.99 & 0.41 & 0.58 & Medicare\\
\hline
\makecell{$A_1$: PT \\ $A_2$: procedures \\ $A_3$: patients} & \makecell{($A_1$.prov\_admissions\_hadm\_\_id)=($A_2$.hadm\_id)\\($A_3$.subject\_id)=($A_2$.subject\_id)} & \makecell{$A_1$.prov\_admissions\_hospital\_\_stay\_\_length<20.0 \\ $A_3$.gender=F} & 0.85 & 0.41 & 0.56 & Medicare\\
\hline
\makecell{} & \makecell{A\_1} & \makecell{$A_1$.prov\_admissions\_admission\_\_type=EMERGENCY \\ $A_1$.prov\_admissions\_marital\_\_status=MARRIED} & 0.93 & 0.36 & 0.52 & Medicare\\
\hline
\makecell{$A_1$: PT \\ $A_2$: procedures \\ $A_3$: patients} & \makecell{($A_1$.prov\_admissions\_hadm\_\_id)=($A_2$.hadm\_id)\\($A_3$.subject\_id)=($A_2$.subject\_id)} & \makecell{$A_1$.prov\_admissions\_hospital\_\_stay\_\_length<17.0 \\ $A_3$.expire\_flag=0} & 0.76 & 0.38 & 0.51 & Medicare\\
\hline
\makecell{$A_1$: PT \\ $A_2$: patients\_admit\_info \\ $A_3$: patients} & \makecell{($A_1$.prov\_admissions\_hadm\_\_id)=($A_2$.hadm\_id)\\($A_3$.subject\_id)=($A_2$.subject\_id)} & \makecell{$A_1$.prov\_admissions\_hospital\_\_stay\_\_length<11.0 \\ $A_3$.expire\_flag=0 \\ $A_2$.age<68.44} & 0.5 & 0.51 & 0.51 & Medicaid\\
\hline
\makecell{} & \makecell{A\_1} & \makecell{$A_1$.prov\_admissions\_hospital\_\_stay\_\_length>4.0 \\ $A_1$.prov\_admissions\_marital\_\_status=MARRIED} & 0.93 & 0.33 & 0.49 & Medicare\\
\hline
\makecell{$A_1$: PT \\ $A_2$: icustays \\ $A_3$: patients} & \makecell{($A_1$.prov\_admissions\_hadm\_\_id)=($A_2$.hadm\_id)\\($A_3$.subject\_id)=($A_2$.subject\_id)} & \makecell{$A_1$.prov\_admissions\_admission\_\_type=EMERGENCY \\ $A_1$.prov\_admissions\_hospital\_\_stay\_\_length<14.0 \\ $A_3$.gender=M} & 0.85 & 0.34 & 0.48 & Medicare\\
\hline
\makecell{} & \makecell{A\_1} & \makecell{$A_1$.prov\_admissions\_admission\_\_location=EMERGENCY ROOM ADMIT \\ $A_1$.prov\_admissions\_admission\_\_type=EMERGENCY \\ $A_1$.prov\_admissions\_hospital\_\_stay\_\_length<12.0} & 0.85 & 0.33 & 0.48 & Medicare\\
\hline
\makecell{$A_1$: PT \\ $A_2$: icustays \\ $A_3$: patients} & \makecell{($A_1$.prov\_admissions\_hadm\_\_id)=($A_2$.hadm\_id)\\($A_3$.subject\_id)=($A_2$.subject\_id)} & \makecell{$A_2$.los>1.1336 \\ $A_3$.expire\_flag=0} & 0.77 & 0.34 & 0.47 & Medicare\\
\hline
\makecell{$A_1$: PT \\ $A_2$: icustays \\ $A_3$: patients} & \makecell{($A_1$.prov\_admissions\_hadm\_\_id)=($A_2$.hadm\_id)\\($A_3$.subject\_id)=($A_2$.subject\_id)} & \makecell{$A_1$.prov\_admissions\_hospital\_\_stay\_\_length<13.0 \\ $A_2$.los<4.41075 \\ $A_3$.gender=F} & 0.84 & 0.32 & 0.46 & Medicare\\
\hline
\end{tabular}
}
\label{tab:tops-mimic2}
\caption{$\query_{mimic2}$ Top-20 patterns}
\end{figure*}

\begin{figure*}[!h]
\centering
{\scriptsize
\centering
\begin{tabular}{|c|P{25em}|P{20em}|c|c|c|P{5em}|}
\hline
\cellcolor{grey}\textbf{Nodes} & \cellcolor{grey}\textbf{Edges} & \cellcolor{grey}\textbf{Pattern Desc} & \cellcolor{grey}\textbf{Precision} & \cellcolor{grey}\textbf{Recall} & \cellcolor{grey}\textbf{\abbrF} & \cellcolor{grey}\textbf{Primary Tuple} \\
\makecell{$A_1$: PT \\ $A_2$: patients \\ $A_3$: procedures \\ $A_4$: admissions} & \makecell{($A_2$.subject\_id)=($A_1$.prov\_icustays\_subject\_\_id)\\($A_2$.subject\_id)=($A_3$.subject\_id)\\($A_4$.hadm\_id)=($A_3$.hadm\_id)} & \makecell{$A_4$.hospital\_stay\_length>9.0 \\ $A_3$.chapter=16} & 0.8 & 0.94 & 0.87 & x>8\\
\hline
\makecell{$A_1$: PT \\ $A_2$: patients \\ $A_3$: icustays \\ $A_4$: admissions} & \makecell{($A_2$.subject\_id)=($A_1$.prov\_icustays\_subject\_\_id)\\($A_2$.subject\_id)=($A_3$.subject\_id)\\($A_4$.hadm\_id)=($A_3$.hadm\_id)} & \makecell{$A_4$.hospital\_stay\_length<6.0 \\ $A_3$.los\_group=0-1} & 0.97 & 0.77 & 0.86 & 0-1\\
\hline
\makecell{$A_1$: PT \\ $A_2$: admissions \\ $A_3$: icustays \\ $A_4$: patients} & \makecell{($A_2$.hadm\_id)=($A_1$.prov\_icustays\_hadm\_\_id)\\($A_2$.hadm\_id)=($A_3$.hadm\_id)\\($A_4$.subject\_id)=($A_3$.subject\_id)} & \makecell{$A_3$.los\_group=0-1 \\ $A_4$.expire\_flag=0} & 0.99 & 0.72 & 0.83 & 0-1\\
\hline
\makecell{$A_1$: PT \\ $A_2$: patients \\ $A_3$: admissions} & \makecell{($A_2$.subject\_id)=($A_1$.prov\_icustays\_subject\_\_id)\\($A_3$.hadm\_id)=($A_1$.prov\_icustays\_hadm\_\_id)} & \makecell{$A_1$.prov\_icustays\_dbsource=carevue \\ $A_3$.hospital\_stay\_length>8.0} & 0.87 & 0.71 & 0.78 & x>8\\
\hline
\makecell{$A_1$: PT \\ $A_2$: patients \\ $A_3$: admissions} & \makecell{($A_2$.subject\_id)=($A_1$.prov\_icustays\_subject\_\_id)\\($A_3$.hadm\_id)=($A_1$.prov\_icustays\_hadm\_\_id)} & \makecell{$A_3$.hospital\_stay\_length<16.0 \\ $A_2$.expire\_flag=0} & 0.86 & 0.7 & 0.77 & 0-1\\
\hline
\makecell{$A_1$: PT \\ $A_2$: patients \\ $A_3$: icustays \\ $A_4$: admissions} & \makecell{($A_2$.subject\_id)=($A_1$.prov\_icustays\_subject\_\_id)\\($A_2$.subject\_id)=($A_3$.subject\_id)\\($A_4$.hadm\_id)=($A_3$.hadm\_id)} & \makecell{$A_4$.hospital\_stay\_length<16.0 \\ $A_2$.expire\_flag=0} & 0.81 & 0.7 & 0.76 & 0-1\\
\hline
\makecell{$A_1$: PT \\ $A_2$: patients \\ $A_3$: icustays \\ $A_4$: admissions} & \makecell{($A_2$.subject\_id)=($A_1$.prov\_icustays\_subject\_\_id)\\($A_2$.subject\_id)=($A_3$.subject\_id)\\($A_4$.hadm\_id)=($A_3$.hadm\_id)} & \makecell{$A_1$.prov\_icustays\_dbsource=carevue \\ $A_4$.hospital\_stay\_length>9.0 \\ $A_3$.dbsource=carevue} & 0.83 & 0.69 & 0.75 & x>8\\
\hline
\makecell{$A_1$: PT \\ $A_2$: patients} & \makecell{($A_2$.subject\_id)=($A_1$.prov\_icustays\_subject\_\_id)} & \makecell{$A_2$.expire\_flag=0} & 0.65 & 0.72 & 0.68 & 0-1\\
\hline
\makecell{$A_1$: PT \\ $A_2$: admissions \\ $A_3$: patients\_admit\_info \\ $A_4$: patients} & \makecell{($A_2$.hadm\_id)=($A_1$.prov\_icustays\_hadm\_\_id)\\($A_2$.hadm\_id)=($A_3$.hadm\_id)\\($A_4$.subject\_id)=($A_3$.subject\_id)} & \makecell{$A_4$.expire\_flag=0} & 0.65 & 0.72 & 0.68 & 0-1\\
\hline
\makecell{$A_1$: PT \\ $A_2$: patients \\ $A_3$: patients\_admit\_info \\ $A_4$: admissions} & \makecell{($A_2$.subject\_id)=($A_1$.prov\_icustays\_subject\_\_id)\\($A_2$.subject\_id)=($A_3$.subject\_id)\\($A_4$.hadm\_id)=($A_3$.hadm\_id)} & \makecell{$A_4$.admission\_type=EMERGENCY \\ $A_4$.hospital\_stay\_length>8.0} & 0.7 & 0.67 & 0.68 & x>8\\
\hline
\makecell{$A_1$: PT \\ $A_2$: admissions \\ $A_3$: icustays \\ $A_4$: patients} & \makecell{($A_2$.hadm\_id)=($A_1$.prov\_icustays\_hadm\_\_id)\\($A_2$.hadm\_id)=($A_3$.hadm\_id)\\($A_4$.subject\_id)=($A_3$.subject\_id)} & \makecell{$A_3$.los<0.9754 \\ $A_4$.gender=M} & 0.99 & 0.51 & 0.67 & 0-1\\
\hline
\makecell{$A_1$: PT \\ $A_2$: admissions \\ $A_3$: icustays \\ $A_4$: patients} & \makecell{($A_2$.hadm\_id)=($A_1$.prov\_icustays\_hadm\_\_id)\\($A_2$.hadm\_id)=($A_3$.hadm\_id)\\($A_4$.subject\_id)=($A_3$.subject\_id)} & \makecell{$A_1$.prov\_icustays\_dbsource=carevue \\ $A_3$.dbsource=carevue \\ $A_3$.los<11.887 \\ $A_4$.expire\_flag=0} & 0.85 & 0.51 & 0.64 & 0-1\\
\hline
\makecell{$A_1$: PT \\ $A_2$: patients \\ $A_3$: procedures \\ $A_4$: admissions} & \makecell{($A_2$.subject\_id)=($A_1$.prov\_icustays\_subject\_\_id)\\($A_2$.subject\_id)=($A_3$.subject\_id)\\($A_4$.hadm\_id)=($A_3$.hadm\_id)} & \makecell{$A_1$.prov\_icustays\_dbsource=carevue \\ $A_4$.hospital\_stay\_length<7.0 \\ $A_2$.expire\_flag=0} & 0.96 & 0.47 & 0.63 & 0-1\\
\hline
\makecell{} & \makecell{A\_1} & \makecell{$A_1$.prov\_icustays\_dbsource=carevue} & 0.57 & 0.68 & 0.62 & 0-1\\
\hline
\makecell{$A_1$: PT \\ $A_2$: patients} & \makecell{($A_2$.subject\_id)=($A_1$.prov\_icustays\_subject\_\_id)} & \makecell{$A_1$.prov\_icustays\_dbsource=carevue} & 0.57 & 0.68 & 0.62 & 0-1\\
\hline
\makecell{$A_1$: PT \\ $A_2$: patients \\ $A_3$: admissions} & \makecell{($A_2$.subject\_id)=($A_1$.prov\_icustays\_subject\_\_id)\\($A_3$.hadm\_id)=($A_1$.prov\_icustays\_hadm\_\_id)} & \makecell{$A_3$.hospital\_stay\_length<20.0 \\ $A_2$.gender=M} & 0.75 & 0.52 & 0.62 & 0-1\\
\hline
\makecell{$A_1$: PT \\ $A_2$: patients \\ $A_3$: patients\_admit\_info \\ $A_4$: admissions} & \makecell{($A_2$.subject\_id)=($A_1$.prov\_icustays\_subject\_\_id)\\($A_2$.subject\_id)=($A_3$.subject\_id)\\($A_4$.hadm\_id)=($A_3$.hadm\_id)} & \makecell{$A_4$.admission\_type=EMERGENCY \\ $A_4$.hospital\_stay\_length<17.0} & 0.68 & 0.56 & 0.61 & 0-1\\
\hline
\makecell{$A_1$: PT \\ $A_2$: patients \\ $A_3$: procedures \\ $A_4$: admissions} & \makecell{($A_2$.subject\_id)=($A_1$.prov\_icustays\_subject\_\_id)\\($A_2$.subject\_id)=($A_3$.subject\_id)\\($A_4$.hadm\_id)=($A_3$.hadm\_id)} & \makecell{$A_4$.hospital\_stay\_length<18.0 \\ $A_2$.gender=M} & 0.69 & 0.54 & 0.6 & 0-1\\
\hline
\makecell{$A_1$: PT \\ $A_2$: patients \\ $A_3$: icustays \\ $A_4$: admissions} & \makecell{($A_2$.subject\_id)=($A_1$.prov\_icustays\_subject\_\_id)\\($A_2$.subject\_id)=($A_3$.subject\_id)\\($A_4$.hadm\_id)=($A_3$.hadm\_id)} & \makecell{$A_4$.hospital\_stay\_length<20.0 \\ $A_2$.gender=M} & 0.7 & 0.53 & 0.6 & 0-1\\
\hline
\makecell{$A_1$: PT \\ $A_2$: admissions \\ $A_3$: icustays \\ $A_4$: patients} & \makecell{($A_2$.hadm\_id)=($A_1$.prov\_icustays\_hadm\_\_id)\\($A_2$.hadm\_id)=($A_3$.hadm\_id)\\($A_4$.subject\_id)=($A_3$.subject\_id)} & \makecell{$A_3$.los<0.9574 \\ $A_4$.gender=F} & 0.99 & 0.42 & 0.58 & 0-1\\
\hline
\end{tabular}
}
\label{tab:tops-mimic3}
\caption{$\query_{mimic3}$ Top-20 patterns}
\end{figure*}

\begin{figure*}[!h]
\centering
{\scriptsize
\centering
\begin{tabular}{|c|P{20em}|P{29.5em}|c|c|c|P{5em}|}
\hline
\cellcolor{grey}\textbf{Nodes} & \cellcolor{grey}\textbf{Edges} & \cellcolor{grey}\textbf{Pattern Desc} & \cellcolor{grey}\textbf{Precision} & \cellcolor{grey}\textbf{Recall} & \cellcolor{grey}\textbf{\abbrF} & \cellcolor{grey}\textbf{Primary Tuple} \\
\makecell{$A_1$: PT \\ $A_2$: patients\_admit\_info \\ $A_3$: patients} & \makecell{($A_1$.prov\_admissions\_hadm\_\_id)=($A_2$.hadm\_id)\\($A_3$.subject\_id)=($A_2$.subject\_id)} & \makecell{$A_3$.expire\_flag=0 \\ $A_2$.age<70.86} & 0.77 & 0.77 & 0.77 & Private\\
\hline
\makecell{} & \makecell{A\_1} & \makecell{$A_1$.prov\_admissions\_admission\_\_type=EMERGENCY} & 0.65 & 0.84 & 0.73 & Medicare\\
\hline
\makecell{$A_1$: PT \\ $A_2$: procedures \\ $A_3$: patients} & \makecell{($A_1$.prov\_admissions\_hadm\_\_id)=($A_2$.hadm\_id)\\($A_3$.subject\_id)=($A_2$.subject\_id)} & \makecell{$A_3$.expire\_flag=0} & 0.6 & 0.79 & 0.68 & Private\\
\hline
\makecell{$A_1$: PT \\ $A_2$: icustays \\ $A_3$: patients} & \makecell{($A_1$.prov\_admissions\_hadm\_\_id)=($A_2$.hadm\_id)\\($A_3$.subject\_id)=($A_2$.subject\_id)} & \makecell{$A_3$.expire\_flag=0} & 0.6 & 0.8 & 0.68 & Private\\
\hline
\makecell{$A_1$: PT \\ $A_2$: procedures \\ $A_3$: patients} & \makecell{($A_1$.prov\_admissions\_hadm\_\_id)=($A_2$.hadm\_id)\\($A_3$.subject\_id)=($A_2$.subject\_id)} & \makecell{$A_1$.prov\_admissions\_hospital\_\_stay\_\_length<22.0 \\ $A_3$.expire\_flag=1} & 0.78 & 0.5 & 0.61 & Medicare\\
\hline
\makecell{$A_1$: PT \\ $A_2$: procedures \\ $A_3$: patients} & \makecell{($A_1$.prov\_admissions\_hadm\_\_id)=($A_2$.hadm\_id)\\($A_3$.subject\_id)=($A_2$.subject\_id)} & \makecell{$A_3$.gender=M} & 0.48 & 0.61 & 0.54 & Private\\
\hline
\makecell{$A_1$: PT \\ $A_2$: icustays \\ $A_3$: patients} & \makecell{($A_1$.prov\_admissions\_hadm\_\_id)=($A_2$.hadm\_id)\\($A_3$.subject\_id)=($A_2$.subject\_id)} & \makecell{$A_3$.gender=F} & 0.6 & 0.48 & 0.53 & Medicare\\
\hline
\makecell{$A_1$: PT \\ $A_2$: patients\_admit\_info \\ $A_3$: patients} & \makecell{($A_1$.prov\_admissions\_hadm\_\_id)=($A_2$.hadm\_id)\\($A_3$.subject\_id)=($A_2$.subject\_id)} & \makecell{$A_3$.gender=M} & 0.48 & 0.6 & 0.53 & Private\\
\hline
\makecell{$A_1$: PT \\ $A_2$: patients\_admit\_info \\ $A_3$: patients} & \makecell{($A_1$.prov\_admissions\_hadm\_\_id)=($A_2$.hadm\_id)\\($A_3$.subject\_id)=($A_2$.subject\_id)} & \makecell{$A_3$.gender=M \\ $A_2$.age<64.72 \\ $A_2$.ethnicity=WHITE} & 0.83 & 0.39 & 0.53 & Private\\
\hline
\makecell{} & \makecell{A\_1} & \makecell{$A_1$.prov\_admissions\_marital\_\_status=MARRIED} & 0.56 & 0.45 & 0.5 & Medicare\\
\hline
\makecell{$A_1$: PT \\ $A_2$: patients\_admit\_info \\ $A_3$: patients} & \makecell{($A_1$.prov\_admissions\_hadm\_\_id)=($A_2$.hadm\_id)\\($A_3$.subject\_id)=($A_2$.subject\_id)} & \makecell{$A_1$.prov\_admissions\_hospital\_\_stay\_\_length<10.0 \\ $A_3$.expire\_flag=0 \\ $A_2$.ethnicity=WHITE} & 0.59 & 0.42 & 0.49 & Private\\
\hline
\makecell{$A_1$: PT \\ $A_2$: procedures \\ $A_3$: patients} & \makecell{($A_1$.prov\_admissions\_hadm\_\_id)=($A_2$.hadm\_id)\\($A_3$.subject\_id)=($A_2$.subject\_id)} & \makecell{$A_1$.prov\_admissions\_hospital\_\_stay\_\_length<20.0 \\ $A_3$.gender=F} & 0.6 & 0.41 & 0.49 & Medicare\\
\hline
\makecell{$A_1$: PT \\ $A_2$: icustays \\ $A_3$: patients} & \makecell{($A_1$.prov\_admissions\_hadm\_\_id)=($A_2$.hadm\_id)\\($A_3$.subject\_id)=($A_2$.subject\_id)} & \makecell{$A_2$.los>1.1306 \\ $A_3$.gender=M} & 0.56 & 0.42 & 0.48 & Medicare\\
\hline
\makecell{$A_1$: PT \\ $A_2$: icustays \\ $A_3$: patients} & \makecell{($A_1$.prov\_admissions\_hadm\_\_id)=($A_2$.hadm\_id)\\($A_3$.subject\_id)=($A_2$.subject\_id)} & \makecell{$A_1$.prov\_admissions\_admission\_\_type=EMERGENCY \\ $A_1$.prov\_admissions\_hospital\_\_stay\_\_length<15.0 \\ $A_2$.los<5.2349 \\ $A_3$.expire\_flag=1} & 0.8 & 0.34 & 0.47 & Medicare\\
\hline
\makecell{$A_1$: PT \\ $A_2$: patients\_admit\_info \\ $A_3$: patients} & \makecell{($A_1$.prov\_admissions\_hadm\_\_id)=($A_2$.hadm\_id)\\($A_3$.subject\_id)=($A_2$.subject\_id)} & \makecell{$A_1$.prov\_admissions\_hospital\_\_stay\_\_length<12.0 \\ $A_3$.gender=F \\ $A_2$.age>48.08} & 0.76 & 0.34 & 0.47 & Medicare\\
\hline
\makecell{$A_1$: PT \\ $A_2$: procedures \\ $A_3$: patients} & \makecell{($A_1$.prov\_admissions\_hadm\_\_id)=($A_2$.hadm\_id)\\($A_3$.subject\_id)=($A_2$.subject\_id)} & \makecell{$A_1$.prov\_admissions\_admission\_\_type=EMERGENCY \\ $A_1$.prov\_admissions\_hospital\_\_stay\_\_length<21.0 \\ $A_3$.gender=M} & 0.61 & 0.38 & 0.46 & Medicare\\
\hline
\makecell{} & \makecell{A\_1} & \makecell{$A_1$.prov\_admissions\_admission\_\_type=EMERGENCY \\ $A_1$.prov\_admissions\_marital\_\_status=MARRIED} & 0.58 & 0.36 & 0.45 & Medicare\\
\hline
\makecell{} & \makecell{A\_1} & \makecell{$A_1$.prov\_admissions\_admission\_\_location=EMERGENCY ROOM ADMIT \\ $A_1$.prov\_admissions\_admission\_\_type=EMERGENCY \\ $A_1$.prov\_admissions\_hospital\_\_stay\_\_length<12.0} & 0.66 & 0.33 & 0.44 & Medicare\\
\hline
\makecell{} & \makecell{A\_1} & \makecell{$A_1$.prov\_admissions\_hospital\_\_stay\_\_length>4.0 \\ $A_1$.prov\_admissions\_marital\_\_status=MARRIED} & 0.58 & 0.33 & 0.42 & Medicare\\
\hline
\makecell{$A_1$: PT \\ $A_2$: icustays \\ $A_3$: patients} & \makecell{($A_1$.prov\_admissions\_hadm\_\_id)=($A_2$.hadm\_id)\\($A_3$.subject\_id)=($A_2$.subject\_id)} & \makecell{$A_1$.prov\_admissions\_hospital\_\_stay\_\_length>4.0 \\ $A_2$.los>1.111825 \\ $A_3$.gender=F} & 0.65 & 0.31 & 0.42 & Medicare\\
\hline
\end{tabular}
}
\label{tab:tops-mimic4}
\caption{$\query_{mimic4}$ Top-20 patterns}
\end{figure*}

\begin{figure*}[!h]
\centering
{\scriptsize
\centering
\begin{tabular}{|c|P{25em}|P{20em}|c|c|c|P{5em}|}
\hline
\cellcolor{grey}\textbf{Nodes} & \cellcolor{grey}\textbf{Edges} & \cellcolor{grey}\textbf{Pattern Desc} & \cellcolor{grey}\textbf{Precision} & \cellcolor{grey}\textbf{Recall} & \cellcolor{grey}\textbf{\abbrF} & \cellcolor{grey}\textbf{Primary Tuple} \\
\hline
\makecell{$A_1$: PT \\ $A_2$: patients \\ $A_3$: patients\_admit\_info \\ $A_4$: admissions} & \makecell{$A_2$.subject\_id=$A_1$.prov\_patients\_\_admit\_\_info\_subject\_\_id\\$A_2$.subject\_id=$A_3$.subject\_id\\$A_4$.hadm\_id=$A_3$.hadm\_id} & \makecell{$A_4$.hospital\_stay\_length<19.0 \\ $A_3$.ethnicity=ASIAN} & 1.0 & 0.82 & 0.9 & ASIAN\\
\hline
\makecell{$A_1$: PT \\ $A_2$: patients \\ $A_3$: patients\_admit\_info \\ $A_4$: admissions} & \makecell{$A_2$.subject\_id=$A_1$.prov\_patients\_\_admit\_\_info\_subject\_\_id\\$A_2$.subject\_id=$A_3$.subject\_id\\$A_4$.hadm\_id=$A_3$.hadm\_id} & \makecell{$A_4$.admission\_type=EMERGENCY \\ $A_4$.hospital\_stay\_length>5.0 \\ $A_3$.ethnicity=HISPANIC} & 1.0 & 0.67 & 0.8 & HISPANIC\\
\hline
\makecell{$A_1$: PT \\ $A_2$: patients \\ $A_3$: icustays \\ $A_4$: admissions} & \makecell{$A_2$.subject\_id=$A_1$.prov\_patients\_\_admit\_\_info\_subject\_\_id\\$A_2$.subject\_id=$A_3$.subject\_id\\$A_4$.hadm\_id=$A_3$.hadm\_id} & \makecell{$A_4$.admission\_type=EMERGENCY} & 0.6 & 0.8 & 0.68 & HISPANIC\\
\hline
\makecell{$A_1$: PT \\ $A_2$: patients \\ $A_3$: patients\_admit\_info \\ $A_4$: admissions} & \makecell{$A_2$.subject\_id=$A_1$.prov\_patients\_\_admit\_\_info\_subject\_\_id\\$A_2$.subject\_id=$A_3$.subject\_id\\$A_4$.hadm\_id=$A_3$.hadm\_id} & \makecell{$A_4$.admission\_type=EMERGENCY} & 0.6 & 0.8 & 0.68 & HISPANIC\\
\hline
\makecell{$A_1$: PT \\ $A_2$: patients \\ $A_3$: procedures \\ $A_4$: admissions} & \makecell{$A_2$.subject\_id=$A_1$.prov\_patients\_\_admit\_\_info\_subject\_\_id\\$A_2$.subject\_id=$A_3$.subject\_id\\$A_4$.hadm\_id=$A_3$.hadm\_id} & \makecell{$A_4$.admission\_type=EMERGENCY} & 0.59 & 0.8 & 0.68 & HISPANIC\\
\hline
\makecell{$A_1$: PT \\ $A_2$: patients \\ $A_3$: admissions} & \makecell{$A_2$.subject\_id=$A_1$.prov\_patients\_\_admit\_\_info\_subject\_\_id\\$A_3$.hadm\_id=$A_1$.prov\_patients\_\_admit\_\_info\_hadm\_\_id} & \makecell{$A_3$.admission\_type=EMERGENCY} & 0.6 & 0.75 & 0.67 & HISPANIC\\
\hline
\makecell{$A_1$: PT \\ $A_2$: admissions \\ $A_3$: icustays \\ $A_4$: patients} & \makecell{$A_2$.hadm\_id=$A_1$.prov\_patients\_\_admit\_\_info\_hadm\_\_id\\$A_2$.hadm\_id=$A_3$.hadm\_id\\$A_4$.subject\_id=$A_3$.subject\_id} & \makecell{$A_4$.expire\_flag=0} & 0.57 & 0.71 & 0.64 & HISPANIC\\
\hline
\makecell{$A_1$: PT \\ $A_2$: patients} & \makecell{$A_2$.subject\_id=$A_1$.prov\_patients\_\_admit\_\_info\_subject\_\_id} & \makecell{$A_2$.expire\_flag=0} & 0.57 & 0.71 & 0.63 & HISPANIC\\
\hline
\makecell{$A_1$: PT \\ $A_2$: admissions \\ $A_3$: patients\_admit\_info \\ $A_4$: patients} & \makecell{$A_2$.hadm\_id=$A_1$.prov\_patients\_\_admit\_\_info\_hadm\_\_id\\$A_2$.hadm\_id=$A_3$.hadm\_id\\$A_4$.subject\_id=$A_3$.subject\_id} & \makecell{$A_4$.expire\_flag=0} & 0.57 & 0.71 & 0.63 & HISPANIC\\
\hline
\makecell{$A_1$} & \makecell{$A_1$} & \makecell{$A_1$.prov\_patients\_\_admit\_\_info\_religion \\ =CATHOLIC} & 0.9 & 0.49 & 0.63 & HISPANIC\\
\hline
\makecell{$A_1$: PT \\ $A_2$: patients \\ $A_3$: procedures \\ $A_4$: admissions} & \makecell{$A_2$.subject\_id=$A_1$.prov\_patients\_\_admit\_\_info\_subject\_\_id\\$A_2$.subject\_id=$A_3$.subject\_id\\$A_4$.hadm\_id=$A_3$.hadm\_id} & \makecell{$A_4$.hospital\_stay\_length<27.0 \\ $A_3$.chapter=16} & 0.54 & 0.74 & 0.63 & HISPANIC\\
\hline
\makecell{$A_1$: PT \\ $A_2$: patients \\ $A_3$: patients\_admit\_info \\ $A_4$: admissions} & \makecell{$A_2$.subject\_id=$A_1$.prov\_patients\_\_admit\_\_info\_subject\_\_id\\$A_2$.subject\_id=$A_3$.subject\_id\\$A_4$.hadm\_id=$A_3$.hadm\_id} & \makecell{$A_4$.hospital\_stay\_length>4.0 \\ $A_2$.expire\_flag=0} & 0.61 & 0.6 & 0.6 & HISPANIC\\
\hline
\makecell{$A_1$: PT \\ $A_2$: admissions \\ $A_3$: icustays \\ $A_4$: patients} & \makecell{$A_2$.hadm\_id=$A_1$.prov\_patients\_\_admit\_\_info\_hadm\_\_id\\$A_2$.hadm\_id=$A_3$.hadm\_id\\$A_4$.subject\_id=$A_3$.subject\_id} & \makecell{$A_4$.gender=M} & 0.57 & 0.64 & 0.6 & HISPANIC\\
\hline
\makecell{$A_1$: PT \\ $A_2$: patients} & \makecell{$A_2$.subject\_id=$A_1$.prov\_patients\_\_admit\_\_info\_subject\_\_id} & \makecell{$A_2$.gender=M} & 0.57 & 0.64 & 0.6 & HISPANIC\\
\hline
\makecell{$A_1$: PT \\ $A_2$: admissions \\ $A_3$: patients\_admit\_info \\ $A_4$: patients} & \makecell{$A_2$.hadm\_id=$A_1$.prov\_patients\_\_admit\_\_info\_hadm\_\_id\\$A_2$.hadm\_id=$A_3$.hadm\_id\\$A_4$.subject\_id=$A_3$.subject\_id} & \makecell{$A_4$.gender=M} & 0.57 & 0.64 & 0.6 & HISPANIC\\
\hline
\makecell{$A_1$: PT \\ $A_2$: admissions \\ $A_3$: patients\_admit\_info \\ $A_4$: patients} & \makecell{$A_2$.hadm\_id=$A_1$.prov\_patients\_\_admit\_\_info\_hadm\_\_id\\$A_2$.hadm\_id=$A_3$.hadm\_id\\$A_4$.subject\_id=$A_3$.subject\_id} & \makecell{$A_1$.prov\_patients\_\_admit\_\_info\_age $>19.68$ \\ $A_4$.expire\_flag=0} & 0.62 & 0.58 & 0.6 & HISPANIC\\
\hline
\makecell{$A_1$: PT \\ $A_2$: patients \\ $A_3$: icustays \\ $A_4$: admissions} & \makecell{$A_2$.subject\_id=$A_1$.prov\_patients\_\_admit\_\_info\_subject\_\_id\\$A_2$.subject\_id=$A_3$.subject\_id\\$A_4$.hadm\_id=$A_3$.hadm\_id} & \makecell{$A_4$.hospital\_stay\_length>5.0 \\ $A_2$.expire\_flag=0} & 0.61 & 0.56 & 0.59 & HISPANIC\\
\hline
\makecell{$A_1$: PT \\ $A_2$: patients \\ $A_3$: procedures \\ $A_4$: admissions} & \makecell{$A_2$.subject\_id=$A_1$.prov\_patients\_\_admit\_\_info\_subject\_\_id\\$A_2$.subject\_id=$A_3$.subject\_id\\$A_4$.hadm\_id=$A_3$.hadm\_id} & \makecell{$A_4$.marital\_status=MARRIED} & 0.56 & 0.56 & 0.56 & ASIAN\\
\hline
\makecell{$A_1$: PT \\ $A_2$: patients \\ $A_3$: admissions} & \makecell{$A_2$.subject\_id=$A_1$.prov\_patients\_\_admit\_\_info\_subject\_\_id\\$A_3$.hadm\_id=$A_1$.prov\_patients\_\_admit\_\_info\_hadm\_\_id} & \makecell{$A_3$.marital\_status=MARRIED} & 0.56 & 0.56 & 0.56 & ASIAN\\
\hline
\makecell{$A_1$: PT \\ $A_2$: patients \\ $A_3$: patients\_admit\_info \\ $A_4$: admissions} & \makecell{$A_2$.subject\_id=$A_1$.prov\_patients\_\_admit\_\_info\_subject\_\_id\\$A_2$.subject\_id=$A_3$.subject\_id\\$A_4$.hadm\_id=$A_3$.hadm\_id} & \makecell{$A_4$.hospital\_stay\_length>5.0 \\ $A_2$.gender=M} & 0.59 & 0.52 & 0.55 & HISPANIC\\
\hline
\end{tabular}
}
\label{tab:tops-mimic5}
\caption{$\query_{mimic5}$ Top-20 patterns}
\end{figure*}


\end{appendix}

\end{document}
\endinput